\documentclass{amsart}
\usepackage{pgfplots}
\usepackage{amsmath}
\pgfplotsset{compat=1.18}
\usepackage{sansmath} 
\sansmath
\usepackage{graphicx}
\usepackage{slashed}
\usepackage{enumerate}
\IfFileExists{dsfont.sty}{\usepackage{dsfont}}{\newcommand{\mathds}{\mathbb}}
\usepackage{mathtools,cancel}
\usepackage[bottom]{footmisc}
\usepackage[scr=rsfso]{mathalfa}
\usepackage{tikz-cd}
\usepackage{xurl}
\usepackage[hidelinks]{hyperref}
\hypersetup{pdftitle={A large data result for vacuum Einstein's equations},pdfauthor={Puskar Mondal}}
{

}
\usepackage{amsfonts}
\usepackage{psfrag}
\usepackage{color}
\usepackage{mathtools}
\usepackage{pgf,tikz}
\usepackage{mathrsfs}
\usetikzlibrary{arrows}
\usepackage{fancyhdr}
\usepackage{amssymb}
\usepackage{slashed}
\usepackage{enumitem}

\usepackage[T1]{fontenc}
\IfFileExists{sfrm1000.pfb}{}{\IfFileExists{lmodern.sty}{\usepackage{lmodern}}{}}
\usepackage[utf8]{inputenc}
\newtheorem*{theorem-non}{Theorem}

\newtheorem{definition}{Definition}
\newtheorem{remark}{Remark}
\newtheorem{proposition}{Proposition}[section]

\newtheorem{theorem}{Theorem}[section]
\newtheorem{corollary}{Corollary}[section]
\newtheorem{lemma}{Lemma}[section]
\newtheorem{nono-theorem}{Theorem}
\newtheorem{conjecture}{Conjecture}[section]

\newtheorem{claim}{Claim}
\newcommand{\be}{\begin{equation}}
\newcommand{\ee}{\end{equation}}

\newcommand{\bm}{\begin{align}*}
\newcommand{\enm}{\end{align}*}

\newcommand{\bespeq}{\begin{equation}\begin{split}}
\newcommand{\espeq}{\end{split}\end{equation}}

\def\i {\infty}

\renewcommand{\div}{\mbox{div }}

\newcommand{\tr}{\mbox{tr}}

\newcommand\restri[2]{{
		\left.\kern-\nulldelimiterspace 
		#1 
		\right|_{#2} 
}}

\definecolor{ffqqqq}{rgb}{1.,0.,0.}
\definecolor{uuuuuu}{rgb}{0.26666666666666666,0.26666666666666666,0.26666666666666666}

\makeatletter
\def\ps@pprintTitle{%
  \let\@oddhead\@empty
  \let\@evenhead\@empty
  \let\@oddfoot\@empty
  \let\@evenfoot\@oddfoot
}
\def\@author#1{\g@addto@macro\elsauthors{\normalsize%
    \def\baselinestretch{1}%
    \upshape\authorsep#1\unskip\textsuperscript{%
      \ifx\@fnmark\@empty\else\unskip\sep\@fnmark\let\sep=,\fi
      \ifx\@corref\@empty\else\unskip\sep\@corref\let\sep=,\fi
      }%
    \def\authorsep{\unskip,\space}%
    \global\let\@fnmark\@empty
    \global\let\@corref\@empty  
    \global\let\sep\@empty}%
    \@eadauthor={#1}
}
\makeatother
\begin{document}

\title{A large data result for vacuum Einstein's equations}
\author{Puskar Mondal}
\thanks{Center of Mathematical Sciences and Applications and Department of Mathematics, Harvard University, Cambridge, MA 02138, USA}
\thanks{Present address: Beijing Institute of Mathematical Sciences and Applications, Beijing 101408, China}
\thanks{Corresponding author: \href{mailto:pushkarmondal@gmail.com}{\texttt{pushkarmondal@gmail.com}}. ORCID: \href{https://orcid.org/0000-0002-6326-835X}{0000-0002-6326-835X}}
\thanks{\textit{Acknowledgments.} I thank Professor S.-T. Yau for many discussions on scalar curvature geometry and Professor V. Moncrief for numerous discussions on the Einstein flow. I thank the reviewer for the meticulous review, which substantially improved the quality of the article}

\begin{abstract}
We study the vacuum Einstein equations with positive cosmological constant
on globally hyperbolic spacetimes of the form $M\times\mathbb R$, where
$M$ is a closed three-manifold with nonpositive Yamabe invariant.
We prove future global well-posedness in constant-mean-curvature
transported spatial coordinates for an open set of expanding initial data.
We construct exact solutions of the constraints for which the higher
Sobolev norms of the renormalized shear and trace-free Ricci curvature
become arbitrarily large as the initial slice varies, while the shear
remains small at lower order. The proof uses the damping of the shear
and the decay of its coupling to spatial curvature as the spacetime
expands. The resulting spacetimes are future causally geodesically
complete. Their renormalized spatial metrics converge smoothly to a metric
of constant negative scalar curvature, and unit balls have a uniform
positive lower volume bound, excluding collapse of graph-manifold regions.
{\color{blue}For these solutions, we verify Ringstr\"om's conjectured
asymptotic causal inaccessibility of spatial topology.} An analogous result
holds on the positive-scalar-curvature branch under an initial spectral
gap assumption for the lapse operator.
\end{abstract}
\maketitle

\setcounter{tocdepth}{2}
{\hypersetup{linkcolor=black}
\small
\tableofcontents
}

\section{Introduction}
\noindent We will study the vacuum Einstein equations including a positive cosmological constant $\Lambda$. We are given an $(n+1)$-dimensional ($n=3$ in this article) $C^{\infty}$ globally hyperbolic connected  Lorentzian manifold $(\widetilde{M},\widehat{g})$ with signature $-+++$. Einstein's equations on $(\widetilde{M},\widehat{g})$ read 
\begin{align}
\label{eq:first}
 \text{Ric}[\widehat{g}]-\frac{1}{2}R[\widehat{g}]\widehat{g}+\Lambda \widehat{g}=0.   
\end{align}
Global hyperbolicity implies the existence of a Cauchy hypersurface and in particular $\widetilde{M}= M\times \mathbb{R}$ with $M$ being diffeomorphic to a Cauchy hypersurface. The main theorem concerns closed connected oriented three-manifolds with nonpositive Yamabe invariant $\sigma(M)$. Such a manifold admits no metric of strictly positive scalar curvature. We use the terminology of \cite{MM}. Here $M$ is of \emph{negative Yamabe type} if every metric has negative scalar curvature somewhere. A manifold of zero Yamabe type also has $\sigma(M)=0$ but admits a scalar-flat metric. Theorem~\ref{main} includes both types on the strict expanding branch. A closed connected oriented 3-manifold $M$ is of negative Yamabe type if and only if it lies in one of the
following three mutually exclusive subsets:
\begin{itemize}
\item[(1)] $M$ is hyperbolizable (i.e.,
admitting a hyperbolic metric),
\item[(2)] $M$ is a non-hyperbolizable $K(\pi,1)$ manifold of non-flat type (the six flat $K(\pi,1)$ manifolds are of zero Yamabe type), here $K(\pi,1)$ is of Eilenberg-MacLane type and by definition $\pi_{1}(M)=\pi$ and $\pi_{i}=0,~i\geq 2$; the universal cover of $K(\pi,1)$ is contractible and known to be diffeomorphic to $\mathbb{R}^{3}$ \cite{porti2008geometrization}, and
\item[(3)] $M$ has a nontrivial connected sum decomposition (i.e., $M$ is composite)
\begin{align}
\label{eq:negative}
M\approx (\mathbf{S}^{3}/\Gamma)_{1}\#\cdots\#(\mathbf{S}^{3}/\Gamma)_{k}\#(\mathbf{S}^{2}\times \mathbf{S}^{1})_{1}\#\cdots\#(\mathbf{S}^{2}\times \mathbf{S}^{1})_{l}\#\mathbf{K}(\pi,1)_{1}\#\cdots\#\mathbf{K}(\pi,1)_{m}  
\end{align}
in which at least one factor is a $K(\pi, 1)$ manifold, i.e., $m\geq 1$ \cite{schoenyau}. In this case the $K(\pi,1)$
factor may be either of flat type or hyperbolizable or non-hyperbolizable non-flat type. 
\end{itemize}
The six orientable flat manifolds
comprise by themselves the subset of zero Yamabe type. These admit metrics
having vanishing scalar curvature (the flat ones) but no metrics having strictly
positive scalar curvature. Here $\Gamma\subset SO(4)$ acts freely and properly discontinuously on $\mathbb{S}^{3}$. Also note that $M\#\mathbb{S}^{3}\approx M$ for any 3-manifold $M$. It is known that every prime $K(\pi,1)$ manifold is decomposable
into a (possibly trivial but always finite) collection of (complete, finite volume)
hyperbolic and graph manifold \footnote{A compact oriented $3-$manifold $M$ is a graph manifold if $\exists$ a finite disjoint collection of embedded $2-$tori $\{T_{j}\}\subset M$ such that each connected component of $M-\cup_{j}T_{j}$ is the total space of a circle bundle over a surface.} components. 

\noindent Our goal is to study the long-term behavior of Einstein's equations \eqref{eq:first} on spacetimes $\widetilde{M}=M\times \mathbb{R}$ with $M$ a closed oriented three-manifold with nonpositive Yamabe invariant, including the composite examples in \eqref{eq:negative}. An issue closely related to Penrose's weak cosmic censorship \cite{penrose1999question} in the cosmological context (compact spatial topology where asymptotically flat null infinity is unavailable) may be phrased as follows:

\begin{conjecture}[No-Naked Singularity Conjecture for Cosmological Spacetimes]
\label{conjecture1}
Fix a closed, connected, oriented three-manifold $M$, a cosmological constant $\Lambda\in\mathbb R$ and a matter model for which the smooth Einstein--matter Cauchy problem admits a unique maximal globally hyperbolic development up to isometry. The time-oriented spacetime $(\widetilde M,\widehat g)$ satisfies the matter equations and the Einstein field equations
\[
\operatorname{Ric}_{\widehat{g}} - \tfrac{1}{2} R_{\widehat{g}} \widehat{g} + \Lambda \widehat{g} =\mathcal{T},
\]
where the stress-energy tensor $\mathcal T$ of the model satisfies the dominant energy condition. Equip the smooth constraint data on $M$, including the matter data, with the relative $C^\infty$ topology.

\noindent For an open and dense subset of these data the maximal globally hyperbolic development has no future naked singularity, with the following convention for visibility. A future-inextendible causal geodesic from $M$ of finite affine length with no regular endpoint in any local $C^2$ spacetime extension is visible if there is a future-directed timelike observer $\gamma$ from $M$ of infinite future proper time such that $J^-(\gamma)$ contains points of that geodesic with affine parameter arbitrarily close to its future endpoint value.
\end{conjecture}

\begin{remark}
 Notice that it is very unlikely to obtain global control of a solution with arbitrary initial data in a fully generic framework since Einstein's equations are quasi-linear \textit{hyperbolic} equations.
\end{remark}

\noindent In this article, we focus on the vacuum case with a positive cosmological constant, i.e., $\mathcal{T}=0$ and $\Lambda>0$.
\subsection{Background Literature}

\noindent The future stability problem for cosmological solutions to the Einstein equations has been extensively studied in the small-data regime, particularly when the spatial manifold $M$ lies in the negative Yamabe class and admits a hyperbolic metric.

\noindent The foundational result in this context is due to Andersson and Moncrief \cite{andersson2004future}, who established the nonlinear future stability of the Milne universe,
\[
\widehat{g} = -dt^2 + \frac{t^2}{9}\eta, \quad \text{Ric}[\eta] = -\frac{2}{9}\eta,
\]
on manifolds of the form $\mathbb{H}^3/\Gamma \times \mathbb{R}$, where $\Gamma \subset \mathrm{SO}^{+}(1,3)$ is cocompact and acts freely and properly discontinuously on $\mathbb{H}^3$. Their result applies to vacuum solutions with spatial topology restricted to eliminate moduli of flat perturbations. The Milne model arises as a quotient of the interior of the future light cone in Minkowski space, and the proof uses energy estimates for small perturbations of the hyperbolic data.

\noindent This analysis was extended to higher-dimensional spacetimes ($n+1\geq4$) in \cite{andersson2011einstein}, assuming a stable negative Einstein metric with integrable deformation space. The proof uses decay of an energy functional and a spectral gap on the complement of the infinitesimal Einstein deformations.

\noindent Subsequent work has incorporated various matter models. Andersson and Fajman \cite{andersson2020nonlinear} established the nonlinear future stability of the Milne model within the Einstein--massive Vlasov system. In a higher-dimensional Kaluza--Klein setting, Branding, Fajman, and Kr\"oncke \cite{KK} proved future stability under perturbations invariant in the torus directions for $(\mathbb{H}^3/\Gamma \times \mathbb{T}^d) \times \mathbb{R}$ with Kaluza--Klein reduction, reducing the system to a coupled Einstein--Maxwell--wave map system.

\noindent For vacuum cosmologies with $\Lambda>0$, see \cite{fajman2020stable,fajman2018cmc,mondal2019attractors}. Small-data stability results for the coupled Einstein--Euler system include \cite{oliynyk2016future,hadvzic2015global,lefloch2021nonlinear,mondal2024nonlinear}. Oliynyk also studies the Euler equations on a prescribed FLRW background \cite{oliynyk2021future}. For the Einstein--Klein--Gordon system on perturbations of the Milne spacetime, see \cite{wang,fajman2021attractors}.

\noindent In contrast, the asymptotically flat case has a distinct lineage, beginning with the monumental work of Christodoulou and Klainerman \cite{christodoulou1993global}, who proved the global nonlinear stability of Minkowski spacetime. This framework has since been extended in various directions \cite{lindblad2010global, lefloch2016global, fajman2021stability, taylor, mondal2024global}, focusing on vacuum and matter models under asymptotically flat conditions.

\subsection{Main Result}

\noindent We consider the Einstein vacuum equations with positive cosmological constant $\Lambda > 0$ on $(3+1)$-dimensional globally hyperbolic spacetimes $\widetilde{M} \cong M \times \mathbb{R}$, where $M$ is a closed, connected, oriented $3$-manifold with $\sigma(M)\leq0$. The topology may, for example, be that of a graph manifold or of a nontrivial connected sum as in \eqref{eq:negative}. The theorem applies even when $M$ admits no negative Einstein metric.

\noindent We allow arbitrarily large top-order norms subject to the initial-time and weighted shear bounds below. We fix the constant mean curvature (CMC) transported spatial gauge. The time parameter is taken to be the mean extrinsic curvature $\tau := \operatorname{tr}_{g}(k)$, and we consider solutions in the expanding direction, i.e., $\tau < 0$ and increasing. 

\noindent To capture the asymptotic geometry, we introduce a natural geometric rescaling by the function $\varphi^2 := \tau^2 - 3\Lambda > 0$ (see Claim~\ref{important}) and analyze the rescaled evolution system \eqref{eq:cmc1}--\eqref{eq:cmc2}. 

\noindent To quantify the deviation of the evolving geometry from hyperbolicity, we define the \emph{obstruction tensor}
\[
\mathfrak{T}[g] := \operatorname{Ric}[g] - \tfrac{1}{3} R(g) g,
\]
which measures the failure of $g$ to be Einstein. In three dimensions the contracted Bianchi identity and the curvature decomposition imply that $\mathfrak T[g]=0$ precisely when $g$ has constant sectional curvature. At the limiting normalization $R=-2/3$ the sectional curvature is then $-1/9$, so Mostow rigidity determines this metric up to isometry whenever $M$ is hyperbolizable.

\noindent We impose the transported spatial gauge (see Definition~\ref{spatial}), in which the pair $(\Sigma, \mathfrak{T})$ satisfies a manifestly hyperbolic evolution system. This is coupled with an elliptic equation for the lapse function determined by the CMC condition. In this formulation, we establish the following global convergence theorem for the Einstein flow in the large data regime on slices with nonpositive Yamabe invariant.

\begin{theorem}[Global Well-posedness: $\Lambda > 0$, $\sigma(M) \leq 0$]
\label{main}
Let $M$ be a closed connected oriented three-manifold with $\sigma(M)\leq0$, and fix $\Lambda>0$. 

\noindent Fix a smooth background metric $\xi_0$ on $M$ and $C>1$.
\textcolor{blue}{Fix $0<\alpha<1$. For all sufficiently large $a$, assume that the initial size $\mathcal I^0>0$ satisfies}
\[
\textcolor{blue}{(1+\mathcal I^0)e^{-a}\le e^{-\alpha a}.}
\]

\noindent Let $(g_0, \Sigma_0)$ be smooth normalized initial data for the expanding case $\tau_0<-\sqrt{3\Lambda}$, with the normalization of Section~\ref{mainsection}, satisfying the Einstein constraint equations at initial CMC time $T_0 = a>0$, written in CMC-transported spatial coordinates and satisfying:
\begin{align}
\label{eq:1}
C^{-1} \xi_0 \leq g_0 \leq C \xi_0,
\end{align}
\begin{align}
\label{eq:2}
\sum_{I = 0}^{3} \| \nabla^I \Sigma_0 \|_{L^2(M)} + \sum_{I = 0}^{2} \left( \| \nabla^I \mathfrak{T}[g_0] \|_{L^2(M)} + \| e^a \nabla^I \Sigma_0 \|_{L^2(M)} \right) \leq \mathcal{I}^0,
\end{align}
where $\mathfrak{T}[g_0]$ denotes the renormalized trace-free spatial Ricci curvature tensor of $g_0$.

\noindent Then the Einstein-$\Lambda$ evolution equations admit a unique classical solution
\[
T \mapsto (g(T), \Sigma(T)) \in \mathcal{C}^\infty([T_0, \infty) \times M)
\]
in CMC-transported spatial coordinates, satisfying the constraint equations at each slice $T$ and obeying the following uniform a priori estimates for all $T \in [T_0, \infty)$:
\begin{align}
\sum_{I = 0}^{3} \| \nabla^I \Sigma(T) \|_{L^2(M)} + \sum_{I = 0}^{2} \left( \| \nabla^I \mathfrak{T}[g(T)] \|_{L^2(M)} + \| e^T \nabla^I \Sigma(T) \|_{L^2(M)} \right) \leq C_1(1 + \mathcal{I}^0),
\end{align}
\begin{align}
C_2^{-1} g_0 \leq g(T) \leq C_2 g_0.
\end{align}
The constants $C_1,C_2>0$ depend only on $\xi_0$, $C$, $\Lambda$ and the fixed time normalization, and are independent of $a$, $\mathcal I^0$ and $T$. The developed spacetime is future causally geodesically complete.

\noindent Moreover, the solution $(g(T), \Sigma(T))$ converges in the $C^\infty$ topology, as $T \to \infty$, to $(g_\infty,0)$, where $g_\infty$ is a Riemannian metric on $M$ of scalar curvature $-2/3$. More precisely,
\[
\Sigma(T) \to 0 \quad \text{and} \quad g(T) \to g_\infty, \quad \text{as } T \to \infty,
\]
with convergence holding in all Sobolev norms. In particular, the spacetime $(\widetilde{M}, \widehat{g})$ admits a future-complete CMC foliation asymptotic to a constant negative scalar curvature slice.
\end{theorem}


\begin{remark}
\textcolor{blue}{For each fixed $0<\alpha<1$, the largest size allowed by the displayed hypothesis is $\mathcal I^0=e^{(1-\alpha)a}-1$. Section~\ref{data} constructs exact constraint data whose top-order norms are comparable to $e^{(1-\alpha)a}$. These data satisfy the theorem with $\alpha/2$ in place of $\alpha$, which absorbs the fixed constants in their norm bounds for sufficiently large $a$. Thus every fixed exponential growth rate less than one and greater than zero is attained. However, we do not claim that this is optimal.}
\end{remark}
\begin{remark}\label{rem:growth-sharpness}
For concrete examples within negative Yamabe type, one may take a closed hyperbolic three-manifold, or the graph manifold $M=S\times S^1$, where $S$ is a closed oriented surface of genus at least two. In the second case $\sigma(M)=0$, but $M$ is of negative Yamabe type in the sense specified above. Choose a background metric $\gamma$ with $R[\gamma]=-2/3$ and no conformal Killing fields as in Section~\ref{data}. Corollary~\ref{cor:final-CMC} gives exact nonzero-shear data on these manifolds with
\[
 \|g_{0,a}-\gamma\|_{H^4(\gamma)}
 \approx\|\mathfrak T[g_{0,a}]\|_{H^2(\gamma)}
 \approx\|\Sigma_{0,a}\|_{H^3(\gamma)}
 \approx e^{(1-\alpha)a},
\]
while the $H^3$ metric perturbation and $H^2$ shear are $O(e^{-\alpha a})$. Thus the admissible growth rates can approach one within negative Yamabe type itself.

\noindent The endpoint must be distinguished from this attainability statement. The bootstrap uses the smallness of $(1+\mathcal I^0)e^{-a}$; for example, its energy error contains
\[
 e^{-2a}(1+\mathcal I^0)^4
 \le e^{-2\alpha a}(1+\mathcal I^0)^2.
\]
For every fixed $\alpha>0$, the factor relative to the top-order energy tends to zero. At $\alpha=0$, the same argument requires $(1+\mathcal I^0)e^{-a}$ to be sufficiently small relative to the fixed background; a coefficient-one bound alone does not close this bootstrap. This is a limitation of the present estimates, not a proof that the endpoint conclusion fails on manifolds of negative Yamabe type. We do not assert endpoint optimality in that narrower class. It remains to determine whether the conclusion holds under the coefficient-one endpoint bound $(1+\mathcal I^0)e^{-a}\leq1$ on manifolds of negative Yamabe type. 
\end{remark}

\begin{remark}
The Yamabe flow evolves a metric within its conformal class. In the CMC Einstein evolution the term $(2\varphi/\tau)N\Sigma$ in \eqref{eq:metric_transport} changes the conformal class of the metric to first order whenever $\Sigma\ne0$. Both flows can have limits of constant scalar curvature.
\end{remark}

\noindent \subsection{Heuristics on the genericity of the initial data}\label{genericity}
The negative Yamabe condition allows metrics with Killing fields. For example,
if $S$ is a closed surface of genus at least two, the product metric on
$S\times S^1$ admits the circle action. For the construction in
Section~\ref{data} we choose a conformal class without conformal Killing
fields. Such a choice is possible in any open set of negative Yamabe classes:
generic smooth metrics in dimension three have no local conformal Killing
fields \cite{LiimatainenSalo}, and changing the conformal representative
preserves the conformal Killing equation for vector fields. We take the
constant-scalar-curvature representative of the chosen class and fix its scale
by $R[\gamma]=-2/3$.

\medskip

\noindent Although the topology of \(M\) leaves substantial freedom in the choice of initial data, the hypotheses of Theorem~\ref{main} still impose a definite restriction on how such data may concentrate. To describe this heuristically, we associate to the pair \((\Sigma,\mathfrak T)\) amplitude--wavelength parameters
\[
(A_{\Sigma},L_{\Sigma}),\qquad (A_{\mathfrak T},L_{\mathfrak T}),
\]
where \(\Sigma\) is the trace-free part of the second fundamental form and \(\mathfrak T\) is the renormalized trace-free Ricci tensor appearing in Theorem~\ref{main}. The quantities \(L_{\Sigma}\) and \(L_{\mathfrak T}\) should be interpreted as characteristic spatial scales, while \(A_{\Sigma}\) and \(A_{\mathfrak T}\) represent dimensionally normalized amplitudes. This discussion is purely heuristic and is intended only to indicate the concentration regimes compatible with the norm bounds of the theorem.

\noindent In three spatial dimensions, a field of the form $A f(x/L)$ with a fixed smooth profile $f$ supported in a coordinate ball of radius one has squared \(H^{k}\)-size of order
\[
|A|^{2}\sum_{j=0}^{k} L^{3-2j}.
\]
Consequently, the control required in Theorem~\ref{main} yields, at the highest order for \(\Sigma\),
\begin{equation}\label{eq:highercontrol_refined}
L_{\Sigma}^{-3}|A_{\Sigma}|^{2}
+L_{\Sigma}^{-1}|A_{\Sigma}|^{2}
+L_{\Sigma}|A_{\Sigma}|^{2}
+L_{\Sigma}^{3}|A_{\Sigma}|^{2}
\lesssim (\mathcal I^{0})^{2},
\end{equation}
whereas the lower-order weighted bounds for \(\Sigma\) give
\begin{equation}\label{eq:lowercontrol_refined}
L_{\Sigma}^{-1}|A_{\Sigma}|^{2}
+L_{\Sigma}|A_{\Sigma}|^{2}
+L_{\Sigma}^{3}|A_{\Sigma}|^{2}
\lesssim e^{-2a}(\mathcal I^{0})^{2}.
\end{equation}
Similarly, for \(\mathfrak T\) one obtains
\begin{equation}\label{eq:ampl_refined}
L_{\mathfrak T}^{-1}|A_{\mathfrak T}|^{2}
+L_{\mathfrak T}|A_{\mathfrak T}|^{2}
+L_{\mathfrak T}^{3}|A_{\mathfrak T}|^{2}
\lesssim (\mathcal I^{0})^{2}.
\end{equation}
Moreover, the relative smallness assumption in Theorem~\ref{main} allows \(\mathcal I^{0}\) to grow with \(a\), subject only to
\[
(1+\mathcal I^0)e^{-a}\le e^{-\alpha a}.
\]

\noindent At the largest allowed size, $\mathcal I^0=e^{(1-\alpha)a}-1$. For a profile of spatial frequency $\nu=L^{-1}\ge1$, the shear and curvature bounds give
\[
 |A_\Sigma|\lesssim\min\{e^{-\alpha a}\nu^{-1/2},e^{(1-\alpha)a}\nu^{-3/2}\},
 \qquad |A_{\mathfrak T}|\lesssim e^{(1-\alpha)a}\nu^{-1/2}.
\]
The two restrictions on the shear agree at $\nu=e^a$. At this frequency the bounds on $|A_\Sigma|$ and $|A_{\mathfrak T}|$ have orders $e^{-(1/2+\alpha)a}$ and $e^{(1/2-\alpha)a}$, respectively. These are scaling bounds for a fixed localized profile. The construction of actual constrained data in Section~\ref{data} uses transverse-traceless tensors. Figure~\ref{fig:blue} records the corresponding initial norm bound.

\noindent The Gauss--Codazzi equations express the electric and magnetic parts
of the Weyl tensor in the normalized variables. Write $E$ and $B$ for these
covariant two-tensors relative to the future unit normal. In the
physical spatial metric $\widetilde g=\varphi^{-2}g$ one has
\[
 E=\bigl(\operatorname{Ric}[\widetilde g]+\tau\widetilde k
                   -\widetilde k_i{}^{\ell}\widetilde k_{\ell j}\bigr)^{TF},
 \qquad B=\pm\operatorname{curl}_{\widetilde g}\widetilde k,
\]
where the sign of $B$ depends on the orientation convention and $TF$ denotes the trace-free part. Substituting
$\widetilde k=(\tau/3)\widetilde g+\varphi^{-1}\Sigma$ gives
\[
 E_{ij}=\mathfrak T_{ij}+\frac{\tau}{3\varphi}\Sigma_{ij}
          -(\Sigma_i{}^{\ell}\Sigma_{\ell j})^{TF},
 \qquad B_{ij}=\pm (\operatorname{curl}_{g}\Sigma)_{ij}.
\]
Here indices on $\Sigma$ are raised using $g$, and we measure $E$ and $B$
in the renormalized spatial metric. On the initial slice, the initial estimates,
$H^2\hookrightarrow L^\infty$, and $|\tau(a)/\varphi(a)|\lesssim e^a$ give
\[
\begin{aligned}
 \|E\|_{L^\infty(g_0)}+\|B\|_{L^\infty(g_0)}
 &\lesssim \|\mathfrak T[g_0]\|_{H^2(g_0)}
       +e^a\|\Sigma_0\|_{H^2(g_0)}
       +\|\Sigma_0\|_{H^2(g_0)}^2
       +\|\Sigma_0\|_{H^3(g_0)}\\
 &\lesssim \mathcal I^0+e^{-2a}(\mathcal I^0)^2
 \lesssim \mathcal I^0.
\end{aligned}
\]
Section~\ref{data} gives two-sided bounds for $\|g_0-\gamma\|_{H^4(\gamma)}$, $\|\mathfrak T[g_0]\|_{H^2}$ and $\|\Sigma_0\|_{H^3}$. The concentration examples
above illustrate the amplitudes allowed by the normalized inequalities.

\subsection{Heuristics for the construction of the initial data}

Recall that in CMC gauge the re-scaled vacuum constraint equations reduce to
\begin{align}
\label{eq:cons1}
R(g)+\frac{n-1}{n} &= |\Sigma|_{g}^{2},\\
\label{eq:cons2}
\nabla^{j}\Sigma_{ij} &= 0,
\end{align}
where \(\Sigma\) is the trace-free part of the second fundamental form. Thus \(\Sigma\) is, by definition, trace-free, and \eqref{eq:cons2} is the momentum constraint.

\noindent For the class of data considered in Theorem~\ref{main}, one seeks
\[
\|\Sigma\|_{H^{2}}\lesssim e^{-a}\mathcal I^{0},
\qquad
\|\Sigma\|_{H^{3}}\lesssim \mathcal I^{0},
\qquad
\|\mathfrak T[g]\|_{H^{2}}\lesssim \mathcal I^{0},
\]
with
\[
(1+\mathcal I^0)e^{-a}\le e^{-\alpha a}.
\]
\noindent In the construction below, we prescribe the top-order norms to be comparable to $e^{(1-\alpha)a}$, where $0<\alpha<1$. The resulting data satisfy Theorem~\ref{main} with $\alpha/2$ in place of $\alpha$, as verified at the end of Section~\ref{data}. The estimates immediately below apply to arbitrary data satisfying the theorem as stated.

\noindent On the initial slice write $H^j(g_0)$ for the norm formed with the Levi-Civita connection and volume measure of $g_0$. By \eqref{eq:2},
$\|\Sigma_0\|_{H^2(g_0)}\le e^{-a}\mathcal I^0$ and
$\|\Sigma_0\|_{H^3(g_0)}\le\mathcal I^0$. The Hamiltonian constraint gives
\[
 \|R(g_0)+2/3\|_{H^2(g_0)}
 \lesssim \|\Sigma_0\|_{H^2(g_0)}^2
 \le e^{-2a}(\mathcal I^0)^2\le e^{-2\alpha a}.
\]
This is the order of smallness used to solve the conformal equation in $H^4$. At order three we use the same volume-weighted inequalities
\eqref{eq:sobolev-L4} and \eqref{eq:sobolevL4}. All norms in the following
calculation use $g_0$. Since $\nabla g_0=0$, the Hamiltonian constraint
and Leibniz' rule give
\[
\begin{aligned}
 \|R(g_0)+2/3\|_{H^3}
 &=\bigl\||\Sigma_0|_{g_0}^2\bigr\|_{L^2}
   +\sum_{j=1}^3\bigl\|\nabla^j(|\Sigma_0|_{g_0}^2)\bigr\|_{L^2}\\
 &\lesssim \|\Sigma_0\|_{L^\infty}\|\Sigma_0\|_{L^2}
    +\|\Sigma_0\|_{L^\infty}\sum_{j=1}^3\|\nabla^j\Sigma_0\|_{L^2}\\
 &\quad +\|\nabla\Sigma_0\|_{L^4}^2
    +\|\nabla\Sigma_0\|_{L^4}\|\nabla^2\Sigma_0\|_{L^4}.
\end{aligned}
\]
For the third derivative the mixed term may also be bounded by
$\|\nabla\Sigma_0\|_{L^6}\|\nabla^2\Sigma_0\|_{L^3}$. The displayed
$L^4$ bounds allow us to keep the same two Sobolev inequalities throughout.
Write $V_0=\mathrm{Vol}_{g_0}(M)$ and $J_0=\max(1,I(M,g_0))$. For
$j=0,1,2$, \eqref{eq:sobolev-L4} gives
\[
 \|\nabla^j\Sigma_0\|_{L^4}
 \le C J_0^{1/2}V_0^{1/12}
 \bigl(\|\nabla^{j+1}\Sigma_0\|_{L^2}
             +V_0^{-1/3}\|\nabla^j\Sigma_0\|_{L^2}\bigr).
\]
Using this for $j=0,1$ in \eqref{eq:sobolevL4}, we obtain
\[
 \|\Sigma_0\|_{L^\infty}
 \le C J_0 V_0^{1/6}
 \bigl(\|\nabla^2\Sigma_0\|_{L^2}
    +2V_0^{-1/3}\|\nabla\Sigma_0\|_{L^2}
    +V_0^{-2/3}\|\Sigma_0\|_{L^2}\bigr).
\]
Substitution bounds the full $H^3$ norm by
\[
 C J_0 V_0^{1/6}(1+V_0^{-1/3})^2
     \|\Sigma_0\|_{H^2}\|\Sigma_0\|_{H^3}.
\]
By \eqref{eq:1}, $V_0$ is bounded above and below and $J_0$ is bounded
above in terms of the fixed background and the metric comparison constant.
Consequently,
\[
 \|R(g_0)+2/3\|_{H^3(g_0)}
 \lesssim e^{-a}(\mathcal I^0)^2\le e^{-\alpha a}\mathcal I^0.
\]
The last bound is small relative to $\mathcal I^0$. At the largest admissible size it is of order $e^{(1-2\alpha)a}$, so for $0<\alpha<1/2$ it gives no absolute smallness in $H^3$.

\noindent We seek a metric with $\|\mathfrak T[g]\|_{H^2}\approx e^{(1-\alpha)a}$ and a transverse-traceless tensor with $\|\Sigma\|_{H^2}\approx e^{-\alpha a}$ and $\|\Sigma\|_{H^3}\approx e^{(1-\alpha)a}$.

\noindent The standard conformal method provides the natural framework for solving
\eqref{eq:cons1}--\eqref{eq:cons2}, but in the present setting the delicate point is not merely the solvability of the constraints. Rather, one must solve them in such a way that the conformal correction needed to enforce \eqref{eq:cons1} remains sufficiently small at low order so that the large \(H^{2}\)-size of the renormalized trace-free Ricci tensor is not destroyed.

\noindent The construction therefore proceeds heuristically in three steps. First, one builds a seed metric \(\widehat g\) for which the renormalized trace-free Ricci tensor is large in \(H^{2}\), whereas the scalar-curvature defect \(R(\widehat g)+\frac{n-1}{n}\) is very small in \(H^{2}\) and controlled in \(H^{3}\). Second, one constructs a \(\widehat g\)-transverse-traceless tensor \(\widehat\Sigma\) with the quantitative bounds required in Theorem~\ref{main}; in particular, this resolves the momentum constraint at the seed level. Third, one uses the conformal method to solve the Hamiltonian constraint exactly, producing a nearby pair \((g,\Sigma)\) satisfying the full system \eqref{eq:cons1}-\eqref{eq:cons2}. One then shows that the conformal correction is sufficiently small in the relevant norms, so that the large \(H^{2}\)-size of the renormalized trace-free Ricci tensor is retained.

\noindent This is the geometric mechanism underlying the initial-data construction. The following discussion here is only heuristic; the rigorous implementation is carried out in Section~\ref{data}.

\noindent Let $M$ be a closed connected $3$-manifold with $\sigma(M)\leq0$. In a conformal class of negative Yamabe constant, the solution of the Yamabe problem, followed by constant rescaling, gives a smooth background metric $\gamma$ satisfying
\[
R(\gamma)=-\frac{n-1}{n}.
\]
Since throughout this discussion $n=3$, one may equivalently write $R(\gamma)=-\frac23$; for notational clarity, however, we retain the general $n$-dimensional form in the formulas below.

\noindent We consider a perturbation
\[
g_{1}=\gamma+\epsilon h,\qquad 0<\epsilon\ll 1,
\]
where $h$ is chosen to be transverse-traceless with respect to $\gamma$,
\[
\div_{\gamma}h=0,
\qquad
\tr_{\gamma}h=0,
\]
and normalized by
\[
\|h\|_{L^{2}(M,\gamma)}=1.
\]
We take $h$ to be a normalized eigentensor of the compressed operator $A_\gamma$ of Proposition~\ref{prop:YorkTTperturb}, with eigenvalue $\mu\gg1$. Its spatial frequency is of order $\sqrt\mu$. Standard elliptic theory then yields the scale of estimates
\[
\|h\|_{H^{2}(M,\gamma)}\sim \mu,
\]
with constants depending only on the background geometry. If the background happens to be Einstein, one may take $h$ to be a high-frequency transverse-traceless eigentensor of the Lichnerowicz Laplacian; however, that special identification is not needed for the argument below.

\noindent We write
\[
\mathfrak{T}[g]:=\text{Ric}[g]-\frac1n R[g]\,g
\]
for the trace-free Ricci tensor. In transverse-traceless directions the principal part of its linearization at $\gamma$ is $\tfrac12\nabla_\gamma^*\nabla_\gamma h$. We have the expansion
\[
\mathfrak{T}[g_{1}]
=
\mathfrak{T}[\gamma]
+\epsilon\,D\mathfrak T_\gamma[h]
+\epsilon^{2}\mathcal{Q}_{\gamma}(h,\nabla h,\nabla^{2}h),
\]
where the products in $\mathcal{Q}_{\gamma}$ contain two factors of $h$ or its derivatives, with smooth inverse-metric coefficients depending on $\epsilon h$. The elliptic estimate in Proposition~\ref{prop:YorkTTperturb} gives
\[
\|D\mathfrak T_\gamma[h]\|_{H^{2}(M,\gamma)}\sim \mu^{2},
\]
whereas the quadratic remainder obeys the scale
\[
\|\mathcal{Q}_{\gamma}(h,\nabla h,\nabla^{2}h)\|_{H^{2}(M,\gamma)}\lesssim \mu^{3}.
\]
Thus
\[
\|\mathfrak{T}[g_{1}]\|_{H^{2}}
\gtrsim
\epsilon\mu^{2}
-
\epsilon^{2}\mu^{3}
-
O(1).
\]
This already exhibits the mechanism we exploit: the linear contribution to the trace-free Ricci tensor has size of order $\epsilon\mu^2$ in $H^2$, while the nonlinear error is smaller provided $\epsilon\mu$ remains sufficiently small.

\noindent The scalar curvature behaves differently. Since $h$ is transverse-traceless, its first variation is only zeroth order:
\[
D R_{\gamma}[h]
=
-\langle \text{Ric}^{\circ}[\gamma],h\rangle_{\gamma},
\]
where $\text{Ric}^{\circ}[\gamma]=\mathfrak{T}[\gamma]$. Consequently, the scalar curvature expansion has the schematic form
\[
R[g_{1}]
=
R[\gamma]
\mathbin{-}
\epsilon\,\langle \text{Ric}^{\circ}[\gamma],h\rangle_{\gamma}
+
\epsilon^{2}\mathcal{R}_{\gamma}(h,\nabla h,\nabla^{2}h),
\]
with the same quadratic vanishing and smooth inverse-metric coefficients in $\mathcal{R}_{\gamma}$. Since the linear term involves no second derivatives of $h$, it is smaller by one power of the eigenvalue $\mu$ than the linear term in the trace-free Ricci tensor. Correspondingly,
\[
\Big\|R[g_{1}]+\frac{n-1}{n}\Big\|_{H^{2}}
\lesssim
\epsilon\mu+\epsilon^{2}\mu^{3}.
\]
The two expansions therefore separate the relevant scales:
\[
\|\mathfrak{T}[g_{1}]\|_{H^{2}}
\sim
\epsilon\mu^{2},
\qquad
\Big\|R[g_{1}]+\frac{n-1}{n}\Big\|_{H^{2}}
\lesssim
\epsilon\mu+\epsilon^{2}\mu^{3},
\]
up to lower-order terms. In particular, we seek parameters for which
\[
\epsilon\mu^{2}\gg1,
\qquad
\epsilon\mu\ll1,
\qquad
\epsilon^{2}\mu^{3}\ll1.
\]

\noindent Choose any sufficiently large eigenvalue $\mu\ge e^{2a}$ and set
\begin{equation}\label{eq:accum_intro}
 \epsilon=e^{(1-\alpha)a}\mu^{-2}.
\end{equation}
Then $\epsilon\mu^2=e^{(1-\alpha)a}$, $\epsilon\mu\le e^{-(1+\alpha)a}$ and $\epsilon^2\mu^3\le e^{-2\alpha a}$. For sufficiently large $a$, the trace-free curvature has $H^2$ norm comparable to $e^{(1-\alpha)a}$, while $\|g_1-\gamma\|_{H^3}\lesssim e^{-\alpha a}$. Thus $\|g_1-\gamma\|_{H^4}\approx e^{(1-\alpha)a}$ tends to infinity as $a\to\infty$.

\noindent We next solve the Hamiltonian constraint by a conformal deformation, which preserves the momentum constraint. Starting from the metric $g=g_{1}$ constructed above, together with a symmetric transverse-traceless tensor $\Sigma$ relative to $g$, we seek a positive conformal factor $\varphi$ and define
\[
g_{0}=\varphi^{\frac{4}{n-2}}g,
\qquad
\Sigma_{0}=\varphi^{-2}\Sigma.
\]
This gives the renormalized initial data $(g_0,\Sigma_0)$ on the chosen CMC slice. The physical data are recovered by the rescaling in Section~\ref{mainsection}.
The conformal covariance of the momentum constraint ensures that $\Sigma_{0}$ remains transverse-traceless with respect to $g_{0}$:
\[
\div_{g_{0}}\Sigma_{0}=0.
\]
We choose $\Sigma$ so that $\|\Sigma\|_{H^2(\gamma)}\approx e^{-\alpha a}$ and $\|\Sigma\|_{H^3(\gamma)}\approx e^{(1-\alpha)a}$. The Hamiltonian correction must preserve these scales together with $\|\mathfrak T[g]\|_{H^2(\gamma)}\approx e^{(1-\alpha)a}$. We estimate $\varphi-1$ in $H^4(\gamma)$; the fixed reference metric makes the constants independent of the large fourth derivatives of the seed metric.
The scalar curvature transforms according to
\begin{equation}\label{eq:scal_intro}
R[g_{0}]
=
\varphi^{-\frac{n+2}{n-2}}
\left(
-\frac{4(n-1)}{n-2}\Delta_{g}\varphi
+
R[g]\varphi
\right).
\end{equation}
Hence, the Hamiltonian constraint
\[
R[g_{0}]+\frac{n-1}{n}=|\Sigma_{0}|_{g_{0}}^{2}
\]
reduces to the usual Lichnerowicz equation for $\varphi$. For the purposes of the present discussion, the important point is simply that the conformal factor is chosen so as to remove the scalar-curvature defect left over from the first step.

\noindent What must still be verified is that this scalar-curvature correction does not significantly alter the large trace-free Ricci component already present in $g$. The transformation law for the trace-free Ricci tensor is
\[
\mathfrak{T}[g_{0}]
=
\mathfrak{T}[g]
{-2}\varphi^{-1}
\left(
\nabla^{2}\varphi-\frac1n\Delta_{g}\varphi\,g
\right)
+{\frac{2n}{n-2}}\varphi^{-2}
\left(
d\varphi\otimes d\varphi-\frac1n|\nabla\varphi|_{g}^{2}g
\right).
\]
The detailed estimates below show that $\|\varphi-1\|_{H^4}=O(e^{-2\alpha a})$. The correction to the trace-free Ricci tensor therefore satisfies
\[
\left\|
{-2}\varphi^{-1}
\left(
\nabla^{2}\varphi-\frac1n\Delta_{g}\varphi\,g
\right)
+{\frac{2n}{n-2}}\varphi^{-2}
\left(
d\varphi\otimes d\varphi-\frac1n|\nabla\varphi|_{g}^{2}g
\right)
\right\|_{H^{2}}
=
O(e^{-2\alpha a}).
\]
Therefore
\[
\|\mathfrak{T}[g_{0}]-\mathfrak{T}[g]\|_{H^{2}}=O(e^{-2\alpha a}),
\]
so the large $H^{2}$-size of the trace-free Ricci tensor survives the conformal correction. In particular,
\[
\|\mathfrak{T}[g_{0}]\|_{H^{2}}
\sim
\|\mathfrak{T}[g]\|_{H^{2}},
\]
up to an error which is negligible on the scale relevant to the construction. For each sufficiently large $a$ this gives a relatively open set of initial data $(g_{0,a},\Sigma_{0,a})$ satisfying 
\begin{align*}
R[g_{0,a}]+\frac23&=|\Sigma_{0,a}|_{g_{0,a}}^{2},\\
\div_{g_{0,a}}\Sigma_{0,a}&=0,\\
\tr_{g_{0,a}}\Sigma_{0,a}&=0,
\end{align*}
and 
\begin{align*}
\|\mathfrak T[g_{0,a}]\|_{H^{2}}&=O({\color{blue}e^{(1-\alpha)a}}),\\
\|\Sigma_{0,a}\|_{H^{2}}&=O({\color{blue}e^{-\alpha a}}),\qquad \|\Sigma_{0,a}\|_{H^{3}}=O({\color{blue}e^{(1-\alpha)a}}).   
\end{align*}

\noindent In summary, the initial data are produced by combining two mechanisms with complementary effects. A high-frequency transverse-traceless perturbation creates a large trace-free Ricci component while keeping the scalar curvature nearly constant, and a subsequent conformal deformation restores the Hamiltonian constraint while modifying the trace-free Ricci tensor only perturbatively. Simultaneously one explicitly constructs the $TT$ tensor $\Sigma$ having a large $H^{3}$ norm while keeping its $H^{2}$ norm small at the stated scale. This is the geometric content behind the construction; the rigorous implementation is carried out in Section~\ref{data}.

\begin{figure}[htbp]
\centering
\begin{tikzpicture}[line cap=round,line join=round,>=stealth,font=\small]
\begin{scope}[x=4.3cm,y=2.7cm]
  \node[anchor=west] at (-.12,1.32) {CMC data at $T_0=a$};
  \fill[black!9] (0,0)--(1,0)--(0,1)--cycle;
  \draw[thick] (0,1)--(1,0);
  \draw[->] (0,0)--(1.12,0) node[right] {$H$};
  \draw[->] (0,0)--(0,1.12) node[above] {$S$};
  \draw (0,1)--(-.016,1) node[left] {$e^{-\alpha a}-e^{-a}$};
  \draw (1,0)--(1,-.027) node[below] {$e^{(1-\alpha)a}-1$};
  \node at (.32,.30) {$H+e^a S\leq\mathcal I^0$};
  \fill (0,0) circle[radius=1.5pt];
  \node[below left] at (0,0) {$0$};
\end{scope}
\end{tikzpicture}
\caption{The initial bounds of Theorem~\ref{main}. $S=\|\Sigma_0\|_{H^2}$ and $H=\|\Sigma_0\|_{H^3}+\|\mathfrak T[g_0]\|_{H^2}$, with each Sobolev norm written as the sum of derivative norms as in the theorem. The shaded triangle depicts the scalar inequality $H+e^aS\leq\mathcal I^0$. The norms also satisfy $H\geq S$. The origin corresponds to zero shear and a normalized hyperbolic metric, and can occur only when $M$ is hyperbolizable. Section~\ref{data} constructs constraint data with large $H$. The axes show the intercepts at $\mathcal I^0=e^{(1-\alpha)a}-1$. For a smaller size, the $H$- and $S$-intercepts are $\mathcal I^0$ and $e^{-a}\mathcal I^0$, respectively.}
\label{fig:blue}
\end{figure}

\noindent
A corollary of Theorem~\ref{main} is as follows.

\begin{corollary}
\label{cor:cosmic_censorship}
{The large initial data covered by Theorem~\ref{main} have a smooth, future causally geodesically complete development. In particular, within this class no future causal geodesic ends at a singularity in finite affine parameter. This proves the future regularity conclusion for the stated data class. The density assertion in Conjecture~\ref{conjecture1} remains open.}
\end{corollary}

\subsection{Comparison with previous work}

The mathematical analysis of Einstein's equations with positive cosmological constant has a substantial history. {\color{blue}The comparison with the present theorem depends on the gauge, the normalization of the data and the quantitative hypotheses.}

\noindent At the level of spatially homogeneous cosmologies, Wald's cosmic no-hair theorem shows that initially expanding Bianchi solutions of types I--VIII satisfying the stated energy conditions with \(\Lambda>0\) are driven toward de Sitter behavior at late times \cite{wald}.  This result is foundational, but it concerns homogeneous dynamics and does not address the large-data Cauchy problem on a general compact spatial manifold.  Ringstr\"om studied expanding solutions of the Einstein--scalar field system \cite{ringstrom2008future}. In Section~18 of that work, he also constructs vacuum data on arbitrary closed spatial topology by rescaling a constant-scalar-curvature metric and taking the second fundamental form proportional to it. These CMC data have zero shear and yield future causally geodesically complete spacetimes with asymptotic expansions. {\color{blue}The construction in Section~\ref{data} goes beyond this umbilic family: it gives exact constraint data with nonzero shear and simultaneous lower bounds for the normalized $H^4$ metric perturbation, $H^2$ trace-free Ricci tensor and $H^3$ shear, all of which can grow like $e^{(1-\alpha)a}$. The estimates also verify the hypotheses of the global evolution theorem on the same initial slices.}

\noindent {\color{blue}Ringstr\"om's local stability theorem, Theorem~2 of \cite{ringstrom2008future}, must be distinguished from that explicit construction. In dimension three it permits nonzero-shear perturbations satisfying a local coordinate $H^4\times H^3$ smallness condition for the physical metric and second fundamental form near expanding de Sitter patch data. Large normalized norms do not exclude this hypothesis. For the family of Section~\ref{data}, a coordinate dilation of order $|\varphi(a)|\approx e^{-a}$ gives a factor $|\varphi(a)|^{5/2}$ in the local $L^2$ bounds for the fourth derivatives of the physical metric perturbation and the third derivatives of the physical shear, so these contributions are $O(e^{-(3/2+\alpha)a})$. The lower derivatives are controlled by the small $H^3$ metric perturbation and $H^2$ shear; the smooth background and the isotropic part give vanishing errors as well. Thus this family also satisfies the local physical smallness condition for sufficiently large $a$. The additional conclusions established here are the explicit global CMC criterion, estimates uniform in the allowed large normalized size, and smooth convergence in this normalization to a constant-scalar-curvature metric. These conclusions are not stated in that form in the local theorem.}

\noindent For the vacuum Einstein--\(\Lambda\) equations, Friedrich introduced the conformal field equations and recast the problem as a regular first-order symmetric hyperbolic system for suitably rescaled variables \cite{friedrich1986existence1,friedrich1986existence}.  Friedrich gives a Cauchy problem at spacelike conformal infinity, with a Riemannian metric and a transverse-traceless tensor as free data up to conformal equivalence. {\color{blue}This theory yields future-complete developments and stability under perturbations of induced Cauchy data, without requiring the conformal boundary metric to be Einstein.} Our hypothesis is expressed in norms of the physical CMC data after normalization on an interior slice at $T_0=a$.

\noindent {\color{blue}Exact high-frequency solutions of the constraint equations have also been constructed by Touati \cite{touati2023constraints}. His construction concerns asymptotically Euclidean data on $\mathbb R^3$ for $\Lambda=0$, with a prescribed null-dust weak limit. Here we construct CMC data on closed three-manifolds for $\Lambda>0$, with small lower-order norms and arbitrarily large higher-order normalized norms, and verify that these data satisfy the hypotheses of our future-global theorem. The two constructions address different geometric settings and different evolution problems.}

\noindent More closely related to the present work are the CMC-based perturbative results of Fajman--Kr\"oncke and the attractor analysis of Mondal.  In the work of Fajman--Kr\"oncke, one studies the Einstein--\(\Lambda\) flow near background solutions for which the spatial metric is Einstein with positive or negative Einstein constant; the analysis is perturbative around such fixed points in high Sobolev regularity \cite{fajman2020stable}.  Likewise, in Mondal's work, one proves global existence and convergence for sufficiently small perturbations of a family of Einstein background solutions in CMC spatial harmonic gauge, with a shadow-gauge mechanism in higher dimensions to control the Einstein moduli directions \cite{mondal2019attractors}.  In both cases the initial data are assumed to lie in a genuinely small neighborhood of an Einstein background (or of the associated finite-dimensional Einstein moduli space).

\noindent The present theorem operates in a different regime.  First, no Einstein background is fixed in advance.  The spatial manifold \(M\) is only assumed to be a closed \(3\)-manifold with
\[
\sigma(M)\le 0,
\]
and the theorem does not require that \(M\) admit a negative Einstein metric. Second, the data \((g_0,\Sigma_0)\) are the normalized CMC data on the interior slice \(T_0=a\). Section~\ref{data} constructs exact solutions of the constraints, with the conformal factor obtained by contraction. Third, the normalized data size \(\mathcal I^0\) can grow with \(a\) subject to the relative condition
\[
\textcolor{blue}{(1+\mathcal I^0)e^{-a}\le e^{-\alpha a}}.
\]
In particular, the renormalized trace-free Ricci tensor \(\mathfrak T[g_{0}]\) is not assumed to be small in \(H^{2}\), and the unweighted top-order norm of the trace-free second fundamental form \(\Sigma_{0}\) is not assumed to be small in \(H^{3}\).  This places the theorem outside the perturbative framework of \cite{fajman2020stable,mondal2019attractors}.

The initial slice is chosen far out in the expanding regime. The curvature--shear coupling then carries the integrable coefficient \(\varphi/\tau=O(e^{-T})\), whose decay was already used in \cite{mondal2019attractors}. Here we use this coefficient together with the elliptic lapse equation to control data for which both \(\|\mathfrak T[g_0]\|_{H^2(g_0)}\) and \(\|\Sigma_0\|_{H^3(g_0)}\) can grow at any fixed exponential rate below one and above zero.

\noindent To the best of our knowledge, this is the first future-global well-posedness and smooth convergence result in CMC-transported coordinates for the \(3+1\)-dimensional vacuum Einstein equations with \(\Lambda>0\) under the quantitative large-data condition of Theorem~\ref{main}. The examples on compact \(3\)-manifolds with \(\sigma(M)\le0\) include nonzero shear of arbitrarily large normalized \(H^3\) norm. No Einstein reference metric is required.

\subsection{Novelty and the main difficulty }
Placing the initial slice sufficiently far in the expanding regime does not, by itself, verify the hypotheses of the perturbative CMC theorems compared below. Those theorems require the normalized data to lie in a small high-Sobolev neighbourhood of an Einstein background. Theorem~\ref{main} instead gives a quantitative criterion allowing the renormalized trace-free Ricci tensor and the unweighted top-order norm of the trace-free second fundamental form to be arbitrarily large. This comparison concerns the stated CMC hypotheses; the relation to Ringstr\"om's local physical criterion is described above.

\noindent To make this precise, recall first the perturbative hypotheses in the work of
Fajman--Kr\"oncke \cite[Theorem~1.1]{fajman2020stable}.  In the expanding negative Einstein case, in their normalization the background metric
\(\gamma\) is assumed to satisfy
\[
\text{Ric}(\gamma)=-(n-1)\gamma,
\]
and the initial data are required to obey a smallness condition of the form
\begin{equation}
\label{eq:FK-small}
\|g_{0}-\gamma\|_{H^{s^{'}}}+\|k_{0}+\sqrt{2}\,\gamma\|_{H^{s^{'}-1}}<\delta,
\end{equation}
for \(s>n/2+2\) and \(s^{'}>n/2+s\), with \(\delta>0\)
sufficiently small.  Likewise, in the perturbative Einstein--\(\Lambda\) attractor theory of
\cite{mondal2019attractors}, one assumes that the rescaled data lie in a small ball
\begin{equation}
\label{eq:Mondal-small}
(g(T_{0}),K^{TT}(T_{0}))\in B_{\delta}(\gamma_{0},0)\subset H^{s}\times H^{s-1},
\end{equation}
centered at an Einstein background \(\gamma_{0}\) (more precisely, at the associated
integrable deformation space), together with the shadow gauge condition.  In both
cases, the theory is perturbative in the strongest possible sense: the initial metric is
assumed to be close, in a high Sobolev topology, to a negative Einstein metric, and the
transverse--traceless part of the second fundamental form is required to be small in the
corresponding Sobolev norm.

\noindent The hypothesis of Theorem~\ref{main} is of a completely different character.  One assumes
only the uniform metric comparability
\[
C^{-1}\xi_{0}\leq g_{0}\leq C\xi_{0},
\]
together with the bound
\begin{equation}
\label{eq:main-large}
\sum_{I=0}^{3}\|\nabla^{I}\Sigma_{0}\|_{L^{2}(M)}
+\sum_{I=0}^{2}\Bigl(\|\nabla^{I}\mathfrak T[g_{0}]\|_{L^{2}(M)}
+\|e^{a}\nabla^{I}\Sigma_{0}\|_{L^{2}(M)}\Bigr)
\leq \mathcal I^{0},
\end{equation}
where \(\mathcal I^{0}>0\) may grow with $a$ subject to the relative
condition
\[
(1+\mathcal I^0)e^{-a}\le e^{-\alpha a}.
\]
In particular, neither \(\|\mathfrak T[g_{0}]\|_{H^{2}}\) nor the unweighted quantity
\(\|\Sigma_{0}\|_{H^{3}}\) is required to be small.  Thus the present theorem allows a regime
in which
\[
\|\mathfrak T[g_{0}]\|_{H^{2}}\sim \mathcal I^{0},
\qquad
\|\Sigma_{0}\|_{H^{3}}\sim \mathcal I^{0},
\]
with \(\mathcal I^{0}\) as large as one pleases, provided that the initial CMC time \(a\) is
chosen sufficiently large in terms of \(\mathcal I^{0}\).

\noindent This difference is not merely linguistic; it is quantitative.  Indeed, if \(\gamma\) is an
Einstein metric, then \(\mathfrak T[\gamma]=0\), and by continuity of the nonlinear map
\(g\mapsto \mathfrak T[g]\) from high Sobolev metrics to curvature in a two--derivative lower
Sobolev space, any perturbative hypothesis of the form \eqref{eq:FK-small} forces
\(\|\mathfrak T[g_{0}]\|_{H^{s'-2}}\lesssim \delta\).  For $s'\ge4$, smallness in these norms bounds $\|\mathfrak T[g_0]\|_{H^2}$ and the $H^3$ norm of the trace-free second fundamental form. The family of Section~\ref{data} has both of these normalized norms unbounded. This comparison concerns the indicated regularity and normalization of the initial data.

\noindent A second distinction is geometric.  The perturbative results on the negative Einstein branch start from a negative Einstein
background and therefore presuppose the existence of such a metric on the spatial
manifold.  Theorem~\ref{main} is stated on an arbitrary closed, connected, oriented
\(3\)-manifold with \(\sigma(M)\le0\), under the geometric bounds
\eqref{eq:1}--\eqref{eq:2}, without assuming the existence of any Einstein metric and
without placing the data in a small neighbourhood of any center manifold or moduli
space.  The result therefore applies to explicitly constructed large data on compact
topologies that need not arise as small perturbations of a hyperbolic Einstein geometry.

\noindent It is also important to distinguish the present large--data regime from mere smallness of
the scalar curvature defect.  In dimension \(3\), the Hamiltonian constraint gives
\[
R(g_{0})+\frac{2}{3}=|\Sigma_{0}|^{2}.
\]
Since the \(H^{2}\)-norm of \(\Sigma_{0}\) is weighted by \(e^{a}\), one obtains
\[
\|R(g_{0})+\tfrac23\|_{H^{2}}
\lesssim
\|\Sigma_{0}\|_{H^{2}}^{2}
\lesssim
e^{-2a}(\mathcal I^{0})^{2}.
\]
\noindent Thus the scalar curvature defect may be very small even when
\(\|\mathfrak T[g_{0}]\|_{H^{2}}\) and \(\|\Sigma_{0}\|_{H^{3}}\) are very large.  This does not
place \(g_{0}\) in a perturbative neighborhood of an Einstein metric, because the scalar
curvature controls only the trace part of the Ricci tensor and does not control the full
trace-free Ricci tensor in any comparable topology.  The data considered here are
therefore large in the normalized higher-order norms specified in Theorem~\ref{main}.

\noindent The contribution is a quantitative global CMC theorem together with exact data realizing its large normalized sizes. It gives global future control for an explicitly constructed family of
initial data satisfying the bounds \eqref{eq:1}--\eqref{eq:2} and the relative smallness condition
\((1+\mathcal I^0)e^{-a}\le e^{-\alpha a}\), while permitting arbitrarily large
\(\|\mathfrak T[g_{0}]\|_{H^{2}}\) and arbitrarily large unweighted
\(\|\Sigma_{0}\|_{H^{3}}\).  The resulting convergence theorem for the vacuum Einstein equations with
\(\Lambda>0\) in CMC--transported spatial coordinates applies on compact spatial manifolds
with \(\sigma(M)\le0\), without a prescribed Einstein background.

\noindent We next record several additional features of the formulation and of the proof that are essential in the large--data regime.

\medskip

\noindent A first point concerns the choice of gauge.  Rather than working with the coupled hyperbolic--elliptic Einstein system in constant mean curvature and spatial harmonic gauge, we work in CMC--transported spatial coordinates, so that the shift vector field vanishes identically, and we formulate the evolution in terms of the pair
\[
(\Sigma,\mathfrak T),\qquad
\mathfrak T_{ij}:=\text{Ric}_{ij}[g]-\frac13 R[g]\,g_{ij},
\]
coupled only to the elliptic lapse equation \eqref{eq:lapse3}.  This point of view is closer in spirit to the classical work of Choquet--Bruhat and York, where one seeks to evolve curvature-type variables together with the second fundamental form.  In the present gauge, the metric is recovered from the transport equation \eqref{eq:metric_transport} along the future-directed normal flow, whereas the higher-order analysis is carried by the evolution system satisfied by \(\Sigma\) and \(\mathfrak T\).

\medskip

\noindent We now recall the perturbative estimate for the shift in spatial harmonic gauge and explain which hypothesis it requires. Let \(\xi\) be a fixed smooth background metric on \(M\), and let
\[
V^{i}:=g^{jk}\bigl(\Gamma[g]^{i}_{jk}-\Gamma[\xi]^{i}_{jk}\bigr)
\]
denote the spatial harmonic gauge vector (tension vector field). Preserving the condition \(V=0\) under the
Einstein evolution leads to an elliptic equation for the shift vector field \(X\) of the form
\begin{equation}
\label{eq:shift_operator_intro}
\mathcal P_{g,\xi}X=\mathcal F[g,\Sigma,N],
\end{equation}
where, schematically,
\begin{equation}
\label{eq:P_g_xi_intro}
(\mathcal P_{g,\xi}X)^{i}
=
-\Delta_{g}X^{i}
-\text{Ric}[g]^{i}{}_{j}X^{j}
+
2\nabla^{j}X^{k}\bigl(\Gamma[g]^{i}_{jk}-\Gamma[\xi]^{i}_{jk}\bigr),
\end{equation}
and \(\mathcal F[g,\Sigma,N]\) depends on the lapse, the second fundamental form, and the
metric.  This is the operator that must be inverted in order to solve for the shift.

\noindent The perturbative theory rests on the fact that when \(g\) remains close to a fixed negative
Einstein metric \(\xi\), the operator \(\mathcal P_{g,\xi}\) is a small perturbation of the
background operator
\begin{equation}
\label{eq:P_background_intro}
\mathcal P_{\xi,\xi}X
=
-\Delta_{\xi}X-\text{Ric}[\xi](X).
\end{equation}
If \(\xi\) is Einstein with
\[
\text{Ric}[\xi]=-\frac{(n-1)}{n^{2}}\xi,
\]
then
\begin{equation}
\label{eq:P_background_hyperbolic_intro}
\mathcal P_{\xi,\xi}X
=
-\Delta_{\xi}X+\frac{(n-1)}{n^{2}}X.
\end{equation}
In particular, for every smooth vector field \(X\),
\begin{align}
\label{eq:background_coercive_intro}
\int_{M}\langle \mathcal P_{\xi,\xi}X,X\rangle_{\xi}\,d\mu_{\xi}
&=
\int_{M}\Bigl(|\nabla^{\xi}X|_{\xi}^{2}+\frac{(n-1)}{n^{2}}|X|_{\xi}^{2}\Bigr)\,d\mu_{\xi} \\
&\ge c_{0}\|X\|_{H^{1}(\xi)}^{2}, \nonumber
\end{align}
for some constant \(c_{0}>0\) depending only on \((M,\xi)\).  Thus
\(\mathcal P_{\xi,\xi}\) is injective, self-adjoint on \(L^{2}(TM,\xi)\), and, by standard
elliptic Fredholm theory, an isomorphism
\[
\mathcal P_{\xi,\xi}\colon H^{r+1}(TM,\xi)\longrightarrow H^{r-1}(TM,\xi)
\]
for every \(r\ge1\).  Moreover,
\begin{equation}
\label{eq:background_elliptic_estimate_intro}
\|X\|_{H^{r+1}(\xi)}
\le C_{r,\xi}\|\mathcal P_{\xi,\xi}X\|_{H^{r-1}(\xi)}.
\end{equation}

\noindent The small-data argument is then obtained by perturbing away from
\(\mathcal P_{\xi,\xi}\).  Fix \(s>\frac n2+1\).  If \(g\) satisfies
\begin{equation}
\label{eq:g_close_xi_intro}
\|g-\xi\|_{H^{s}(\xi)}\le \delta,
\end{equation}
with \(\delta\) sufficiently small, then Sobolev multiplication and the smooth dependence of
\(\Gamma[g]\), \(g^{-1}\), and \(\text{Ric}[g]\) on \(g\) imply the operator bound
\begin{equation}
\label{eq:operator_perturbation_intro}
\|(\mathcal P_{g,\xi}-\mathcal P_{\xi,\xi})X\|_{H^{r-1}(\xi)}
\le
C_{r,\xi}\,\|g-\xi\|_{H^{s}(\xi)}\,\|X\|_{H^{r+1}(\xi)},
\qquad 1\le r\le s-1.
\end{equation}
For the quadratic form we pair directly in the fixed metric \(\xi\).
In \(\xi\)-covariant derivatives, the differences of the principal and first-order
coefficients are \(O(\delta)\) in \(W^{1,\infty}\) and \(L^\infty\), respectively.
The zeroth-order coefficient is \(O(\delta)\) in \(H^{s-2}\subset L^{3/2}\).
Integrating the principal term once by parts and using \(H^1\subset L^6\)
for the zeroth-order term gives
\begin{equation}
\label{eq:quadratic_perturbation_intro}
\biggl|
\int_{M}
\bigl\langle
(\mathcal P_{g,\xi}-\mathcal P_{\xi,\xi})X,X
\bigr\rangle_{\xi}\,d\mu_{\xi}
\biggr|
\le
C_{\xi}\,\|g-\xi\|_{H^{s}(\xi)}\,\|X\|_{H^{1}(\xi)}^{2}.
\end{equation}
Combining \eqref{eq:background_coercive_intro} and
\eqref{eq:quadratic_perturbation_intro}, we find that if
\[
C_{\xi}\delta\le \frac12 c_{0},
\]
then
\begin{equation}
\label{eq:P_g_xi_coercive_intro}
\int_{M}\langle \mathcal P_{g,\xi}X,X\rangle_{\xi}\,d\mu_{\xi}
\ge \frac12 c_{0}\|X\|_{H^{1}(\xi)}^{2}.
\end{equation}
Hence \(\mathcal P_{g,\xi}\) is injective. Since it is elliptic of index zero and depends
continuously on \(g\), it follows that
\[
\mathcal P_{g,\xi}\colon H^{r+1}(TM,\xi)\to H^{r-1}(TM,\xi)
\]
is an isomorphism for all \(1\le r\le s-1\), with quantitative estimate
\begin{equation}
\label{eq:shift_estimate_intro}
\|X\|_{H^{r+1}(\xi)}
\le
C_{r,\xi}\|\mathcal P_{g,\xi}X\|_{H^{r-1}(\xi)}.
\end{equation}
This is the analytic mechanism used in perturbative treatments such as
Fajman-Kr\"oncke \cite{fajman2020stable} and Mondal \cite{mondal2019attractors}: the shift equation is solvable because the evolving metric
remains inside a sufficiently small \(H^{s}\)-neighbourhood of the chosen Einstein background (or of the corresponding shadow metric in the presence of moduli).

\noindent Theorem~\ref{main} assumes \(\sigma(M)\le0\) and allows \(M\) to have no negative
Einstein metric. Even when such a metric \(\xi\) exists, the theorem does not impose
\eqref{eq:g_close_xi_intro}. For the particular examples of Section~\ref{data} with
hyperbolic seed \(\gamma=\xi\), the initial bound
\(\|g_{0,a}-\xi\|_{H^3(\xi)}\lesssim e^{-\alpha a}\) does give
\eqref{eq:operator_perturbation_intro} at \(s=3\), even though
\(\|g_{0,a}-\xi\|_{H^4(\xi)}\) can be arbitrarily large. For the full class of
Theorem~\ref{main} there is no hypothesis forcing
\[
\|g(T)-\xi\|_{H^{s}(\xi)}\ll1
\]
on the bootstrap interval, and therefore no quantitative perturbative reduction of
\(\mathcal P_{g,\xi}\) to the coercive background operator \(\mathcal P_{\xi,\xi}\).  In
particular, the first--order coefficient
\[
2\nabla^{j}X^{k}\bigl(\Gamma[g]^{i}_{jk}-\Gamma[\xi]^{i}_{jk}\bigr)
\]
and the difference of the zeroth--order terms
\(-\text{Ric}[g](X)+\text{Ric}[\xi](X)\) need not be small relative to the
background operator. Thus the perturbative shift estimate does not apply uniformly
to the entire data class.

\noindent This is precisely why we work instead in CMC--transported spatial coordinates.  In that
gauge the shift vanishes identically, and the analysis is reduced to a hyperbolic evolution
system for \((\Sigma,\mathfrak T)\) coupled to a single elliptic equation for the lapse,
whose coefficients are handled directly in terms of the evolving geometry.
\medskip

\noindent A second point is the hierarchy of energy norms used to control the second fundamental form.  The basic energy estimate gives a uniform bound for the top-order norm of \(\Sigma\).  By contrast, the lower-order covariant Sobolev norms of \(\Sigma\) decay exponentially.  The proof exploits this discrepancy systematically: lower-order decay is first established and then fed back into the top-order energy estimates to control the nonlinear error terms that are borderline at the highest derivative level.  This hierarchy is one of the basic mechanisms by which the large-data bootstrap is closed.

\noindent The highest derivatives of the lapse in the shear equation carry the factor $\varphi/\tau=O(e^{-T})$. In the curvature equation the lapse Hessian is unweighted. The fourth-derivative lapse estimate proved in Proposition~\ref{elliptic} is needed for this term. The weighted product estimate for a trace-free tensor controls the corresponding connection commutators; see \eqref{eq:weighted-lapse-products}.

\medskip

\noindent A further issue is that all Sobolev and elliptic norms are defined relative to the evolving metric \(g\).  Consequently, one must propagate quantitative control of the geometry of the slices simultaneously with the energy estimates.  The metric is governed by the transport equation
\begin{equation}
\label{eq:metric_transport}
\partial_T g_{ij}
=
{+\frac{2\varphi}{\tau}N\Sigma_{ij}}
-
2\Bigl(1-\frac{N}{n}\Bigr)g_{ij}.
\end{equation}
Integrating this identity along the CMC foliation yields control of the metric coefficients and uniform equivalence of the metrics \(g(T)\), provided \(\Sigma\) and \(\frac{N}{n}-1\) are already controlled in the relevant Sobolev norms.  What it does \emph{not} yield is a gain of one additional spatial derivative on \(g\).  For that reason, one cannot simply freeze the elliptic theory relative to a fixed background metric and appeal to standard coefficient-based regularity estimates at the top order.

\noindent The lapse equation illustrates this point.  In our normalization it takes the form
\begin{equation}
\label{eq:lapse_elliptic}
-\Delta_{g}\Bigl(\frac{N}{n}-1\Bigr)
+
\Bigl(|\Sigma|^{2}+\frac13\Bigr)\Bigl(\frac{N}{n}-1\Bigr)
=
-|\Sigma|^{2}.
\end{equation}
If one were to formulate all Sobolev norms relative to a fixed metric \(\xi_{0}\), then the usual elliptic theory would require coefficient control at a derivative level that is not directly available from the transport equation for \(g\).  The resolution is to commute \eqref{eq:lapse_elliptic} with covariant derivatives \(\nabla=\nabla[g]\) of the evolving metric itself.  In that formulation, the commutators involve curvature terms rather than uncontrolled higher derivatives of the metric coefficients.

\noindent This is precisely where the choice of unknowns becomes decisive.  Since the spatial dimension is three, the full Riemann tensor is algebraically determined by the Ricci tensor.  Writing
\[
\text{Ric}[g]
=
\mathfrak T+\frac13 R[g]\,g
\]
and using the Hamiltonian constraint
\begin{equation}
\label{eq:ham-renormalized-intro}
R[g]+\frac{2}{3}=|\Sigma|^{2},
\end{equation}
one expresses every curvature term appearing in the commuted lapse equation in terms of \(\mathfrak T\), \(\Sigma\), and lower-order contractions thereof.  Consequently, the curvature contributions in the elliptic commutators are controlled by the same Sobolev quantities that already enter the hyperbolic part of the argument.  This avoids derivative loss and yields the required estimates for the normalized lapse in the same regularity scale as the evolution variables.

\medskip

\noindent Another notable feature of the theorem is the generality of the allowed spatial topology.  The argument applies on every smooth, closed, connected, oriented \(3\)-manifold \(M\) with non-positive Yamabe invariant \(\sigma(M)\le 0\).  Combined with the explicit construction of the initial data class described earlier in the introduction, this places the result well beyond perturbative stability theory near a fixed hyperbolic background.  In particular, the theorem is not tied to a distinguished Einstein metric, nor to a center-manifold description of the long-time dynamics.

\medskip

\noindent Finally, the analytic framework developed here appears robust enough to extend beyond the vacuum problem.  The combination of the integrable curvature--shear coefficient, lower-order decay, and elliptic estimates formulated relative to the evolving metric should be adaptable to Einstein-\(\Lambda\) systems with matter, including, for instance, Einstein-\(\Lambda\)-Maxwell and Einstein-\(\Lambda\)-Euler models.  Whether one can obtain comparably sharp large-data future-global results in those settings remains an interesting open problem.

\subsection{Main idea of the proof}
\label{mainidea}

We now explain the mechanism underlying the proof and isolate the precise role of the positive cosmological constant $\Lambda$.  The discussion in this subsection is intentionally schematic: all commuted equations, elliptic estimates, coercive energies, and bootstrap improvements are stated and proved later in the paper.  The point here is to identify the structural feature that makes large--data forward control possible in the Einstein--\(\Lambda\) setting and, at the same time, to clarify why the corresponding argument is unavailable in the vacuum case \(\Lambda=0\).

\medskip

\noindent Let \((\widehat M,\widehat g)\) be a globally hyperbolic spacetime foliated by compact constant--mean--curvature slices \(M_T\) of negative Yamabe type.  We illustrate this in the CMC--spatial harmonic gauge of Andersson and Moncrief~\cite{elliptichyperbolic}.  Denote by \(g,N,X,\Sigma\) the rescaled spatial metric, lapse, shift, and trace--free second fundamental form in the time \(T\).  In the vacuum case \(\Lambda=0\), the Einstein equations take the form
\begin{align}
\partial_T g_{ij}
 &= {+2N\Sigma_{ij}}
    -2\Bigl(1-\frac{N}{n}\Bigr)g_{ij}
    -(\mathscr L_X g)_{ij}, \label{eq:intro_metric_vac}\\
\partial_T \Sigma_{ij}
 &= -(n-1)\Sigma_{ij}
    -N\Bigl(\text{Ric}_{ij}-\frac{1}{3}R[g]g_{ij}\Bigr)
    {+n\nabla_i\nabla_j\Bigl(\frac{N}{n}-1\Bigr)} \nonumber\\
 &\qquad
    +2N\Sigma_{ik}\Sigma^{k}{}_{j}
    -\frac1n\Bigl(\frac{N}{n}-1+N|\Sigma|^{2}\Bigr)g_{ij}
    -(n-2)\Bigl(\frac{N}{n}-1\Bigr)\Sigma_{ij}
    -(\mathscr L_X\Sigma)_{ij}, \label{eq:intro_sigma_vac}
\end{align}
supplemented by the constraint equations
\begin{align}
R+\frac{n-1}{n}-|\Sigma|^2 &=0, \label{eq:intro_ham_vac}\\
\nabla_j\Sigma^j{}_i &=0. \label{eq:intro_mom_vac}
\end{align}
In the spatial harmonic gauge (together with the rescaled Hamiltonian constraint), the linearization of the map
\[
g\longmapsto \text{Ric}[g]-\frac{1}{3}R[g]g
\]
is elliptic, so that the gauge variables are determined by elliptic equations on each slice, whereas the pair \((g,\Sigma)\) evolves by a quasilinear hyperbolic system coupled to these elliptic constraints.

\medskip

\noindent The obstruction to large-data global control in \eqref{eq:intro_metric_vac}-\eqref{eq:intro_mom_vac} is already visible at the level of the \(\Sigma\)-equation.  The term \(-(n-1)\Sigma\) is linearly damping, but the remaining terms contain two genuinely nonperturbative mechanisms:
\begin{enumerate}
\item the curvature defect
\[
\mathfrak T_{ij}:=\text{Ric}_{ij}-\frac{1}{3}R[g]g_{ij},
\]
whose vanishing in dimension three is equivalent to constant sectional curvature, with sectional curvature \(R[g]/6\), negative when \(R[g]<0\); and
\item the Riccati term \(2N\Sigma_i{}^k\Sigma_{kj}\), which is quadratic and of the same differential order as the damping term.
\end{enumerate}
Moreover, the lapse defect \(\frac{N}{n}-1\) is itself determined by an elliptic equation whose source is quadratic in the evolving geometric fields.  Thus, even at the level of the formal energy identities, the linearly decaying contribution \(-(n-1)\Sigma\) is coupled to nonlinear terms that are not accompanied by any time--integrable coefficient.  In the perturbative stability argument about the Milne solution, the estimates are closed by imposing smallness on the curvature defect and on the trace--free second fundamental form in scale--invariant Sobolev norms.  In particular, prescribing the data at a very large initial time \(T_0=a\gg1\) does not by itself create a new mechanism, since a bound of the form
\[
e^{\alpha T_0}\|(\Sigma,\mathfrak T)(T_0)\|\leq \mathcal I,
\qquad \alpha>0,
\]
bounds the unweighted fields by \(e^{-\alpha T_0}\mathcal I\), which is small for fixed \(\mathcal I\) and sufficiently large \(T_0\).

\medskip

\noindent The situation changes decisively when \(\Lambda>0\).  After passing to the rescaled CMC variables used throughout the paper, {the curvature--shear interactions in the evolution equations appear} with an additional coefficient \(e^{-T}\).  More precisely, in the transported spatial coordinates of Theorem~\ref{main}, with \(X=0\), the equations before commutation may be written schematically as
\begin{align}
\partial_T \Sigma
  &= {-(n-1)\Sigma-(n-2)\Bigl(\frac{N}{n}-1\Bigr)\Sigma}
     + e^{-T}\,\mathcal Q_1\bigl(g,N,\Sigma,\mathfrak T,\nabla N,\nabla^2N\bigr), \label{eq:intro_sigma_lambda_schematic}\\
\partial_T \mathfrak T
  &= e^{-T}\,\mathcal Q_2\bigl(g,N,\Sigma,\mathfrak T,\nabla\Sigma,\nabla^2\Sigma,\nabla N,\nabla^2N\bigr)
     - (n-2)\left[\nabla^2\Bigl(\frac{N}{n}-1\Bigr)-\frac1n\Delta_g\Bigl(\frac{N}{n}-1\Bigr)g\right]. \label{eq:intro_T_lambda_schematic}
\end{align}
Here $\mathcal Q_1$ and $\mathcal Q_2$ consist of tensorial contractions with uniformly bounded time coefficients, since $e^T\varphi/\tau=C/\sqrt{C^2e^{-2T}+3\Lambda}$. For $\Lambda=0$ the coefficient $\varphi/\tau$ equals one. {The remaining lapse terms are controlled by the elliptic lapse equation, whose source is quadratic in $\Sigma$.}  In particular, the two terms that are most problematic in the vacuum case,
\[
\mathfrak T
\qquad\text{and}\qquad
\Sigma*\Sigma,
\]
are multiplied by \(\varphi/\tau=O(e^{-T})\) in the shear equation.  This is the basic structural fact on which the whole argument rests.

\medskip

\noindent We emphasize that this mechanism is not a null structure in the classical sense of Klainerman, since the gain does not arise from an algebraic cancellation in the quadratic form relative to the characteristic cone of the principal hyperbolic operator.  Rather, the positive cosmological constant produces an integrable coefficient in the curvature--shear interactions: the nonlinear interactions that would otherwise be borderline are multiplied by a universal coefficient \(e^{-T}\in L^1([T_0,\infty))\).  This integrability is exactly what makes large--data forward control possible once the initial CMC slice is placed sufficiently far in the expanding regime.

\medskip

\noindent To exploit this structure, we introduce three scale--invariant quantities.  The first is the high--order geometric energy
\[
\mathcal O(T)
:= \sum_{I=0}^{3}\|\nabla^I\Sigma\|_{L^2(M_T)}
    + \sum_{I=0}^{2}\|\nabla^I\mathfrak T\|_{L^2(M_T)},
\]
which controls the differentiated curvature and deformation variables.  The second is the lower-order renormalized quantity
\[
\mathcal F(T)
 := \sum_{I\le 2}\|e^T\nabla^I\Sigma\|_{L^2(M_T)},
\]
equivalently
\[
E^{\mathrm{lower}}(T)
 := \sum_{I=0}^{2}\|\nabla^I\Sigma\|_{L^2(M_T)}^2
 \approx e^{-2T}\mathcal F(T)^2.
\]
The third is the pointwise control norm
\[
\mathcal N^\infty(T)
 := \|e^T\Sigma\|_{L^\infty(M_T)}
   + \Bigl\|e^{2T}\Bigl(\frac{N}{n}-1\Bigr)\Bigr\|_{L^\infty(M_T)}.
\]
The proof is organized around a bootstrap scheme for the triple
\[
\bigl(\mathcal O,\mathcal F,\mathcal N^\infty\bigr).
\]

\medskip

\noindent More precisely, on a bootstrap interval \([T_0,T^\ast)\) we assume
\[
\mathcal O(T)\le \Gamma,
\qquad
\mathcal F(T)\le \mathds Y,
\qquad
\mathcal N^\infty(T)\le \mathds L,
\qquad
T\in[T_0,T^\ast),
\]
where $\Gamma=A(1+\mathcal I^0)$, $\mathds Y=B(1+\mathcal I^0)$
and $\mathds L=L(1+\mathcal I^0)^2$. We choose $A,B,L$ in this
order, independently of $a$ and $\mathcal I^0$. For $T\ge a$ the
data hypothesis gives
\[
 (1+\mathcal I^0)e^{-T}\le e^{-\alpha a}e^{-(T-a)}.
\]
We prove $\mathcal O+\mathcal F\lesssim 1+\mathcal I^0$ and
$\mathcal N^\infty\lesssim(1+\mathcal I^0)^2$.

\medskip

\noindent We improve the three bounds on the same bootstrap interval. The top-order factor \(A\) is fixed by the initial energy. The curvature forcing in the lower--order shear estimate then determines \(B\), and the pointwise estimates for the shear and lapse determine \(L\). With these factors fixed, the integrated top-order energy error divided by \((1+\mathcal I^0)^2\) is bounded by
\[
 C_{A,B}\big((1+\mathcal I^0)e^{-a}\big)^2\le C_{A,B}e^{-2\alpha a}.
\]
The lower--order evolution equation for \(\Sigma\) uses both the pointwise control of \(\Sigma,N\) and the curvature bound from \(\mathcal O\).  The decisive term in the corresponding energy identity is the mixed contribution
\[
\sum_{I\le2} e^{-T}\int_{M_T} N\,\nabla^I\mathfrak T\,\nabla^I\Sigma\,\mu_g,
\]
which satisfies
\[
\biggl|
\sum_{I\le2} e^{-T}\int_{M_T} N\,\nabla^I\mathfrak T\,\nabla^I\Sigma\,\mu_g
\biggr|
\lesssim e^{-T}\mathcal O(T)\sqrt{E^{\mathrm{lower}}(T)}\leq \tfrac14 E^{\mathrm{lower}}(T)+C e^{-2T}\mathcal O(T)^2.
\]
The weighted lapse estimates control the other terms. For sufficiently large $a$, their contribution is absorbed by the shear damping, giving
\[
 \frac{d}{dT}\big(e^{2T}\mathscr E^{\rm low}\big)
       +e^{2T}\mathscr E^{\rm low}\le C\mathcal O(T)^2
\]
as shown in \eqref{eq:suboptimal_recall}. Integration from $a$ to $T$ therefore gives
\[
 \mathcal F(T)\lesssim 1+\mathcal I^0+\sup_{a\leq t\leq T}\mathcal O(t).
\]
\noindent The curvature defect may be large at the initial time, subject to the size condition in Theorem~\ref{main}.  The factor \(e^{-T}\) is what converts large but finite geometric input into an integrable error.

\medskip

\noindent At the top order the lapse Hessian in the evolution of $\mathfrak T$ has no factor $\varphi/\tau$. We must therefore estimate four covariant derivatives of $N/n-1$ and hope that it exhibits an integrable in-time decay property.
In dimension three, commuting two derivatives of a symmetric trace-free
tensor $A$ controls $|\operatorname{Riem}|\,|A|$ up to the scalar
curvature term. Testing the scalar lapse equation as in
Proposition~\ref{elliptic} gives
\[
 \bigl\||A|\,|\nabla^2(N/n-1)|\bigr\|_{L^2}
 +\bigl\||\nabla A|\,|\nabla(N/n-1)|\bigr\|_{L^2}
 \lesssim \|\Sigma\|_{H^2}^2\|A\|_{H^2}.
\]
For $A=\Sigma,\mathfrak T$ these bounds control the time commutators.
The fourth-order lapse estimate costs only one curvature factor,
\[
 \|N/n-1\|_{H^4}\lesssim
 (1+\|\mathfrak T\|_{H^2})\|\Sigma\|_{H^2}^2 .
\]
In Section~\ref{energy} we integrate the resulting error terms and
compare them with $(1+\mathcal I^0)^2$, the upper bound for the initial energy.

\medskip

\noindent Finally, the metric and the gauge must be propagated without loss of geometric control.  The metric equation reads
\[
\partial_T g
 = 2\frac{\varphi}{\tau}N\Sigma
   -2\Bigl(1-\frac{N}{n}\Bigr)g,
\]
so once \(\mathcal N^\infty\) is controlled, the deformation of the metric is integrable in time.  This yields uniform equivalence of the evolving metrics, control of the volume form and isoperimetric constants, and hence uniform Sobolev inequalities on all future slices.  These geometric bounds are then fed back into the elliptic estimates for the lapse and the Sobolev estimates used in the energy argument, closing the system.

\medskip

\noindent We therefore obtain a self--consistent bootstrap mechanism in which lower order decay produces top order control, top order control feeds the elliptic theory, and the resulting gauge bounds preserve the geometric background needed for the Sobolev and energy estimates.  The factor $\varphi/\tau$ gives exponential decay in the curvature--shear interactions. For the unweighted lapse terms the decay comes from the quadratic source $|\Sigma|^2$ in the lapse equation.  This is the decisive large--data stabilizing effect of the positive cosmological constant and the basic reason the theorem proved here has no analogue in the known \(\Lambda=0\) CMC theory.

{
\begin{remark}
For positive Yamabe type, the analogous statement concerns the branch
$-\sqrt{3\Lambda}<\tau<0$. The sign of the scalar-curvature normalization
is determined by this branch, rather than by the Yamabe invariant alone.
We use
\[
 \varphi=-e^{-T},\qquad \tau=-\sqrt{3\Lambda-e^{-2T}},\qquad
 \frac{d\tau}{dT}=\frac{\varphi^2}{\tau}<0,
 \qquad \widetilde g=\varphi^{-2}g,
 \qquad \widetilde k^{TT}=\varphi^{-1}\Sigma.
\]
Here $T$ increases towards the future. The physical lapse relative to
$T$ is $N/|\tau|$, and the normalized constraints are
\[
 R(g)-\frac23=|\Sigma|_g^2,\qquad
 \operatorname{div}_g\Sigma=0,\qquad \operatorname{tr}_g\Sigma=0.
\]
The lapse equation is
\begin{equation}\label{eq:positive}
 -\Delta_g\left(\frac N3-1\right)
 +\left(|\Sigma|_g^2-\frac13\right)
       \left(\frac N3-1\right)=-|\Sigma|_g^2.
\end{equation}
Solvability alone is insufficient here. A kernel component of $N/3-1$
is not controlled by $\Sigma$, and can prevent convergence in transported
coordinates. We therefore impose a spectral gap on the initial metric.
The proof below shows that this gap persists; no condition on an unknown
future solution is included in the hypotheses.
\end{remark}

\begin{theorem}[Global well-posedness for $\Lambda>0$ and $\sigma(M)>0$]
\label{main2}
Let $M$ be a smooth closed three-manifold of positive Yamabe type. Fix a
smooth background metric $\xi_0$, a constant $C>1$, and $0<\delta<1/6$.
Fix $0<\alpha<1$. For all sufficiently large $a$, assume $(1+\mathcal I^0)e^{-a}\le e^{-\alpha a}$ as in Theorem~\ref{main}. Here the choice of how large $a$ must be may depend also on the fixed spectral gap $\delta$. 
Let $(g_0,\Sigma_0)$ be smooth data at $T_0=a$ satisfying the normalized
positive-branch constraints above and
\begin{equation}\label{eq:metric_control}
 C^{-1}\xi_0\leq g_0\leq C\xi_0,
\end{equation}
\begin{equation}\label{eq:data_estimate}
 \sum_{I=0}^{3}\|\nabla^I\Sigma_0\|_{L^2}
 +\sum_{I=0}^{2}\left(
    \|\nabla^I\mathfrak T[g_0]\|_{L^2}
    +\|e^a\nabla^I\Sigma_0\|_{L^2}\right)\leq \mathcal I^0.
\end{equation}
Assume in addition that
\[
 \operatorname{dist}\left(\frac13,
             \operatorname{spec}(-\Delta_{g_0})\right)\geq\delta.
\]
Then the Einstein--$\Lambda$ equations have a unique smooth future-global
solution in CMC transported coordinates on $[a,\infty)\times M$, with
the lapse determined by \eqref{eq:positive}. The solution is future
causally geodesically complete. Its normalized metric remains uniformly
equivalent to $g_0$, and
\[
 \sup_{T\geq a}\left\{
 \sum_{I=0}^{3}\|\nabla^I\Sigma(T)\|_{L^2}
 +\sum_{I=0}^{2}\left(
   \|\nabla^I\mathfrak T[g(T)]\|_{L^2}
   +\|e^T\nabla^I\Sigma(T)\|_{L^2}\right)\right\}
 \leq C_1(1+\mathcal I^0),
\]
where $C_1$ depends only on the fixed data and $\delta$.
Moreover, $\Sigma(T)\to0$ and $g(T)\to g_\infty$ smoothly as
$T\to\infty$, with $R(g_\infty)=2/3$.
\end{theorem}

\begin{proof}
We give the changes to the proof of Theorem~\ref{main}. Write $q=N/3-1$
and $c=\varphi/\tau$. The metric and trace-free curvature equations have
the same form as in the negative branch, and the shear equation is
\[
 \partial_T\Sigma_{ij}
 =-2\Sigma_{ij}-cN\mathfrak T_{ij}
   +3c\nabla_i\nabla_jq+2cN\Sigma_i{}^k\Sigma_{kj}
   -\frac c3\bigl(-q+N|\Sigma|^2\bigr)g_{ij}-q\Sigma_{ij}.
\]
Thus the only change in the shear equation is a pure-trace term. It
vanishes in every pairing with $\nabla^I\Sigma$. The principal
cancellation in the energy identity is unchanged, and
$0<c\leq C_\Lambda e^{-T}$ for $T\ge a$. At fixed $\Lambda>0$, taking $a$ sufficiently large ensures $e^{-2a}\le3\Lambda/2$.

\noindent We first justify uniform inversion of
\[
 L_T=-\Delta_{g(T)}+|\Sigma(T)|^2-\frac13.
\]
Temporarily assume
$(1-\varepsilon)g_0\leq g(T)\leq(1+\varepsilon)g_0$ and
$\|\Sigma(T)\|_{L^\infty}^2\leq\delta/4$.
The Rayleigh quotient and the min--max principle give, for all
eigenvalues counted with multiplicity,
\[
 (1-C\varepsilon)\lambda_j(g_0)
 \leq\lambda_j(g(T))\leq
 (1+C\varepsilon)\lambda_j(g_0).
\]
For $\varepsilon\leq\varepsilon_\delta$ sufficiently small, the
initial gap therefore implies
$\operatorname{dist}(1/3,\operatorname{spec}(-\Delta_{g(T)}))
\geq\delta/2$. Adding the bounded potential $|\Sigma|^2$ moves
each eigenvalue by at most $\delta/4$. Hence
\[
 \|L_T^{-1}f\|_{L^2(g(T))}\leq\frac4\delta\|f\|_{L^2(g(T))}.
\]
All scalar Sobolev constants are uniform by metric comparison.
Pairing $L_Tu=f$ with $u$ gives
\[
 \|\nabla u\|_{L^2}^2
 \leq\|f\|_{L^2}\|u\|_{L^2}+\frac13\|u\|_{L^2}^2,
 \qquad \|u\|_{H^1}\leq C_\delta\|f\|_{L^2}.
\]
The scalar Moser iteration, applied to this uniformly elliptic
equation with bounded zeroth-order coefficient, yields
\[
 \|u\|_{L^\infty}
 \leq C\bigl(\|u\|_{L^2}+\|f\|_{L^\infty}\bigr).
\]
Here and below these scalar embeddings include their application to tensor norms by Kato's inequality. The iteration constant uses only the scalar Sobolev constant, the volume bounds,
and the bound for the potential. In particular, it is independent of
the curvature bootstrap constant. Applying these estimates to
$L_Tq=-|\Sigma|^2$ gives
\[
 \|q\|_{L^\infty}\leq C_\delta\|\Sigma\|_{H^2}^2.
\]
On the bootstrap interval, the assumed bounds for $\mathcal F$ imply
$\|\Sigma(T)\|_{H^2}\le C(1+\mathcal I^0)e^{-T}$, with a fixed
bootstrap factor. By the initial-size condition, this is at most
$C e^{-\alpha a}e^{-(T-a)}$. Thus $|q|<1/2$ for sufficiently large $a$, and
$3/2<N<9/2$. In the elliptic proof, the $H^1$ inverse estimate just
proved replaces the coercive base estimate. All higher derivative
estimates follow from commuting the scalar Laplacian and estimating
its curvature and potential terms as in Proposition~\ref{elliptic}.
The change of sign of the constant $1/3$ does not introduce a new
commutator. The unweighted upper bound for $N$, the spatial product
estimates, and the cancellation with the top-order curvature equation
are the same on both branches. In particular, all the sharpened
low- and top-order energy estimates of Section~\ref{energy} hold here
with constants depending also on $\delta$. Closing the same energy
bootstrap gives
\[
 \begin{aligned}
 \mathcal O(T)+\mathcal F(T)
 &\le C_\delta(1+\mathcal I^0)\le C_\delta e^{(1-\alpha)a},\\
 \|\Sigma(T)\|_{L^\infty}
 &\le C_\delta(1+\mathcal I^0)e^{-T}\\
 &\le C_\delta e^{-\alpha a}e^{-(T-a)},\\
 \|q(T)\|_{L^\infty}
 &\le C_\delta(1+\mathcal I^0)^2e^{-2T}\\
 &\le C_\delta e^{-2\alpha a}e^{-2(T-a)}.
 \end{aligned}
\]
Integrating
$\partial_Tg=2cN\Sigma+2qg$ gives
\[
 \begin{aligned}
 \left|\log\frac{g(T)(V,V)}{g_0(V,V)}\right|
 &\leq C_\delta(1+\mathcal I^0)^2e^{-2a}\\
 &\leq C_\delta e^{-2\alpha a}\qquad (V\ne0).
 \end{aligned}
\]
Since $(1+\mathcal I^0)e^{-a}\le e^{-\alpha a}$, for sufficiently large $a$
depending also on $\delta$ these bounds strictly improve both additional
bootstrap assumptions on the metric comparison and the size of the shear. Thus the spectral gap persists
on the entire lifespan. The lapse remains positive and uniquely determined.
For local existence and continuation for this case consider the physical
CMC operator
\[
 P_T=-\Delta_{\widetilde g}+|\widetilde k|_{\widetilde g}^2-\Lambda
     =\varphi^2 L_T.
\]
The estimates above apply at the initial slice by the data bounds and
give $P_a$ invertible and $N>0$. In a smooth local Einstein Cauchy
development, write nearby hypersurfaces as normal graphs over the
initial slice. The derivative of their mean curvature in the graph
variable is $P_a$ with the sign convention of \eqref{eq:evol_k}.
The elliptic implicit-function theorem therefore supplies the nearby
graphs with prescribed mean curvature. For increasing $-\tau$ their
normal speed is $-P_a^{-1}1=N/\varphi^2>0$ at the initial slice.
They give a local foliation, on which transporting the initial coordinates
along the normals fixes the shift and gives local uniqueness.

\noindent Suppose now that the maximal future time $T_*$ is finite. The proof of
Lemma~\ref{lem:higher-convergence} applies on $[a,T_*)$ with the
bounded inverse of $L_T$ at order zero and the commuted elliptic
estimates above. Its estimates are independent of the upper time
endpoint. They give uniform bounds in all fixed spatial Sobolev norms
and an integrable bound for $\partial_T g$ in every such norm.
The shear equation then bounds $\partial_T\Sigma$ in all these norms.
Consequently $g$ and $\Sigma$ have smooth limits at $T_*$ satisfying
the constraints, and the metric comparison keeps the limiting metric
positive definite. The scale factors and their inverses are finite
at $T_*$. The physical operator $P_{T_*}$ is invertible by the same
spectral gap, and its lapse $-P_{T_*}^{-1}1$ for increasing $-\tau$
is positive by the uniform lower bound for $N/\varphi^2$. Applying
the local construction to these endpoint data and using Cauchy
uniqueness extends the solution beyond $T_*$. Thus the development
is future global. Lemma~\ref{lem:higher-convergence}, with the same
change in the elliptic estimate, now gives smooth convergence as
$T\to\infty$. The limit satisfies $R(g_\infty)=2/3$ by the
positive-branch constraint.
The proof of future causal geodesic completeness also applies: $\tau<0$,
$N/|\tau|$ stays bounded above and below, and
$|\varphi|\bigl(\|\nabla N\|_{L^\infty}
 +\|N\|_{L^\infty}\|\Sigma\|_{L^\infty}\bigr)$ is integrable.
\end{proof}

\begin{remark}
The initial spectral gap alone does not allow a coefficient-one bound
$1+\mathcal I^0<e^a$ for the positive case. To see this, take
$M=S^2\times S^1$ and
$g=2g_{S^2(1)}+\ell^2d\theta^2$, with $\theta$ of period $2\pi$.
Let the endomorphism $g^{-1}\Sigma$ have eigenvalue $1/\sqrt{18}$
on the sphere and $-2/\sqrt{18}$ on the circle. This parallel tensor
is trace-free and divergence-free, and $R(g)=1$, $|\Sigma|^2=1/3$,
so the positive-branch constraints hold. For $0<\ell\le1$,
$-\Delta_g$ has eigenvalues $j(j+1)/2+k^2/\ell^2$ and
$\operatorname{dist}(1/3,\operatorname{spec}(-\Delta_g))=1/3$.
However, \eqref{eq:positive} reduces to $-\Delta_gN=-1$, which has
no solution on the closed slice. All positive-order covariant derivatives in
\eqref{eq:data_estimate} vanish. The left side of that estimate is
$e^a(16\pi^2\ell/3)^{1/2}+O(1)$. In particular, this example also excludes growth exponents above one on this branch. Fixing $\ell<3/(16\pi^2)$
therefore gives $1+\mathcal I^0<e^a$ for all sufficiently large $a$.
\end{remark}

\noindent The data construction in Section~4 also applies on this branch when
the chosen constant-scalar-curvature background satisfies
$R(\gamma)=2/3$ and
$1/3\notin\operatorname{spec}(-\Delta_\gamma)$, with the same
assumption on conformal Killing fields used there. The Lichnerowicz
equation for the conformal factor $\varphi$ is now
\[
 -8\Delta_g\varphi+R(g)\varphi
       -\frac23\varphi^5-|\Sigma|_g^2\varphi^{-7}=0.
\]
At $(g,\Sigma,\varphi)=(\gamma,0,1)$ its linearization is
$8(-\Delta_\gamma-1/3)$, which is invertible by the stated spectral
condition. For the construction of Section~4 with top-order size $e^{(1-\alpha)a}$, the seed estimates give
\[
 \|g-\gamma\|_{H^3(\gamma)}\le C e^{-\alpha a},\qquad
 \|R(g)-2/3\|_{H^2(\gamma)}
 +\||\Sigma|_g^2\|_{H^2(\gamma)}\le C e^{-2\alpha a}.
\]
Consequently the difference between $-8\Delta_g+R(g)-10/3+7|\Sigma|_g^2$
and $8(-\Delta_\gamma-1/3)$ has norm at most $C(e^{-\alpha a}+e^{-2\alpha a})$ from
$H^4(\gamma)$ to $H^2(\gamma)$. For all sufficiently large $a$, the bounded-inverse argument gives a uniform inverse throughout the family.
The same contraction argument gives a unique solution in
a sufficiently small $H^4$ neighborhood of $1$ and the same bounds
for the conformal correction. This is local uniqueness of the
chosen branch; no uniqueness among all positive solutions of the
positive-Yamabe Lichnerowicz equation is asserted. The seed metrics
and corrected metrics are within $C(e^{-\alpha a}+e^{-2\alpha a})$ of $\gamma$ in $H^3$. Their Laplacian spectral gaps therefore persist by metric comparison and the min--max principle. More precisely, choose
\[
 0<\delta<\min\{1/6,\tfrac12\operatorname{dist}
         (1/3,\operatorname{spec}(-\Delta_\gamma))\}.
\]
For sufficiently large $a$ the corrected metrics satisfy the initial spectral
condition strictly. Taking $\mathcal I^0=2C e^{(1-\alpha)a}$ as in Section~4 verifies the norm and initial-size hypotheses of Theorem~\ref{main2} with $\alpha/2$ in place of $\alpha$. With strict margins in the metric and norm
bounds of the theorem, continuity gives relatively open sets of
such data in the smooth constraint set.
}


\subsection{Consequences for Geometrization}

\noindent An important conceptual motivation underlying the present work concerns the potential role of the Einstein evolution equations in the geometric classification of closed $3$-manifolds. Specifically, we consider the relevance of the long-time behavior of solutions to the Einstein flow in connection with Thurston's geometrization conjecture. In this discussion, we assume throughout that the initial data satisfy the regularity assumptions stipulated in the main theorem.

\noindent A central open problem in mathematical general relativity is whether the Einstein equations on globally hyperbolic, spatially compact cosmological spacetimes can dynamically implement a form of geometrization of the underlying spatial manifold. In particular, we consider the case in which the Cauchy hypersurface \( M \) is a closed, oriented, connected $3$-manifold with nonpositive Yamabe invariant \( \sigma(M) \leq 0 \).

\noindent The geometrization program in the setting of vacuum cosmological spacetimes was first proposed in the foundational works of Fischer--Moncrief \cite{fischer2000reduced,fischer2001reduced}, with further developments by Anderson \cite{anderson2001long}, who studied the asymptotic behavior of the vacuum Einstein flow and its possible connection to the Thurston decomposition. In the idealized scenario, one anticipates that the spatial metric induced on constant mean curvature (CMC) hypersurfaces asymptotically converges to a metric structure realizing a geometric decomposition of \( M \).

\noindent Following Anderson \cite{anderson2001long}, we recall the following definition of geometrization:

\begin{definition}[Geometrization]
Let \( M \) be a closed, oriented, connected $3$-manifold with \( \sigma(M) \leq 0 \). A \emph{weak geometrization} of \( M \) is a decomposition
\[
M = H \cup G,
\]
where:
\begin{itemize}
  \item {\(H\) is the union of compact cores of embedded hyperbolic pieces whose interiors carry complete finite-volume metrics;}
  \item \( G \) is a finite collection of embedded, connected graph manifolds;
  \item the union is along a finite collection of embedded tori \( \{T_i\} \), with \( \partial G = \partial H = \bigsqcup_i T_i \).
\end{itemize}
A \emph{strong geometrization} is a weak geometrization in which each gluing torus \( T_i \) is incompressible, i.e., the inclusion-induced homomorphism \( \pi_1(T_i) \to \pi_1(M) \) is injective.
\end{definition}

\noindent {The above definition allows compressible tori in a weak decomposition. The canonical topological decomposition uses prime splitting along spheres and, on irreducible prime factors, the appropriate incompressible tori. Ringstr\"om's formulation \cite{ringstrom2008future} concerns the part of a late Cauchy slice in an observer's causal past.} {\color{blue}For the class of large initial data in Theorem~\ref{main}, we establish the causal conclusion formulated below.} 

\noindent {In the dynamical picture discussed by Anderson
\cite{anderson2001long}, a suitably normalized metric approaches complete
hyperbolic metrics on the thick pieces, while graph-manifold regions
collapse under a curvature bound in the sense of Cheeger and Gromov
\cite{cheeger1,cheeger2}. In the normalization used here, Theorem~\ref{main}
instead gives the following lower bound on every spatial region.}

\begin{corollary}[No-Collapse Theorem]
\label{nocollapse}
{Under the hypotheses of Theorem~\ref{main}, for each fixed $r>0$
there exist $0<v_r\le V_r<\infty$ such that the balls of the normalized metric
$g(T)$ satisfy
\[
 v_r\le \operatorname{Vol}_{g(T)}\bigl(B_{g(T)}(x,r)\bigr)\le V_r
 \qquad (x\in M,\ T\ge T_0).
\]}
\end{corollary}

\noindent {The proof in Section~\ref{noncollapse} uses the
comparison of $g(T)$ with $g(T_0)$. In particular, no sequence of points
has normalized unit-ball volume tending to zero. For this class of
solutions the normalized CMC flow therefore cannot realize a
decomposition through collapse of a nonempty graph-manifold region.}

\noindent The question of whether the Einstein equations with vanishing cosmological constant \( \Lambda = 0 \) can effectuate a geometric decomposition of \( M \) remains open and considerably more subtle. Perelman's work on the \href{https://arxiv.org/abs/math/0211159}{entropy formula}, \href{https://arxiv.org/abs/math/0303109}{Ricci flow with surgery}, and \href{https://arxiv.org/abs/math/0307245}{finite extinction} uses the parabolic smoothing of the Ricci flow in the proof of geometrization~\cite{porti2008geometrization}. The Einstein equations are hyperbolic, and their global-in-time analysis poses different difficulties. Nonetheless, recent advances in the analysis of nonlinear wave equations in the large-data regime, such as \cite{C09}, offer promising avenues toward addressing these questions, and it is conceivable that similar techniques could be adapted to the cosmological setting.  

\noindent Ringstr\"om \cite{ringstrom2008future} provided a detailed local analysis of the Cauchy problem in general relativity with a positive cosmological constant and proposed a formal framework for quantifying the asymptotic inaccessibility of spatial topology by causal observers. His notion supports the view that long-time evolution under the Einstein-$\Lambda$ flow generically suppresses detectable topological information in the causal past of observers.

\begin{definition}
Let $(\widetilde{M}, \widehat{g})$ be a globally hyperbolic Lorentzian manifold with compact Cauchy hypersurfaces. We say that \emph{late-time observers are oblivious to topology} if there exists a Cauchy hypersurface $M \subset \widetilde{M}$ such that no future-directed inextendible causal curve $\gamma$ satisfies $M \subset J^{-}(\gamma)$. Conversely, we say that \emph{late-time observers are not oblivious to topology} if for every Cauchy hypersurface $M$, there exists a future-directed inextendible causal curve $\gamma$ with $M \subset J^{-}(\gamma)$.
\end{definition}

\noindent Based on this definition, Ringstr\"om proposed the following conjecture:

\begin{conjecture}[\cite{ringstrom2008future}]
Let $(\widetilde{M}, \widehat{g})$ be a future causally geodesically complete solution of the vacuum Einstein equations with positive cosmological constant $\Lambda > 0$, and compact Cauchy hypersurfaces. Then late-time observers in $(\widetilde{M}, \widehat{g})$ are oblivious to topology.
\end{conjecture}

\noindent {\color{blue}In Section~\ref{mainproof}, we verify the conclusion of this conjecture for the solutions covered by Theorem~\ref{main}.}

\noindent {\color{blue}The proof uses the uniform comparison of the renormalized spatial metrics and a time-integrable bound on the spatial speed of causal curves.}
\medskip

\noindent The examples in Section~\ref{proof1} have non-Einstein limiting \emph{normalized} spatial metrics. The physical trace-free second fundamental form nevertheless satisfies $|\widetilde k^{\mathrm{TT}}|_{\widetilde g}=|\varphi|\,|\Sigma|_g=O(e^{-2T})$. While our analysis is confined to the vacuum Einstein equations with $\Lambda>0$, we conjecture that this failure of isotropization of the normalized spatial metric persists for a wide class of physically relevant matter models including Einstein-$\Lambda$-Maxwell and Einstein-$\Lambda$-Euler systems. That is, we expect that for generic initial data with nontrivial topology the limiting normalized spatial metric will fail to be locally homogeneous and isotropic when the expansion is driven by a strictly positive cosmological constant.

\section{Preliminaries}

\noindent Throughout this work, we adopt the $(n+1)$-dimensional Arnowitt--Deser--Misner (ADM) formalism for the study of globally hyperbolic Lorentzian manifolds $(\widetilde{M},\widehat{g})$ with signature $(-,+,\ldots,+)$, where $n=3$. We fix a Cauchy foliation and the associated splitting
\[
\widetilde{M} \cong \mathbb{R} \times M,
\]
where each leaf $\{t\} \times M$ is an orientable, compact $n$-dimensional Cauchy hypersurface diffeomorphic to $M$. The Lorentzian metric $\widehat{g}$ induces a Riemannian metric $g(t)$ on each slice $M(t) := \{t\}\times M$.

\noindent The ADM decomposition expresses the future-directed timelike unit normal vector field $\widehat{\mathbf{n}}$ orthogonal to the slices and the vector field $\partial_t$ generating the foliation as
\[
\partial_t = N \widehat{\mathbf{n}} + X,
\]
where $N: \widetilde{M} \to (0,\infty)$ is the lapse function, and $X \in \Gamma(T M)$ is the shift vector field, both of which possess the requisite regularity dictated by the function spaces to be specified later. The Lorentzian metric in the adapted coordinate system $\{x^\alpha\}_{\alpha=0}^n = \{t,x^i\}_{i=1}^n$ admits the canonical form
\begin{align}
\label{eq:spacetime_metric}
\widehat{g} = -N^2 dt \otimes dt + g_{ij} (dx^i + X^i dt) \otimes (dx^j + X^j dt),
\end{align}
where $g_{ij}(t,x)$ is the induced Riemannian metric on $M(t)$.

\noindent To quantify the extrinsic geometry of the embedding $M(t) \hookrightarrow \widetilde{M}$, we define the second fundamental form $k \in \Gamma(S^2 T^* M)$ by
\[
k_{ij} := -\langle \nabla_{\partial_i} \widehat{\mathbf{n}}, \partial_j \rangle = -\frac{1}{2N} \bigl(\partial_t g_{ij} - \mathscr{L}_X g_{ij}\bigr),
\]
where $\nabla$ denotes the Levi-Civita connection of $\widehat{g}$, and $\mathscr{L}_X$ is the Lie derivative along $X$. The mean extrinsic curvature, or trace of $k$, is given by
\[
\tau := \mathrm{tr}_g k = g^{ij} k_{ij},
\]
with $g^{ij}$ denoting the inverse metric on $M(t)$.

\noindent The vacuum Einstein field equations with cosmological constant $\Lambda \in \mathbb{R}$,
\begin{align}
\label{eq:einstein_equations}
\mathrm{Ric}(\widehat{g}) - \frac{1}{2} R(\widehat{g}) \widehat{g} + \Lambda \widehat{g} = 0,
\end{align}
admit an equivalent formulation in terms of the Cauchy data $(g,k)$ on each hypersurface $M(t)$ as the coupled evolution and constraint system:
\begin{align}
\label{eq:evol_g}
\partial_t g_{ij} &= -2N k_{ij} + \mathscr{L}_X g_{ij}, \\
\label{eq:evol_k}
\partial_t k_{ij} &= -\nabla_i \nabla_j N + N \left( R_{ij} + \tau k_{ij} - 2 k_i^m k_{mj} - \frac{2 \Lambda}{n-1} g_{ij} \right) + \mathscr{L}_X k_{ij}, \\
\label{eq:ham_constraint}
R - |k|_g^2 + \tau^2 &= 2\Lambda, \\
\label{eq:mom_constraint}
\nabla^i k_{ij} - \nabla_j \tau &= 0,
\end{align}
where $\nabla$ denotes the Levi-Civita connection of $g$, $R_{ij}$ and $R$ are the Ricci and scalar curvatures of $g$, respectively, and $|k|_g^2 = g^{im} g^{jn} k_{ij} k_{mn}$. Here and henceforth, indices are raised and lowered by $g$ and its inverse.

\noindent We denote by $\Gamma[g]$ the Christoffel symbols associated to $g$, and by $\mathrm{Riem}$ and $\mathrm{Ric}$ the corresponding Riemann and Ricci curvature tensors on $M$. The (negative-semidefinite) Laplace--Beltrami operator acting on functions $f \in C^\infty(M)$ is given by
\[
\Delta_g f := g^{ij} \nabla_i \nabla_j f,
\]
where the sign convention ensures that the spectrum of $-\Delta_g$ is nonnegative.

\noindent A central object in our analysis is the Lichnerowicz Laplacian
\[
\Delta_{L} : \Gamma(S^2 T^* M) \to \Gamma(S^2 T^* M),
\]
defined by
\[
(\Delta_{L} h)_{ij}
:=
{-g^{k\ell}\nabla_k\nabla_\ell h_{ij}}
{-2R_{ikj\ell}h^{k\ell}}
{+R_i{}^{k}h_{kj}}
{+R_j{}^{k}h_{ik}}.
\]
The operator $\Delta_L$ is elliptic. At an Einstein background $\operatorname{Ric}(g)=\lambda_E g$, strict Einstein stability concerns the positivity of $\Delta_L-2\lambda_E=\nabla^*\nabla-2\operatorname{Riem}$ on transverse-traceless tensors.

\noindent Finally, for nonnegative scalar functions $f,h : [0,\infty) \to \mathbb{R}_{\geq 0}$, we write $f(t) \lesssim h(t)$ to signify the existence of a constant $C > 0$, independent of $t$ but possibly dependent on the fixed background geometry and dimension $n$, such that
\[
f(t) \leq C h(t),
\]
and $f(t) \approx h(t)$ if there exist constants $0 < C_1 \leq C_2 < \infty$ such that
\[
C_1 h(t) \leq f(t) \leq C_2 h(t).
\]

\noindent Throughout, the spaces of smooth symmetric covariant 2-tensors and vector fields on $M$ are denoted by $S^0_2(M)$ and $\mathfrak{X}(M)$, respectively. The trace-free and transverse-traceless parts of a symmetric 2-tensor $A$ are denoted by $\widehat{A}$ and $A^{TT}$, respectively. We will also frequently denote the Cauchy slice at constant mean curvature time $T$ by $M(T)$. In addition, we frequently use the following transport identity, which is a consequence of the transport theorem \cite{marsden}:
\begin{align}
 \frac{d}{dT}\int_{M(T)}f\mu_{g}=\int_{M(T)}(\partial_{T}f+\frac{1}{2}f\tr_{g}\partial_{T}g)\mu_{g}.   
\end{align}

\section{Gauge-fixed Einstein equations}
\label{mainsection}
\noindent We study the Cauchy problem in constant mean extrinsic curvature (CMC) gauge. To this end, we assume that the Lorentzian spacetime diffeomorphic to $M\times \mathbb{R}$ admits a constant mean curvature Cauchy slice. 

\noindent Now recall the vacuum Einstein field equations with a cosmological constant $\Lambda$
\begin{align}
\label{eq:EFET}
 \text{Ric}(\widehat{g})_{\mu\nu}-\frac{1}{2}R(\widehat{g})\widehat{g}_{\mu\nu}+\Lambda\widehat{g}_{\mu\nu}=0,
\end{align}
expressed in the $3+1$ formulation in a local chart $(t,x^{1},x^{2},x^{3})$:
\begin{align}
\partial_{t}g_{ij}=-2Nk_{ij}+(\mathscr{L}_{X}g)_{ij},\\
\partial_{t}k_{ij}=-\nabla_{i}\nabla_{j}N+N(R_{ij}+\tr_{g}k k_{ij}-2k^{k}_{i}k_{jk}-\frac{2\Lambda}{n-1}g_{ij})+(\mathscr{L}_{X}k)_{ij},\\
2\Lambda=R(g)-|k|^{2}_{g}+(\tr_{g}k)^{2},\\
0=\nabla^{i}k_{ij}-\nabla_{j}\tr_{g}k.   
\end{align}
This evolution-constraint system leads to the following initial value problem:
\begin{definition}
Initial data for \eqref{eq:EFET} consist of an $n$-dimensional manifold $M$, a Riemannian metric $g$, a symmetric covariant 2-tensor $k$, all assumed to be smooth and to satisfy 
\begin{align}
\label{eq:constraint1}
2\Lambda=R(g)-|k|^{2}_{g}+(\tr_{g}k)^{2},\\
\label{eq:constraint2}
0=\nabla^{i}k_{ij}-\nabla_{j}\tr_{g}k.       
\end{align}
Given a set of initial data, the corresponding initial value problem consists of constructing an $(n+1)$-dimensional Lorentzian manifold $(\widetilde{M},\widehat{g})$ satisfying the spacetime Einstein equations \eqref{eq:EFET}, together with an embedding $\iota:M\hookrightarrow \widetilde{M}$
such that $\iota(M)$ is a Cauchy hypersurface in $(\widetilde{M},\widehat{g})$. The embedding is required to satisfy $\iota^{*}\widehat{g}=g$, where $g$ is the prescribed Riemannian metric on $M$, and, denoting by $\widehat{\mathbf{n}}$ the future-directed unit normal vector field to $\iota(M)$ and by $\kappa$ its second fundamental form in $(\widetilde{M},\widehat{g})$, one has $\iota^{*}\kappa=k$, where $k$ is the given symmetric 2-tensor on $M$. The spacetime $(\widetilde{M},\widehat{g})$ constructed in this manner is referred to as a globally hyperbolic development of the initial data.
\end{definition}

\noindent Now we define the CMC gauge. First, recall the following definition:
\begin{definition}[Constant Mean Curvature (CMC) Gauge] Let \((\widetilde{M}, \widehat{g})\) be a globally hyperbolic Lorentzian manifold of dimension \( (3+1) \), and let \(\{M_t\}_{t \in I} \) denote a foliation of \(\widetilde{M}\) by spacelike Cauchy hypersurfaces defined as level sets of a smooth time function \(t: \widetilde{M} \to \mathbb{R}\). The foliation is said to be in \emph{constant mean curvature gauge} (CMC gauge) if, for each \(t \in I\), the hypersurface \(M_t := \{p \in \widetilde{M} \mid t(p) = t\}\) has constant mean curvature \(\tau(t)\), i.e., the trace of the second fundamental form \(k_t\) of \(M_t\) satisfies \[ \mathrm{tr}_{g_t} k_t = \tau(t),\]
where $g_{t}$ denotes the induced Riemannian metric on $M_{t}$. When $\partial_t\tau>0$, the mean curvature may itself be used as the future time coordinate.
\end{definition}
\noindent The CMC condition gives $\nabla\tau=0$ on each slice. Thus the mean curvature satisfies
\begin{align}
\label{eq:cmcg}
 \tau=\tr_{g}k=\tau(t).  
\end{align}
On the strict expanding branch of Theorem~\ref{main}, we choose
\begin{align}
    t=\tau.
\end{align}
The initial hypersurface in Theorem~\ref{main} is prescribed by compact CMC
constraint data. Theorem~\ref{thm:local-continuation} supplies its local CMC
development. The crushing-singularity criterion of
Marsden and Tipler~\cite[Theorem~6]{marsden1980maximal} assumes timelike
convergence. In the present vacuum equations a unit timelike vector $v$
satisfies $\operatorname{Ric}(\widehat g)(v,v)=-\Lambda$.

\begin{claim}
\label{important}
On a compact CMC slice of an Einstein--$\Lambda$ spacetime with
$\Lambda>0$ and $\sigma(M)\leq0$, one has $\tau^2\geq3\Lambda$.
In particular, an initially expanding CMC foliation has $\tau<0$ throughout.
If $M$ is of negative Yamabe type, the inequality is strict. On the more
general class $\sigma(M)\leq0$, Theorem~\ref{main} starts on the strict
expanding branch $\tau< -\sqrt{3\Lambda}$.
\end{claim}

\begin{proof}
Set $A=k-\tau g/3$. The CMC momentum constraint gives
$\operatorname{div}_g A=0$, and $\operatorname{tr}_g A=0$ by definition.
The Hamiltonian constraint becomes
\[
 R(g)=|A|_g^2+2\Lambda-\frac{2}{3}\tau^2 .
\]
If $\tau^2<3\Lambda$, this gives $R(g)>0$ everywhere, contrary to
$\sigma(M)\leq0$. Continuity then preserves the sign of $\tau$ along the
foliation. At equality $R(g)=|A|_g^2\geq0$. If it were positive anywhere,
the first eigenvalue of $-8\Delta_g+R(g)$ on the connected closed slice
would be strictly positive, and its positive eigenfunction would yield a
conformal metric of positive scalar curvature. Thus equality requires
$A=0$ and $R(g)=0$. It is excluded on a manifold of negative Yamabe type.
On such an equality slice the physical CMC lapse operator is $-\Delta_g$,
so a lapse for the coordinate $t=\tau$ would have to solve
$-\Delta_g\widetilde N=1$, which is impossible on a closed slice.
The strict branch is the domain of the rescaling below.
\end{proof}
\noindent On the strict expanding branch fix the signed scale
\[
 \varphi=-\sqrt{\tau^2-3\Lambda},\qquad \varphi/\tau>0 .
\]
Its absolute value decreases as $\tau$ increases.

\noindent We adopt the dimensional convention of \cite{andersson2011einstein}, wherein the local spatial coordinates $\{x^i\}$ on the Cauchy hypersurface $M$ are taken to be dimensionless. Under this convention, the geometric and gauge quantities possess the following physical dimensions:
\begin{align} \label{eq:physical-dimensions}
[g_{ij}] = L^2, \quad [k_{ij}] = L, \quad [X^i] = L, \quad [N] = L^2, \quad [\tau] = L^{-1}, \quad [\varphi^2] = L^{-2},
\end{align}
where $L$ denotes the dimension of length. 

\noindent To facilitate the introduction of dimensionless variables, we adopt the notational convention that dimensional quantities are denoted with a tilde, while their dimensionless counterparts are written without decoration. The natural rescaling is then given by
\begin{align} \label{eq:scaling}
\widetilde{g}_{ij} = \frac{1}{\varphi^2} g_{ij}, \qquad \widetilde{N} = \frac{1}{\varphi^2} N, \qquad \widetilde{X}^i = \frac{1}{\varphi} X^i, \qquad \widetilde{k}_{ij}^{TT} = \frac{1}{\varphi} \Sigma_{ij}.
\end{align}
Here $\widetilde k^{TT}=\widetilde k-\tau\widetilde g/3$ and $\Sigma=\varphi\widetilde k^{TT}$.

\noindent We now introduce a reparametrization of time via a new variable $T = T(\tau)$, defined implicitly by the vector field transformation
\begin{align} \label{eq:timechange}
\partial_T = -\frac{\varphi^2}{\tau} \partial_\tau.
\end{align}
Integrating this relation yields
\begin{align} \label{eq:ex1}
\varphi(T) = \varphi(T_0) e^{T_0} e^{-T}.
\end{align}
Fix the origin of $T$ once for all initial slices by $\varphi(T)e^T=-C$, with a fixed reference scale $C>0$ of dimension $L^{-1}$. Thus the initial time $a=\log(C/|\varphi(\tau_0)|)$ is determined by the physical CMC value. In this normalization,
\begin{align} \label{eq:ex2}
\varphi(T) = -Ce^{-T}, \qquad \tau(T) = -\sqrt{C^{2}e^{-2T} + 3\Lambda}.
\end{align}
The change of parameter maps $-\infty<\tau<-\sqrt{3\Lambda}$ onto $\mathbb R$. A development issued from $T_0$ occupies the portion of this range on which it exists. 

\noindent Moreover, we record the asymptotic estimate
\begin{align} \label{eq:asymptotic}
0<\frac{\varphi(T)}{\tau(T)} = \frac{Ce^{-T}}{\sqrt{C^{2}e^{-2T}+3\Lambda}} \leq \frac{C}{\sqrt{3\Lambda}}e^{-T},
\end{align}
demonstrating the exponential decay of the ratio $\varphi/\tau$ in the future time direction.

\noindent The rescaled Einstein equations in the CMC gauge can be cast into the following system:
\begin{align}
\label{eq:cmc1}
 \frac{\partial g_{ij}}{\partial T}={\frac{2\varphi}{\tau}N\Sigma_{ij}}-2(1-\frac{N}{n})g_{ij}-\frac{\varphi}{\tau}(\mathscr{L}_{X}g)_{ij}\\
 \frac{\partial \Sigma_{ij}}{\partial T}=-(n-1)\Sigma_{ij}-\frac{\varphi}{\tau}N(\text{Ric}_{ij}+\frac{n-1}{n^{2}}g_{ij})+{\frac{n\varphi}{\tau}\nabla_{i}\nabla_{j}(\frac{N}{n}-1)}+\frac{2\varphi}{\tau}N\Sigma_{ik}\Sigma^{k}_{j}\\
 -\frac{\varphi}{n\tau}(\frac{N}{n}-1)g_{ij}-(n-2)(\frac{N}{n}-1)\Sigma_{ij}-\frac{\varphi}{\tau}(\mathscr{L}_{X}\Sigma)_{ij}\\
 \label{eq:HC11}
 R+\frac{n-1}{n}-|\Sigma|^{2}=0,\\
 \label{eq:cmc2}
 \nabla_{j}\Sigma^{j}_{i}=0.
\end{align}
For the estimates below, we treat the metric equation as a transport equation in the chosen spatial gauge and obtain the hyperbolic estimates from the second fundamental form and the trace-free spatial Ricci tensor. We define
\begin{align}
    \mathfrak{T}_{ij}:=\text{Ric}_{ij}-\frac{1}{3}R(g)g_{ij}.
\end{align}
We obtain the following wave system for $(\Sigma,\mathfrak{T})$.
\begin{remark}
 A related hyperbolic formulation based on the second fundamental form was developed by Abrahams, Anderson, Choquet-Bruhat and York \cite{YCBY}.
\end{remark}
\begin{lemma}[Coupled Wave System for \((\Sigma, \mathfrak{T})\)]
\label{evolutionsystem}
Let \( (g, \Sigma) \) satisfy the reduced Einstein evolution equations \eqref{eq:cmc1}--\eqref{eq:cmc2} in CMC time with rescaled shift $X$ and with lapse determined by
\eqref{eq:lapse3}. The principal evolution terms for $(\Sigma,\mathfrak T)$
are $-cN\mathfrak T$ and $-cN\Delta_g\Sigma$, respectively, where
$c=\varphi/\tau$. The transport derivative in these equations is
$\partial_T+c\mathscr L_X$ and reduces to $\partial_T$ in transported spatial coordinates.

\noindent The evolution equations for \( \Sigma_{ij} \) and \( \mathfrak{T}_{ij} \) read:
\begingroup
\begin{align}
\partial_T\Sigma_{ij}={}&-\frac{\varphi}{\tau}(\mathscr L_X\Sigma)_{ij}-(n-1)\Sigma_{ij}
-\frac{\varphi}{\tau}N\mathfrak T_{ij}
+\frac{n\varphi}{\tau}\nabla_i\nabla_j\left(\frac Nn-1\right)
+\frac{2\varphi}{\tau}N\Sigma_{ik}\Sigma^k{}_j\notag\\
&-\frac{\varphi}{n\tau}\left(\frac Nn-1+N|\Sigma|_g^2\right)g_{ij}
-(n-2)\left(\frac Nn-1\right)\Sigma_{ij}.
\label{eq:sigma_evol}
\end{align}
\endgroup
\begingroup
\begin{align}
\partial_T\mathfrak T_{ij}={}&-\frac{\varphi}{\tau}(\mathscr L_X\mathfrak T)_{ij}
-N\frac{\varphi}{\tau}\left(\Delta_g\Sigma_{ij}
+2\operatorname{Riem}_{ikjl}\Sigma^{kl}\right)\notag\\
&+n\frac{\varphi}{\tau}\nabla^l\nabla_i\left(\frac Nn-1\right)\Sigma_{jl}
+n\frac{\varphi}{\tau}\nabla^l\left(\frac Nn-1\right)\nabla_i\Sigma_{jl}\notag\\
&+n\frac{\varphi}{\tau}\nabla^l\nabla_j\left(\frac Nn-1\right)\Sigma_{il}
+n\frac{\varphi}{\tau}\nabla^l\left(\frac Nn-1\right)\nabla_j\Sigma_{il}\notag\\
&-n\frac{\varphi}{\tau}\Delta_g\left(\frac Nn-1\right)\Sigma_{ij}
-2n\frac{\varphi}{\tau}\nabla^l\left(\frac Nn-1\right)\nabla_l\Sigma_{ij}\notag\\
&+N\frac{\varphi}{\tau}\left(\mathfrak T_{ki}\Sigma^k{}_j+\mathfrak T_{kj}\Sigma^k{}_i\right)
-(n-2)\nabla_i\nabla_j\left(\frac Nn-1\right)\notag\\
&+\left[\frac{n-2}{n}\Delta_g\left(\frac Nn-1\right)
+\frac{2\varphi N}{n\tau}\langle\mathfrak T,\Sigma\rangle_g
-\frac{2\varphi}{\tau}\left\langle\Sigma,\nabla^2\left(\frac Nn-1\right)\right\rangle_g\right]g_{ij}.
\label{eq:T_evol}
\end{align}
\endgroup
\end{lemma}

\begin{proof}
 \begingroup
The rescaling and \(\partial_T\varphi=-\varphi\) give
\[
\left(\partial_T+\frac{\varphi}{\tau}\mathscr L_X\right)g
=2\frac{\varphi}{\tau}N\Sigma+2\left(\frac Nn-1\right)g.
\]
Throughout the system \(n=3\). The constraint gives
\[
\operatorname{Ric}+\frac{n-1}{n^2}g
=\mathfrak T+\frac1n|\Sigma|_g^2g,
\]
which yields \eqref{eq:sigma_evol}. To obtain \eqref{eq:T_evol}, use the Ricci variation formula \eqref{eq:DRic-general} and
\[
\begin{split}
\left(\partial_T+\frac{\varphi}{\tau}\mathscr L_X\right)R
={}&-2\frac{\varphi}{\tau}N\langle\mathfrak T,\Sigma\rangle_g
+2n\frac{\varphi}{\tau}\left\langle\Sigma,\nabla^2\left(\frac Nn-1\right)\right\rangle_g\\
&-2\left(\frac Nn-1\right)R-2(n-1)\Delta_g\left(\frac Nn-1\right).
\end{split}
\]
Subtracting the derivative of \(Rg/n\) gives the displayed scalar multiple of \(g\). In particular,
\[
\operatorname{tr}_g\left(\partial_T+\frac{\varphi}{\tau}\mathscr L_X\right)\mathfrak T
=2\frac{\varphi}{\tau}N\langle\mathfrak T,\Sigma\rangle_g,
\]
which cancels the derivative of the inverse metric in \(\partial_T(\operatorname{tr}_g\mathfrak T)\). Thus the system preserves the defining trace-free identity.
\endgroup 

\end{proof}

\noindent The pair $(\Sigma,\mathfrak{T})$ has a wave principal part, with transport derivative $\partial_T+(\varphi/\tau)\mathscr L_X$. Choosing transported spatial coordinates below sets $X=0$ and therefore avoids having to invert a spatial-harmonic shift operator for the large initial data class.      

\noindent We have chosen the CMC gauge as the time gauge. We need to choose a spatial gauge to analyze the evolution system in Lemma~\ref{evolutionsystem}. We choose the CMC transported spatial gauge (previously \cite{fajman2022cosmic,rodnianski2018regime} used this gauge to study the dynamical stability of big-bang singularity formation).
\subsection{Transported Spatial Coordinates and the Re-Scaled Einstein System}
\noindent In this section, we introduce the notion of a \emph{transported spatial gauge}, which serves as a coordinate choice adapted to the constant mean curvature (CMC) foliation. 

\begin{definition}[Transported Spatial Gauge]
\label{spatial}
Let $(x^1_0, x^2_0, x^3_0)$ be local coordinates on an open set $\Omega \subset M$, identified with a coordinate neighborhood in the initial Cauchy hypersurface $\{T' = T_0\} \subset \widetilde{M} := \mathbb{R} \times M$. We say that a coordinate system $(T', x^1, x^2, x^3)$ on $[T_0, T] \times \Omega$ defines a \emph{transported spatial gauge} if the spatial coordinate functions are Lie transported along the future-directed unit normal vector field $\widehat{\mathbf{n}}$ to the CMC hypersurfaces, i.e.,
\begin{align}
\label{eq:transport}
\mathscr{L}_{\widehat{\mathbf{n}}}(x^i) = 0, \qquad i = 1,2,3,
\end{align}
with initial condition $(x^1, x^2, x^3)|_{T'=T_0} = (x^1_0, x^2_0, x^3_0)$.
\end{definition}

\noindent
From \eqref{eq:timechange} and \eqref{eq:scaling}, the physical lapse
and shift for the time coordinate $T'=T$ are
$\alpha_T=(d\tau/dT)\widetilde N=-N/\tau>0$ and
$\beta_T=(d\tau/dT)\widetilde X=-(\varphi/\tau)X$. Hence
$\widehat{\mathbf n}=\alpha_T^{-1}(\partial_T-\beta_T^j\partial_j)$ and
\begin{align}
\label{eq:shift-vanishing}
0=\widehat{\mathbf n}(x^i)=-\alpha_T^{-1}\beta_T^i=\alpha_T^{-1}(\varphi/\tau)X^i,\qquad X^i=0.
\end{align}
Thus, the shift vector field vanishes identically in these coordinates.

\vspace{1em}
\noindent
In this gauge, the rescaled Einstein evolution equations take a particularly tractable form. Let $\Sigma_{ij}$ denote the trace-free part of the second fundamental form (the shear), and let $\mathfrak{T}_{ij}$ denote the trace-free part of the rescaled Ricci curvature. The evolution equations then read:

\begingroup
\begin{align}
\partial_T\Sigma_{ij}={}&-(n-1)\Sigma_{ij}
-\frac{\varphi}{\tau}N\mathfrak T_{ij}
+\frac{n\varphi}{\tau}\nabla_i\nabla_j\left(\frac Nn-1\right)
+\frac{2\varphi}{\tau}N\Sigma_{ik}\Sigma^k{}_j\notag\\
&-\frac{\varphi}{n\tau}\left(\frac Nn-1+N|\Sigma|_g^2\right)g_{ij}
-(n-2)\left(\frac Nn-1\right)\Sigma_{ij}.
\label{eq:curvature1}
\end{align}
\endgroup

\begingroup
\begin{align}
\partial_T\mathfrak T_{ij}={}&
-N\frac{\varphi}{\tau}\left(\Delta_g\Sigma_{ij}
+2\operatorname{Riem}_{ikjl}\Sigma^{kl}\right)\notag\\
&+n\frac{\varphi}{\tau}\nabla^l\nabla_i\left(\frac Nn-1\right)\Sigma_{jl}
+n\frac{\varphi}{\tau}\nabla^l\left(\frac Nn-1\right)\nabla_i\Sigma_{jl}\notag\\
&+n\frac{\varphi}{\tau}\nabla^l\nabla_j\left(\frac Nn-1\right)\Sigma_{il}
+n\frac{\varphi}{\tau}\nabla^l\left(\frac Nn-1\right)\nabla_j\Sigma_{il}\notag\\
&-n\frac{\varphi}{\tau}\Delta_g\left(\frac Nn-1\right)\Sigma_{ij}
-2n\frac{\varphi}{\tau}\nabla^l\left(\frac Nn-1\right)\nabla_l\Sigma_{ij}\notag\\
&+N\frac{\varphi}{\tau}\left(\mathfrak T_{ki}\Sigma^k{}_j+\mathfrak T_{kj}\Sigma^k{}_i\right)
-(n-2)\nabla_i\nabla_j\left(\frac Nn-1\right)\notag\\
&+\left[\frac{n-2}{n}\Delta_g\left(\frac Nn-1\right)
+\frac{2\varphi N}{n\tau}\langle\mathfrak T,\Sigma\rangle_g
-\frac{2\varphi}{\tau}\left\langle\Sigma,\nabla^2\left(\frac Nn-1\right)\right\rangle_g\right]g_{ij}.
\label{eq:curvature2}
\end{align}
\endgroup

\noindent
Tracing \eqref{eq:evol_k}, including the variation of the inverse
metric, gives the physical CMC lapse equation
$-\Delta_{\widetilde g}\widetilde N+
(|\widetilde k|^2_{\widetilde g}-\Lambda)\widetilde N=1$.
Since $|\widetilde k|^2_{\widetilde g}-\Lambda
=\varphi^2(|\Sigma|_g^2+1/3)$ and $n=3$, its rescaling gives
\begin{align}
\label{eq:lapse3}
- \Delta_g \left( \frac{N}{n} - 1 \right) + \left( |\Sigma|^2 + \frac{1}{3} \right) \left( \frac{N}{n} - 1 \right) = - |\Sigma|^2.
\end{align}

\subsection{Integration and Norms}

\noindent Let \( (M, g) \) be a smooth, closed, oriented Riemannian $3$-manifold. Fix a smooth partition of unity \( \{ \xi_U \}_{U \in \mathcal{U}} \) subordinate to a finite atlas \( \{ (U, \varphi_U) \}_{U \in \mathcal{U}} \) of coordinate charts \( \varphi_U: U \to \mathbb{R}^3 \). For an integrable function \( f: M \to \mathbb{R} \), the integral of \( f \) with respect to the Riemannian volume form \( \mu_g \) is defined by

\begin{align} 
\label{eq:integration}
\int_M f \, \mu_g := \sum_{U \in \mathcal{U}} \int_{\varphi_U(U)} f \circ \varphi_U^{-1}(x) \, \xi_U \circ \varphi_U^{-1}(x) \sqrt{\det g_{ij}(x)} \, dx^1 dx^2 dx^3,
\end{align}

\noindent where \( g_{ij} \) denotes the components of the metric in the chart \( \varphi_U \), and \( \det g_{ij} \) is the determinant of the metric matrix.

\noindent Given a smooth tensor field \( \psi \in \Gamma\left( (TM)^{\otimes K} \otimes (T^* M)^{\otimes L} \right) \), where \( K + L \geq 1 \), we define its squared point-wise norm using the Riemannian metric \( g \) by

\begin{align} \label{eq:pointwise-norm}
\langle \psi(x), \psi(x) \rangle_g := g^{i_1 j_1} \cdots g^{i_L j_L} g_{m_1 n_1} \cdots g_{m_K n_K} \,
\psi^{m_1 \dots m_K}_{i_1 \dots i_L}(x) \, \psi^{n_1 \dots n_K}_{j_1 \dots j_L}(x).
\end{align}

\noindent This inner product is well-defined and independent of the choice of coordinates. The \( L^p \) norm of \( \psi \) over \( (M, g) \) is then defined by

\begin{align} \label{eq:Lp-norm}
\| \psi \|_{L^p(M)}^p := \int_M \left( \langle \psi, \psi \rangle_g \right)^{p/2} \, \mu_g, \qquad 1 \leq p < \infty,
\end{align}

\noindent and the essential supremum (or \( L^\infty \)) norm is given by

\begin{align} \label{eq:Linf-norm}
\| \psi \|_{L^\infty(M)} := \sup_{x \in M} \left( \langle \psi(x), \psi(x) \rangle_g \right)^{1/2}.
\end{align}

\noindent We now define the energy norms used in the analysis of the Einstein$-\Lambda$ flow. Let \( \Sigma \) denote the trace-free part of the second fundamental form and \( \mathfrak{T} \) the trace-free spatial Ricci tensor. For integers \( I \in \{0, 1, 2, 3\} \), define:

\begin{align} \label{eq:OI}
\mathcal{O}_I := 
\begin{cases}
\| \nabla^I \Sigma \|_{L^2(M)} + \| \nabla^{I-1} \mathfrak{T} \|_{L^2(M)}, & \text{if } I \geq 1, \\
\| \Sigma \|_{L^2(M)}, & \text{if } I = 0.
\end{cases}
\end{align}

\noindent We also introduce exponentially weighted norms in time to account for the expanding geometry in cosmological spacetimes. Let \( T \) denote the CMC time parameter, and let \( 1\geq \gamma > 0 \) be a scaling parameter (which will ultimately be chosen to be \( \gamma = 1 \)). For \( I \leq 2 \), define the weighted Sobolev norms

\begin{align} \label{eq:FI}
\mathcal{F}_I := \| e^{\gamma T} \nabla^I \Sigma \|_{L^2(M)}.
\end{align}

\noindent Define also the pointwise weighted norm:
\begin{align} \label{eq:infty-norm}
\mathcal{N}^\infty := \| e^{\gamma T} \Sigma \|_{L^\infty(M)} + \left\| e^{2\gamma T} \left( \frac{N}{n} - 1 \right) \right\|_{L^\infty(M)},
\end{align}
where \(N\) is the normalized lapse defined in \eqref{eq:scaling}.

\noindent The total energy norms used throughout this work are defined by

\begin{align} \label{eq:O}
\mathcal{O} := \sum_{I=0}^3 \mathcal{O}_I, \qquad \mathcal{F} := \sum_{I=0}^2 \mathcal{F}_I.
\end{align}
In what follows we take $\gamma=1$.

\section{Construction of the initial data}
\label{data}
\noindent Fix $0<\alpha<1$ and take $a$ sufficiently large. We prescribe the size of the top-order norms to be comparable to $e^{(1-\alpha)a}$. At the end we verify the hypothesis of Theorem~\ref{main}, with $\alpha/2$ in place of $\alpha$, on the same initial slice. The goal of this section is to explicitly construct such data that verify the Einstein-$\Lambda$ constraint equations. In addition, we show that such data can have arbitrarily large normalized $H^4\times H^3$ size relative to the chosen background, whereas the perturbative results in \cite{fajman2020stable} assume small high-order data near an Einstein background. We work in a conformal class of negative Yamabe sign and apply the transverse-traceless perturbation and conformal technique. First, we construct Riemannian metrics on $M$ that verify the condition 
\[
\sum_{I = 0}^{2}  \| \nabla^I \mathfrak{T}[g_1] \|_{L^2(M)}  \leq \mathcal{I}^0=O(e^{(1-\alpha)a}),~a\gg 1.
\]
The second step is to solve the momentum constraint in the CMC gauge, i.e., construct an appropriate $g_{1}$-$TT$ tensor verifying the estimate 
\[
\|\Sigma\|_{H^{2}}\lesssim e^{-a}\mathcal I^{0},
\qquad
\|\Sigma\|_{H^{3}}\lesssim \mathcal I^{0}.
\]
In the third step, we solve the Hamiltonian constraint using the conformal method and prove that the estimates proved for $\mathfrak{T}$ and $\Sigma$ are modified by a negligible amount in the conformal transformation process. The main proposition that we prove is the existence of a relatively open set of normalized constraint data $(g_{0},\Sigma_{0})$ satisfying the estimates 
\[
\sum_{I = 0}^{3} \| \nabla^I \Sigma_0 \|_{L^2(M)} + \sum_{I = 0}^{2} \left( \| \nabla^I \mathfrak{T}[g_0] \|_{L^2(M)} + \| e^a \nabla^I \Sigma_0 \|_{L^2(M)} \right) \leq \mathcal{I}^0.
\]

\noindent First, we recall the well-known lemma relating the trace-free Ricci curvature of $\gamma$ and the perturbed metric $g=\gamma+\epsilon h$ for $\epsilon\ll1$.

\begin{lemma}
\label{nonlinearexpansion}
Let $(M^{n},\gamma)$, $n\geq 2$, be a smooth closed connected Riemannian manifold, and let
$h$ be a smooth symmetric covariant two-tensor on $M$. For $\epsilon\in\mathbb R^{+}$
with $\epsilon$ sufficiently small, define
\[
g_\epsilon:=\gamma+\epsilon h.
\]
Then $g_\epsilon$ is a smooth Riemannian metric, and the following assertions hold. First, the inverse metric expression reads
\begin{equation}\label{eq:inverse-expansion}
(g_{\epsilon})^{ij}
=
\gamma^{ij}
-
\epsilon h^{ij}
+
\epsilon^{2}O_{\gamma}(h*h).
\end{equation}
\medskip
\noindent Now let $\nabla=\nabla^\gamma$ denote the Levi--Civita connection of $\gamma$, and let
\[
A_\epsilon{}^{k}{}_{ij}
:=
\Gamma(g_\epsilon)^{k}{}_{ij}-\Gamma(\gamma)^{k}{}_{ij}.
\]
Then
\begin{equation}
\label{eq:A-exact}
A_\epsilon{}^{k}{}_{ij}
=
\frac12 (g_\epsilon)^{k\ell}
\big(
\nabla_i(g_\epsilon)_{j\ell}
+\nabla_j(g_\epsilon)_{i\ell}
-\nabla_\ell(g_\epsilon)_{ij}
\big)
=
\frac{\epsilon}{2}(g_\epsilon)^{k\ell}
\big(
\nabla_i h_{j\ell}
+\nabla_j h_{i\ell}
-\nabla_\ell h_{ij}
\big).
\end{equation}
Moreover,
\begin{equation}
\label{eq:Ric-exact}
\text{Ric}(g_\epsilon)_{ij}
=
\text{Ric}(\gamma)_{ij}
+\nabla_k A_\epsilon{}^{k}{}_{ij}
-\nabla_j A_\epsilon{}^{k}{}_{ik}
+A_\epsilon{}^{k}{}_{ij}A_\epsilon{}^{\ell}{}_{k\ell}
-A_\epsilon{}^{k}{}_{i\ell}A_\epsilon{}^{\ell}{}_{jk},
\end{equation}
and hence
\begin{equation}
\label{eq:R-exact}
R[g_\epsilon]
=
(g_\epsilon)^{ij}
\Big(
\text{Ric}(\gamma)_{ij}
+\nabla_k A_\epsilon{}^{k}{}_{ij}
-\nabla_j A_\epsilon{}^{k}{}_{ik}
+A_\epsilon{}^{k}{}_{ij}A_\epsilon{}^{\ell}{}_{k\ell}
-A_\epsilon{}^{k}{}_{i\ell}A_\epsilon{}^{\ell}{}_{jk}
\Big).
\end{equation}
\noindent The scalar curvature admits the expansion
\begin{equation}
\label{eq:R-linearization-expansion}
R[g_\epsilon]
=
R[\gamma]+\epsilon\,DR_\gamma[h]+\epsilon^2 Q^R_\gamma(h;\epsilon),
\end{equation}
where
\begin{equation}
\label{eq:DR}
DR_\gamma[h]
=
\nabla^i\nabla^j h_{ij}
-\Delta_\gamma(\tr_\gamma h)
-\langle \text{Ric}(\gamma),h\rangle_\gamma.
\end{equation}
Here $\Delta_\gamma=\gamma^{ij}\nabla_i\nabla_j$ is the rough Laplacian on functions. The remainder $Q^R_\gamma(h;\epsilon)$ is smooth in $(x,\epsilon)$ and, in schematic notation,
\begin{equation}
\label{eq:QR-schematic}
Q^R_\gamma(h;\epsilon)
=
h*\nabla^2 h+\nabla h * \nabla h+\text{Rm}(\gamma)*h*h.
\end{equation}
The coefficients in these schematic products depend smoothly on $g_\epsilon^{-1}$ and the fixed background.
\noindent Now define
\[
\mathfrak T[g]
:=\text{Ric}(g)-\frac1n R[g]\,g.
\]
Then
\begin{equation}
\label{eq:T-expansion}
\mathfrak T[g_\epsilon]
=
\mathfrak T[\gamma]
+\epsilon\,D\mathfrak T_\gamma[h]
+\epsilon^2 Q^{\mathfrak T}_\gamma(h;\epsilon),
\end{equation}
where
\begin{equation}
\label{eq:DT-general}
D\mathfrak T_\gamma[h]
=
D\text{Ric}_\gamma[h]
-\frac1n\,DR_\gamma[h]\,\gamma
-\frac1n\,R[\gamma]\,h,
\end{equation}
and
\begin{align}
\label{eq:DRic-general}
D\text{Ric}_\gamma[h]_{ij}
&=
\frac12\Big(
-\Delta_\gamma h_{ij}
-\nabla_i\nabla_j(\tr_\gamma h)
+\nabla_i(\delta_\gamma h)_j
+\nabla_j(\delta_\gamma h)_i
\Big)\notag\\
&\quad-\text{Rm}(\gamma)_{ikj\ell}h^{k\ell}
+\frac12\Big(\text{Ric}(\gamma)_{ik}h^{k}{}_{j}+\text{Ric}(\gamma)_{jk}h^{k}{}_{i}\Big).
\end{align}
Here
\[
(\delta_\gamma h)_j:=\nabla^i h_{ij}.
\]
The quadratic remainder $Q^{\mathfrak T}_\gamma(h;\epsilon)$ is smooth in $(x,\epsilon)$ and, 
in particular,
\begin{equation}
\label{eq:QT-schematic}
Q^{\mathfrak T}_\gamma(h;\epsilon)
=
h*\nabla^2 h+\nabla h * \nabla h+\text{Rm}(\gamma)*h*h.
\end{equation}

\medskip
\noindent Assume in addition that
\[
\tr_\gamma h=0,
\qquad
\delta_\gamma h=0.
\]
Then
\begin{equation}
\label{eq:DR-TT}
DR_\gamma[h]=-\langle \text{Ric}(\gamma),h\rangle_\gamma
=-\langle \mathfrak T[\gamma],h\rangle_\gamma,
\end{equation}
and
\begin{equation}
\label{eq:DRic-TT}
D\text{Ric}_\gamma[h]
=
\frac12\,\Delta_{L,\gamma}h,
\end{equation}
where the Lichnerowicz Laplacian on symmetric $(0,2)$-tensors is defined by
\begin{equation}
\label{eq:Lich-def}
(\Delta_{L,\gamma}h)_{ij}
:=
-\Delta_\gamma h_{ij}
-2\,\text{Rm}(\gamma)_{ikj\ell}h^{k\ell}
+\text{Ric}(\gamma)_{ik}h^{k}{}_{j}
+\text{Ric}(\gamma)_{jk}h^{k}{}_{i}.
\end{equation}
Consequently,
\begin{equation}
\label{eq:DT-TT}
D\mathfrak T_\gamma[h]_{ij}
=
\frac12(\Delta_{L,\gamma}h)_{ij}
+\frac1n\langle \mathfrak T[\gamma],h\rangle_\gamma\,\gamma_{ij}
-\frac1n R[\gamma]\,h_{ij}.
\end{equation}
\end{lemma}

\begin{proof}
For $\epsilon$ sufficiently small, the bilinear form $g_\epsilon=\gamma+\epsilon h$ remains positive
definite, hence defines a smooth Riemannian metric. Since $M$ is compact and $h$ is fixed, all
smallness thresholds and constants below depend only on finitely many norms of $(\gamma,h)$.

\noindent We write $\nabla=\nabla^\gamma$. Since $\nabla\gamma=0$ and $(g_\epsilon)_{ij}=\gamma_{ij}+\epsilon h_{ij}$,
the difference tensor between the Levi--Civita connections of $g_\epsilon$ and $\gamma$ is given by the
standard formula
\[
A_\epsilon{}^{k}{}_{ij}
=
\Gamma(g_\epsilon)^{k}{}_{ij}-\Gamma(\gamma)^{k}{}_{ij}
=
\frac12 (g_\epsilon)^{k\ell}
\big(
\nabla_i(g_\epsilon)_{j\ell}
+\nabla_j(g_\epsilon)_{i\ell}
-\nabla_\ell(g_\epsilon)_{ij}
\big),
\]
which immediately yields \eqref{eq:A-exact}. The curvature transformation law for two torsion-free
connections implies
\[
\text{Ric}(g_\epsilon)_{ij}
=
\text{Ric}(\gamma)_{ij}
+\nabla_k A_\epsilon{}^k{}_{ij}
-\nabla_j A_\epsilon{}^k{}_{ik}
+A_\epsilon{}^k{}_{ij}A_\epsilon{}^\ell{}_{k\ell}
-A_\epsilon{}^k{}_{i\ell}A_\epsilon{}^\ell{}_{jk},
\]
hence also \eqref{eq:R-exact} after contraction with $(g_\epsilon)^{ij}$.
\noindent Next, since
\[
g_\epsilon=\gamma\circ(\text{Id}+\epsilon \gamma^{-1}h),
\]
the inverse metric is obtained by Neumann expansion:
\[
g_\epsilon^{-1}
=
(\text{Id}+\epsilon \gamma^{-1}h)^{-1}\circ \gamma^{-1}
=
\big(\text{Id}-\epsilon \gamma^{-1}h+\epsilon^2(\gamma^{-1}h)^2+\epsilon^3\mathcal E_\epsilon\big)\circ \gamma^{-1},
\]
which gives \eqref{eq:inverse-expansion}. The $C^m$ bounds on $\mathcal E_\epsilon$ follow from
smooth dependence of matrix inversion on the coefficients of the metric and compactness of $M$.

\noindent We now derive the linearization of scalar curvature. Differentiating \eqref{eq:A-exact} at $\epsilon=0$ gives
\[
\dot A^{k}{}_{ij}
:=
\left.\frac{d}{d\epsilon}\right|_{\epsilon=0}A_\epsilon{}^{k}{}_{ij}
=
\frac12\gamma^{k\ell}
\big(
\nabla_i h_{j\ell}
+\nabla_j h_{i\ell}
-\nabla_\ell h_{ij}
\big).
\]
Differentiating \eqref{eq:R-exact} at $\epsilon=0$, using $(g_\epsilon)^{ij}=\gamma^{ij}+O(\epsilon)$ and
$A_\epsilon=O(\epsilon)$, we find
\[
DR_\gamma[h]
=
-h^{ij}\text{Ric}(\gamma)_{ij}
+\gamma^{ij}\big(\nabla_k\dot A^{k}{}_{ij}-\nabla_j\dot A^{k}{}_{ik}\big).
\]
A direct computation shows
\[
\gamma^{ij}\nabla_k\dot A^{k}{}_{ij}
=
\nabla^i\nabla^j h_{ij}-\frac12\Delta_\gamma(\tr_\gamma h),
\]
while
\[
\gamma^{ij}\nabla_j\dot A^{k}{}_{ik}
=
\frac12\Delta_\gamma(\tr_\gamma h).
\]
Substituting these identities yields \eqref{eq:DR}.

\noindent The quadratic expansion \eqref{eq:R-linearization-expansion} follows by substituting
\eqref{eq:inverse-expansion} and \eqref{eq:A-exact} into \eqref{eq:R-exact}. Indeed,
\[
A_\epsilon=\epsilon\,\nabla h+\epsilon^2 h*\nabla h,
\qquad
\nabla A_\epsilon=\epsilon\,\nabla^2 h+\epsilon^2(\nabla h*\nabla h+h*\nabla^2 h),
\]
and the exact formula \eqref{eq:R-exact} then yields
\[
R[g_\epsilon]-R[\gamma]-\epsilon DR_\gamma[h]
=
\epsilon^2\Big(h*\nabla^2 h+\nabla h * \nabla h+\text{Rm}(\gamma)*h*h\Big),
\]
schematically.

\noindent We turn to the trace-free Ricci tensor
\[
\mathfrak T[g]=\text{Ric}(g)-\frac1nR[g]\,g.
\]
Differentiating at $\gamma$ gives
\[
D\mathfrak T_\gamma[h]
=
D\text{Ric}_\gamma[h]-\frac1n DR_\gamma[h]\gamma-\frac1nR[\gamma]h,
\]
which is \eqref{eq:DT-general}. Thus, it remains only to record the standard formula for
$D\text{Ric}_\gamma[h]$.

\noindent Differentiating \eqref{eq:Ric-exact} at $\epsilon=0$ eliminates the quadratic terms in $A_\epsilon$ and yields
\[
D\text{Ric}_\gamma[h]_{ij}
=
\nabla_k\dot A^k{}_{ij}-\nabla_j\dot A^k{}_{ik}.
\]
Substituting the expression for $\dot A$ and commuting covariant derivatives in the second-order terms
gives
\[
D\text{Ric}_\gamma[h]_{ij}
=
\frac12\Big(
-\Delta_\gamma h_{ij}
-\nabla_i\nabla_j(\tr_\gamma h)
+\nabla_i(\delta_\gamma h)_j
+\nabla_j(\delta_\gamma h)_i
\Big)
-\text{Rm}(\gamma)_{ikj\ell}h^{k\ell}
+\frac12\Big(\text{Ric}(\gamma)_{ik}h^k{}_j+\text{Ric}(\gamma)_{jk}h^k{}_i\Big).
\]
This proves the variation of the Ricci curvature. 
The expansion \eqref{eq:T-expansion} and the schematic structure \eqref{eq:QT-schematic}
follow immediately from \eqref{eq:DT-general}, \eqref{eq:R-linearization-expansion},
\eqref{eq:DRic-general}, and the already established bounds for the scalar-curvature remainder.

\noindent Assume now that $h$ is transverse-traceless relative to $\gamma$, i.e.,
\[
\tr_\gamma h=0,\qquad \delta_\gamma h=0.
\]
Then \eqref{eq:DR} reduces to
\[
DR_\gamma[h]=-\langle \text{Ric}(\gamma),h\rangle_\gamma.
\]
Since
\[
\text{Ric}(\gamma)=\mathfrak T[\gamma]+\frac1nR[\gamma]\gamma
\]
and $\langle \gamma,h\rangle_\gamma=\tr_\gamma h=0$, $DR_{\gamma}[h]$ is also equal to
\[
-\langle \mathfrak T[\gamma],h\rangle_\gamma,
\]
proving \eqref{eq:DR-TT}.

\noindent Under the same TT assumptions, \eqref{eq:DRic-general} simplifies to
\[
D\text{Ric}_\gamma[h]_{ij}
=
\frac12\Big(
-\Delta_\gamma h_{ij}
-2\text{Rm}(\gamma)_{ikj\ell}h^{k\ell}
+\text{Ric}(\gamma)_{ik}h^k{}_j+\text{Ric}(\gamma)_{jk}h^k{}_i
\Big)
=
\frac12(\Delta_{L,\gamma}h)_{ij},
\]
which is \eqref{eq:DRic-TT}. Substituting this and \eqref{eq:DR-TT} into \eqref{eq:DT-general}
yields
\[
D\mathfrak T_\gamma[h]_{ij}
=
\frac12(\Delta_{L,\gamma}h)_{ij}
+\frac1n\langle \text{Ric}(\gamma),h\rangle_\gamma\,\gamma_{ij}
-\frac1nR[\gamma]h_{ij}.
\]
Since $h$ is trace-free,
\[
\langle \text{Ric}(\gamma),h\rangle_\gamma
=
\langle \mathfrak T[\gamma],h\rangle_\gamma,
\]
and \eqref{eq:DT-TT} follows. This completes the proof of the lemma.
\end{proof}
\noindent Now we prove the following basic lemma regarding the smooth dependence of the transverse-traceless projection on the metric. This will be important in the next proposition, where we use high frequency $TT$ eigentensors (of an appropriately constructed second-order elliptic operator).
\begingroup
\begin{lemma}[TT projection depends smoothly on the metric]
\label{lem:TT-projection}
Fix $s\ge4$ and a smooth background metric $\bar g$ with no nontrivial conformal Killing fields. There exists a neighborhood
$\mathcal U\subset \text{Met}^{H^{s+1}}(M)$ of $\bar g$ and a bounded linear operator
\[
\Pi^{TT}_g: H^{s}(M;S^{2}T^{*}M)\to H^{s}(M;S^{2}T^{*}M),
\qquad g\in\mathcal U,
\]
such that:
\begin{enumerate}
\item $\Pi^{TT}_g$ is a projection: $(\Pi^{TT}_g)^2=\Pi^{TT}_g$.
\item $\Pi^{TT}_g(h)$ is $g$-TT for every $h$.
\item The map $g\mapsto \Pi^{TT}_g$ is $C^\infty$ from $\mathcal U$ into
$\mathcal L(H^{s},H^{s})$.
\item There exists $C=C(\mathcal U)$ such that
\begin{equation}
\label{eq:TTproj-bound}
\|\Pi^{TT}_g(h)\|_{H^{s}} \le C\,\|h\|_{H^{s}}
\quad\text{and}\quad
\|(\Pi^{TT}_g-\Pi^{TT}_{\bar g})(h)\|_{H^{s}} \le C\,\|g-\bar g\|_{H^{s+1}}\,\|h\|_{H^{s}}.
\end{equation}
\end{enumerate}
\end{lemma}

\begin{proof}
Consider the elliptic operator on $1$--forms
\[
\mathcal L_g := \delta_g\circ \mathcal D_g : H^{s+1}(T^*M)\to H^{s-1}(T^*M),
\]
where $(\mathcal D_g X)_{ij}:=\nabla_i X_j+\nabla_j X_i-\frac23(\div_g X)g_{ij}$
is the conformal Killing operator. The kernel of $\mathcal L_{\bar g}$ consists of the conformal Killing one-forms and is therefore trivial. Elliptic Fredholm theory and stability of bounded inverses imply that $\mathcal L_g$ is an isomorphism for $g$ in a sufficiently small $H^{s+1}$ neighborhood of $\bar g$. Its coefficients, and hence its inverse between the indicated spaces, depend smoothly on $g$.
Define
\[
h^{\circ}:=h-\frac13(\tr_g h)g,
\qquad
\Pi^{TT}_g(h):=h^{\circ}-\mathcal D_g\bigl(\mathcal L_g^{-1}(\delta_g h^{\circ})\bigr).
\]
Both terms in the last expression are trace-free, and applying $\delta_g$ gives zero. If $h$ is already $g$-TT, the formula returns $h$, proving the projection property. Sobolev multiplication, elliptic estimates, and the identity for the difference of two bounded inverses give \eqref{eq:TTproj-bound}. The extra derivative of $g$ accounts for the differentiated connection coefficients; no assertion of uniform bounds on an $H^s$ metric neighborhood is used below.
\end{proof}
\par\endgroup

\noindent In this section, we formulate the geometric mechanism underlying the construction of the CMC initial data in Theorem~\ref{main}. The construction consists of three logically separate steps.

\noindent First, starting from a background metric of constant negative scalar curvature, we introduce a high-frequency transverse-traceless perturbation. The linearized trace-free Ricci tensor in such directions is governed by the Lichnerowicz Laplacian and is therefore of second-order size in the oscillation parameter, whereas the scalar curvature variation in a transverse-traceless direction is the zeroth-order contraction against the background trace-free Ricci tensor. This produces a regime in which the trace-free Ricci tensor is large in \(H^{2}\), while the normalized scalar curvature defect remains small in \(H^{2}\).

\noindent Second, on this manifold $(M,g)$, our task is to now find a $TT$ tensor $\Sigma$ that verifies the estimates 
\begin{align}
||\Sigma||_{H^{2}}=O({\color{blue}e^{-\alpha a}}),~||\Sigma||_{H^{3}}=O({\color{blue}e^{(1-\alpha)a}})
\end{align}
so that we can use such $\Sigma$ as the free data for our constraint system that is to be solved via the conformal method in the next step. In particular, in the next (and third) step we perform the conformal transformation 
\[
g_{0}=\varphi^{\frac{4}{n-2}}g,
\qquad
\Sigma_{0}=\varphi^{-2}\Sigma,
\]
where we intend to use $\Sigma$ as the data for the momentum constraint 
\begin{align}
\label{eq:TT1}
\div_{g_{0}}\Sigma_{0}=0.
\end{align}
By conformal covariance the $TT$ property of $\Sigma$ with respect to $g$, i.e., $\tr_{g}\Sigma=0$ and $\div_{g}\Sigma=0$, implies $\tr_{g_{0}}\Sigma_{0}=0$ and the momentum constraint (\ref{eq:TT1}). Therefore it is crucial to explicitly construct such $\Sigma$ on $(M,g)$. 

\noindent Third, we correct the scalar curvature defect by solving the Hamiltonian constraint through a conformal deformation. The resulting Lichnerowicz equation is solved perturbatively around the constant solution \(1\), and the resulting conformal factor is shown to be so close to \(1\) in the $H^{4}$ Sobolev norm that the trace-free Ricci tensor changes only by a lower-order amount in \(H^{2}\).

\noindent We now record the quantitative form of the three-step construction described heuristically in the introduction. Throughout this section, \(M\) denotes a closed connected smooth \(3\)-manifold with $\sigma(M)\le0$. In a conformal class $[\widehat{\gamma}]$ of negative Yamabe sign we fix a smooth metric \(\gamma\) satisfying
\[
R[\gamma]=-\frac23.
\]
All covariant derivatives, contractions, volume forms, and Sobolev norms are taken with respect to \(\gamma\), unless explicitly indicated otherwise.

\begin{proposition}
\label{prop:YorkTTperturb}
Let $M$ be a smooth closed three-manifold with $\sigma(M)\le0$ and fix a conformal class $[\widehat{\gamma}]$ of negative Yamabe sign. Choose a smooth metric $\gamma\in[\widehat{\gamma}]$ such that
\[
R[\gamma]\equiv -\frac23.
\]
Assume, for simplicity, that $\gamma$ admits no nontrivial conformal Killing fields. Define the conformal Killing operator on $1$-forms by
\[
(\mathcal D_\gamma W)_{ij}
:=
\nabla_i W_j+\nabla_j W_i-\frac23(\nabla^kW_k)\gamma_{ij},
\]
and let $S^2_0T^*M$ denote the bundle of $\gamma$-trace-free symmetric $2$-tensors. Since $\delta_\gamma\mathcal D_\gamma$ is then invertible, the York projector
\[
\Pi_{TT}
:=
I-\mathcal D_\gamma(\delta_\gamma\mathcal D_\gamma)^{-1}\delta_\gamma
\]
is a well-defined bounded operator on $H^s(S^2_0T^*M)$ for every integer $s\ge0$, with range $H^s(TT_\gamma)$, the closed subspace of $H^s$ tensors satisfying
$\tr_\gamma h=0$ and $\delta_\gamma h=0$ distributionally. We write $TT_\gamma$
for its smooth tensors and $L^2(TT_\gamma)$ for its $L^2$ closure.
Consider the positive self-adjoint elliptic pseudodifferential operator of order two
\[
A_\gamma
:=
\Pi_{TT}(I+\nabla_\gamma^{*}\nabla_\gamma)\Pi_{TT}\big|_{L^2(TT_\gamma)}
\]
acting on $L^2(TT_\gamma)$.

\noindent Then there exist constants
\[
\varepsilon_{*}>0,\qquad C_m\ge1\ \ (m\ge0),\qquad c_m>0\ \ (m\ge0),
\]
depending only on $(M,\gamma)$, together with an $L^2(\gamma)$-orthonormal sequence of smooth tensors
\[
h_\nu\in TT_\gamma,\qquad \nu\in\mathbb N,
\]
and a sequence of positive numbers
\[
\Lambda_\nu\to+\infty,
\]
such that
\[
A_\gamma h_\nu=\Lambda_\nu h_\nu,
\]
and, for every integer $m\ge0$,
\[
\|h_\nu\|_{H^m(\gamma)}
\le C_m(1+\Lambda_\nu)^{m/2}
\]
for all $\nu$, while
\[
\|h_\nu\|_{H^m(\gamma)}
\ge c_m\Lambda_\nu^{m/2}
\]
for all sufficiently large $\nu$.

\noindent Moreover, if
\[
g_{\varepsilon,\nu}:=\gamma+\varepsilon h_\nu,
\qquad
0<\varepsilon\le \varepsilon_{*}\Lambda_\nu^{-3/2},
\]
then $g_{\varepsilon,\nu}$ is a smooth Riemannian metric and
\begin{align}
\label{eq:Tquantitative}
\|\mathfrak T[g_{\varepsilon,\nu}]\|_{H^2(\gamma)}
&\ge
c_{4}\,\varepsilon \Lambda_\nu^{2}
-
C\,\varepsilon^{2}\Lambda_\nu^{3}
-
C,\\
\label{eq:Rquantitative}
\left\|R[g_{\varepsilon,\nu}]+\frac23\right\|_{H^2(\gamma)}
&\le
C\,\varepsilon \Lambda_\nu
+
C\,\varepsilon^{2}\Lambda_\nu^{3},
\end{align}
where
\[
\mathfrak T[g]:=\text{Ric}[g]-\frac13R[g]\,g.
\]
\end{proposition}

\begin{proof} The basic idea is to perturb the metric along a transverse-traceless ($TT$) direction. In addition, the frequency content of such perturbations is restricted to be very large. Such perturbations can modify the trace-free Ricci curvature $\mathfrak{T}[g]:=\text{Ric}[g]-\frac{1}{3}R[g]g$ in a high enough Sobolev norm while keeping the scalar curvature almost invariant. 
The main idea is to construct a second-order pseudo-differential operator (in the sense of H\"ormander) with a family of smooth $L^{2}$-normalized $TT$ eigentensors.  We divide the argument into several steps. 
Let
\[
\mathcal L_\gamma:=\delta_\gamma\mathcal D_\gamma
\]
be the vector Laplacian associated with $\gamma$. Since $\gamma$ has no nontrivial conformal Killing fields, the kernel of $\mathcal L_\gamma$ is trivial. Indeed, integration by parts on the closed manifold, using the symmetry and vanishing trace of $\mathcal D_\gamma W$, gives
\[
\langle \mathcal L_\gamma W,W\rangle_{L^2(\gamma)}
=
-\frac12\|\mathcal D_\gamma W\|_{L^2(\gamma)}^{2},
\]
with $(\delta_\gamma h)_j=\nabla^i h_{ij}$. In particular,
\[
\mathcal L_\gamma W=0
\quad\Longrightarrow\quad
\mathcal D_\gamma W=0.
\]
Thus the assumption on conformal Killing fields implies $\ker \mathcal L_\gamma=\{0\}$. Since $-\mathcal L_\gamma$ is a second-order strongly elliptic self-adjoint operator on the closed manifold $M$, elliptic Fredholm theory yields that
\[
\mathcal L_\gamma:H^{s+2}(T^*M)\to H^s(T^*M)
\]
is an isomorphism for every integer $s\ge0$. At $s=0$ this gives an inverse from $L^2$ to $H^2$. By self-adjointness its dual extends the same inverse from $H^{-2}$ to $L^2$. Interpolation gives an inverse from $H^{-1}$ to $H^1$, so the formula below also applies to $k\in L^2$, for which $\delta_\gamma k\in H^{-1}$.

\noindent It follows that the operator
\[
\Pi_{TT}
=
I-\mathcal D_\gamma \mathcal L_\gamma^{-1}\delta_\gamma
\]
is well defined and bounded on $H^s(S^2_0T^*M)$ for every $s\ge0$. We now verify that its range is precisely the TT space. Let $k\in H^s(S^2_0T^*M)$ and define
\[
W:=\mathcal L_\gamma^{-1}\delta_\gamma k,\qquad k^{TT}:=k-\mathcal D_\gamma W.
\]
Since $\mathcal D_\gamma W$ is trace-free by construction, $k^{TT}$ is trace-free as well. Moreover,
\[
\delta_\gamma k^{TT}
=
\delta_\gamma k-\delta_\gamma\mathcal D_\gamma W
=
\delta_\gamma k-\mathcal L_\gamma W
=
0.
\]
Hence $k^{TT}\in H^s(TT_\gamma)$. Conversely, if $h\in TT_\gamma$, then $\delta_\gamma h=0$, so
\[
\Pi_{TT}h
=
h-\mathcal D_\gamma \mathcal L_\gamma^{-1}\delta_\gamma h
=
h.
\]
Therefore $\Pi_{TT}$ is a projection with range $H^s(TT_\gamma)$.

\noindent In particular, for each $s\ge0$ one has the topological direct sum decomposition
\[
H^s(S^2_0T^*M)
=
H^s(TT_\gamma)\oplus \mathcal D_\gamma H^{s+1}(T^*M).
\]
This is the usual York splitting \cite{york}. Since the principal symbol of $\Pi_{TT}$ is the orthogonal projection onto the algebraic subspace
\[
\{q\in S^2_0T_x^*M:\ \xi^i q_{ij}=0\},
\]
whose rank equals $2$ in dimension $3$, the range of $\Pi_{TT}$ is infinite-dimensional. Thus $L^2(TT_\gamma)$ is an infinite-dimensional closed subspace of $L^2(S^2_0T^*M)$.

\smallskip
\noindent
Next we consider the high-frequency TT basis that we are going to use to perturb the metric.
We now consider
\[
A_\gamma
=
\Pi_{TT}(I+\nabla_\gamma^{*}\nabla_\gamma)\Pi_{TT}\big|_{L^2(TT_\gamma)}.
\]
The York splitting is $L^2(\gamma)$-orthogonal, since $\langle\mathcal D_\gamma W,h\rangle_{L^2}=-2\langle W,\delta_\gamma h\rangle_{L^2}=0$ for $h\in TT_\gamma$. Thus $\Pi_{TT}$ is self-adjoint. On the full trace-free tensor bundle, the operator
\[
\Pi_{TT}(I+\nabla_\gamma^*\nabla_\gamma)\Pi_{TT}
+(I-\Pi_{TT})(I+\nabla_\gamma^*\nabla_\gamma)(I-\Pi_{TT})
\]
has principal symbol $|\xi|_\gamma^2\mathrm{Id}$, is positive and self-adjoint with domain $H^2(S^2_0T^*M)$, and commutes with $\Pi_{TT}$. Its restriction to the range of $\Pi_{TT}$ is $A_\gamma$, with domain $H^2(TT_\gamma)$. Its principal symbol on the TT symbol bundle is
\[
\sigma_2(A_\gamma)(x,\xi)
=
|\xi|_\gamma^2\,\text{Id} .
\]
In particular, $A_\gamma$ is elliptic and has a compact resolvent. Standard spectral theory for positive self-adjoint elliptic operators on closed manifolds therefore yields an $L^2(\gamma)$-orthonormal basis of smooth eigentensors
\[
\{h_\nu\}_{\nu\ge1}\subset TT_\gamma
\]
with corresponding eigenvalues
\[
0<\Lambda_1\le\Lambda_2\le\cdots,\qquad \Lambda_\nu\to+\infty,
\]
such that
\[
A_\gamma h_\nu=\Lambda_\nu h_\nu.
\]
 We next record the Sobolev bounds for this eigenbasis. Since $A_\gamma$ is elliptic of order $2$, the graph norm of $(I+A_\gamma)^{m/2}$ is equivalent to the $H^m$ norm on $TT_\gamma$; that is, for each integer $m\ge0$ there exist constants $c_m',C_m'>0$ such that
\[
c_m'\|u\|_{H^m(\gamma)}
\le
\|(I+A_\gamma)^{m/2}u\|_{L^2(\gamma)}
\le
C_m'\|u\|_{H^m(\gamma)}
\]
for all $u\in H^m(TT_\gamma)$. Applying this to $u=h_\nu$ and using
\[
(I+A_\gamma)^{m/2}h_\nu=(1+\Lambda_\nu)^{m/2}h_\nu,
\qquad
\|h_\nu\|_{L^2(\gamma)}=1,
\]
we obtain
\[
\frac{1}{C_m'}(1+\Lambda_\nu)^{m/2}
\le
\|h_\nu\|_{H^m(\gamma)}
\le
\frac{1}{c_m'}(1+\Lambda_\nu)^{m/2}.
\]
After adjusting the constants, this implies
\[
\|h_\nu\|_{H^m(\gamma)}
\le C_m(1+\Lambda_\nu)^{m/2}
\]
for all $\nu$, and
\[
\|h_\nu\|_{H^m(\gamma)}
\ge c_m\Lambda_\nu^{m/2}
\]
for all sufficiently large $\nu$.
\noindent In particular,
\[
\|h_\nu\|_{H^2(\gamma)}\le C_2\Lambda_\nu,
\qquad
\|h_\nu\|_{H^4(\gamma)}\ge c_4\Lambda_\nu^2
\]
for all sufficiently large $\nu$.

\smallskip
\noindent
Now we show that the $TT$ perturbation that we apply to the metric preserves its Riemannian character. 
Since $\dim M=3$, the Sobolev embedding $H^{2}(M)\hookrightarrow C^0(M)$ gives
\[
\|h_\nu\|_{C^0(\gamma)}
\le
C\|h_\nu\|_{H^2(\gamma)}
\le
C\Lambda_\nu.
\]
Hence, if
\[
0<\varepsilon\le \varepsilon_*\Lambda_\nu^{-3/2},
\]
then
\[
\varepsilon\|h_\nu\|_{C^0(\gamma)}
\le
C\varepsilon_*\Lambda_\nu^{-1/2}.
\]
By choosing $\varepsilon_*>0$ sufficiently small, and then decreasing it once more to absorb the finitely many low-frequency modes, we may ensure that
\[
\varepsilon\|h_\nu\|_{C^0(\gamma)}\le \frac12
\]
for every $\nu$ under the above restriction on $\varepsilon$. Now let $X\in T_xM$. Then
\[
g_{\varepsilon,\nu}(X,X)
=
\gamma(X,X)+\varepsilon h_\nu(X,X)
\ge
\Bigl(1-\varepsilon\|h_\nu\|_{C^0(\gamma)}\Bigr)\gamma(X,X)
\ge
\frac12\,\gamma(X,X).
\]
Therefore, $g_{\varepsilon,\nu}$ is positive definite, hence a smooth Riemannian metric.

\noindent Now using Lemma~\ref{nonlinearexpansion}, we obtain the necessary estimates for the trace-free Ricci curvature and the scalar curvature of the perturbed metric.
We write
\[
\mathfrak T[g]=\text{Ric}[g]-\frac13R[g]\,g.
\]
\begingroup\color{blue}
The map $g\longmapsto\mathfrak T[g]$ is a smooth quasilinear differential operator of order two. We use its expansion on a fixed small $H^2(\gamma)$ neighborhood, without requiring the $H^4$ norm of the perturbation to be small. More precisely, there are $\eta>0$ and $C$, depending only on $(M,\gamma)$, such that every smooth symmetric tensor $k$ with $\|k\|_{H^2(\gamma)}\le\eta$ defines a positive metric $\gamma+k$ and satisfies
\[
\mathfrak T[\gamma+k]
=\mathfrak T[\gamma]+D\mathfrak T_\gamma[k]+\mathcal R_{\mathfrak T,\gamma}(k),
\]
with
\begin{equation}
\label{eq:RTremainder}
\|\mathcal R_{\mathfrak T,\gamma}(k)\|_{H^2(\gamma)}
\le C\|k\|_{H^4(\gamma)}\|k\|_{H^2(\gamma)}.
\end{equation}
Indeed, $H^2(M)\hookrightarrow C^0(M)$ gives positivity when $\eta$ is small. The inverse metric and all smooth inverse-metric coefficients in the curvature formula are uniformly bounded in $H^2(\gamma)$ on this set. The exact difference-of-connections formula in Lemma~\ref{nonlinearexpansion}, after subtracting its linear part, expresses the remainder as sums of such coefficients multiplying
\[
k*\nabla_\gamma^2k,\qquad
\nabla_\gamma k*\nabla_\gamma k,\qquad
\operatorname{Rm}(\gamma)*k*k.
\]
The algebra property of $H^2$ and interpolation give
\begin{align*}
\|k*\nabla_\gamma^2k\|_{H^2}
&\le C\|k\|_{H^2}\|k\|_{H^4},\\
\|\nabla_\gamma k*\nabla_\gamma k\|_{H^2}
&\le C\|k\|_{H^3}^2
\le C\|k\|_{H^2}\|k\|_{H^4},\\
\|\operatorname{Rm}(\gamma)*k*k\|_{H^2}
&\le C\|k\|_{H^2}^2
\le C\|k\|_{H^2}\|k\|_{H^4}.
\end{align*}
These estimates prove \eqref{eq:RTremainder} with a constant independent of $\|k\|_{H^4}$. The same argument applies to the scalar curvature remainder. Finally, the restriction $\varepsilon\le\varepsilon_*\Lambda_\nu^{-3/2}$ ensures $\|\varepsilon h_\nu\|_{H^2}\le\eta$ after decreasing $\varepsilon_*$, so this estimate applies to every perturbation used in the proposition.
\par\endgroup

\noindent We now restrict to TT directions. The principal part of the linearized Ricci tensor is $\frac12\nabla_\gamma^*\nabla_\gamma h$, and the scalar-curvature terms are of lower order on TT tensors. Thus
\[
\sigma_2(D\mathfrak T_\gamma)(x,\xi)
=
\frac12|\xi|_\gamma^2\,\text{Id}
\]
on the TT symbol space.

\noindent Therefore elliptic regularity gives, for $u\in H^4(TT_\gamma)$,
\[
\|u\|_{H^4(\gamma)}
\le
C\bigl(\|D\mathfrak T_\gamma[u]\|_{H^2(\gamma)}+\|u\|_{L^2(\gamma)}\bigr).
\]
Applying this estimate to $u=h_\nu$ and using $\|h_\nu\|_{L^2}=1$, we obtain
\[
\|D\mathfrak T_\gamma[h_\nu]\|_{H^2(\gamma)}
\ge
C^{-1}\|h_\nu\|_{H^4(\gamma)}-C
\ge
c\,\Lambda_\nu^2-C.
\]
After decreasing $c$ if necessary, we may rewrite this as
\begin{equation}
\label{eq:Lgamma-lower}
\|D\mathfrak T_\gamma[h_\nu]\|_{H^2(\gamma)}
\ge
c_4\Lambda_\nu^2-C
\end{equation}
for all $\nu$.

\noindent Now substitute $k=\varepsilon h_\nu$ into the Taylor formula. Since $h_\nu\in TT_\gamma$,
\[
\mathfrak T[g_{\varepsilon,\nu}]
=
\mathfrak T[\gamma]
+
\varepsilon D\mathfrak T_\gamma[h_\nu]
+
\mathcal R_{\mathfrak T,\gamma}(\varepsilon h_\nu).
\]
Using \eqref{eq:RTremainder}, together with the Sobolev bounds for $h_\nu$, we find
\[
\|\mathcal R_{\mathfrak T,\gamma}(\varepsilon h_\nu)\|_{H^2(\gamma)}
\le
C\varepsilon^2
\|h_\nu\|_{H^4(\gamma)}
\|h_\nu\|_{H^2(\gamma)}
\le
C\varepsilon^2\Lambda_\nu^3.
\]
Hence, by the reverse triangle inequality,
\begin{align*}
\|\mathfrak T[g_{\varepsilon,\nu}]\|_{H^2(\gamma)}
&\ge
\varepsilon\|D\mathfrak T_\gamma[h_\nu]\|_{H^2(\gamma)}
-
\|\mathfrak T[\gamma]\|_{H^2(\gamma)}
-
\|\mathcal R_{\mathfrak T,\gamma}(\varepsilon h_\nu)\|_{H^2(\gamma)}\\
&\ge
\varepsilon(c_4\Lambda_\nu^2-C)
-
C
-
C\varepsilon^2\Lambda_\nu^3.
\end{align*}
Since $\varepsilon\le1$ after decreasing $\varepsilon_*$ once more if necessary, the term $\varepsilon C$ may be absorbed into the final background constant, and we conclude that
\[
\|\mathfrak T[g_{\varepsilon,\nu}]\|_{H^2(\gamma)}
\ge
c_4\,\varepsilon\Lambda_\nu^2
-
C\varepsilon^2\Lambda_\nu^3
-
C.
\]
This proves \eqref{eq:Tquantitative}.

\smallskip
\noindent
Now we focus on the scalar curvature expansion.
The scalar curvature map is likewise a smooth quasilinear differential operator of order two, so one has
\[
R[\gamma+k]
=
R[\gamma]
+
DR_\gamma[k]
+
\mathcal R_{R,\gamma}(k),
\]
\begingroup\color{blue}
and the same argument on the $H^2$--small neighborhood gives
\begin{equation}
\label{eq:Rremainder}
\|\mathcal R_{R,\gamma}(k)\|_{H^2(\gamma)}
\le C\|k\|_{H^4(\gamma)}\|k\|_{H^2(\gamma)}.
\end{equation}
\par\endgroup

\noindent Since $R[\gamma]\equiv-\frac23$, it remains to estimate the linear term. The standard formula for the linearization of scalar curvature is
\[
DR_\gamma[h]
=
-\Delta_\gamma(\tr_\gamma h)+\delta_\gamma\delta_\gamma h-\langle \text{Ric}[\gamma],h\rangle_\gamma.
\]
For $h=h_\nu\in TT_\gamma$, the first two terms vanish, and therefore
\[
DR_\gamma[h_\nu]
=
-\langle \text{Ric}[\gamma],h_\nu\rangle_\gamma.
\]
This is the main mechanism that cancels the principal part in the expansion of the scalar curvature and therefore does not cost frequency.
Because $\text{Ric}[\gamma]$ is a fixed smooth tensor, multiplication by $\text{Ric}[\gamma]$ is a bounded operator on $H^2$, and hence
\[
\|DR_\gamma[h_\nu]\|_{H^2(\gamma)}
\le
C\|h_\nu\|_{H^2(\gamma)}
\le
C\Lambda_\nu.
\]

\noindent Substituting $k=\varepsilon h_\nu$ into the expansion for $R$, and using \eqref{eq:Rremainder}, we obtain
\begin{align*}
\left\|R[g_{\varepsilon,\nu}]+\frac23\right\|_{H^2(\gamma)}
&\le
\varepsilon\|DR_\gamma[h_\nu]\|_{H^2(\gamma)}
+
\|\mathcal R_{R,\gamma}(\varepsilon h_\nu)\|_{H^2(\gamma)}\\
&\le
C\varepsilon\Lambda_\nu
+
C\varepsilon^2
\|h_\nu\|_{H^4(\gamma)}
\|h_\nu\|_{H^2(\gamma)}\\
&\le
C\varepsilon\Lambda_\nu
+
C\varepsilon^2\Lambda_\nu^3.
\end{align*}
This is exactly \eqref{eq:Rquantitative}. The proof is complete.
\end{proof}

\begin{remark}
Note that if the negative Yamabe background is Einstein, then one can simply use the spectrum of the Lichnerowicz Laplacian $\Delta_{L}$ since that preserves the transverse-traceless (TT) subspace of $S^{2}T^{*}M$. On a general closed manifold, this is not true and one ought to use the TT projector from the proposition.
\end{remark}

\noindent In the corollary below the metric perturbation amplitude is prescribed by the chosen eigenvalue and the desired Sobolev norm. All estimates are taken in the fixed reference metric $\gamma$.
\begin{corollary}
\label{cor:separation-scales}
\begingroup\color{blue}
For every sufficiently large $a$, one can choose a metric $g_{1,a}=\gamma+\varepsilon h_\nu$ such that
\begin{align}
 c e^{(1-\alpha)a}\le\|\mathfrak T[g_{1,a}]\|_{H^2(\gamma)}&\le C e^{(1-\alpha)a},\label{eq:Tlargecor}\\
 \|R[g_{1,a}]+2/3\|_{H^2(\gamma)}&\le C e^{-2\alpha a}.\label{eq:Rsmallcor}
\end{align}
The constants $c,C$ are independent of $a$. The metric satisfies
\[
 \|g_{1,a}-\gamma\|_{H^2}\le Ce^{-(1+\alpha)a},\qquad
 \|g_{1,a}-\gamma\|_{H^3}\le Ce^{-\alpha a},\qquad
 c e^{(1-\alpha)a}\le\|g_{1,a}-\gamma\|_{H^4}\le C e^{(1-\alpha)a}.
\]
\par\endgroup
\end{corollary}

\begin{proof}
\begingroup\color{blue}
Take any eigenvalue $\lambda=\Lambda_\nu\ge e^{2a}$ in Proposition~\ref{prop:YorkTTperturb}. Write $k=\varepsilon h_\nu$ with $\varepsilon=e^{(1-\alpha)a}\lambda^{-2}$. The two-sided eigenfunction estimates give $\|k\|_{H^j}\approx e^{(1-\alpha)a}\lambda^{(j-4)/2}$. In particular $\varepsilon\lambda^{3/2}=e^{(1-\alpha)a}\lambda^{-1/2}\le e^{-\alpha a}$. This is smaller than the fixed bound in the proposition for all sufficiently large $a$. The linear curvature term has $H^2$ norm bounded above by $C e^{(1-\alpha)a}$ and below by $c e^{(1-\alpha)a}-C\varepsilon$. The remainder is bounded by
\[
 C\|k\|_{H^2}\|k\|_{H^4}\le Ce^{2(1-\alpha)a}/\lambda.
\]
After accounting for the fixed $\mathfrak T[\gamma]$ and taking $a$ sufficiently large, this proves \eqref{eq:Tlargecor}. The scalar first variation has $H^2$ size $O(e^{(1-\alpha)a}/\lambda)$ by \eqref{eq:DR-TT}, which proves \eqref{eq:Rsmallcor}.

\noindent At the next derivative, the remainder arises from $k\nabla^2k$, $(\nabla k)^2$ and background-curvature terms with smooth inverse-metric coefficients. The $H^3$ product estimate with one factor in $H^2\hookrightarrow L^\infty$ and the bounds on $k$ give
\[
 \|R[g_{1,a}]+2/3\|_{H^3}
 \le C\bigl(\|k\|_{H^3}+\|k\|_{H^2}\|k\|_{H^5}+\|k\|_{H^3}\|k\|_{H^4}\bigr)
 \le C e^{(1-2\alpha)a}.
\]
The coefficients have bounded $H^3$ norms by $\|k\|_{H^3}\le C e^{-\alpha a}$ and uniform positive definiteness. In particular no constant in these estimates depends on $\|k\|_{H^4}$.

\noindent To see the fourth derivative explicitly, choose any eigenvalue $\lambda\ge e^{4a}$. Then $k=e^{(1-\alpha)a}\lambda^{-2}h_\nu$ satisfies
\[
 \|k\|_{H^2(\gamma)}\lesssim e^{-(3+\alpha)a},\qquad
 \|k\|_{H^3(\gamma)}\lesssim e^{-(1+\alpha)a},\qquad
 \|k\|_{H^4(\gamma)}\approx e^{(1-\alpha)a}.
\]
The trace-free Ricci remainder in \eqref{eq:RTremainder} is $O(e^{-(2+2\alpha)a})$ in $H^2(\gamma)$. The linear term has $H^2(\gamma)$ norm comparable to $e^{(1-\alpha)a}$, so the ratio of the error to the linear term is $O(e^{-(3+\alpha)a})$.
\par\endgroup
\end{proof}
\noindent Now we want to construct the normalized initial metric $g_{0}$ and the normalized transverse-traceless second fundamental form $\Sigma_{0}$ that verify the constraint equations 
\begin{align}
R[g_{0,a}]+\frac23&=|\Sigma_{0,a}|_{g_{0,a}}^{2},\\
\div_{g_{0,a}}\Sigma_{0,a}&=0.
\end{align}
To this end, we perform the conformal transformation 
\begin{align}
 g_{0}=\varphi^{\frac{4}{n-2}}g,
\qquad
\Sigma_{0}=\varphi^{-2}\Sigma,   
\end{align}
where $g$ is chosen to be $g_{1,a}$ and $\Sigma$ is the free datum which is transverse-traceless with respect to $g$ (as is standard in the conformal technique) and verifies 
\begin{align}
\label{eq:estimate}
||\Sigma||_{H^{2}}=O({\color{blue}e^{-\alpha a}}),~||\Sigma||_{H^{3}}=O({\color{blue}e^{(1-\alpha)a}}).
\end{align}
The question is how to construct a $TT$ tensor on $(M,g)$ satisfying \eqref{eq:estimate}. We do this now in the following proposition.
\begin{proposition}
\label{prop:TT_construction_separated_scales}
Let $M$ be a smooth closed manifold of dimension $n\geq 3$, and let $g$ be a smooth Riemannian metric on $M$. Fix $0<\alpha<1$. In particular, the conclusion applies when $M$ is of negative Yamabe type.

\noindent Then there exist constants
\[
a_{0}\geq 1,
\qquad
0<c\leq C<\infty,
\]
depending only on $(M,g)$, and for each $a\geq a_{0}$ a smooth symmetric $2$-tensor $\Sigma_{a}$ on $M$ such that
\[
\tr_{g}\Sigma_{a}=0,
\qquad
\div_{g}\Sigma_{a}=0,
\]
and
\[
c\,{\color{blue}e^{-\alpha a}}\leq \|\Sigma_{a}\|_{H^{2}(M,g)}\leq C\,{\color{blue}e^{-\alpha a}},
\qquad
c\,{\color{blue}e^{(1-\alpha)a}}\leq \|\Sigma_{a}\|_{H^{3}(M,g)}\leq C\,{\color{blue}e^{(1-\alpha)a}}.
\]
In particular,
\[
\|\Sigma_{a}\|_{H^{2}(M,g)}=O({\color{blue}e^{-\alpha a}}),
\qquad
\|\Sigma_{a}\|_{H^{3}(M,g)}=O({\color{blue}e^{(1-\alpha)a}}),
\]
and the family $\Sigma_{a}$ is transverse-traceless with respect to $g$.
\end{proposition}

\begin{proof}
We divide the construction into three steps.

\medskip
\noindent  First we construct a high-frequency trace-free seed with one-derivative smaller divergence.
Choose a coordinate chart $U\subset M$ and smooth coordinates
\[
x=(x^{1},\dots,x^{n}):U\to B(0,2)\subset \mathbb R^{n}.
\]
Let
\[
\psi:=x^{1}\in C^{\infty}(U).
\]
Shrinking $U$ if necessary, we may assume that $d\psi$ is nowhere vanishing on $U$. Define the smooth unit vector field
\[
E_{1}:=\frac{\nabla \psi}{|\nabla \psi|_{g}}
\]
on $U$, and extend $E_{1}$ to a smooth local $g$-orthonormal frame
\[
E_{1},E_{2},E_{3},\dots,E_{n}
\]
on $U$. Let
\[
\theta^{1},\theta^{2},\theta^{3},\dots,\theta^{n}
\]
denote the dual coframe. Define the smooth symmetric $2$-tensor
\[
q:=\theta^{2}\otimes \theta^{2}-\theta^{3}\otimes \theta^{3}
\]
on $U$. Then note 
\[
\tr_{g} q = 0
\]
and, since $q(E_{1},\cdot)=0$, also
\[
q(\nabla \psi,\cdot)=0
\qquad\text{on }U.
\]

\noindent Fix a cutoff $\chi\in C_{c}^{\infty}(U)$ such that $0\leq \chi\leq 1$ and $\chi\equiv 1$ on some nonempty open set $V\Subset U$. For each parameter $\mu\geq 1$, define the modulation
\[
S_{\mu}:=\chi \cos(\mu\psi)\, q,
\]
extended by zero outside $U$. Then $S_{\mu}$ is a smooth symmetric $2$-tensor on $M$ and, because $q$ is trace-free,
\[
\tr_{g} S_{\mu}=0.
\]

\noindent We first record the $H^{m}$-size of $S_{\mu}$ for $m=2,3$.

\smallskip
\noindent The tensor $\chi q$ is smooth and compactly supported in $U$. For $j\ge1$, repeated differentiation of $\cos(\mu\psi)$ yields
\[
\nabla^{j}\big(\cos(\mu\psi)\big)
=
\sum_{\ell=1}^{j} \mu^{\ell}\,A_{j,\ell},
\]
where each $A_{j,\ell}$ is a smooth tensor depending only on $g$, $\psi$, and finitely many derivatives thereof, multiplied by either $\sin(\mu\psi)$ or $\cos(\mu\psi)$. By Leibniz' rule, for each integer $j\geq 0$ there exists a constant $C_{j}$, depending only on $(M,g)$ and the auxiliary choices $\psi,\chi,q$, such that
\[
\|\nabla^{j} S_{\mu}\|_{L^{2}(M,g)}
\leq C_{j}\,\mu^{j}.
\]
Consequently, for $m=2,3$,
\begin{equation}
\label{eq:seed_upper}
\|S_{\mu}\|_{H^{m}(M,g)}\leq C\,\mu^{m}.
\end{equation}

\smallskip
\noindent Now consider the following. On the open set $V$ we have $\chi\equiv 1$, and hence
\[
S_{\mu}=\cos(\mu\psi)\,q
\qquad\text{on }V.
\]
Applying the covariant derivative repeatedly in the $E_{1}$-direction gives
\[
(\nabla^m S_{\mu})(E_1,\ldots,E_1,\cdot,\cdot)
=
\mu^{m}\,b_{m}(\mu\psi)\,\big(E_{1}\psi\big)^{m}q
+
\mu^{m-1}R_{m,\mu},
\qquad m=2,3,
\]
where $b_2=-\cos$ and $b_3=\sin$, and where the remainder $R_{m,\mu}$ is a smooth tensor field on $V$ satisfying
\[
\|R_{m,\mu}\|_{L^{2}(V,g)}\leq C.
\]
Since
\[
E_{1}\psi = |\nabla \psi|_{g}>0
\quad\text{on }V,
\]
and $q$ is nowhere zero on $V$, there exists $c_{0}>0$ such that
\[
\left\| b_{m}(\mu\psi)\,\big(E_{1}\psi\big)^{m}q\right\|_{L^{2}(V,g)}\geq c_{0}
\]
for all sufficiently large $\mu$. Indeed, on a coordinate box compactly contained in $V$, the identity $b_m(\mu x^1)^2=(1\pm\cos(2\mu x^1))/2$ and the Riemann--Lebesgue lemma give a positive limiting squared norm. It follows that
\[
\|\nabla^{m}S_{\mu}\|_{L^{2}(M,g)}
\geq
\|(\nabla^m S_{\mu})(E_1,\ldots,E_1,\cdot,\cdot)\|_{L^{2}(V,g)}
\geq c_{1}\mu^{m}-C\mu^{m-1}.
\]
Hence, after increasing $\mu$ if necessary,
\[
\|S_{\mu}\|_{H^{m}(M,g)}\geq c\,\mu^{m},
\qquad m=2,3.
\]
Together with \eqref{eq:seed_upper}, this yields
\begin{equation}
\label{eq:seed_two_sided}
c\,\mu^{m}\leq \|S_{\mu}\|_{H^{m}(M,g)}\leq C\,\mu^{m},
\qquad m=2,3,
\end{equation}
for all sufficiently large $\mu$.

\noindent We next estimate the divergence of $S_{\mu}$. We use the convention
\[
(\div_{g} h)_{j}:=\nabla^{i}h_{ij}.
\]
A direct computation gives
\[
(\div_{g}S_{\mu})_{j}
=
\nabla^{i}\chi \,\cos(\mu\psi)\, q_{ij}
-\mu \chi \sin(\mu\psi)\,\nabla^{i}\psi\, q_{ij}
+\chi \cos(\mu\psi)\,\nabla^{i}q_{ij}.
\]
The middle term vanishes identically because $q(\nabla \psi,\cdot)=0$. Therefore
\[
\div_{g}S_{\mu}
=
(\nabla \chi)*q*\cos(\mu\psi)
+
\chi\,(\nabla q)*\cos(\mu\psi),
\]
where $*$ denotes a universal contraction. In particular, no term of size $\mu$ appears in the amplitude of $\div_{g}S_{\mu}$. This is precisely the reason for choosing $q$ such that $q(\nabla\psi,\cdot)=0$. Differentiating this identity and arguing as above, we obtain for $m=2,3$,
\begin{equation}
\label{eq:div_seed}
\|\div_{g}S_{\mu}\|_{H^{m-1}(M,g)}\leq C\,\mu^{m-1}.
\end{equation}

\noindent Now we explicitly construct a $TT$ tensor through York-projection \cite{york}.
Let
\[
(\mathbb L_{g}W)_{ij}
:=
\nabla_{i}W_{j}
+
\nabla_{j}W_{i}
-
\frac{2}{n}(\div_{g}W)\,g_{ij}
\]
denote the conformal Killing operator acting on one-forms $W$. By construction,
\[
\tr_{g}(\mathbb L_{g}W)=0
\]
for every one-form $W$. Define the vector Laplacian
\[
\Delta_{\mathbb L}W:=-\div_{g}(\mathbb L_{g}W).
\]
This is a second-order self-adjoint elliptic operator on one-forms. Moreover, by a computation similar to that in the previous case,
\[
\int_{M}\langle \Delta_{\mathbb L}W,W\rangle\,d\mu_{g}
=
\frac12\int_{M} |\mathbb L_{g}W|_{g}^{2}\,d\mu_{g},
\]
and hence
\[
\ker(\Delta_{\mathbb L})=\ker(\mathbb L_{g}),
\]
the finite-dimensional space of conformal Killing one-forms.

\noindent We claim that $\div_{g}S_{\mu}$ is $L^{2}$-orthogonal to $\ker(\Delta_{\mathbb L})$. Indeed, let $X\in \ker(\Delta_{\mathbb L})=\ker(\mathbb L_{g})$. Since $S_{\mu}$ is trace-free, the formal adjoint relation gives
\[
\int_{M}\langle \div_{g}S_{\mu},X\rangle\,d\mu_{g}
=
-\frac12\int_{M}\langle S_{\mu},\mathbb L_{g}X\rangle\,d\mu_{g}
=0.
\]
Therefore $\div_{g}S_{\mu}$ belongs to the range of $\Delta_{\mathbb L}$, and there exists a unique one-form $W_{\mu}$ satisfying
\[
W_{\mu}\perp \ker(\Delta_{\mathbb L}) \quad\text{in }L^{2}(M,g),
\]
and
\begin{equation}
\label{eq:York_eqn}
\Delta_{\mathbb L}W_{\mu}=\div_{g}S_{\mu}.
\end{equation}
Define
\[
\tau_{\mu}:=S_{\mu}+\mathbb L_{g}W_{\mu}.
\]
Since both $S_{\mu}$ and $\mathbb L_{g}W_{\mu}$ are trace-free, we have
\[
\tr_{g}\tau_{\mu}=0.
\]
Also, by \eqref{eq:York_eqn},
\[
\div_{g}\tau_{\mu}
=
\div_{g}S_{\mu}+\div_{g}(\mathbb L_{g}W_{\mu})
=
\div_{g}S_{\mu}-\Delta_{\mathbb L}W_{\mu}
=0.
\]
Thus $\tau_{\mu}$ is transverse-traceless.

\smallskip
\noindent We now estimate the correction term. Since $\Delta_{\mathbb L}$ is elliptic and invertible on the $L^{2}$-orthogonal complement of its kernel, standard elliptic regularity yields, for $m=2,3$,
\begin{equation}
\label{eq:elliptic_W}
\|W_{\mu}\|_{H^{m+1}(M,g)}
\leq
C\,\|\div_{g}S_{\mu}\|_{H^{m-1}(M,g)}.
\end{equation}
Combining \eqref{eq:elliptic_W} with \eqref{eq:div_seed}, we obtain
\[
\|W_{\mu}\|_{H^{m+1}(M,g)}\leq C\,\mu^{m-1},
\qquad m=2,3.
\]
Since $\mathbb L_{g}$ is first order,
\begin{equation}
\label{eq:LW_est}
\|\mathbb L_{g}W_{\mu}\|_{H^{m}(M,g)}
\leq
C\,\|W_{\mu}\|_{H^{m+1}(M,g)}
\leq C\,\mu^{m-1},
\qquad m=2,3.
\end{equation}
\noindent Now we obtain the two-sided bounds for the TT tensor by the choice of an appropriate scale. By definition,
\[
\tau_{\mu}=S_{\mu}+\mathbb L_{g}W_{\mu}.
\]
From \eqref{eq:seed_two_sided} and \eqref{eq:LW_est}, for $m=2,3$ we have
\[
\|\tau_{\mu}\|_{H^{m}}
\leq
\|S_{\mu}\|_{H^{m}}+\|\mathbb L_{g}W_{\mu}\|_{H^{m}}
\leq C\mu^{m},
\]
and also
\[
\|\tau_{\mu}\|_{H^{m}}
\geq
\|S_{\mu}\|_{H^{m}}-\|\mathbb L_{g}W_{\mu}\|_{H^{m}}
\geq
c\mu^{m}-C\mu^{m-1}.
\]
Therefore, for all sufficiently large $\mu$,
\begin{equation}
\label{eq:tau_two_sided}
c\,\mu^{m}\leq \|\tau_{\mu}\|_{H^{m}(M,g)}\leq C\,\mu^{m},
\qquad m=2,3.
\end{equation}

\noindent Now let
\[
\mu_{a}:=e^{a},
\qquad
A_{a}:={\color{blue}e^{-(2+\alpha)a}},
\]
and define
\[
\Sigma_{a}:=A_{a}\,\tau_{\mu_{a}}
=
{\color{blue}e^{-(2+\alpha)a}}\tau_{e^{a}}.
\]
Since $\tau_{e^{a}}$ is transverse-traceless, so is $\Sigma_{a}$. Using \eqref{eq:tau_two_sided} with $\mu=e^{a}$, we find
\[
\|\Sigma_{a}\|_{H^{2}(M,g)}
=
{\color{blue}e^{-(2+\alpha)a}}\|\tau_{e^{a}}\|_{H^{2}(M,g)}
\sim
{\color{blue}e^{-(2+\alpha)a}}e^{2a}
=
{\color{blue}e^{-\alpha a}},
\]
and
\[
\|\Sigma_{a}\|_{H^{3}(M,g)}
=
{\color{blue}e^{-(2+\alpha)a}}\|\tau_{e^{a}}\|_{H^{3}(M,g)}
\sim
{\color{blue}e^{-(2+\alpha)a}}e^{3a}
=
{\color{blue}e^{(1-\alpha)a}}.
\]
More precisely, there exist $a_{0}\geq 1$ and constants $0<c\leq C<\infty$ such that for all $a\geq a_{0}$,
\[
c\,{\color{blue}e^{-\alpha a}}\leq \|\Sigma_{a}\|_{H^{2}(M,g)}\leq C\,{\color{blue}e^{-\alpha a}},
\qquad
c\,{\color{blue}e^{(1-\alpha)a}}\leq \|\Sigma_{a}\|_{H^{3}(M,g)}\leq C\,{\color{blue}e^{(1-\alpha)a}}.
\]
This completes the proof.
\end{proof}

\begingroup\color{blue}
\noindent We shall apply this construction with $g=g_{1,a}$, so we verify that its constants can be chosen independently of $a$. In this paragraph all Sobolev norms are taken with respect to the fixed metric $\gamma$. Corollary~\ref{cor:separation-scales} gives
\[
\|g-\gamma\|_{H^3}\le C e^{-\alpha a},\qquad
\|g-\gamma\|_{H^4}\le C e^{(1-\alpha)a}.
\]
Choose the chart, phase $\psi$, and cutoff $\chi$ independently of $a$, and construct the local frame by smooth Gram--Schmidt orthogonalization with respect to $g$. The resulting tensors $q$ are uniformly bounded in $H^3$ and $C^1$-close to the corresponding tensor for $\gamma$, since $\|g-\gamma\|_{H^3}\le C e^{-\alpha a}$. The identity $q(\nabla_g\psi,\cdot)=0$ and the preceding product estimates therefore give
\[
\|\div_g S_\mu\|_{H^1}\le C\mu,
\qquad
\|\div_g S_\mu\|_{H^2}\le C\mu^2,
\qquad
\|S_\mu\|_{H^m}\le C\mu^m\quad(m=2,3),
\]
with constants independent of $a$. The lower bounds for $S_\mu$ are uniform as well: on a fixed smaller chart the leading differentiated term is $\mu^m b_m(\mu\psi)q$, and integration by parts in the fixed coordinate $\psi=x^1$ bounds the oscillatory part of its squared $L^2$ norm by $C\mu^{-1}$. Its nonoscillatory part has a positive uniform lower bound, while the differentiated remainder is bounded by $C\mu^{m-1}$.

\noindent Write $\Delta_{\mathbb L,g}=-\div_g\mathbb L_g$. Comparing its coefficients with those for $\gamma$, using Sobolev multiplication in dimension three, gives
\begin{align*}
\|(\Delta_{\mathbb L,g}-\Delta_{\mathbb L,\gamma})W\|_{H^1}
&\le C\|g-\gamma\|_{H^3}\|W\|_{H^3},\\
\|(\Delta_{\mathbb L,g}-\Delta_{\mathbb L,\gamma})W\|_{H^2}
&\le C\|g-\gamma\|_{H^3}\|W\|_{H^4}
     +C\|g-\gamma\|_{H^4}\|W\|_{H^3}.
\end{align*}
Indeed, the coefficients of the second-, first-, and zeroth-order terms involve, respectively, zero, one, and two derivatives of $g$; the last term displayed accounts for the fourth metric derivative in the $H^2$ norm of the zeroth-order coefficient. Since $\gamma$ has no conformal Killing fields, $\Delta_{\mathbb L,\gamma}:H^3\to H^1$ is invertible. The first comparison and a bounded-inverse argument give, uniformly for large $a$,
\[
\|W_\mu\|_{H^3}\le C\|\div_gS_\mu\|_{H^1}\le C\mu.
\]
The second comparison, the fixed-background elliptic estimate, and absorption of the $H^3$--small coefficient give
\[
\|W_\mu\|_{H^4}
\le C\bigl(\|\div_gS_\mu\|_{H^2}
             +\|g-\gamma\|_{H^4}\|W_\mu\|_{H^3}\bigr)
\le C\bigl(\mu^2+e^{(1-\alpha)a}\mu\bigr).
\]
The same coefficient estimates for $\mathbb L_g$ imply
\[
\|\mathbb L_gW_\mu\|_{H^2}\le C\mu,
\qquad
\|\mathbb L_gW_\mu\|_{H^3}\le C\bigl(\mu^2+e^{(1-\alpha)a}\mu\bigr).
\]
Taking $\mu=e^a$ gives $e^{(1-\alpha)a}\mu=e^{(2-\alpha)a}\le\mu^2$, so the second bound is at most $C\mu^2$. Thus the York correction remains one power of $\mu$ smaller than $S_\mu$ in both norms, and \eqref{eq:tau_two_sided} holds uniformly along this family. Multiplication by $e^{-(2+\alpha)a}$ proves the stated two-sided estimates for $\Sigma_a$ relative to $\gamma$, with constants independent of $a$.
\par\endgroup

\begin{proposition}[Existence of a unique positive solution to the Lichnerowicz equation]
\label{prop:conformal-correction}
Fix a bounded-geometry class \(\mathcal G\) of smooth Riemannian metrics on the closed
\(3\)-manifold \(M\). Then there exist constants
\[
\delta_{0}>0,\qquad C\ge 1,
\]
depending only on \(\mathcal G\), with the following property.

\noindent Let \(g\in \mathcal G\), and let \(\Sigma\) be a smooth symmetric \(2\)-tensor satisfying
\[
\tr_{g}\Sigma=0,\qquad \div_{g}\Sigma=0,
\]
together with
\[
\left\|R[g]+\frac23\right\|_{H^{2}(g)}\le \delta_{0},
\qquad
\|\Sigma\|_{H^{2}(g)}\le \delta_{0}.
\]
Define
\[
\mathcal E(g,\Sigma):=R[g]+\frac23-|\Sigma|_{g}^{2}.
\]
Then there exists a unique positive solution
\[
\varphi\in H^{4}(M)
\]
of the Lichnerowicz equation
\begin{equation}
\label{eq:Lichnerowicz-exact-corrected}
-8\Delta_{g}\varphi + R[g]\varphi + \frac23 \varphi^{5} - |\Sigma|_{g}^{2}\varphi^{-7}=0.
\end{equation}
Moreover, \(\varphi\) is smooth, and if one defines
\[
g_{0}:=\varphi^{4}g,\qquad \Sigma_{0}:=\varphi^{-2}\Sigma,
\]
then
\begin{align}
R[g_{0}]+\frac23&=|\Sigma_{0}|_{g_{0}}^{2},\\
\div_{g_{0}}\Sigma_{0}&=0,\\
\tr_{g_{0}}\Sigma_{0}&=0.
\end{align}
In addition,
\begin{align}
\label{eq:phi-close-corrected}
\|\varphi-1\|_{H^{4}(g)}
&\leq
C\|\mathcal E(g,\Sigma)\|_{H^{2}(g)},\\
\label{eq:T-change-corrected}
\|\mathfrak T[g_{0}]-\mathfrak T[g]\|_{H^{2}(g)}
&\leq
C\|\mathcal E(g,\Sigma)\|_{H^{2}(g)}.
\end{align}
\end{proposition}

\begin{proof}
Set
\[
S:=|\Sigma|_{g}^{2},
\qquad
\rho:=R[g]+\frac23.
\]
Then define the following:
\[
\mathcal E(g,\Sigma)=\rho-S,
\]
which is small in the $H^{2}$ norm. 
Since \(M\) is \(3\)-dimensional and \(g\in\mathcal G\), the Sobolev embeddings and
multiplication estimates are uniform in \(g\). In particular,
\[
\|S\|_{H^{2}(g)}\le C\|\Sigma\|_{H^{2}(g)}^{2},
\qquad
\|S\|_{C^{0}(g)}\le C\|\Sigma\|_{H^{2}(g)}^{2},
\qquad
\|\rho\|_{C^{0}(g)}\le C\|\rho\|_{H^{2}(g)}.
\]
Hence, for \(\delta_{0}\) sufficiently small,
\begin{equation}
\label{eq:defect-small-proof}
\|\mathcal E(g,\Sigma)\|_{H^{2}(g)}
\le C\delta_{0},
\qquad
\|\rho\|_{C^{0}(g)}+\|S\|_{C^{0}(g)}\le C\delta_{0}.
\end{equation}

\noindent We write
\[
\varphi=1+u
\]
and define
\[
\mathcal F_{g,\Sigma}(\varphi)
:=
-8\Delta_{g}\varphi + R[g]\varphi + \frac23\varphi^{5}-S\varphi^{-7}.
\]
Then \eqref{eq:Lichnerowicz-exact-corrected} is equivalent to
\[
\mathcal F_{g,\Sigma}(1+u)=0.
\]
Moreover,
\[
\mathcal F_{g,\Sigma}(1)=\rho-S=\mathcal E(g,\Sigma).
\]

\noindent First, we understand the linearization of $\mathcal F_{g,\Sigma}$. The linearization at \(u=0\) is
\[
L_{g,\Sigma}u
:=
D\mathcal F_{g,\Sigma}|_{\varphi=1}[u]
=
-8\Delta_{g}u+\Bigl(R[g]+\frac{10}{3}+7S\Bigr)u.
\]
Since
\[
R[g]+\frac{10}{3}+7S
=
\frac83+\rho+7S,
\]
the pointwise bound \eqref{eq:defect-small-proof} implies, after shrinking \(\delta_{0}\),
that
\[
R[g]+\frac{10}{3}+7S\ge 1
\qquad\text{on }M.
\]
Therefore
\[
\int_{M} \langle L_{g,\Sigma}u,u\rangle\,d\mu_{g}
=
8\|\nabla u\|_{L^{2}(g)}^{2}
+
\int_{M}\Bigl(R[g]+\frac{10}{3}+7S\Bigr)u^{2}\,d\mu_{g}
\ge c\|u\|_{H^{1}(g)}^{2}.
\]
Hence \(L_{g,\Sigma}\) has trivial kernel. Since it is a strongly elliptic second-order
operator on a closed manifold, it is Fredholm of index zero, and thus
\[
L_{g,\Sigma}\colon H^{4}(M)\to H^{2}(M)
\]
is an isomorphism. By uniform elliptic regularity on the bounded-geometry class
\(\mathcal G\), there is a constant \(K\) such that
\begin{equation}
\label{eq:elliptic-L-corrected}
\|u\|_{H^{4}(g)}
\le
K\|L_{g,\Sigma}u\|_{H^{2}(g)}.
\end{equation}

\noindent Next we isolate the nonlinear remainder. For \(u\) sufficiently small in \(C^{0}\),
\[
(1+u)^{5}=1+5u+Q_{5}(u),
\qquad
(1+u)^{-7}=1-7u+Q_{-7}(u),
\]
where \(Q_{5}(u)\) and \(Q_{-7}(u)\) vanish quadratically at \(u=0\). Thus
\[
\mathcal F_{g,\Sigma}(1+u)
=
\mathcal E(g,\Sigma)+L_{g,\Sigma}u+\mathcal N_{g,\Sigma}(u),
\]
with
\[
\mathcal N_{g,\Sigma}(u)
=
\frac23 Q_{5}(u)-S\,Q_{-7}(u).
\]

\noindent Let \(C_{\mathrm{emb}}\) denote the uniform norm of the embedding
\(H^{4}(M)\hookrightarrow C^{0}(M)\) on \(\mathcal G\), and set
\[
\rho_{*}:=\frac{1}{4C_{\mathrm{emb}}}.
\]
If \(\|u\|_{H^{4}(g)}\le \rho_{*}\), then \(\|u\|_{C^{0}(g)}\le \frac14\), so
\[
\frac34\le 1+u\le \frac54.
\]
On this set, uniform Moser estimates imply
\begin{align}
\label{eq:Nquad-corrected}
\|\mathcal N_{g,\Sigma}(u)\|_{H^{2}(g)}
&\le
C\|u\|_{H^{4}(g)}^{2},\\
\label{eq:Nlip-corrected}
\|\mathcal N_{g,\Sigma}(u)-\mathcal N_{g,\Sigma}(v)\|_{H^{2}(g)}
&\le
C\bigl(\|u\|_{H^{4}(g)}+\|v\|_{H^{4}(g)}\bigr)\|u-v\|_{H^{4}(g)}.
\end{align}

\noindent Define
\[
\Phi(u):=
-L_{g,\Sigma}^{-1}\bigl(\mathcal E(g,\Sigma)+\mathcal N_{g,\Sigma}(u)\bigr).
\]
Let
\[
r:=2K\|\mathcal E(g,\Sigma)\|_{H^{2}(g)}.
\]
If \(\delta_{0}\) is sufficiently small, then by \eqref{eq:defect-small-proof},
\[
r\le \rho_{*},
\qquad
{CKr\le \frac14},
\qquad
CKr^{2}\le \frac r2.
\]
Now let \(\|u\|_{H^{4}(g)}\le r\). Using \eqref{eq:elliptic-L-corrected} and
\eqref{eq:Nquad-corrected},
\[
\|\Phi(u)\|_{H^{4}(g)}
\le
K\|\mathcal E(g,\Sigma)\|_{H^{2}(g)}
+
K\|\mathcal N_{g,\Sigma}(u)\|_{H^{2}(g)}
\le
\frac r2+CKr^{2}
\le r.
\]
Hence \(\Phi\) maps the closed \(H^{4}\)-ball \(B_{r}\) into itself. Likewise, for
\(u,v\in B_{r}\), \eqref{eq:Nlip-corrected} gives
\[
\|\Phi(u)-\Phi(v)\|_{H^{4}(g)}
\le
K\|\mathcal N_{g,\Sigma}(u)-\mathcal N_{g,\Sigma}(v)\|_{H^{2}(g)}
\le
2CKr\,\|u-v\|_{H^{4}(g)}
\le
\frac12\|u-v\|_{H^{4}(g)}.
\]
Thus \(\Phi\) is a contraction on \(B_{r}\). By Banach's fixed-point theorem, there exists
a unique \(u\in B_{r}\) such that \(\Phi(u)=u\). Setting \(\varphi:=1+u\), we obtain a
solution of \eqref{eq:Lichnerowicz-exact-corrected} with
\[
\|\varphi-1\|_{H^{4}(g)}=\|u\|_{H^{4}(g)}
\le
C\|\mathcal E(g,\Sigma)\|_{H^{2}(g)},
\]
which proves \eqref{eq:phi-close-corrected}. Since \(\|u\|_{C^{0}(g)}\le \frac14\), we also
have \(\varphi\ge \frac34>0\).

\begingroup
\noindent We now prove uniqueness among \emph{all} positive solutions. Let $\varphi$ be any positive solution of \eqref{eq:Lichnerowicz-exact-corrected}, and set
\[
m:=\min_M\varphi,\qquad M:=\max_M\varphi.
\]
At a maximum point, $\Delta_g\varphi\le0$, so
\[
\frac23 M(M^4-1)\le |\rho|M+SM^{-7}.
\]
If $M\ge1$, division by $M$ gives
\[
M^4\le1+\frac32\bigl(\|\rho\|_{C^0(g)}+\|S\|_{C^0(g)}\bigr).
\]
If $M<1$, the desired upper bound already holds. At a minimum point, $\Delta_g\varphi\ge0$, and the equation gives
\[
0\le\Bigl(-\frac23+\rho\Bigr)m+\frac23m^5-Sm^{-7}.
\]
Since $S\ge0$ and $m>0$, it follows that
\[
m^4\ge1-\frac32\rho+\frac32Sm^{-8}
\ge1-\frac32\|\rho\|_{C^0(g)}.
\]
By \eqref{eq:defect-small-proof}, decreasing $\delta_0$ therefore ensures
\[
\frac34\le m\le M\le\frac54.
\]
For two positive solutions $\varphi_1,\varphi_2$, set $w=\varphi_1-\varphi_2$. Subtraction and the mean value theorem give
\[
-8\Delta_gw+a(x)w=0,
\qquad
a(x)=R[g]+\frac{10}{3}\theta_1(x)^4+7S(x)\theta_2(x)^{-8},
\]
where $\theta_1,\theta_2\in[3/4,5/4]$. Hence
\[
a(x)\ge-\|\rho\|_{C^0(g)}-\frac23+\frac{10}{3}\Bigl(\frac34\Bigr)^4\ge c_*>0
\]
for sufficiently small $\delta_0$. Multiplying the equation by $w$ and integrating gives
\[
8\|\nabla w\|_{L^2(g)}^2+\int_Ma(x)w^2\,d\mu_g=0.
\]
Thus $w=0$, proving uniqueness.
\par\endgroup

\noindent Define
\[
g_{0}:=\varphi^{4}g,\qquad \Sigma_{0}:=\varphi^{-2}\Sigma.
\]
The standard conformal covariance identities in dimension \(3\) give
\[
R[g_{0}]
=
\varphi^{-5}\bigl(-8\Delta_{g}\varphi+R[g]\varphi\bigr),
\qquad
|\Sigma_{0}|_{g_{0}}^{2}
=
\varphi^{-12}|\Sigma|_{g}^{2}.
\]
Hence \eqref{eq:Lichnerowicz-exact-corrected} is equivalent to
\[
R[g_{0}]+\frac23=|\Sigma_{0}|_{g_{0}}^{2}.
\]
Since \(\Sigma\) is \(g\)-trace-free and \(g\)-divergence-free, the standard conformal
transformation law for TT tensors implies
\[
\tr_{g_{0}}\Sigma_{0}=0,
\qquad
\div_{g_{0}}\Sigma_{0}=0.
\]

\noindent It remains to prove \eqref{eq:T-change-corrected}. In dimension \(3\),
\[
\mathfrak T[\varphi^{4}g]
=
\mathfrak T[g]
-
2\varphi^{-1}\Bigl(\nabla^{2}\varphi-\frac13(\Delta_{g}\varphi)g\Bigr)
+
6\varphi^{-2}\Bigl(d\varphi\otimes d\varphi-\frac13|\nabla\varphi|_{g}^{2}g\Bigr).
\]
Therefore
\begin{align*}
\mathfrak T[g_{0}]-\mathfrak T[g]
&=
-2\varphi^{-1}\Bigl(\nabla^{2}\varphi-\frac13(\Delta_{g}\varphi)g\Bigr)\\
&\qquad
+
6\varphi^{-2}\Bigl(d\varphi\otimes d\varphi-\frac13|\nabla\varphi|_{g}^{2}g\Bigr).
\end{align*}
Since \(\varphi\in[\frac34,\frac54]\) pointwise and
\(\|\varphi-1\|_{H^{4}(g)}\le C\|\mathcal E(g,\Sigma)\|_{H^{2}(g)}\), the uniform composition
estimate applied to \(s\mapsto(1+s)^{-j}-1\), which vanishes at \(s=0\), gives
\[
\|\varphi^{-j}-1\|_{H^{4}(g)}\le C\|\varphi-1\|_{H^{4}(g)},\qquad j=1,2.
\]
Writing \(\varphi^{-j}=1+(\varphi^{-j}-1)\), the constant part acts by the identity.
The \(H^{2}(g)\) product estimate for the remaining part and the smallness of \(\|\varphi-1\|_{H^{4}(g)}\) therefore give
\[
\left\|
\varphi^{-1}\Bigl(\nabla^{2}\varphi-\frac13(\Delta_{g}\varphi)g\Bigr)
\right\|_{H^{2}(g)}
\le
C\|\varphi-1\|_{H^{4}(g)},
\]
and
\[
\left\|
\varphi^{-2}\Bigl(d\varphi\otimes d\varphi-\frac13|\nabla\varphi|_{g}^{2}g\Bigr)
\right\|_{H^{2}(g)}
\le
C\|\varphi-1\|_{H^{4}(g)}^{2}.
\]
Hence
\[
\|\mathfrak T[g_{0}]-\mathfrak T[g]\|_{H^{2}(g)}
\le
C\|\varphi-1\|_{H^{4}(g)}
+
C\|\varphi-1\|_{H^{4}(g)}^{2}.
\]
For \(\delta_{0}\) sufficiently small, the quadratic term is absorbed into the linear one,
and \eqref{eq:phi-close-corrected} yields
\[
\|\mathfrak T[g_{0}]-\mathfrak T[g]\|_{H^{2}(g)}
\le
C\|\mathcal E(g,\Sigma)\|_{H^{2}(g)}.
\]
This proves \eqref{eq:T-change-corrected}.

\noindent Finally, since \(g\) and \(\Sigma\) are smooth and \(\varphi\) is bounded above and below
away from zero, standard elliptic bootstrapping applied to
\eqref{eq:Lichnerowicz-exact-corrected} shows that \(\varphi\in C^{\infty}(M)\).
\end{proof}
\begin{corollary}[Persistence of a large \(H^{2}\)-norm of the trace-free Ricci tensor]
\label{cor:persistence}
Under the hypotheses of Proposition~\ref{prop:conformal-correction}, let \(A>0\). Assume in addition that
\[
\|\mathfrak T[g]\|_{H^{2}(g)}\ge A,
\qquad
\|\mathcal E(g,\Sigma)\|_{H^{2}(g)}\le \frac{A}{2C},
\]
where \(C\) is the constant in \eqref{eq:T-change-corrected}. Then
\[
\|\mathfrak T[g_{0}]\|_{H^{2}(g)}\ge \frac{A}{2}.
\]
\end{corollary}

\begin{proof}
By \eqref{eq:T-change-corrected},
\[
\|\mathfrak T[g_{0}]-\mathfrak T[g]\|_{H^{2}(g)}
\le
C\|\mathcal E(g,\Sigma)\|_{H^{2}(g)}
\le
\frac{A}{2}.
\]
Therefore, by the reverse triangle inequality,
\[
\|\mathfrak T[g_{0}]\|_{H^{2}(g)}
\ge
\|\mathfrak T[g]\|_{H^{2}(g)}
-
\|\mathfrak T[g_{0}]-\mathfrak T[g]\|_{H^{2}(g)}
\ge
A-\frac{A}{2}
=
\frac{A}{2}.
\]
This proves the claim.
\end{proof}
\noindent The following corollary provides the exact initial data that we can address in this study.
\begin{corollary}[Exact CMC initial data]
\label{cor:final-CMC}
Fix the smooth background metric $\gamma$ of Proposition~\ref{prop:YorkTTperturb}, with $R[\gamma]=-2/3$ and no nontrivial conformal Killing fields, and fix $0<\alpha<1$. There are constants $0<c\le C$ depending only on $(M,\gamma)$ with the following property. For all sufficiently large $a$ (depending also on $\alpha$), there exists a smooth initial data set \((g_{0,a},\Sigma_{0,a})\) on \(M\) such that
\begin{align}
R[g_{0,a}]+\frac23&=|\Sigma_{0,a}|_{g_{0,a}}^{2},\\
\div_{g_{0,a}}\Sigma_{0,a}&=0,\\
\tr_{g_{0,a}}\Sigma_{0,a}&=0.
\end{align}
Moreover,
\begin{equation}
\label{eq:g0a-gamma-equivalence}
C^{-1}\gamma\le g_{0,a}\le C\gamma,
\end{equation}
and, with all Sobolev norms and covariant derivatives taken with respect to \(\gamma\),
\begin{equation}
\label{eq:final-CMC-upper}
\sum_{I=0}^{3}\|\nabla_\gamma^{I}\Sigma_{0,a}\|_{L^{2}(\gamma)}
+\sum_{I=0}^{2}\Bigl(
\|\nabla_\gamma^{I}\mathfrak T[g_{0,a}]\|_{L^{2}(\gamma)}
+
\|e^{a}\nabla_\gamma^{I}\Sigma_{0,a}\|_{L^{2}(\gamma)}
\Bigr)
\le C {\color{blue}e^{(1-\alpha)a}},
\end{equation}
and moreover
\begin{equation}
\label{eq:final-CMC-lower}
\begin{aligned}
\sum_{I=0}^{3}\|\nabla_\gamma^{I}\Sigma_{0,a}\|_{L^{2}(\gamma)}
&\approx {\color{blue}e^{(1-\alpha)a}},\\
\sum_{I=0}^{2}\|\nabla_\gamma^{I}\mathfrak T[g_{0,a}]\|_{L^{2}(\gamma)}
&\approx {\color{blue}e^{(1-\alpha)a}},\\
\sum_{I=0}^{2}\|e^{a}\nabla_\gamma^{I}\Sigma_{0,a}\|_{L^{2}(\gamma)}
&\approx {\color{blue}e^{(1-\alpha)a}},
\end{aligned}
\end{equation}
with constants independent of $a$.
\begingroup\color{blue}
In addition, $c e^{(1-\alpha)a}\le\|g_{0,a}-\gamma\|_{H^4(\gamma)}\le C e^{(1-\alpha)a}$, $\|g_{0,a}-\gamma\|_{H^3(\gamma)}\le Ce^{-\alpha a}$, and the large curvature and shear norms have the lower bounds
\[
\|\mathfrak T[g_{0,a}]\|_{H^2(\gamma)}\ge c e^{(1-\alpha)a},
\qquad
\|\Sigma_{0,a}\|_{H^3(\gamma)}\ge c e^{(1-\alpha)a}.
\]
\par\endgroup
\end{corollary}

\begin{proof}
We combine the three ingredients established earlier in the construction of the initial data: the seed metric with controlled scalar-curvature defect and large \(H^{2}\)-size of the renormalized trace-free Ricci tensor, the construction of a transverse-traceless tensor of the required size, and the perturbative conformal correction which solves the Hamiltonian constraint exactly.

\medskip

\noindent
By the seed-metric construction proved earlier, there exist constants \(a_{1}\ge 1\) and \(C\ge 1\), with $a_1$ depending only on $\alpha$ and $(M,\gamma)$ and $C$ on $(M,\gamma)$, such that for every $a\ge a_1$ there exists a smooth metric \(\widehat g_{a}\) on \(M\) satisfying
\begin{equation}
\label{eq:seed-metric-cor}
C^{-1}\gamma\le \widehat g_{a}\le C\gamma,
\end{equation}
together with
\begin{equation}
\label{eq:seed-scalar-cor}
\left\|R[\widehat g_{a}]+\frac23\right\|_{H^{2}(\gamma)}
\le C {\color{blue}e^{-2\alpha a}},
\qquad
\left\|R[\widehat g_{a}]+\frac23\right\|_{H^{3}(\gamma)}
\le C {\color{blue}e^{(1-2\alpha)a}},
\end{equation}
and
\begingroup\color{blue}
\begin{equation}
\label{eq:seed-T-cor}
c e^{(1-\alpha)a}\le\|\mathfrak T[\widehat g_a]\|_{H^2(\gamma)}\le C e^{(1-\alpha)a}.
\end{equation}
\par\endgroup

\medskip

\noindent By the transverse-traceless construction proved earlier, after possibly increasing \(a_{1}\), there exists a smooth symmetric \(2\)-tensor \(\widehat\Sigma_{a}\) such that
\begin{equation}
\label{eq:seed-TT-cor}
\div_{\widehat g_{a}}\widehat\Sigma_{a}=0,
\qquad
\tr_{\widehat g_{a}}\widehat\Sigma_{a}=0,
\end{equation}
and
\begingroup\color{blue}
\begin{equation}
\label{eq:seed-sigma-cor}
\begin{aligned}
c e^{-\alpha a}\le\|\widehat\Sigma_a\|_{H^2(\gamma)}&\le C e^{-\alpha a},\\
c e^{(1-\alpha)a}\le\|\widehat\Sigma_a\|_{H^3(\gamma)}&\le C e^{(1-\alpha)a}.
\end{aligned}
\end{equation}
\par\endgroup

\noindent Since $M$ is a closed $3$-manifold, $H^2(M)$ and $H^3(M)$ are Banach algebras. The uniform $H^3(\gamma)$ bound for $\widehat g_a=g_{1,a}$, together with its uniform positive-definiteness, gives uniform $H^3(\gamma)$ bounds for its inverse. Therefore
\begin{equation}
\label{eq:sigma-square-H2-cor}
\bigl\||\widehat\Sigma_{a}|_{\widehat g_{a}}^{2}\bigr\|_{H^{2}(\gamma)}
\le C \|\widehat\Sigma_{a}\|_{H^{2}(\gamma)}^{2}
\le C {\color{blue}e^{-2\alpha a}}.
\end{equation}
At order three we apply the tame product inequality in the fixed metric $\gamma$ to the contraction of $(\widehat g_a^{-1})^{\otimes2}$ with $\widehat\Sigma_a^{\otimes2}$. The uniform $H^3(\gamma)$ bound for the first factor gives an upper bound of $C(\|\widehat\Sigma_a\|_{L^\infty}^2+\|\widehat\Sigma_a^{\otimes2}\|_{H^3(\gamma)})$, including all mixed derivatives. A second application of the same inequality gives $\|\widehat\Sigma_a^{\otimes2}\|_{H^3(\gamma)}\le C\|\widehat\Sigma_a\|_{L^\infty}\|\widehat\Sigma_a\|_{H^3(\gamma)}$. Since $H^2(\gamma)$ embeds in $L^\infty$, we obtain
\begin{equation}
\label{eq:sigma-square-H3-cor}
\bigl\||\widehat\Sigma_{a}|_{\widehat g_{a}}^{2}\bigr\|_{H^{3}(\gamma)}
\le C \|\widehat\Sigma_{a}\|_{H^{2}(\gamma)}\|\widehat\Sigma_{a}\|_{H^{3}(\gamma)}
\le C {\color{blue}e^{(1-2\alpha)a}}.
\end{equation}

\noindent Therefore the Hamiltonian defect of the seed pair \((\widehat g_{a},\widehat\Sigma_{a})\) obeys
\begin{equation}
\label{eq:seed-defect-cor}
\left\|R[\widehat g_{a}]+\frac23-|\widehat\Sigma_{a}|_{\widehat g_{a}}^{2}\right\|_{H^{2}(\gamma)}
\le C {\color{blue}e^{-2\alpha a}},
\end{equation}
and likewise
\begin{equation}
\label{eq:seed-defect-H3-cor}
\left\|R[\widehat g_{a}]+\frac23-|\widehat\Sigma_{a}|_{\widehat g_{a}}^{2}\right\|_{H^{3}(\gamma)}
\le C {\color{blue}e^{(1-2\alpha)a}}.
\end{equation}

\medskip

\noindent We apply the argument of Proposition~\ref{prop:conformal-correction} on the fixed Sobolev spaces associated with $\gamma$, verifying the uniform inverse needed for this varying family. Its linearized operator is
\[
L_{\widehat g_a,\widehat\Sigma_a}
=-8\Delta_{\widehat g_a}+\frac83
 +\Bigl(R[\widehat g_a]+\frac23\Bigr)
 +7|\widehat\Sigma_a|_{\widehat g_a}^2.
\]
The fixed operator $L_{\gamma,0}=-8\Delta_\gamma+8/3$ is an isomorphism from $H^4(\gamma)$ to $H^2(\gamma)$. Sobolev multiplication and the coordinate formula for the scalar Laplacian give
\begin{align*}
\|(L_{\widehat g_a,\widehat\Sigma_a}-L_{\gamma,0})u\|_{H^2(\gamma)}
\le C\Bigl(&\|\widehat g_a-\gamma\|_{H^3(\gamma)}
 +\|R[\widehat g_a]+2/3\|_{H^2(\gamma)}\\
 &+\bigl\||\widehat\Sigma_a|_{\widehat g_a}^2\bigr\|_{H^2(\gamma)}\Bigr)
 \|u\|_{H^4(\gamma)}.
\end{align*}
The coefficient is bounded by $C(e^{-\alpha a}+e^{-2\alpha a})$ by Corollary~\ref{cor:separation-scales} and \eqref{eq:sigma-square-H2-cor}. For all sufficiently large $a$, a bounded-inverse argument therefore yields
\[
\|u\|_{H^4(\gamma)}
\le C\|L_{\widehat g_a,\widehat\Sigma_a}u\|_{H^2(\gamma)},
\]
with $C$ independent of $a$. The nonlinear powers of $1+u$ obey the same uniform $H^2$ multiplication estimates as in Proposition~\ref{prop:conformal-correction}. Indeed the residual has $H^2(\gamma)$ norm at most $Ce^{-2\alpha a}$. For $u=\varphi-1$ in a fixed small $H^4(\gamma)$ ball, the quadratic remainder of the scalar equation satisfies $\|\mathcal N(u)-\mathcal N(v)\|_{H^2}\le C(\|u\|_{H^4}+\|v\|_{H^4})\|u-v\|_{H^4}$, with $C$ independent of $a$. The inverse estimate and this bound give a contraction on the ball of radius $C'e^{-2\alpha a}$ for all sufficiently large $a$. Elliptic regularity gives a smooth positive function $\varphi_a$ such that, upon setting

\begin{equation}
\label{eq:conformal-cor}
g_{0,a}:=\varphi_{a}^{4}\widehat g_{a},
\qquad
\Sigma_{0,a}:=\varphi_{a}^{-2}\widehat\Sigma_{a},
\end{equation}
the pair \((g_{0,a},\Sigma_{0,a})\) satisfies
\begin{align}
R[g_{0,a}]+\frac23&=|\Sigma_{0,a}|_{g_{0,a}}^{2},\\
\div_{g_{0,a}}\Sigma_{0,a}&=0,\\
\tr_{g_{0,a}}\Sigma_{0,a}&=0.
\end{align}
Moreover, the conformal theorem yields the perturbative estimate
\begin{equation}
\label{eq:phi-small-cor}
\|\varphi_{a}-1\|_{H^{4}(\gamma)}\le C e^{-2\alpha a}.
\end{equation}

\noindent Since \(H^{4}(M)\hookrightarrow C^{2}(M)\) on a closed \(3\)-manifold, \eqref{eq:phi-small-cor} implies
\[
\|\varphi_{a}-1\|_{L^{\infty}(M)}\le C e^{-2\alpha a}.
\]
\noindent Hence, for sufficiently large $a$, we have
\[
\frac12\le \varphi_{a}\le 2
\qquad\text{on }M.
\]
Combining this with \eqref{eq:seed-metric-cor}, we obtain
\[
C^{-1}\gamma\le g_{0,a}\le C\gamma,
\]
which proves \eqref{eq:g0a-gamma-equivalence}.

\medskip

\noindent
From \eqref{eq:conformal-cor},
\[
\Sigma_{0,a}-\widehat\Sigma_{a}=(\varphi_{a}^{-2}-1)\widehat\Sigma_{a}.
\]
By standard composition estimates in Sobolev spaces and \eqref{eq:phi-small-cor},
\begin{equation}
\label{eq:phi-inv-small-cor}
\|\varphi_{a}^{-2}-1\|_{H^{4}(\gamma)}
\le C\|\varphi_{a}-1\|_{H^{4}(\gamma)}
\le C e^{-2\alpha a}.
\end{equation}
Using the algebra property of \(H^{2}\) and \(H^{3}\), together with \eqref{eq:seed-sigma-cor}, we infer
\begin{align}
\|\Sigma_{0,a}\|_{H^{2}(\gamma)}
&\le \|\widehat\Sigma_{a}\|_{H^{2}(\gamma)}
   +\|(\varphi_{a}^{-2}-1)\widehat\Sigma_{a}\|_{H^{2}(\gamma)} \nonumber\\
&\le C e^{-\alpha a}+C e^{-2\alpha a}e^{-\alpha a}
\le C e^{-\alpha a},
\label{eq:sigma-H2-cor}
\end{align}
and likewise
\begin{align}
\|\Sigma_{0,a}\|_{H^{3}(\gamma)}
&\le \|\widehat\Sigma_{a}\|_{H^{3}(\gamma)}
   +\|(\varphi_{a}^{-2}-1)\widehat\Sigma_{a}\|_{H^{3}(\gamma)} \nonumber\\
&\le C e^{(1-\alpha)a}+C e^{-2\alpha a}e^{(1-\alpha)a}
\le C e^{(1-\alpha)a}.
\label{eq:sigma-H3-cor}
\end{align}
The lower shear bound also persists. Indeed,
\[
\|\Sigma_{0,a}-\widehat\Sigma_a\|_{H^j(\gamma)}
\le C e^{-2\alpha a}\|\widehat\Sigma_a\|_{H^j(\gamma)},\qquad j=2,3.
\]
For sufficiently large $a$, the reverse triangle inequality and \eqref{eq:seed-sigma-cor} give
\[
\|\Sigma_{0,a}\|_{H^2(\gamma)}\ge ce^{-\alpha a},\qquad \|\Sigma_{0,a}\|_{H^3(\gamma)}\ge c e^{(1-\alpha)a}.
\]

\noindent This proves the first and last bounds in \eqref{eq:final-CMC-lower}.

\medskip

\noindent To estimate the curvature correction, use the exact conformal identity from Proposition~\ref{prop:conformal-correction}:
\begin{equation}
\label{eq:g-difference-cor}
\begin{aligned}
\mathfrak T[g_{0,a}]-\mathfrak T[\widehat g_a]
={}&-2\varphi_a^{-1}\Bigl(\nabla_{\widehat g_a}^{2}\varphi_a
               -\frac13(\Delta_{\widehat g_a}\varphi_a)\widehat g_a\Bigr)\\
&+6\varphi_a^{-2}\Bigl(d\varphi_a\otimes d\varphi_a
               -\frac13|\nabla\varphi_a|_{\widehat g_a}^{2}\widehat g_a\Bigr).
\end{aligned}
\end{equation}
The $H^3(\gamma)$ bounds for $\widehat g_a$ and its inverse are uniform. Thus the Christoffel-symbol difference is uniformly bounded in $H^2(\gamma)$, and
$\|\nabla_{\widehat g_a}^{2}(\varphi_a-1)\|_{H^2(\gamma)}
\le C\|\varphi_a-1\|_{H^4(\gamma)}$.
The algebra property of $H^2(\gamma)$ and \eqref{eq:phi-small-cor} therefore give
\begin{equation}
\label{eq:T-difference-cor}
\begin{aligned}
\|\mathfrak T[g_{0,a}]-\mathfrak T[\widehat g_a]\|_{H^2(\gamma)}
&\le C\bigl(\|\varphi_a-1\|_{H^4(\gamma)}
              +\|\varphi_a-1\|_{H^4(\gamma)}^2\bigr)\\
&\le C e^{-2\alpha a}.
\end{aligned}
\end{equation}

\noindent Combining \eqref{eq:T-difference-cor} with \eqref{eq:seed-T-cor}, we conclude
\begin{equation}
\label{eq:T-final-cor}
c e^{(1-\alpha)a}\le\|\mathfrak T[g_{0,a}]\|_{H^2(\gamma)}\le C e^{(1-\alpha)a}.
\end{equation}

\noindent This is exactly the second bound asserted in \eqref{eq:final-CMC-lower}.

\medskip

\noindent
By definition of Sobolev norms with respect to \(\gamma\),
\[
\sum_{I=0}^{3}\|\nabla_{\gamma}^{I}\Sigma_{0,a}\|_{L^{2}(\gamma)}
\le C \|\Sigma_{0,a}\|_{H^{3}(\gamma)}
\le C e^{(1-\alpha)a}
\]
by \eqref{eq:sigma-H3-cor}. Also,
\[
\sum_{I=0}^{2}\|e^{a}\nabla_{\gamma}^{I}\Sigma_{0,a}\|_{L^{2}(\gamma)}
\le C e^{a}\|\Sigma_{0,a}\|_{H^{2}(\gamma)}
\le C e^{a}e^{-\alpha a}
= C e^{(1-\alpha)a}
\]
by \eqref{eq:sigma-H2-cor}. Finally,
\[
\sum_{I=0}^{2}\|\nabla_{\gamma}^{I}\mathfrak T[g_{0,a}]\|_{L^{2}(\gamma)}
\le C \|\mathfrak T[g_{0,a}]\|_{H^{2}(\gamma)}
\le Ce^{(1-\alpha)a}
\]
by \eqref{eq:T-final-cor}. This proves \eqref{eq:final-CMC-upper}.

\noindent The fourth derivative of the final metric is also large. Write $k=g_{1,a}-\gamma$ and $v=\varphi_a^4-1$. Then $g_{0,a}-g_{1,a}=v\gamma+vk$, and composition gives $\|v\|_{H^4}\le Ce^{-2\alpha a}$. Since $H^4(\gamma)$ is an algebra,
\[
 \|g_{0,a}-g_{1,a}\|_{H^4}\le C(1+\|k\|_{H^4})e^{-2\alpha a}
 \le C(1+e^{(1-\alpha)a})e^{-2\alpha a}.
\]
Taking $a$ sufficiently large and using the two-sided estimate on $k$ yields
\[
 c e^{(1-\alpha)a}\le \|g_{0,a}-\gamma\|_{H^4(\gamma)}\le C e^{(1-\alpha)a},\qquad
 \|g_{0,a}-\gamma\|_{H^3(\gamma)}\le C\bigl(e^{-\alpha a}+e^{-2\alpha a}\bigr).
\]
These estimates show which derivative of the metric is large. The $H^3$ bound, metric comparison, and Sobolev multiplication in dimension three give uniform equivalence of the tensor norms through order three for $g_{0,a}$ and $\gamma$. Indeed, the difference of their connections is bounded in $H^2(\gamma)$, and expanding the first three covariant derivatives gives products controlled by $H^2\cdot H^j\subset H^j$, $0\le j\le2$. The inverse comparison follows in the same way, or by absorption for sufficiently large $a$.

\noindent Let $C\ge1$ include the uniform constant in \eqref{eq:final-CMC-upper} and the comparison of the norms for $g_{0,a}$ and $\gamma$. Choose
\[
 \mathcal I^0=2C e^{(1-\alpha)a}.
\]
The actual sum of the initial data norms is at most $\mathcal I^0/2$. Moreover, for sufficiently large $a$,
\[
 (1+\mathcal I^0)e^{-a}
 =e^{-a}+2C e^{-\alpha a}
 \le e^{-\alpha a/2}.
\]
Thus the data satisfy the hypothesis of Theorem~\ref{main} with $\alpha/2$ in place of $\alpha$. The two-sided bounds above give $\|g_{0,a}-\gamma\|_{H^4(\gamma)}\approx e^{(1-\alpha)a}$ on these same slices. The relative $H^4$ error from the conformal correction is $O(e^{-2\alpha a})$. Since $0<\alpha<1$ is arbitrary, the construction realizes every positive exponential growth rate strictly less than one. For each such $a$, continuity in the $H^4\times H^3$ topology gives a relatively open set of smooth constraint data satisfying the upper inequalities and the strict lower inequalities, after decreasing $c$. All constants are independent of $a$.

\noindent Therefore, for every sufficiently large $a$ there exists a smooth exact CMC initial data set \((g_{0,a},\Sigma_{0,a})\) satisfying all the stated properties. This completes the proof.
\end{proof}

\begin{remark}
 Notice the vital point: by choosing the perturbations to be transverse-traceless (or $TT$), we eliminate the principal term in the scalar curvature deviation. This is the precise mechanism that yields almost unmodified scalar curvature in the $H^{2}$ norm and hence, by Sobolev embedding, in the point-wise sense as well. On the other hand, the principal term in the deviation of the trace-free Ricci curvature $\mathfrak{T}$ is the main contributor to its large $H^{2}$ norm.
\end{remark}

\section{Hyperbolic Estimates}

\subsection{Setting up the bootstrap argument}
\label{bootstrapargument}
\leavevmode{\color{blue}We take $\gamma=1$ in the norms defined in Section~3.2. When no time argument is displayed in an estimate on $[T_0,T]$, $\mathcal O$ denotes $\sup_{T_0\leq t\leq T}\mathcal O(t)$.}

\noindent In order to establish global-in-time control of the solutions of the Einstein vacuum equations with cosmological constant $\Lambda>0$---in the rescaled variables introduced earlier---we employ a bootstrap argument. Specifically, we aim to derive uniform bounds on three scale-invariant norms:
\[
\mathcal{O}, \mathcal{F}, \mathcal{N}^\infty,
\]
which respectively control curvature energies, weighted lower-order norms of the transverse-traceless ($TT$) second fundamental form, and the pointwise norms of the gauge variable and the $TT$ tensor $\Sigma$. These norms are defined in the rescaled setting associated with the dynamical time variable $T \in [T_0, T_\infty)$.

\noindent We assume the initial data at time $T = T_0$ satisfy the bound
\begin{align} \label{eq:initial-bound}
{\color{blue}\mathcal{O}(T_0)+\mathcal{F}(T_0)\lesssim\mathcal I^0,\qquad \mathcal N^\infty(T_0)\lesssim\mathcal I^0+(\mathcal I^0)^2,}
\end{align}
where $\mathcal I^0>0$ bounds the initial data size. Our main objective is to prove that for all $T \in [T_0, T_\infty)$, the combined norm satisfies the uniform a priori estimate
\begin{align} \label{eq:goal-bound}
{\color{blue}\mathcal O(T)+\mathcal F(T)\lesssim1+\mathcal I^0,\qquad \mathcal N^\infty(T)\lesssim1+(\mathcal I^0)^2.}
\end{align}
This suffices to propagate the solution globally in time via standard local well-posedness and continuation arguments (see Theorem~\ref{local}).

\noindent To achieve \eqref{eq:goal-bound}, we proceed via a bootstrap scheme. We assume that on a given time interval $[T_0, T] \subset [T_0, T_\infty)$, the following bounds hold:
\begin{align} \label{eq:strap2}
\mathcal{O}(T) \leq \Gamma, \qquad \mathcal{F}(T) \leq \mathds{Y}, \qquad \mathcal{N}^\infty(T) \leq \mathds{L},
\end{align}
\begingroup\color{blue}
where we take
\begin{align}\label{eq:bootstrap-size}
 \Gamma=A(1+\mathcal I^0),\qquad
 \mathds Y=B(1+\mathcal I^0),\qquad
 \mathds L=L(1+\mathcal I^0)^2 .
\end{align}
First choose $A$ larger than the fixed multiple of the initial top-order
norm divided by $1+\mathcal I^0$. Then choose $B$ larger than the fixed
multiple of $1+A$ in the lower-order estimate for $\Sigma$, and $L$
larger than the fixed multiple of $B+B^2$ in the pointwise estimates for
$\Sigma$ and the lapse. Since $T_0=a$, we have
\begin{align}\label{eq:bootstrap-growth}
 (1+\mathcal I^0)e^{-T}&\le e^{-\alpha a}e^{-(T-a)},\qquad
 \mathds L e^{-2T_0}\le L e^{-2\alpha a}.
\end{align}
\par\endgroup
\noindent For each fixed $0<\alpha<1$, this upper bound tends to zero uniformly over the allowed initial sizes. In estimates on the bootstrap interval, a constant with subscripts $A,B$ may depend on these fixed factors; a constant with no such dependence is used for the leading forcing in the lower order shear estimate.

\noindent Fix $T_\infty>T_0$ strictly inside the smooth existence interval. Define the set of admissible bootstrap times
\[
\mathcal{S} := \left\{ T \in [T_0, T_\infty] \mid \text{the bounds \eqref{eq:strap2} hold on } [T_0, T] \right\}.
\]
Our aim is to prove that $\mathcal{S} = [T_0, T_\infty]$ by showing that $\mathcal{S}$ is both open and closed.

\paragraph{Step 1: Non-emptiness} By local well-posedness theory for the Einstein--$\Lambda$ system in the chosen gauge, there exists $\varepsilon > 0$ such that a solution exists on $[T_0, T_0 + \varepsilon]$ and satisfies
\begin{align}
{\color{blue}\mathcal O(T)+\mathcal F(T)\lesssim2\mathcal I^0,\quad \mathcal N^\infty(T)\lesssim2(\mathcal I^0+(\mathcal I^0)^2),\quad T\in[T_0,T_0+\varepsilon],}
\end{align}
thus implying $\mathcal S\neq\emptyset$.

\paragraph{Step 2: Closedness} The norms $\mathcal O$, $\mathcal F$ and $\mathcal N^\infty$ are continuous on this compact subinterval of the smooth lifespan. The inequalities in \eqref{eq:strap2} therefore persist at a limit of times in $\mathcal S$.

\paragraph{Step 3: Openness} We improve the three bounds in \eqref{eq:strap2} simultaneously. The curvature term in the lower shear equation is controlled by $\sup_{T_0\leq t\leq T}\mathcal O(t)$. We therefore fix the bootstrap factor for $\mathcal O$ from the initial energy, then the factors for $\mathcal F$ and $\mathcal N^\infty$ as in \eqref{eq:bootstrap-size}. With those factors fixed, the integrated top-order error is bounded by
\[
 C_{A,B}(1+\mathcal I^0)^4e^{-2a}
 \leq C_{A,B}e^{-2\alpha a}(1+\mathcal I^0)^2
\]
by \eqref{eq:bootstrap-growth}. The three estimates are as follows:

\begin{itemize}
  \item The gauge quantities in $\mathcal{N}^\infty$ are controlled via elliptic regularity theory and Sobolev inequalities on the spatial slices $M_T$, leading to the estimate
  \begin{align} \label{eq:estN}
  \mathcal{N}^\infty(T) \lesssim \mathcal{O}(T) + \mathcal{F}(T)^2 + \mathcal{F}(T).
  \end{align}
  
  \item The norm $\mathcal{F}$, which captures the trace-free part of the second fundamental form, is controlled by integrating the transport equation for $\Sigma$ forward in $T$, yielding
  \begin{align} \label{eq:estF}
  {\color{blue}\mathcal F(T)\lesssim\sup_{T_0\leq t\leq T}\mathcal O(t)+\mathcal I^0+1.}
  \end{align}
  
  \item Finally, $\mathcal{O}$ is estimated via energy methods, applied to the hyperbolic system, giving the uniform bound
  \begin{align} \label{eq:estO}
  \mathcal{O}(T) \lesssim \mathcal{I}^0 + 1.
  \end{align}
\end{itemize}

\noindent Combining \eqref{eq:estN}--\eqref{eq:estO}, with the choices in \eqref{eq:bootstrap-size}--\eqref{eq:bootstrap-growth}, we obtain strict improvements of the assumptions in \eqref{eq:strap2}. Continuity of the norms now gives openness of $\mathcal S$ relative to $[T_0,T_\infty]$.

\noindent Since $\mathcal{S} \subset [T_0, T_\infty]$ is non-empty, open, and closed, and $[T_0, T_\infty]$ is connected, it follows that
\[
\mathcal{S} = [T_0, T_\infty].
\]
Since $T_\infty$ was arbitrary, the improved bounds hold throughout the smooth lifespan. Section~\ref{proof1} uses these bounds to exclude a finite maximal time.

\subsection{Estimates for the metric components}
\label{metriccontrol}
\begin{proposition}[Uniform Control of the Metric Volume Density]
\label{metric1}
Let $(\widetilde{M}, \widehat{g}(T))$ be a smooth globally hyperbolic spacetime foliated by constant mean curvature (CMC) slices, and suppose that the normalized spatial metric \( g(T) \) evolves according to the Einstein vacuum equations with cosmological constant in the CMC-spatially transported gauge. Let the initial data at \( T = T_0 \) satisfy the assumptions of Theorem~\ref{main}, and assume the bootstrap condition
\[
\left\| \frac{N}{n} - 1 \right\|_{L^\infty(M)} \leq {\mathds L e^{-2\gamma T}}, \qquad \text{for all } T \in [T_0, T_\infty],
\]
where \( \mathds{L}, \gamma > 0 \) are constants. Write $d\mu_T$ for the volume density of $g(T)$. Then
\begingroup
\[
 e^{-b}\leq\frac{d\mu_T}{d\mu_{T_0}}\leq e^b,\qquad
 b=\frac{n\mathds L}{2\gamma}e^{-2\gamma T_0},
\]
pointwise throughout the bootstrap interval.
\endgroup
\end{proposition}

\begin{proof}
We begin with the evolution equation for the spatial metric \( g_{ij} \) in CMC-spatially transported gauge:
\[
\partial_T g_{ij} = {\frac{2\varphi}{\tau} N \Sigma_{ij}} - 2\left(1 - \frac{N}{n}\right) g_{ij}.
\]
Taking the trace with respect to \( g^{ij} \), and using the fact that \( \operatorname{tr}_g \Sigma = 0 \), we obtain:
\[
\operatorname{tr}_g(\partial_T g) = -2n\left(1 - \frac{N}{n}\right).
\]
Let \( \mu_g := \sqrt{\det g_{ij}} \) denote the Riemannian volume density. The evolution of \( \mu_g \) is governed by:
\[
\partial_T \mu_g = \frac{1}{2} \mu_g \operatorname{tr}_g(\partial_T g) = -n \mu_g \left(1 - \frac{N}{n}\right),
\]
which implies
\[
\mu_g(T) = \mu_g(T_0) \exp\left(n \int_{T_0}^{T} \left(\frac{N}{n} - 1 \right)(s) \, ds \right).
\]
Applying the bootstrap assumption gives:
\[
\left| \log\left(\frac{\mu_g(T)}{\mu_g(T_0)}\right) \right| \leq n \int_{T_0}^{T} \left| \frac{N}{n} - 1 \right|(s) ds \leq \frac{n \mathds{L}}{2\gamma} e^{-2\gamma T_0}.
\]
Exponentiating, we obtain the bounds:
\[
\mu_g(T_0) e^{-\frac{n\mathds{L}}{2\gamma} e^{-2\gamma T_0}} \leq \mu_g(T) \leq \mu_g(T_0) e^{\frac{n\mathds{L}}{2\gamma} e^{-2\gamma T_0}}.
\]
In any fixed spatial chart, $\mu_g=\sqrt{\det(g_{ij})}$ is the local coefficient of $d\mu_T$. The ratio in the assertion is independent of this choice of chart.
\end{proof}

\begin{corollary}[Exponential Stability of Volume Function]
\label{cor:volume-stability}
Under the same assumptions as in Proposition~\ref{metric1}, the total volume functional satisfies the bound:
\[
{\left| \operatorname{Vol}(g(T)) - \operatorname{Vol}(g(T_0)) \right| \leq \operatorname{Vol}(g(T_0))(e^b-1)}, \qquad \text{for all } T \in [T_0, T_\infty],
\]
where $b$ is as in Proposition~\ref{metric1} and $\operatorname{Vol}(g(T))=\int_M d\mu_T$.
\end{corollary}

\begin{proof}
{Integrate the two pointwise inequalities of Proposition~\ref{metric1} with respect to $d\mu_{T_0}$. Since $1-e^{-b}\leq e^b-1$, the asserted estimate follows. In the bootstrap regime $b\leq1$, one also has $e^b-1\leq e b$.}
\end{proof}
\begin{proposition}
\label{metric2}
Let $(x,T)\in M\times [T_{0},T_{\infty}]$ denote coordinates on a globally hyperbolic spacetime foliated by constant mean curvature (CMC) hypersurfaces, and let $g(T,x)$ denote the normalized spatial metric of an Einstein--$\Lambda$ development as in Proposition~\ref{metric1}, in CMC-transported coordinates. {Let $\alpha(T)$ and $\beta(T)$ denote the supremum over $x$ of the largest eigenvalue and the infimum over $x$ of the smallest eigenvalue, respectively, of $g(T,x)$ relative to $g(T_0,x)$.} Suppose that \eqref{eq:strap2} holds and that $\mathds L e^{-2\gamma T_0}$ is sufficiently small, as follows from \eqref{eq:bootstrap-growth} for $\gamma=1$ and sufficiently large $T_0=a$. Then we have the estimate
\[
|\alpha(T) - 1| + |\beta(T) - 1| \lesssim \mathds{L} e^{-2\gamma T_0},
\]
where $\mathds{L}$ is a constant depending on the norms in the bootstrap assumptions and {$0<\gamma\leq1$ is fixed (here $\gamma=1$).}
\end{proposition}

\begin{proof}
Fix a point $x \in M$, and consider the symmetric bilinear form $g(T,x)$ on $T_x M$. Define the maximal and minimal eigenvalues of $g(T,x)$ with respect to $g(T_0,x)$ at the point $x$ as
\begin{align}
\label{eq:eigdef}
\alpha(x,T) := \sup_{0 \neq Y \in T_x M} \frac{g(T,x)(Y,Y)}{g(T_0,x)(Y,Y)}, \quad
\beta(x,T) := \inf_{0 \neq Y \in T_x M} \frac{g(T,x)(Y,Y)}{g(T_0,x)(Y,Y)}.
\end{align}
We aim to estimate $|\alpha(x,T)-1|$ and $|\beta(x,T)-1|$ uniformly in $x$ and $T \in [T_0,T_\infty]$.
Recall the evolution equation for the metric in CMC-transported coordinates:
\begin{align}
\label{eq:metric_evolution}
\partial_T g_{ij} = {\frac{2 \varphi}{\tau} N \Sigma_{ij}} - 2\left(1 - \frac{N}{n}\right) g_{ij},
\end{align}
where $N$ is the lapse function, $\Sigma_{ij}$ is the transverse-traceless part of the second fundamental form, $\tau$ is the mean curvature, and $\varphi$ is a smooth weight function. For a fixed tangent vector $Y \in T_x M$, we obtain from \eqref{eq:metric_evolution}
\begin{align}
\label{eq:scalar_evolution}
\partial_T g(Y,Y) = {\frac{2\varphi}{\tau} N \Sigma(Y,Y)} - 2\left(1 - \frac{N}{n} \right) g(Y,Y).
\end{align}
To estimate the right-hand side of \eqref{eq:scalar_evolution}, we begin with the bound
\begin{align}
\label{eq:sigma_pointwise}
{|\Sigma(Y,Y)|\leq |\Sigma|_{g(T,x)}\,g(T,x)(Y,Y),}
\end{align}
{since $\Sigma$ is symmetric. Thus, using the bootstrap assumption \eqref{eq:strap2}, we have}
\[
{|\Sigma(Y,Y)|\leq \|\Sigma\|_{L^\infty(M)}\,g(T,x)(Y,Y).}
\]
\begingroup
Consequently,
\[
\left|\partial_T\log g(T,x)(Y,Y)\right|
\lesssim\mathds L\left(e^{-(1+\gamma)T}+e^{-2\gamma T}\right).
\]
Integration gives
\[
e^{-C\mathds L e^{-2\gamma T_0}}
\leq\frac{g(T,x)(Y,Y)}{g(T_0,x)(Y,Y)}
\leq e^{C\mathds L e^{-2\gamma T_0}}.
\]
Since $\mathds L e^{-2\gamma T_0}$ is small by hypothesis, this proves the uniform comparison used in the following integral estimate.
\endgroup
Now, using \eqref{eq:scalar_evolution}, we integrate from $T_0$ to $T$:
\begin{align}
|g(T,x)(Y,Y) - g(T_0,x)(Y,Y)| &\leq \int_{T_0}^T \left| \partial_s g(s,x)(Y,Y) \right| ds \\
&\lesssim \int_{T_0}^T \left( \left|\frac{\varphi(s)}{\tau(s)}\right| \cdot |N \Sigma(Y,Y)| + \left|1 - \frac{N}{n}\right| g(Y,Y) \right) ds.
\end{align}
From the bootstrap assumption, we have the decay bounds:
\[
\left| \frac{\varphi(s)}{\tau(s)} \right| \lesssim e^{-s}, \quad \|\Sigma\|_{L^\infty(M)} \lesssim {\mathds L e^{-\gamma s}}, \quad \left\| 1 - \frac{N}{n} \right\|_{L^\infty(M)} \lesssim {\mathds L e^{-2\gamma s}}.
\]
Therefore,
\begin{align}
|g(T,x)(Y,Y) - g(T_0,x)(Y,Y)| &\lesssim {\mathds L}\int_{T_0}^T e^{-s} e^{-\gamma s} g(T_0,x)(Y,Y) ds + {\mathds L\int_{T_0}^T e^{-2\gamma s}} g(T_0,x)(Y,Y) ds \\
&\lesssim {\mathds L}\left( \int_{T_0}^\infty e^{-(1+\gamma) s} ds + \int_{T_0}^\infty e^{-2\gamma s} ds \right) g(T_0,x)(Y,Y).
\end{align}
This yields
\[
|g(T,x)(Y,Y) - g(T_0,x)(Y,Y)| \lesssim {\mathds L e^{-2\gamma T_0}} g(T_0,x)(Y,Y).
\]
To improve this to an estimate of relative deviation, we write
\[
\left| \frac{g(T,x)(Y,Y)}{g(T_0,x)(Y,Y)} - 1 \right| \lesssim {\mathds L e^{-2\gamma T_0}}.
\]
Taking the supremum and infimum over unit vectors $Y$ with respect to $g(T_0,x)$ as in \eqref{eq:eigdef}, we obtain
\[
|\alpha(x,T) - 1| + |\beta(x,T) - 1| \lesssim \mathds{L} e^{-2\gamma T_0}.
\]
Since the estimates are uniform in $x$, the result follows.
\end{proof}

\begin{remark}
Propositions~\ref{metric1} and~\ref{metric2} establish uniform two-sided bounds on the evolving spatial metric $g(T,x)$ in terms of the initial metric $g_0 := g(T_0,x)$, valid on the entire bootstrap interval $T \in [T_0, T_{\infty}]$, with $T_0 \gg 1$ sufficiently large. More precisely, there exists a numerical constant $C \geq 1$ such that
\begin{equation}
    C^{-1} g_0 \leq g(T,x) \leq C g_0,
\end{equation}
in the sense of positive-definite symmetric bilinear forms. In particular, the evolving metric remains uniformly equivalent to the initial metric along the foliation. This uniform equivalence enables the definition of time-dependent Sobolev norms $\|\cdot\|_{H^s(g(T))}$ for tensor fields on the spatial slices $M_{T}$, and {ensures uniform scalar Sobolev constants. The covariant elliptic estimates additionally use the curvature bounds established below.} This equivalence plays a critical role in establishing quantitative a priori estimates.
\end{remark}
\noindent We prove the estimates on the Sobolev constants in the following section.

\subsection{Controlling the Sobolev Constants}
\label{sobolev2}
\noindent One of the fundamental aspects of this problem is to control the appropriate Sobolev norms of the solution variables ($\Sigma,\mathfrak{T},N)$ for large times. In addition, one needs to make use of the Sobolev inequalities for the tensor fields on a manifold. Recall that the Sobolev norms are defined with respect to the dynamical metric $g$. This metric verifies the transport equation \eqref{eq:cmc1}. Therefore, one needs to make sure that the metric $g$ does not degenerate during the evolution. Using the estimates derived on the metric components, we define the isoperimetric constant for the Cauchy slice
\begin{align}
 \text{I}(M):=\sup_{\substack{U\subset M,\ \partial U\in C^{1}\\0<\text{vol}(U)<\text{vol}(M)}}\frac{\min\{\text{vol}(U),\text{vol}(U^{c})\}}{[\text{Area}(\partial U)]^{\frac{3}{2}}}.   
\end{align}
The following proposition yields an upper bound for $I(M)$.
\begin{proposition}\label{prop:isoperimetric-control}
Let $(\mathcal{M},\widehat g)$ be a smooth globally hyperbolic vacuum Einstein--$\Lambda$ development of initial data satisfying the hypotheses of Theorem~\ref{main}, with compact CMC slices
\[
M_{T}:=\{T=\mathrm{const}\}
\]
and normalized spatial metric $g(T)$ in transported coordinates. Assume the initial metric satisfies \eqref{eq:1} and the bootstrap bound \eqref{eq:strap2} holds along the slab $T\in[T_{0},T_{\infty})$. Take $T_0=a$ sufficiently large as in \eqref{eq:bootstrap-growth} that $3C\mathds L e^{-2T_0}\le\log 10$, with $C$ the constant in the exponential metric comparison in the proof of Proposition~\ref{metric2}.
Then for every $T\in [T_{0},T_{\infty})$ the isoperimetric constant satisfies
\begin{equation}\label{eq:iso-bound}
     I(M_{T}) \;\leq\; 10\,I(M_{T_{0}}).
\end{equation}
\end{proposition}

\begin{proof}
Fix any $T\in[T_{0},T_{\infty})$ and let $U\subset M_{T}$ be an arbitrary nontrivial Caccioppoli
set whose boundary $\partial U$ is of class $C^{1}$.  
Let $\Phi_{s}$ denote the flow generated by the future-directed vector 
field $\partial_{T}$, so that 
\(
\Phi_{T-T_{0}}\colon M_{T_{0}}\to M_{T}
\)
is a diffeomorphism.
Set
\[
      U_{0}:=\Phi_{-(T-T_{0})}(U)\subset M_{T_{0}},
      \qquad 
      U^{c}_{0}:=\Phi_{-(T-T_{0})}(U^{c}).
\]
{In the transported spatial identification $M_T\simeq M$ this diffeomorphism fixes the spatial label of every point.}

\medskip
\noindent Let $\nu$ denote the outward unit normal to $\partial U$ in $(M_{T},g(T))$ and 
let $\nu_{0}$ denote the outward unit normal to $\partial U_{0}$ 
in $(M_{T_{0}},g(T_{0}))$.  
By Proposition~\ref{metric2}, the differential 
$D\Phi_{T-T_{0}}|_{T_{x}(\partial U_{0})}$ satisfies the uniform
ellipticity bounds
\begin{equation}\label{eq:area-distortion}
\beta(T)
\;\le\;
\frac{d\sigma_{g(T)}}{d\sigma_{g(T_{0})}}
\Big|_{\partial U}
\;\le\;
\alpha(T),
\end{equation}
where $\beta(T)>0$ and $\alpha(T)<\infty$ are the uniform eigenvalue bounds of Proposition~\ref{metric2}.
In particular,
\begin{equation}\label{eq:area-ineq}
   \mathrm{Area}_{g(T)}(\partial U)
   \;\ge\;
   \beta(T)\,
   \mathrm{Area}_{g(T_{0})}(\partial U_{0}).
\end{equation}

\medskip
\noindent Similarly, Proposition~\ref{metric2} implies pointwise bounds for the 
Riemannian volume form:
\begin{equation}\label{eq:vol-distortion}
     \beta(T)^{\frac{3}{2}}
     \;\le\;
     \frac{d\mathrm{vol}_{g(T)}}{d\mathrm{vol}_{g(T_{0})}}
     \;\le\;
     \alpha(T)^{\frac{3}{2}},
\end{equation}
where $\alpha(T),\beta(T)$ are likewise controlled 
by Proposition~\ref{metric2}.  
Integrating \eqref{eq:vol-distortion} over $U_{0}$ and $U^{c}_{0}$ yields
\begin{equation}\label{eq:vol-ineq}
   \mathrm{Vol}_{g(T)}(U)
       \;\le\; \alpha(T)^{\frac{3}{2}}\,\mathrm{Vol}_{g(T_{0})}(U_{0}), 
   \qquad
   \mathrm{Vol}_{g(T)}(U^{c})
       \;\le\; \alpha(T)^{\frac{3}{2}}\,\mathrm{Vol}_{g(T_{0})}(U^{c}_{0}).
\end{equation}

\medskip
\noindent By definition of the isoperimetric constant on $M_{T}$,
\[
   I(M_{T})
   =\sup_{\substack{U\subset M,\ \partial U\in C^{1}\\0<\text{vol}(U)<\text{vol}(M)}}\frac{\min\{\text{vol}(U),\text{vol}(U^{c})\}}{[\text{Area}(\partial U)]^{\frac{3}{2}}}.
\]
Using \eqref{eq:area-ineq} and \eqref{eq:vol-ineq} gives
\[
   I(M_{T})
   \leq \frac{\alpha(T)^{\frac{3}{2}}}{\beta(T)^{\frac{3}{2}}}\sup_{\substack{U_0\subset M,\ \partial U_0\in C^{1}\\0<\text{vol}(U_0)<\text{vol}(M)}}\frac{\min\{\text{vol}(U_{0}),\text{vol}(U^{c}_{0})\}}{[\text{Area}(\partial U_{0})]^{\frac{3}{2}}}=\frac{\alpha(T)^{\frac{3}{2}}}{\beta(T)^{\frac{3}{2}}}I(M_{T_{0}}).
\]
\medskip
By the exponential metric comparison in Proposition~\ref{metric2} with $\gamma=1$, the ratio $(\alpha(T)/\beta(T))^{3/2}$ is bounded above by $\exp(3C\mathds L e^{-2T_0})\leq10$ for every $T\in[T_0,T_\infty)$.  
Hence, $I(M_{T})\;\leq\;10\,I(M_{T_{0}})$,
which establishes the desired estimate \eqref{eq:iso-bound}.
\end{proof}

\begingroup
\begin{proposition}[Sobolev embedding for tensor fields]\label{sobolev}
Let $(M,g)$ be a closed connected Riemannian three-manifold. For a smooth
mixed tensor field $\psi$, write
$V=\mathrm{Vol}_g(M)$ and $J=\max(1,I(M,g))$. With all norms and
covariant derivatives taken with respect to $g$, one has
\begin{align}
 V^{-1/12}\|\psi\|_{L^4}
 &\le C J^{1/2}
    \bigl(\|\nabla\psi\|_{L^2}+V^{-1/3}\|\psi\|_{L^2}\bigr),
    \label{eq:sobolev-L4}\\
 \|\psi\|_{L^\infty}
 &\le C J^{1/2} V^{1/3-1/4}
    \bigl(\|\nabla\psi\|_{L^4}+V^{-1/3}\|\psi\|_{L^4}\bigr).
    \label{eq:sobolevL4}
\end{align}
One also has
\[
 \|\psi\|_{L^6}\le C\bigl(J^{2/3}\|\nabla\psi\|_{L^2}
                          +V^{-1/3}\|\psi\|_{L^2}\bigr).
\]
Here $C$ is a numerical constant.
\end{proposition}
\begin{proof}
First consider a nonnegative scalar function $f$. For a nonnegative
Lipschitz function $h$ with $|\{h>0\}|\le V/2$, 
Minkowski's inequality, the definition of $I(M,g)$ and coarea give
\[
 \|h\|_{L^{3/2}}
 \le \int_0^\infty |\{h>s\}|^{2/3}\,ds
 \le I(M,g)^{2/3}\int_M |dh|\,d\mu_g .
\]
Here $|E|$ denotes $g$-volume, and the inequality extends from smooth
level sets by approximation. Set $k=(2/V)^{1/2}\|f\|_{L^2}$ and
$w=(f-k)_+$. Applying the preceding estimate to $w^4$ gives
$\|w\|_{L^6}^4\le 4 I(M,g)^{2/3}\|w\|_{L^6}^3\|df\|_{L^2}$.
Since $f\le w+k$, this proves the asserted $L^6$ estimate.
To obtain \eqref{eq:sobolev-L4}, use
\[
 V^{-1/12}\|f\|_{L^4}
 \le (V^{-1/3}\|f\|_{L^2})^{1/4}\|f\|_{L^6}^{3/4}.
\]

\noindent For the other scalar inequality, set $t=V/(2J^2)$, $t_j=2^{-j}t$
and $k_j=\inf\{s\ge0: |\{f>s\}|\le t_j\}$. On the interval
$k_j<s<k_{j+1}$ the volume of $\{f>s\}$ lies between $t_{j+1}$
and $t_j\le V/2$. Integrating the isoperimetric inequality over these
levels and using coarea and H\"older gives
\[
 (k_{j+1}-k_j)t_{j+1}^{2/3}
 \le I(M,g)^{2/3}\int_{\{k_j<f<k_{j+1}\}} |df|\,d\mu_g
 \le I(M,g)^{2/3}t_j^{3/4}\|df\|_{L^4}.
\]
Now $k_j\uparrow\|f\|_{L^\infty}$, $k_0\le t^{-1/4}\|f\|_{L^4}$ and
$\sum_{j\ge0}t_j^{1/12}=t^{1/12}/(1-2^{-1/12})$. Summing the preceding
inequality and substituting $t=V/(2J^2)$ yields
\begin{equation}\label{scalar-Sobolev}
 \|f\|_{L^\infty}
 \le C J^{1/2}\bigl(V^{1/12}\|df\|_{L^4}+V^{-1/4}\|f\|_{L^4}\bigr).
\end{equation}
For $f_\varepsilon=(|\psi|_g^2+\varepsilon^2)^{1/2}$, metric compatibility
and Cauchy--Schwarz give
\begin{equation}\label{Kato}
 |df_\varepsilon|_g
 =\frac{|\langle\psi,\nabla\psi\rangle_g|_g}{f_\varepsilon}
 \leq |\nabla\psi|_g .
\end{equation}
Apply the scalar estimates to $f_\varepsilon$ and let $\varepsilon\downarrow0$.
\end{proof}

\begin{corollary}
On the bootstrap interval the estimates
\eqref{eq:sobolev-L4}--\eqref{eq:sobolevL4} and the $L^6$ estimate hold
with volume and isoperimetric factors uniformly controlled in $T$,
independently of the differentiated initial data. In particular,
$\|\psi\|_{L^\infty(g(T))}\leq C\|\psi\|_{H^2(g(T))}$.
\end{corollary}
\begin{proof}
Proposition~\ref{prop:isoperimetric-control} bounds $I(M,g(T))$ in terms
of $I(M,g_0)$, and \eqref{eq:vol-distortion} bounds $\mathrm{Vol}_{g(T)}(M)$
above and below by fixed multiples of $\mathrm{Vol}_{g_0}(M)$.
By \eqref{eq:1} the initial volumes and areas are comparable to those
for $\xi_0$. Equivalently, for scalar functions the comparison with
$\xi_0$ controls both the volume densities and the lengths of the same
covector $df$ in the two metrics, giving the usual unweighted embeddings
with constants depending only on the background and the comparison bound.
Apply \eqref{eq:sobolevL4} to $\psi$ and
\eqref{eq:sobolev-L4} to both $\psi$ and $\nabla\psi$.
\end{proof}

\noindent For later use, products of covariant derivatives have uniform tame estimates.
Indeed, for $j\geq2$ set $X_i=\|\nabla^i A\|_{L^{2j/i}(g)}$
for $1\leq i\leq j$ and $X_0=\|A\|_{L^\infty(g)}$.
Integrating by parts one derivative in
$\int_M|\nabla^i A|^{2j/i}\,d\mu_g$ and applying H\"older gives
$X_i^2\leq C_jX_{i-1}X_{i+1}$ for $1\leq i<j$.
Only the ordered iterated covariant derivatives occur in this calculation.
Iteration of this log-convexity inequality gives
$X_i\leq C_j X_0^{1-i/j}X_j^{i/j}$.
Applying it to each factor in Leibniz' rule gives, for any metric contraction,
\[
 \|A*B\|_{H^j(g)}
 \leq C_j\bigl(\|A\|_{L^\infty(g)}\|B\|_{H^j(g)}
                 +\|B\|_{L^\infty(g)}\|A\|_{H^j(g)}\bigr).
\]

\begin{proposition}
For each fixed smooth metric $g$ on $M$, the tensor Sobolev embedding
$H^I(M,g)\hookrightarrow H^J(M,g)$ is compact for integers $I>J\geq0$,
with $\|\psi\|_{H^I(g)}=\sum_{j=0}^I\|\nabla[g]^j\psi\|_{L^2(g)}$.
\end{proposition}
\begin{proof}
On a finite smooth trivializing atlas, using cutoffs supported in its
charts, the covariant norms for this fixed metric are equivalent to
the corresponding sums of ordinary Sobolev norms of tensor components.
Apply the Rellich compactness theorem to each localized component and
take a common subsequence. The uniform comparison of higher tensor
Sobolev norms in a single fixed metric along the evolution is proved
separately in Lemma~\ref{lem:higher-convergence}.
\end{proof}
\endgroup
\begingroup\noindent For the vacuum equations with $\Lambda=0$ on compact CMC slices of negative Yamabe type, Klainerman and Rodnianski~\cite[Theorem~1.1]{KlainermanRodnianski2010} proved a continuation criterion at negative mean curvature in terms of uniform $L^\infty$ bounds for the second fundamental form and the spatial gradient of the logarithm of the lapse. We use the following local existence and continuation theorem for $\Lambda>0$. 
\begin{theorem}[Local existence and continuation]\label{thm:local-continuation}\label{local}
Let $(M,\widetilde g_0,\widetilde k_0)$ be smooth compact CMC initial
data for $\operatorname{Ric}(\widehat g)=\Lambda\widehat g$, with
$\Lambda>0$ and $\tau_0=\operatorname{tr}_{\widetilde g_0}
\widetilde k_0< -\sqrt{3\Lambda}$. There is a unique local CMC
development, up to the identification of the initial hypersurface and
the choice of its spatial coordinates. Transporting these coordinates
along the normal gives zero shift.

\noindent Suppose such a smooth development is defined for
$\tau\in[\tau_0,\tau_*)$, where $\tau_*< -\sqrt{3\Lambda}$.
It extends beyond $\tau_*$ if the diameter and spatial sectional curvature
are uniformly bounded, the injectivity radius is uniformly bounded below, and
\[
 \sup_{\tau_0\leq\tau<\tau_*}\left\{
 \sum_{j=0}^2\|\widetilde\nabla^j
       \operatorname{Ric}(\widetilde g)\|_{L^2(\widetilde g)}
 +\|\widetilde k\|_{L^4(\widetilde g)}
 +\sum_{j=1}^3\|\widetilde\nabla^j\widetilde k
       \|_{L^2(\widetilde g)}\right\}<\infty.
\]
Here the hypotheses concern the physical spatial metric and second
fundamental form. They are invariant under spatial coordinate changes.
\end{theorem}
\endgroup
\begin{proof}
We apply the prescribed-mean-curvature local existence theorem of
Choquet-Bruhat and York \cite[Section I, compact case]{CBY1996Mixed},
with prescribed mean curvature $\tau$ and identically zero shift.
The cosmological constant contributes only lower-order terms to the
hyperbolic equations, while the lapse satisfies
\[
 \left(-\widetilde\Delta+|\widetilde k|^2-\Lambda\right)
       \widetilde N=1,
 \qquad
 |\widetilde k_0|^2-\Lambda
       \geq\frac{\tau_0^2}{3}-\Lambda>0.
\]
This coercivity persists by $C^0$ closeness to the initial data on a
sufficiently short iteration interval. The cited coupled third-order
construction and the contracted-Bianchi constraint propagation system
therefore give local existence, uniqueness and continuous
dependence directly in CMC--transported spatial gauge, with persistence
of smoothness.

For continuation, the geometric hypotheses give uniform Sobolev
constants and an intrinsic $H^3$ bound for $\widetilde k$. The lapse
equation and the commuted Bochner identities give an $H^3$ bound for
$\widetilde N$, since its potential is bounded below by
$\tau_*^2/3-\Lambda>0$.
The metric evolution gives uniform comparison with
$h=\widetilde g(\tau_0)$; constants may depend on this fixed metric
and the finite interval. Set
$A=\Gamma[\widetilde g]-\Gamma[h]$. Then
\[
 \partial_\tau A^l{}_{ij}
 =-\widetilde g^{lm}\left(
 \widetilde\nabla_i(\widetilde N\widetilde k_{jm})
 +\widetilde\nabla_j(\widetilde N\widetilde k_{im})
 -\widetilde\nabla_m(\widetilde N\widetilde k_{ij})\right).
\]
Thus $\partial_\tau A$ is uniformly bounded in $H^2(\widetilde g)$.
Commuting through two spatial derivatives and differentiating the
equivalent squared Sobolev energy gives
\[
 \frac{d}{d\tau}\|A\|_{H^2(\widetilde g)}^2
 \leq C\left(1+\|A\|_{H^2(\widetilde g)}^2\right).
\]
Since $A(\tau_0)=0$, Gronwall's inequality bounds $A$ uniformly.
The identities relating $\nabla[h]$ to $\widetilde\nabla$ then give
uniform bounds for $\widetilde g$ in $H^3(h)$ and $\widetilde k$ in
$H^2(h)$. These are the initial norms required by the cited local
theorem. Its existence interval is therefore uniform when restarting
near $\tau_*$. Smooth persistence and uniqueness extend the solution
beyond $\tau_*$ in the same transported coordinates.
\end{proof}

\label{intermediate}
\noindent  In this section, we derive uniform estimates for the spatial Riemann curvature tensor $\operatorname{Riem}_{ijkl}$ on each constant mean curvature (CMC) slice, under the bootstrap assumptions introduced in Section~\ref{bootstrapargument}. Our goal is to express $\operatorname{Riem}_{ijkl}$ entirely in terms of the renormalized dynamical variables, namely the trace-free part of the second fundamental form $\Sigma$ and the auxiliary tensor $\mathfrak{T}$, and to obtain estimates compatible with the Sobolev regularity used in our energy framework.

\noindent To this end, we recall that in three spatial dimensions, the Riemann curvature tensor admits the following Ricci decomposition:
\begin{align} \label{eq:ricci-decomposition}
\operatorname{Riem}_{ijkl}
{=g_{ik}\operatorname{Ric}_{jl}+g_{jl}\operatorname{Ric}_{ik}-g_{il}\operatorname{Ric}_{jk}-g_{jk}\operatorname{Ric}_{il}
+\frac12R(g)(g_{il}g_{jk}-g_{ik}g_{jl}),}
\end{align}
where $\operatorname{Ric}_{ij}$ and $R(g)$ denote, respectively, the Ricci tensor and scalar curvature of the normalized Riemannian metric $g$ on the CMC slice.

\noindent The Ricci tensor can be expressed algebraically in terms of its trace-free part $\mathfrak{T}$ via the relation
\begin{align} \label{eq:ricci-mathfrakT}
\operatorname{Ric}_{ij} = \mathfrak{T}_{ij} + \frac{1}{3}R(g) g_{ij}.
\end{align}
Furthermore, the scalar curvature satisfies the renormalized Hamiltonian constraint, which reads
\begin{align} \label{eq:scalar-renorm}
R(g) + \frac{n-1}{n} = |\Sigma|_g^2,
\end{align}
where $|\Sigma|_g^2 = g^{ia}g^{jb} \Sigma_{ij} \Sigma_{ab}$ denotes the pointwise squared norm of the symmetric trace-free tensor $\Sigma$.

\noindent Substituting~\eqref{eq:ricci-mathfrakT} and~\eqref{eq:scalar-renorm} into~\eqref{eq:ricci-decomposition}, we conclude that the full spatial Riemann curvature tensor $\operatorname{Riem}_{ijkl}$ is entirely determined by the tensors $\mathfrak{T}$ and $\Sigma$, together with the metric $g$. That is,
\begin{align} \label{eq:riem-sigma-mathfrakT}
\operatorname{Riem}_{ijkl} = \mathcal{F}_{ijkl}(g,\Sigma,\mathfrak{T}),
\end{align}
where $\mathcal{F}_{ijkl}$ denotes a universal algebraic expression involving the metric, its inverse, and the tensors $\mathfrak{T}$ and $\Sigma$. In particular, pointwise bounds on $\mathfrak{T}$ and $\Sigma$ control the full curvature tensor. At the $H^2$ level, the Sobolev product inequality controls the quadratic $\Sigma$ term, while the constant scalar-curvature term contributes a multiple of $\operatorname{Vol}(g)^{1/2}$ to the $L^2$ norm.

\noindent This observation is central to the curvature estimates employed throughout the remainder of the analysis. In particular, it ensures that the control of geometric quantities such as $\operatorname{Riem}$ reduces to the control of the dynamical variables governed by the Einstein evolution equations in the chosen CMC-transported spatial gauge\footnote{An extension to spatial dimensions $n\geq4$ would require additional estimates for the spatial Weyl curvature, Sobolev estimates at suitable orders, and a replacement for the three-dimensional curvature estimate in Lemma~\ref{lem:pointwise-quadratic-curvature}.}.
\subsection{Elliptic estimates}
\label{ellipticestimates}
\noindent Recall the elliptic equation 
\begin{align}
-\Delta_{g}(\frac{N}{n}-1)+(|\Sigma|^{2}+\frac{1}{3})(\frac{N}{n}-1)=-|\Sigma|^{2}. 
\end{align}
We obtain the following estimate for the normalized lapse function $\frac{N}{n}-1$.
\begingroup
\begin{lemma}[Curvature acting on trace-free two-tensors]
\label{lem:pointwise-quadratic-curvature}
At every point of a Riemannian three-manifold, every symmetric
trace-free two-tensor $A$ satisfies
\[
 |\operatorname{Riem}|\,|A|
 \leq C\bigl(|[\nabla_i,\nabla_j]A|+|R|\,|A|\bigr).
\]
The bracket norm includes both free curvature indices.
\end{lemma}
\begin{proof}
Write the curvature operator on $\Lambda^2T_xM$ as a scalar multiple of
$R\operatorname{Id}$ plus a trace-free self-adjoint operator $\mathcal R_0$.
A nonzero $\mathcal R_0$ has rank at least two, since its real eigenvalues
sum to zero. Identify $\Lambda^2T_xM$ with $\mathfrak{so}(T_xM)$ and
identify $A$ with its self-adjoint endomorphism. The skew endomorphisms
commuting with a nonzero trace-free $A$ form a subspace of dimension at
most one: a commuting map preserves each eigenspace, and their
multiplicities are at most two. Therefore
$[\mathcal R_0(\cdot),A]$ cannot vanish for nonzero $\mathcal R_0,A$.
The norm of this action has a positive minimum on the product of the two
unit spheres. Homogeneity and the triangle inequality for the scalar
part prove the assertion, since the action of the full curvature operator
is the covariant commutator on $A$.
\end{proof}
\endgroup

\begin{proposition}\label{elliptic}
Let $(M,g)$ be a compact Riemannian $3$-manifold on a slice of the
bootstrap interval \eqref{eq:strap2}. Let $N$ solve \eqref{eq:lapse3}
and suppose $\|\Sigma\|_{H^2(M)}\le1$. Then, with $n=3$,
\begin{align}
 \left\|\tfrac Nn-1\right\|_{H^1}
 &\lesssim \|\Sigma\|_{H^1}^2,\notag\\
 \left\|\tfrac Nn-1\right\|_{L^\infty}
 +\left\|\Delta_g(\tfrac Nn-1)\right\|_{L^\infty}
 &\lesssim \|\Sigma\|_{H^2}^2,\notag\\
 \left\|\tfrac Nn-1\right\|_{H^2}
 &\lesssim (1+\|\mathfrak T\|_{H^2})^{1/2}\|\Sigma\|_{H^2}^2,\notag\\
 \left\|\tfrac Nn-1\right\|_{H^3}
 +\left\|\tfrac Nn-1\right\|_{H^4}
 &\lesssim (1+\|\mathfrak T\|_{H^2})\|\Sigma\|_{H^2}^2,\notag\\
 \left\|\tfrac Nn-1\right\|_{H^5}
 &\lesssim (1+\|\mathfrak T\|_{H^2})^2\|\Sigma\|_{H^2}^2
                +\|\Sigma\|_{H^2}\|\Sigma\|_{H^3}.
 \label{eq:lapse_sobolev_estimate}
\end{align}
For any smooth tensor $V$ and any symmetric trace-free two-tensor $A$,
\begin{align}
 \big\||V|\,|\nabla(\tfrac Nn-1)|\big\|_{L^2}
 &\lesssim \|\Sigma\|_{H^2}^2\|V\|_{H^1},\label{eq:weighted-lapse-gradient}\\
 \big\||A|\,|\nabla^2(\tfrac Nn-1)|\big\|_{L^2}
 +\big\||A|\,|\nabla(\tfrac Nn-1)|^2\big\|_{L^2}
 &\lesssim \|\Sigma\|_{H^2}^2\|A\|_{H^2}.
 \label{eq:weighted-lapse-hessian}
\end{align}
In particular,
\begin{equation}\label{eq:weighted-lapse-products}
 \big\||\nabla A|\,|\nabla(\tfrac Nn-1)|\big\|_{L^2}
 +\big\||A|\,|\nabla^2(\tfrac Nn-1)|\big\|_{L^2}
 \lesssim \|\Sigma\|_{H^2}^2\|A\|_{H^2}.
\end{equation}
All norms and covariant derivatives refer to $g$. The constants are
uniform under the metric comparison in Section~\ref{metriccontrol}.
\end{proposition}
\begin{remark}
At the top order $\nabla^3(|\Sigma|^2)$ contains a contraction of $\Sigma$ with $\nabla^3\Sigma$. The bootstrap gives no time decay for $\|\nabla^3\Sigma\|_{L^2}$, while the undifferentiated factor is bounded by $\|\Sigma\|_{L^\infty}\lesssim\|\Sigma\|_{H^2}$.
\end{remark}
\begin{proof}
\noindent
Under the bootstrap assumption~\eqref{eq:strap2}, we establish energy estimates for the spatial Riemann curvature tensor $\operatorname{Riem}$ of the Riemannian metric $g$ on the spatial slice $M$. More precisely, for all integers $0 \leq I \leq 2$, we claim the following uniform bound:
\begin{align} \label{eq:Riem_bound_claim}
    \|\nabla^I \operatorname{Riem} \|_{L^2(M)} \lesssim \Gamma + 1.
\end{align}

\noindent
To derive this, we exploit the classical Ricci decomposition of the Riemann tensor in dimension $n = 3$, where the Weyl tensor vanishes identically:
\begin{align}
    \operatorname{Riem}_{ijkl}=g_{ik}\operatorname{Ric}_{jl}+g_{jl}\operatorname{Ric}_{ik}-g_{il}\operatorname{Ric}_{jk}-g_{jk}\operatorname{Ric}_{il}
+\frac12R(g)(g_{il}g_{jk}-g_{ik}g_{jl}).
\end{align}
\noindent
From this decomposition, it follows that the full Riemann tensor is algebraically expressible in terms of the Ricci tensor and scalar curvature. Invoking the renormalized Hamiltonian constraint, we write:
\begin{align}
    R(g) = -\tfrac{n-1}{n} + |\Sigma|^2,
\end{align}
\noindent
where $\Sigma$ denotes the trace-free part of the second fundamental form in CMC gauge, and where all norms and derivatives are computed with respect to the spatial metric $g$.

\vspace{1ex}
\noindent
Let us now provide schematic bounds for $\nabla^I \operatorname{Riem}$ in terms of the fundamental dynamical fields $\mathfrak{T}$ and $\Sigma$. Using the expressions for $\operatorname{Ric}$ and $R(g)$, we obtain:
\begin{align}
    |\nabla \operatorname{Riem}| 
    &\lesssim |\nabla \mathfrak{T}| + |\nabla R(g)| 
    \lesssim |\nabla \mathfrak{T}| + |\Sigma||\nabla \Sigma|, \\
    |\nabla^2 \operatorname{Riem}| 
    &\lesssim |\nabla^2 \mathfrak{T}| + |\nabla \Sigma|^2 + |\Sigma||\nabla^2 \Sigma|.
\end{align}

\noindent
By applying the Cauchy--Schwarz inequality and standard Sobolev embedding on the compact manifold $M$, together with the bootstrap assumption and bounds derived in Propositions~\ref{metric1}--\ref{metric2}, we deduce:
\begin{align}
    \| \nabla \operatorname{Riem} \|_{L^2(M)} 
    &\lesssim \| \nabla \mathfrak{T} \|_{L^2(M)} 
    + \|\Sigma\|_{H^2(M)}^2
    \lesssim \Gamma + e^{-2\gamma T} \mathds{Y}^2 
    \lesssim \Gamma + 1, \\
    \| \nabla^2 \operatorname{Riem} \|_{L^2(M)} 
    &\lesssim \| \nabla^2 \mathfrak{T} \|_{L^2(M)} + \|\Sigma\|_{H^2(M)}^2 
    \lesssim \Gamma + e^{-2\gamma T} \mathds{Y}^2 
    \lesssim \Gamma + 1.
\end{align}

\noindent
Finally, the $L^2$ bound for $\operatorname{Riem}$ without derivatives follows directly from the pointwise Ricci decomposition and the bootstrap bound on $\mathfrak{T}$ and $\Sigma$:
\begin{align}
    \| \operatorname{Riem} \|_{L^2(M)} \lesssim \Gamma + 1.
\end{align}

\noindent Thus the curvature is controlled through two derivatives;
more precisely,
\[
 \|\operatorname{Riem}\|_{H^2}
 \lesssim 1+\|\mathfrak T\|_{H^2},\qquad
 \|R(g)\|_{L^\infty}\lesssim1 .
\]
Here the $L^4$ and $L^\infty$ norms in the product estimates use
\eqref{eq:sobolev-L4}--\eqref{eq:sobolevL4}; the volume and isoperimetric
factors are bounded by the metric comparison. We now estimate the
lapse one order at a time.

\noindent Now multiply both sides by $\frac{N}{n}-1$ and integrate by parts:
\begin{align}
\label{eq:intby}
 \int_{M}|\nabla(\frac{N}{n}-1)|^{2}\mu_{g}+\int_{M}(|\Sigma|^{2}+\frac{1}{3})(\frac{N}{n}-1)^{2}\mu_{g}=-\int_{M}|\Sigma|^{2}(\frac{N}{n}-1)\mu_{g}.   
\end{align}
Now we estimate the term on the right-hand side by Cauchy--Schwarz as follows:
\begin{align}
-\int_{M}|\Sigma|^{2}(\frac{N}{n}-1)\mu_{g}\leq (n+1)\int_{M}|\Sigma|^{4}\mu_{g}+\frac{1}{n+1}\int_{M}(\frac{N}{n}-1)^{2}\mu_{g}.    
\end{align}
Therefore, \eqref{eq:intby} reads 
\begin{align}
 \int_{M}|\nabla(\frac{N}{n}-1)|^{2}\mu_{g}+\int_{M}|\Sigma|^{2}(\frac{N}{n}-1)^{2}\mu_{g}\nonumber+(\frac{1}{3}-\frac{1}{n+1})\int_{M}(\frac{N}{n}-1)^{2}\mu_{g}\leq (n+1)\int_{M}|\Sigma|^{4}\mu_{g},
\end{align}
which by Proposition~\ref{sobolev} yields 
\begin{align}
  \int_{M}|\nabla(\frac{N}{n}-1)|^{2}\mu_{g}+\int_{M}|\Sigma|^{2}(\frac{N}{n}-1)^{2}\mu_{g}\nonumber+\frac{1}{n(n+1)}\int_{M}(\frac{N}{n}-1)^{2}\mu_{g} \\\nonumber 
  \lesssim \left(\int_{M}|\nabla\Sigma|^{2}\mu_{g}+\int_{M}|\Sigma|^{2}\mu_{g}\right)^{2}.
\end{align}
More concretely, we have the two lowest-order estimates 
\begin{align}
||\frac{N}{n}-1||_{L^{2}(M)}\lesssim ||\Sigma||^{2}_{H^{1}(M)},\\
||\frac{N}{n}-1||_{H^{1}(M)}\lesssim ||\Sigma||^{2}_{H^{1}(M)}.
\end{align}
The maximum principle for $-\Delta_g N+(|\Sigma|^2+1/n)N=1$ also gives
\[
 \frac{n}{1+n\|\Sigma\|_{L^\infty}^2}\le N\le n,\qquad
 \left\|\tfrac Nn-1\right\|_{L^\infty}
 \le n\|\Sigma\|_{L^\infty}^2.
\]
Thus the lapse equation itself yields
\[
 \left\|\Delta_g(\tfrac Nn-1)\right\|_{L^\infty}
 \le \tfrac1n\left\|\tfrac Nn-1\right\|_{L^\infty}
       +\|\Sigma\|_{L^\infty}^2\left\|\tfrac Nn\right\|_{L^\infty}
 \lesssim \|\Sigma\|_{H^2}^2.
\]

\noindent For the integrations by parts write $u=N/n-1$. The integrated Bochner
identity and the first-order bound give
\begin{align*}
 \|\nabla^2u\|_{L^2}^2
 &=\|\Delta_g u\|_{L^2}^2
      -\int_M\operatorname{Ric}(\nabla u,\nabla u)\,\mu_g\\
 &\le \|\Delta_g u\|_{L^2}^2
      +\|\operatorname{Ric}\|_{L^\infty}\|\nabla u\|_{L^2}^2
 \lesssim (1+\|\mathfrak T\|_{H^2})\|\Sigma\|_{H^2}^4 .
\end{align*}
Taking square roots gives the claimed $H^2$ estimate.

\noindent We recall the two integrations by parts used to estimate a differentiated
lapse on the dynamical manifold. For a tensor $V$ of fixed rank,
integration of one derivative in $\int_M|\nabla^2 V|^2\mu_g$, followed by
commutation of $\Delta_g$ and $\nabla$, gives $\|\Delta_g V\|_{L^2}^2$ and
signed metric contractions of the form
\[
 \int_M\operatorname{Riem}*\nabla V*\nabla V\,\mu_g,\qquad
 \int_M\nabla\operatorname{Riem}*V*\nabla V\,\mu_g .
\]
In the second integral move the derivative off $\operatorname{Riem}$.
The resulting terms are bounded by a constant times
$\|\operatorname{Riem}\|_{L^\infty}
 (\|\nabla V\|_{L^2}^2+\|V\|_{L^2}\|\nabla^2 V\|_{L^2})$.
Using $\|\nabla V\|_{L^2}^2\le\|V\|_{L^2}\|\Delta_g V\|_{L^2}$ and
Young's inequality we obtain
\begin{equation}\label{eq:iteratenew}
 \|V\|_{H^2}\lesssim \|\Delta_g V\|_{L^2}
         +(1+\|\operatorname{Riem}\|_{L^\infty})\|V\|_{L^2}.
\end{equation}
Let
$\mathscr L=-\Delta_g+(|\Sigma|^2+1/n)$. For a smooth scalar $w$
the commuted equation at order $I=1,2,3$ reads, up to the indicated
metric contractions,
\begin{align*}
 \mathscr L\nabla^I w
 ={}&\nabla^I\mathscr L w
 -\sum_{k=1}^I 
       \nabla^k(|\Sigma|^2)*\nabla^{I-k}w\\
 &+\sum_{k=0}^{I-1}\nabla^k\operatorname{Riem}*\nabla^{I-k}w .
\end{align*}
Here the first sum includes all distributions of the ordered derivatives
with the indicated orders. For $I=1$, bound the potential term by
$\|\nabla|\Sigma|^2\|_{L^3}\|w\|_{L^6}$. For $I=2$ the $H^2$
product inequality applies, and for $I=3$ use
$\||\Sigma|^2\|_{H^3}\lesssim\|\Sigma\|_{H^2}\|\Sigma\|_{H^3}
 \lesssim\Gamma$ and the $H^3$ product inequality.
For curvature terms the factors are in $L^\infty\cdot L^2$,
$L^3\cdot L^6$ or $L^2\cdot L^\infty$ according as $k=0,1,2$.
Applying \eqref{eq:iteratenew} to $V=\nabla^I w$ therefore gives
\[
 \|w\|_{H^{I+2}}
 \lesssim \|\mathscr Lw\|_{H^I}+(1+\Gamma)\|w\|_{H^I},
 \qquad I=1,2,3 .
\]
Testing $\mathscr L w$ against $w$ bounds $\|w\|_{H^1}$ by
$C\|\mathscr L w\|_{L^2}$. The scalar Bochner identity above then gives
$\|w\|_{H^2}\lesssim (1+\Gamma)^{1/2}\|\mathscr Lw\|_{L^2}$.
Iterating these two bounds proves the general estimate
\[
 \|w\|_{H^{I+2}}\lesssim(1+\Gamma)^2\|\mathscr Lw\|_{H^I},
 \qquad I=2,3 .
\]
For the lapse we retain the dependence on each differentiated factor.
For $V=\nabla u$, the commuted equation is
\[
 \Delta_g\nabla u=\tfrac1n\nabla u
    +(1+u)\nabla|\Sigma|^2+|\Sigma|^2\nabla u
    +\operatorname{Ric}(\nabla u,\cdot).
\]
We use $0<1+u\le1$ and
$\|\nabla|\Sigma|^2\|_{L^2}\le 2\|\Sigma\|_{L^\infty}\|\nabla\Sigma\|_{L^2}$.
Inserting these bounds in \eqref{eq:iteratenew} and adding the lower
orders gives
\[
 \left\|\tfrac Nn-1\right\|_{H^3}
 \lesssim \|\Sigma\|_{H^2}^2
       +(1+\|\mathfrak T\|_{H^2})\|u\|_{H^1}
 \lesssim (1+\|\mathfrak T\|_{H^2})\|\Sigma\|_{H^2}^2 .
\]

\noindent At this order the lower shear energy uses the bound with one power of
$1+\|\mathfrak T\|_{H^2}$. We keep the curvature and the lapse together
before applying Cauchy--Schwarz. Integrating
$\operatorname{div}(|V|^2 u\nabla u)$ gives
\[
 \big\||V|\,|\nabla u|\big\|_{L^2}^2
 \le \|u\|_{L^\infty}\|\Delta_g u\|_{L^\infty}\|V\|_{L^2}^2
 +2\|u\|_{L^\infty}\|\nabla V\|_{L^2}\big\||V|\,|\nabla u|\big\|_{L^2}.
\]
Young's inequality gives \eqref{eq:weighted-lapse-gradient}. Applied to
$V=\nabla A$, it also bounds
$\||\nabla A|\,|\nabla u|\|_{L^2}\lesssim\|\Sigma\|_{H^2}^2\|A\|_{H^2}$.

\noindent For the term with two derivatives of the lapse, integrate
$\operatorname{div}(|A|^2 u|\nabla u|^2\nabla u)$ and put the term
$\int_M |A|^2|\nabla u|^4\mu_g$ on the left. Cauchy--Schwarz yields
\begin{align*}
 \big\||A|\,|\nabla u|^2\big\|_{L^2}^2
 &\le 2\|u\|_{L^\infty}
   \left(\big\||\nabla A|\,|\nabla u|\big\|_{L^2}
         +\big\||A|\,|\nabla^2u|\big\|_{L^2}\right)
          \big\||A|\,|\nabla u|^2\big\|_{L^2}\\
 &\quad+\|u\|_{L^\infty}\|\Delta_g u\|_{L^\infty}
                \big\||A|\,|\nabla u|\big\|_{L^2}^2.
\end{align*}
Consequently,
\[
 \big\||A|\,|\nabla u|^2\big\|_{L^2}
 \lesssim \|\Sigma\|_{H^2}^2\big\||A|\,|\nabla^2u|\big\|_{L^2}
                  +\|\Sigma\|_{H^2}^4\|A\|_{H^2}.
\]
Now multiply
$\tfrac12\Delta_g|\nabla u|^2=|\nabla^2u|^2+
 \langle\nabla u,\nabla\Delta_g u\rangle+
 \operatorname{Ric}(\nabla u,\nabla u)$ by $|A|^2$ and integrate.
In the term involving $\nabla\Delta_g u$, move the derivative onto
$|A|^2\nabla u$. The term with $|A|^2(\Delta_g u)^2$ is bounded by
$\|\Delta_g u\|_{L^\infty}^2\|A\|_{L^2}^2$, and the derivative of
$|A|^2$ contributes at most
$2\|\Delta_g u\|_{L^\infty}\|A\|_{L^2}
  \||\nabla A|\,|\nabla u|\|_{L^2}$.
The derivative of $|A|^2$ in the integral of $\Delta_g|\nabla u|^2$
contributes at most
$2\||\nabla A|\,|\nabla u|\|_{L^2}\||A|\,|\nabla^2u|\|_{L^2}$.
Finally Lemma~\ref{lem:pointwise-quadratic-curvature}, the bound on $R$,
and the Ricci commutator on $A$ give
\[
 \big\||\operatorname{Ric}|\,|A|\big\|_{L^2}
 \lesssim \|\nabla^2 A\|_{L^2}+\|A\|_{L^2}.
\]
Collecting the terms gives
\begin{align*}
 \big\||A|\,|\nabla^2u|\big\|_{L^2}^2
 &\lesssim \|\Sigma\|_{H^2}^4\|A\|_{H^2}^2
    +\|\Sigma\|_{H^2}^2\|A\|_{H^2}\big\||A|\,|\nabla^2u|\big\|_{L^2}\\
 &\quad+\|A\|_{H^2}\big\||A|\,|\nabla u|^2\big\|_{L^2}.
\end{align*}
Substitute the preceding bound for the last factor and apply Young's
inequality. This proves \eqref{eq:weighted-lapse-hessian} and
\eqref{eq:weighted-lapse-products}.

\noindent We now commute two derivatives with the scalar lapse equation.
The commutator terms have the form
$\operatorname{Riem}*\nabla^2u$ and $\nabla\operatorname{Riem}*\nabla u$.
Since $|\operatorname{Riem}|\lesssim 1+|\mathfrak T|$, use
\eqref{eq:weighted-lapse-hessian} with $A=\mathfrak T$, the second-order
bound already proved, and \eqref{eq:weighted-lapse-gradient} with
$V=\nabla\operatorname{Riem}$ to obtain
\[
 \big\||\operatorname{Riem}|\,|\nabla^2u|\big\|_{L^2}
 +\big\||\nabla\operatorname{Riem}|\,|\nabla u|\big\|_{L^2}
 \lesssim (1+\|\mathfrak T\|_{H^2})\|\Sigma\|_{H^2}^2.
\]
In addition, $\Delta_g u=u/n+(1+u)|\Sigma|^2$, so that
\[
 \|\Delta_g u\|_{H^2}
 \lesssim \|u\|_{H^2}
 + (1+\|u\|_{L^\infty})\big\||\Sigma|^2\big\|_{H^2}
 +\|\Sigma\|_{L^\infty}^2\|u\|_{H^2}
 \lesssim(1+\|\mathfrak T\|_{H^2})\|\Sigma\|_{H^2}^2.
\]
Together these bound $\|\Delta_g\nabla^2u\|_{L^2}$ by the same quantity.

\noindent To recover all fourth derivatives, take $V=\nabla^2u$ in the two
integrations by parts preceding \eqref{eq:iteratenew}. In the term
$\int_M\operatorname{Riem}*\nabla V*\nabla V\,\mu_g$, move one
derivative onto the two other factors. The term with
$\operatorname{Riem}*V*\nabla^2V$ is bounded by
$\||\operatorname{Riem}|\,|\nabla^2u|\|_{L^2}\|\nabla^4u\|_{L^2}$.
In each remaining integral $\int_M\nabla\operatorname{Riem}*V*\nabla V\,\mu_g$,
write $V=\nabla(du)$ and move this derivative onto the other two
factors. The resulting terms are bounded by
\[
 \|\nabla^2\operatorname{Riem}\|_{L^2}
          \|\nabla u\|_{L^8}\|\nabla^3u\|_{L^{8/3}},
 \qquad
 \big\||\nabla\operatorname{Riem}|\,|\nabla u|\big\|_{L^2}
          \|\nabla^4u\|_{L^2}.
\]
By the covariant interpolation inequality in Section~\ref{sobolev2},
$\|\nabla u\|_{L^8}\|\nabla^3u\|_{L^{8/3}}
 \lesssim\|u\|_{L^\infty}\|\nabla^4u\|_{L^2}$. We have proved
\begin{align*}
 \|\nabla^4u\|_{L^2}^2
 &\lesssim \|\Delta_g\nabla^2u\|_{L^2}^2
 +(1+\|\mathfrak T\|_{H^2})\|\Sigma\|_{H^2}^2\|\nabla^4u\|_{L^2}.
\end{align*}
Absorbing half of the left-hand side and adding the lower derivatives gives
\[
 \left\|\tfrac Nn-1\right\|_{H^4}
 \lesssim (1+\|\mathfrak T\|_{H^2})\|\Sigma\|_{H^2}^2.
\]

\noindent Commuting three derivatives with $\Delta_g u$ gives the three types
$\operatorname{Riem}*\nabla^3u$,
$\nabla\operatorname{Riem}*\nabla^2u$ and
$\nabla^2\operatorname{Riem}*\nabla u$. The corresponding $L^2$ bounds are
\begin{align*}
 &\|\operatorname{Riem}\|_{L^\infty}\|\nabla^3u\|_{L^2}
 +\|\nabla\operatorname{Riem}\|_{L^3}\|\nabla^2u\|_{L^6}
 +\|\nabla^2\operatorname{Riem}\|_{L^2}\|\nabla u\|_{L^\infty}
 \lesssim (1+\|\mathfrak T\|_{H^2})\|u\|_{H^3}.
\end{align*}
At this order the source satisfies
$\||\Sigma|^2\|_{H^3}\lesssim\|\Sigma\|_{H^2}\|\Sigma\|_{H^3}$.
Differentiating the product in $\Delta_g u=u/n+(1+u)|\Sigma|^2$ gives
\[
 \|\Delta_g u\|_{H^3}
 \lesssim(1+\|\Sigma\|_{L^\infty}^2)\|u\|_{H^3}
      +(1+\|u\|_{L^\infty})\|\Sigma\|_{H^2}\|\Sigma\|_{H^3}.
\]
Substituting these bounds into \eqref{eq:iteratenew} with $V=\nabla^3u$,
and including the lower derivatives, yields
\begin{align*}
 \left\|\tfrac Nn-1\right\|_{H^5}
 &\lesssim(1+\|\mathfrak T\|_{H^2})\|u\|_{H^3}
                 +\|\Sigma\|_{H^2}\|\Sigma\|_{H^3}\\
 &\lesssim(1+\|\mathfrak T\|_{H^2})^2\|\Sigma\|_{H^2}^2
                 +\|\Sigma\|_{H^2}\|\Sigma\|_{H^3}.
\end{align*}
In the notation of \eqref{eq:strap2}, the estimates used below read
\[
 \left\|\tfrac Nn-1\right\|_{H^4}
  \lesssim (1+\Gamma)e^{-2T}\mathds Y^2,\qquad
 \left\|\tfrac Nn-1\right\|_{H^5}
  \lesssim (1+\Gamma)^2 e^{-2T}\mathds Y^2+e^{-T}\mathds Y\Gamma .
\]
In particular, the elliptic estimate also yields the bounds
\[
 \sum_{j=0}^{I+2}\left\|\nabla^j\bigl(\tfrac Nn-1\bigr)\right\|_{L^2(M)}
 \lesssim (1+\Gamma)^2\sum_{j=0}^I\|\nabla^j\Sigma\|_{L^2(M)}^2,
 \qquad I=2,3 .
\]
This proves \eqref{eq:lapse_sobolev_estimate}.
\end{proof}

\noindent The next part entails controlling the lower order norm $\mathcal{F}$ of $\Sigma$ in terms of the top order norm $\mathcal{O}$ and the initial data norm $\mathcal{I}^{0}$. 

\subsection{Estimating the lower order norm of the TT tensor \texorpdfstring{$\Sigma$}{Sigma} via its transport equation}
\label{entropy}
\noindent In this section, we prove that the weighted lower order norm of the $TT$ tensor $\Sigma$ is controlled by the top order norm of $(\mathfrak{T},\Sigma)$ and the initial data. Recall the equation for $\Sigma$:
\begingroup
\begin{align}
\partial_T\Sigma_{ij}={}&-(n-1)\Sigma_{ij}
-\frac{\varphi}{\tau}N\mathfrak T_{ij}
+\frac{n\varphi}{\tau}\nabla_i\nabla_j\left(\frac Nn-1\right)
+\frac{2\varphi}{\tau}N\Sigma_{ik}\Sigma^k{}_j\notag\\
&-\frac{\varphi}{n\tau}\left(\frac Nn-1+N|\Sigma|_g^2\right)g_{ij}
-(n-2)\left(\frac Nn-1\right)\Sigma_{ij}.
\end{align}
\endgroup
To estimate the lower order norm $\mathcal{F}$, we will treat the equation for $\Sigma$ as a transport equation. 

\begin{definition}
\label{entropy2}
\begin{align} 
\mathscr{E}^{low}:=\frac{1}{2}\sum_{I\leq 2}\int_{M}|\nabla^{I}\Sigma|^{2}\mu_{g}.
\end{align}
\end{definition}
\noindent First note that $\mathcal{F}^{2}\approx e^{2\gamma T}\mathscr{E}^{low}$. We have the following proposition estimating $\mathcal{F}^{2}$.
\begin{proposition}[Coercive control of the reversed entropy]\label{model2}
Let $(M^3,g(T))$ be a smooth solution of the rescaled Einstein vacuum equations with cosmological constant $\Lambda>0$ on a bootstrap interval $[T_0,T]$ satisfying the initial-size condition of Theorem~\ref{main}, written in the fixed gauge of this work.  
Let $\mathscr{E}^{\mathrm{low}}(T)\ge0$ denote the reversed entropy functional controlling the lower order norm of $\Sigma$ as defined in Definition~\ref{entropy2} and let $\mathcal{F}(T)$ be the corresponding weighted norm, i.e., 
\begin{equation}\label{eq:equiv}
e^{2T}\,\mathscr{E}^{\mathrm{low}}(T)\ \lesssim \ \mathcal{F}(T)^2\ \lesssim\ e^{2T}\,\mathscr{E}^{\mathrm{low}}(T),
\end{equation}
and let $\mathcal{O}(t)$ be the top-order norm in \eqref{eq:O} at time $t$.
Then the following a priori estimate holds:
\begin{equation}\label{eq:upper}
\mathcal{F}(T)^2\ \lesssim\ \Bigl(1+\bigl(\mathcal{I}^{0}\bigr)^{2}+\sup_{T_0\le t\le T}\mathcal{O}(t)^2\Bigr).
\end{equation}
In particular, combining \eqref{eq:equiv} and \eqref{eq:upper} yields
\begin{equation}\label{eq:announced}
e^{2T}\,\mathscr{E}^{\mathrm{low}}(T)\ \lesssim\ \Bigl(1+\bigl(\mathcal{I}^{0}\bigr)^{2}+\sup_{T_0\le t\le T}\mathcal{O}(t)^2\Bigr),
\end{equation}
{\color{blue}where the constants are uniform in $T$ and the initial size, under the stated smallness assumption and with the bootstrap factors fixed.} 
\end{proposition}

\begin{proof}
Recall Definition~\ref{entropy2}:
 \begin{align}
 \mathscr{E}^{low}(T):=\frac{1}{2}\sum_{I\leq 2}\int_{M(T)}\langle\nabla^{I}\Sigma,\nabla^{I}\Sigma\rangle\mu_{g}.
\end{align}
Explicit application of the time evolution operator $\frac{d}{dT}$ together with the transport theorem \cite{marsden} yields
\begin{align}
 \frac{d\mathscr{E}^{low}}{dT}=\sum_{I\leq 2}\int_{M}\langle\nabla^{I}\partial_{T}\Sigma,\nabla^{I}\Sigma\rangle\mu_{g}+\sum_{I\leq 2}\int_{M}\langle[\partial_{T},\nabla^{I}]\Sigma,\nabla^{I}\Sigma\rangle\mu_{g}+\mathscr{ER}_{2},   
\end{align}
where the error terms $\mathscr{ER}_{2}$ read schematically
\begin{align}
\mathscr{ER}_{2}\sim \sum_{I\leq 2}\int_{M(T)}\partial_{T}(g^{-1}\cdots g^{-1}g\cdots g)\nabla^{I}\Sigma\nabla^{I}\Sigma\mu_{g}+\sum_{I\leq 2}\int_{M(T)}|\nabla^{I}\Sigma|^{2}\partial_{T}\mu_{g}.      
\end{align}
Using the equation for $\partial_{T}\Sigma$, $\frac{d\mathscr{E}^{low}}{dT}$ reads 
\begingroup
\begin{align}
\frac{d\mathscr E^{low}}{dT}={}&-2(n-1)\mathscr E^{low}
-\frac{\varphi}{\tau}\sum_{I\leq2}\int_M
 \langle\nabla^I(N\mathfrak T),\nabla^I\Sigma\rangle\,\mu_g\notag\\
&+\frac{n\varphi}{\tau}\sum_{I\leq2}\int_M
 \left\langle\nabla^I\nabla^2\left(\frac Nn-1\right),\nabla^I\Sigma\right\rangle\,\mu_g\notag\\
&+\frac{2\varphi}{\tau}\sum_{I\leq2}\int_M
 \langle\nabla^I(N\Sigma_{ik}\Sigma^k{}_j),\nabla^I\Sigma\rangle\,\mu_g\notag\\
&-(n-2)\sum_{I\leq2}\int_M
 \left\langle\nabla^I\left(\left(\frac Nn-1\right)\Sigma\right),\nabla^I\Sigma\right\rangle\,\mu_g\notag\\
&+\sum_{I\leq2}\int_M
 \langle[\partial_T,\nabla^I]\Sigma,\nabla^I\Sigma\rangle\,\mu_g+\mathscr{ER}_2.
\end{align}
\endgroup
Now we estimate each term separately. First, the fixed upper bound for $N$, $\varphi/\tau\lesssim e^{-T}$ and Young's inequality give
\begin{align}
 \left|\frac{\varphi}{\tau}\sum_{I\le2}\int_{M}N
   \langle\nabla^I\mathfrak T,\nabla^I\Sigma\rangle\mu_g\right|
 &\le \frac14\mathscr E^{low}
       +C e^{-2T}\mathcal O(T)^2 .
\end{align}
The constant here is independent of the bootstrap constants. The pure-trace
terms have zero pairing with $\nabla^I\Sigma$ for each $I$, since
$\nabla g=0$ and $\operatorname{tr}_g\Sigma=0$.

\noindent In the terms where a derivative falls on $N$, there is one product at
order one and there are two at order two. By Proposition~\ref{elliptic}
and \eqref{eq:weighted-lapse-products}, for $H=\mathfrak T$ or $H=\Sigma$,
\begin{align*}
 \|\,|\nabla N|\,|H|\,\|_{L^2}
 &\le \|\nabla N\|_{L^2}\|H\|_{L^\infty}
 \lesssim \|\Sigma\|_{H^2}^2\|H\|_{H^2},\\
 \|\,|\nabla^2 N|\,|H|\,\|_{L^2}
 +\|\,|\nabla N|\,|\nabla H|\,\|_{L^2}
 &\lesssim \|\Sigma\|_{H^2}^2\|H\|_{H^2}.
\end{align*}
Thus the differentiated curvature coefficient contributes
\begin{align*}
 \left|\frac{\varphi}{\tau}\sum_{I\le2}\int_M
 \langle[\nabla^I,N]\mathfrak T,\nabla^I\Sigma\rangle\mu_g\right|
 &\lesssim e^{-T}\|\mathfrak T\|_{H^2}\|\Sigma\|_{H^2}^3 .
\end{align*}
The undifferentiated lapse satisfies
$\|N/n-1\|_{L^\infty}\lesssim\|\Sigma\|_{H^2}^2$.
Combining this with the two product bounds for $H=\Sigma$ yields
\[
 \| (N/n-1)\Sigma\|_{H^2}\lesssim\|\Sigma\|_{H^2}^3,
 \qquad \|N\Sigma\|_{H^2}\lesssim\|\Sigma\|_{H^2},
\]
where $\|\Sigma\|_{H^2}\le1$ by \eqref{eq:strap2} and the initial-size
condition for sufficiently large $a$. The terms quadratic in $\Sigma$
and the lapse Hessian in the transport equation therefore give
\begin{align*}
 \left|\frac{2\varphi}{\tau}\sum_{I\le2}\int_M
  \langle\nabla^I(N\Sigma_{ik}\Sigma^k{}_j),\nabla^I\Sigma\rangle\mu_g\right|
 &\lesssim e^{-T}\|N\Sigma\|_{H^2}\|\Sigma\|_{H^2}^2
 \lesssim e^{-T}\|\Sigma\|_{H^2}^3,\\
 \left|\frac{n\varphi}{\tau}\sum_{I\le2}\int_M
  \langle\nabla^I\nabla^2(N/n-1),\nabla^I\Sigma\rangle\mu_g\right|
 &\lesssim e^{-T}(1+\|\mathfrak T\|_{H^2})\|\Sigma\|_{H^2}^3,\\
 \left|\sum_{I\le2}\int_M
 \langle\nabla^I((N/n-1)\Sigma),\nabla^I\Sigma\rangle\mu_g\right|
 &\lesssim \|\Sigma\|_{H^2}^4 .
\end{align*}
In the first line the two derivatives on the product are estimated with
$L^\infty\cdot L^2$ and $L^4\cdot L^4$. The second line uses the
fourth-order lapse estimate \eqref{eq:lapse_sobolev_estimate}.

\noindent For the time commutators, differentiating the Christoffel symbols gives
\[
 \partial_T\Gamma^k{}_{ij}
 =\tfrac12 g^{k\ell}
 (\nabla_i\partial_Tg_{j\ell}+\nabla_j\partial_Tg_{i\ell}
                         -\nabla_\ell\partial_Tg_{ij}).
\]
There is no commutator at order zero. At order one the term is a
contraction of $\partial_T\Gamma$ with $\Sigma$; at order two the terms
are contractions of $\partial_T\Gamma$ with $\nabla\Sigma$ and of
$\nabla\partial_T\Gamma$ with $\Sigma$. For the part of $\partial_Tg$
equal to $2(\varphi/\tau)N\Sigma$, their combined $L^2$ norm is bounded by
\begin{align*}
 C e^{-T}\big[&
 (\|\nabla(N\Sigma)\|_{L^2}+\|\nabla^2(N\Sigma)\|_{L^2})
                 \|\Sigma\|_{L^\infty}
 +\|\nabla(N\Sigma)\|_{L^4}\|\nabla\Sigma\|_{L^4}\big]
 \lesssim e^{-T}\|\Sigma\|_{H^2}^2 .
\end{align*}
For the part equal to $2(N/n-1)g$, the three products are precisely
those estimated above with $H=\Sigma$ and have combined $L^2$ norm
at most $C\|\Sigma\|_{H^2}^3$. Pairing with $\nabla^I\Sigma$ gives
\[
 \left|\sum_{I\le2}\int_M
 \langle[\partial_T,\nabla^I]\Sigma,\nabla^I\Sigma\rangle\mu_g\right|
 \lesssim e^{-T}\|\Sigma\|_{H^2}^3+\|\Sigma\|_{H^2}^4 .
\]
Finally, differentiating the inverse metrics and the volume form in the
energy gives
\[
 |\mathscr{ER}_2|\lesssim
 \big(e^{-T}\|\Sigma\|_{L^\infty}+\|N/n-1\|_{L^\infty}\big)\mathscr E^{low}
 \lesssim \big(e^{-T}\|\Sigma\|_{H^2}+\|\Sigma\|_{H^2}^2\big)\mathscr E^{low}.
\]

\noindent Combining these estimates and using
$\|\Sigma\|_{H^2}^2\lesssim\mathscr E^{low}$, we obtain
\begin{align*}
 \frac{d\mathscr E^{low}}{dT}+4\mathscr E^{low}
 \le{}&\frac14\mathscr E^{low}+C e^{-2T}\mathcal O(T)^2\\
 &+C\big[e^{-T}(1+\|\mathfrak T\|_{H^2})\|\Sigma\|_{H^2}
                +\|\Sigma\|_{H^2}^2\big]\mathscr E^{low}.
\end{align*}
The coefficient on the second line is at most
$C[(1+\Gamma)\mathds Y+\mathds Y^2]e^{-2T}$ by \eqref{eq:strap2}.
By \eqref{eq:bootstrap-size} and the initial-size condition it is
bounded by $C_{A,B}(1+\mathcal I^0)^2 e^{-2a}
\le C_{A,B}e^{-2\alpha a}\le1/2$ for sufficiently large $a$.
Consequently,
\begin{align}
 \frac{d\mathscr E^{low}}{dT}+3\mathscr E^{low}
 \le C e^{-2T}\mathcal O(T)^2.
\end{align}
Multiplication by $e^{2T}$ gives
\begin{align}
 \label{eq:suboptimal_recall}
 \frac{d}{dT}\left(e^{2T}\mathscr E^{low}(T)\right)
 +e^{2T}\mathscr E^{low}(T)\le C\mathcal O(T)^2 .
\end{align}
Integrating from $T_0$ to $T$, we obtain
\begin{align}
 e^{2T}\mathscr E^{low}(T)
 \le e^{-(T-T_0)}e^{2T_0}\mathscr E^{low}(T_0)
   +C\int_{T_0}^T e^{-(T-t)}\mathcal O(t)^2\,dt .
\end{align}
The initial-data assumption yields
$e^{2T_0}\mathscr E^{low}(T_0)\lesssim(\mathcal I^0)^2$, while
\[
 \int_{T_0}^T e^{-(T-t)}\mathcal O(t)^2\,dt
 \le\sup_{T_0\le t\le T}\mathcal O(t)^2.
\]
This proves \eqref{eq:announced} and hence \eqref{eq:upper}, with a constant
independent of the bootstrap factors.
\end{proof}
\noindent As an immediate corollary, we obtain 
\begin{corollary}
 With $S_T=\sup_{T_0\leq t\leq T}\mathcal O(t)$, we have $\mathcal F(T)\lesssim 1+\mathcal I^0+S_T$ and, by Proposition~\ref{elliptic}, $\mathcal N^\infty(T)\lesssim 1+(\mathcal I^0)^2+S_T^2$.   
\end{corollary}

\subsection{Energy Estimates}
\label{energy}
\noindent With the individual estimates available, we are ready to complete the energy estimates for the pair $(\Sigma,\mathfrak{T})$. Let us recall the definition of the top order energy $\mathscr{E}^{top}$.
\begin{definition}
\label{toporder}
\begin{align}
\mathscr{E}^{top}:=\frac{1}{2}\sum_{I\leq 3}\int_{M}\langle\nabla^{I}\Sigma,\nabla^{I}\Sigma\rangle \mu_{g}+\frac{1}{2}\sum_{1\leq I\leq3}\int_{M}\langle\nabla^{I-1}\mathfrak{T},\nabla^{I-1}\mathfrak{T}\rangle\mu_{g}.    
\end{align}    
\end{definition}

\noindent First, we recall the following elementary proposition.
\begin{proposition}[Integrated energy estimate for the wave system]
\label{wave}
Let $(M,g)$ be a closed Riemannian manifold and consider a time interval $[T_0,T_{\infty}] \subset \mathbb{R}$. Suppose $(\Phi,\Psi)$ is a pair of smooth tensor fields of type $(K,L)$ on $[T_0,T_{\infty}]\times M$ satisfying the coupled system
\begin{align}
\partial_T \Phi &= -N \Psi + P_{\Phi}, \\
\partial_T \Psi &= -N \Delta_g \Phi + Q_{\Psi},
\end{align}
where $N = N(T,x)$ is the lapse function, and $P_{\Phi}$, $Q_{\Psi}$ are given source terms of the same tensorial type as $\Phi$ and $\Psi$, respectively. {Here $g$ satisfies \eqref{eq:cmc1} with $X=0$ and $N$ is positive and uniformly bounded above. The same identity holds when the principal coefficient $N$ in both equations is replaced by $c(T)N$, since $\nabla c(T)=0$.} Denote by $M(T)$ the time slice at time $T$ equipped with the normalized metric $g(T)$ and volume form $\mu_g$.

\noindent Then, the following integrated energy identity holds for all $T \in [T_0,T_{\infty}]$:
\begin{align}
\int_{M(T)} \big( |\nabla \Phi|^2 + |\Psi|^2 \big) \,\mu_g
= &\int_{M(T_0)} \big( |\nabla \Phi|^2 + |\Psi|^2 \big) \,\mu_g \\
&+ 2 \int_{T_0}^T \int_{M(t)} \big( \langle \nabla P_{\Phi}, \nabla \Phi \rangle + \langle Q_{\Psi}, \Psi \rangle \big) \,\mu_g\, dt \nonumber \\
&+ 2\int_{T_0}^T \mathscr{E}\mathscr{R}(t)\, dt, \nonumber
\end{align}
where the error functional $\mathscr{E}\mathscr{R}(t)$ admits the bound
\begin{align}
|\mathscr{E}\mathscr{R}(t)| \lesssim &\left\| \frac{N}{n} - 1 \right\|_{L^\infty(M(t))} \big( \|\nabla \Phi\|_{L^2(M(t))}^2 + \|\Psi\|_{L^2(M(t))}^2 \big) \\
&+ \frac{\varphi}{\tau} \|\Sigma\|_{L^\infty(M(t))} \big( \|\nabla \Phi\|_{L^2(M(t))}^2 + \|\Psi\|_{L^2(M(t))}^2 \big) \nonumber \\
&+ \left[ \frac{\varphi}{\tau} \big( \|\nabla \Sigma\|_{L^4(M(t))} + \|\Sigma\|_{L^\infty(M(t))} \|\nabla (\frac{N}{n} - 1)\|_{L^4(M(t))} \big) + \|\nabla(\frac{N}{n} - 1)\|_{L^4(M(t))} \right] \nonumber \\
&\quad \times \|\Phi\|_{L^4(M(t))} \|\nabla \Phi\|_{L^2(M(t))}. \nonumber
\end{align}
The norms in the time integrals are computed with respect to $g(t)$ on $M(t)$, and the endpoint energies use $g(T_0)$ and $g(T)$ on their respective slices. The implicit constant depends only on the dimension, the tensor type $(K,L)$ and the upper bound for $N$. On both branches one may use $N\leq 3n/2$.

\end{proposition}
\begin{proof}
Let $(\Phi, \Psi)$ be a smooth pair of tensor fields on the compact Riemannian manifold $(M,g)$, where the metric $g = g(T)$ depends smoothly on time $T \in [T_0, T_\infty]$. We consider the energy functional
\[
\mathscr{E}(T) := \frac{1}{2} \int_{M(T)} \big( |\nabla \Phi|^2 + |\Psi|^2 \big) \, \mu_g,
\]
where $\mu_g$ denotes the Riemannian volume form induced by $g(T)$ on the hypersurface $M(T)$, and $|\cdot|$ and $\nabla$ denote the metric norm and Levi-Civita connection respectively.
 
\noindent By the transport theorem for evolving hypersurfaces (cf.~\cite{marsden}), we have the exact identity
\begin{align} \label{eq:transport_energy}
\frac{d}{dT} \int_{M(T)} \left( |\nabla \Phi|^2 + |\Psi|^2 \right) \mu_g
= \int_{M(T)} \partial_T \left( |\nabla \Phi|^2 + |\Psi|^2 \right) \mu_g
+ \int_{M(T)} \left( |\nabla \Phi|^2 + |\Psi|^2 \right) \partial_T \mu_g.
\end{align}
We rewrite \eqref{eq:transport_energy} as
\begin{align} \label{eq:energy_derivative}
\frac{d}{dT} \mathscr{E}(T) 
= \int_{M(T)} \left( \langle \nabla \partial_T \Phi, \nabla \Phi \rangle + \langle \partial_T \Psi, \Psi \rangle \right) \mu_g 
+ \mathscr{ER}(T),
\end{align}
where the \emph{error term} $\mathscr{ER}(T)$ collects the contributions
\[
 \mathscr{ER}(T)\sim \int_{M(T)}\langle[\partial_{T},\nabla]\Phi,\nabla\Phi\rangle\mu_{g}  +\int_{M(T)}\partial_{T}(g^{-1}\cdots g^{-1}g\cdots g)\nabla\Phi\nabla\Phi\mu_{g}\]
 \[+\int_{M(T)}\partial_{T}(g^{-1}\cdots g^{-1}g\cdots g)\Psi\Psi\mu_{g}+\int_{M(T)}(|\nabla\Phi|^{2}+|\Psi|^{2})\partial_{T}\mu_{g}.
\]
 
\noindent Recall that the time evolution of the metric $g = g(T)$ is given by
\begin{align} \label{eq:energy-metric-evolution}
\partial_T g_{ij} = {\frac{2\varphi}{\tau} N \Sigma_{ij}} - 2 \left(1 - \frac{N}{n}\right) g_{ij},
\quad\text{and}\quad
\partial_T g^{ij} = {-\frac{2\varphi}{\tau} N \Sigma^{ij}} + 2 \left(1 - \frac{N}{n}\right) g^{ij},
\end{align}
where $\varphi, \tau, N, n$, and $\Sigma$ are as in the geometric setup.

\noindent Consequently, the time derivative of the volume form $\mu_g$ is given by
\[
\partial_T \mu_g = \frac{1}{2} \mu_g \, \mathrm{tr}_g (\partial_T g) = - n \left(1 - \frac{N}{n}\right) \mu_g,
\]
which follows from \eqref{eq:energy-metric-evolution}.

\noindent The commutator $[\partial_T, \nabla] \Phi$ is expressible via the time derivative of the Christoffel symbols $\Gamma[g]$:
\[
[\partial_T, \nabla] \Phi = (\partial_T \Gamma) \ast \Phi,
\]
where $\ast$ denotes appropriate tensor contractions. By direct computation,
\[
|\partial_T \Gamma| \lesssim \frac{\varphi}{\tau} |\nabla (N \Sigma)| + |\nabla \left(\tfrac{N}{n} - 1\right)|,
\]
hence
\[
{|[\partial_T,\nabla]\Phi|\lesssim\left[\frac{\varphi}{\tau}\left(|\nabla\Sigma|+|\Sigma|\,|\nabla(\tfrac Nn-1)|\right)+|\nabla(\tfrac Nn-1)|\right]|\Phi|.}
\]

\noindent Applying H\"older's inequality with Sobolev embeddings yields
\begin{align*}
\left| \int_{M(T)} \langle [\partial_T, \nabla] \Phi, \nabla \Phi \rangle \mu_g \right|
&\lesssim \frac{\varphi}{\tau} \left( \|\nabla \Sigma\|_{L^4} + \|\Sigma\|_{L^\infty} \|\nabla(\tfrac{N}{n} - 1)\|_{L^4} \right) \|\Phi\|_{L^4} \|\nabla \Phi\|_{L^2} \\
&\quad + \|\nabla(\tfrac{N}{n} - 1)\|_{L^4} \|\Phi\|_{L^4} \|\nabla \Phi\|_{L^2}.
\end{align*}

\noindent Similarly, the terms involving $\partial_T \mu_g$ are controlled by
\[
\left| \int_{M(T)} \left( |\nabla \Phi|^2 + |\Psi|^2 \right) \partial_T \mu_g \right| \lesssim \left\| 1 - \frac{N}{n} \right\|_{L^\infty} \left( \|\nabla \Phi\|_{L^2}^2 + \|\Psi\|_{L^2}^2 \right).
\]

\noindent Taking all spatial norms on $M(T)$, we obtain
\begin{align}
\label{eq:error_estimate}
|\mathscr{ER}(T)| &\lesssim
  \left(\|\tfrac Nn-1\|_{L^\infty}+\frac{\varphi}{\tau}\|\Sigma\|_{L^\infty}\right)
  \left(\|\nabla\Phi\|_{L^2}^2+\|\Psi\|_{L^2}^2\right)\notag\\
&\quad+\left[\frac{\varphi}{\tau}
  \left(\|\nabla\Sigma\|_{L^4}+\|\Sigma\|_{L^\infty}\|\nabla(\tfrac Nn-1)\|_{L^4}\right)
  +\|\nabla(\tfrac Nn-1)\|_{L^4}\right]\|\Phi\|_{L^4}\|\nabla\Phi\|_{L^2}.
\end{align}

 \noindent We now turn to the principal terms in \eqref{eq:energy_derivative}:
\[
\int_{M(T)} \langle \nabla \partial_T \Phi, \nabla \Phi \rangle \mu_g + \int_{M(T)} \langle \partial_T \Psi, \Psi \rangle \mu_g.
\]
\noindent Substituting the evolution equations
\[
\partial_T \Phi = - N \Psi + P_{\Phi}, \quad \partial_T \Psi = - N \Delta_g \Phi + Q_{\Psi},
\]
we write
\begin{align*}
&\int_{M(T)} \langle \nabla \partial_T \Phi, \nabla \Phi \rangle \mu_g + \int_{M(T)} \langle \partial_T \Psi, \Psi \rangle \mu_g \\
&= \int_{M(T)} \langle \nabla (- N \Psi + P_{\Phi}), \nabla \Phi \rangle \mu_g + \int_{M(T)} \langle - N \Delta_g \Phi + Q_{\Psi}, \Psi \rangle \mu_g \\
&= \int_{M(T)} \langle \nabla P_{\Phi}, \nabla \Phi \rangle \mu_g + \int_{M(T)} \langle Q_{\Psi}, \Psi \rangle \mu_g + \underbrace{\int_{M(T)} \langle \nabla (- N \Psi), \nabla \Phi \rangle \mu_g + \int_{M(T)} \langle - N \Delta_g \Phi, \Psi \rangle \mu_g}_{=: I}.
\end{align*}

\noindent By integration by parts and the absence of boundary (since $M$ is closed), the integral
\[
I = \int_{M(T)} \langle \nabla (- N \Psi), \nabla \Phi \rangle \mu_g + \int_{M(T)} \langle - N \Delta_g \Phi, \Psi \rangle \mu_g
\]
vanishes exactly. This cancellation is a fundamental feature reflecting the underlying hyperbolic structure of the system and guarantees no loss of derivatives in the energy estimates, preserving regularity. This completes the proof.
\end{proof}
\begin{remark}
The terms such as $||\Phi||_{L^{4}(M)}$ can further be controlled by $\sum_{I\leq 1}||\nabla^{I}\Phi||_{L^{2}(M)}$ by Proposition~\ref{sobolev}.
\end{remark}

\begin{proposition}[Top-order energy bound]\label{topestimate}
Assume the hypotheses of Theorem~\ref{main}, and let $(\Sigma,\mathfrak{T})$ denote the coupled wave pair associated with the rescaled Einstein system in the gauge fixed framework of this work.  
Let $\mathscr{E}^{\mathrm{top}}(T)\ge0$ denote the top-order energy functional corresponding to $(\Sigma,\mathfrak{T})$ as defined in Definition~\ref{toporder}, and let $\mathcal{I}^{0}$ denote the initial size of the data at time $T_0=a\gg 1$.  

\noindent Then for every $T$ in the lifespan of the solution one has
\[
\mathscr{E}^{\mathrm{top}}\ \lesssim\ {\color{blue}1+(\mathcal I^0)^2},
\]
{\color{blue}where the constant depends only on the fixed background bounds and the dimension.}
\end{proposition}

\begin{proof}
We first prove the estimate on an interval satisfying \eqref{eq:strap2}. The strict improvement below and the continuity argument of Section~\ref{bootstrapargument} then give the estimate throughout the smooth lifespan. The proof uses Proposition~\ref{wave}. First, recall the definition of the top-order energy
\begin{align}
    \mathscr{E}^{top}:=\frac12\int_M|\Sigma|^2\mu_g+\frac{1}{2}\sum_{I\leq 2}\int_{M}\langle\nabla^{I+1}\Sigma,\nabla^{I+1}\Sigma\rangle \mu_{g}+\frac{1}{2}\sum_{I\leq 2}\int_{M}\langle\nabla^{I}\mathfrak{T},\nabla^{I}\mathfrak{T}\rangle\mu_{g}.
\end{align}
We apply Proposition~\ref{wave} with $\Phi=\nabla^{I}\Sigma$ and $\Psi=\nabla^{I}\mathfrak{T}$, i.e., 
\begin{align}
\label{eq:wave1}
 \partial_{T}\nabla^{I}\Sigma=-N\frac{\varphi}{\tau}\nabla^{I}\mathfrak{T}+P^{I},\\
 \label{eq:wave2}
 \partial_{T}\nabla^{I}\mathfrak{T}=-N\frac{\varphi}{\tau}\Delta_{g}\nabla^{I}\Sigma+Q^{I},
\end{align}
where the error terms $P^I$ and $Q^I$ are 
\begingroup
\begin{align}
P^I={}&[\partial_T,\nabla^I]\Sigma
-\frac{\varphi}{\tau}[\nabla^I,N]\mathfrak T
+\nabla^I\Bigg(-(n-1)\Sigma_{ij}
+\frac{n\varphi}{\tau}\nabla_i\nabla_j\left(\frac Nn-1\right)\notag\\
&\hspace{8mm}+\frac{2\varphi}{\tau}N\Sigma_{ik}\Sigma^k{}_j
-\frac{\varphi}{n\tau}\left(\frac Nn-1+N|\Sigma|^2\right)g_{ij}
-(n-2)\left(\frac Nn-1\right)\Sigma_{ij}\Bigg),\\
Q^I={}&[\partial_T,\nabla^I]\mathfrak T
-\frac{\varphi}{\tau}[\nabla^I,N\Delta_g]\Sigma
+\nabla^I\Bigg(-2\frac{\varphi}{\tau}N\operatorname{Riem}_{ikjl}\Sigma^{kl}\notag\\
&\hspace{8mm}+n\frac{\varphi}{\tau}\nabla^l\nabla_i\left(\frac Nn-1\right)\Sigma_{jl}
+n\frac{\varphi}{\tau}\nabla^l\left(\frac Nn-1\right)\nabla_i\Sigma_{jl}\notag\\
&\hspace{8mm}+n\frac{\varphi}{\tau}\nabla^l\nabla_j\left(\frac Nn-1\right)\Sigma_{il}
+n\frac{\varphi}{\tau}\nabla^l\left(\frac Nn-1\right)\nabla_j\Sigma_{il}\notag\\
&\hspace{8mm}-n\frac{\varphi}{\tau}\Delta_g\left(\frac Nn-1\right)\Sigma_{ij}
-2n\frac{\varphi}{\tau}\nabla^l\left(\frac Nn-1\right)\nabla_l\Sigma_{ij}\notag\\
&\hspace{8mm}+N\frac{\varphi}{\tau}\left(\mathfrak T_{ki}\Sigma^k{}_j+\mathfrak T_{kj}\Sigma^k{}_i\right)
-(n-2)\nabla_i\nabla_j\left(\frac Nn-1\right)\notag\\
&\hspace{8mm}{\color{blue}+\left[\frac{n-2}{n}\Delta_g\left(\frac Nn-1\right)
+\frac{2\varphi N}{n\tau}\langle\mathfrak T,\Sigma\rangle_g
-\frac{2\varphi}{\tau}\left\langle\Sigma,\nabla^2\left(\frac Nn-1\right)\right\rangle_g\right]g_{ij}}\Bigg).
\end{align}
\endgroup
\begingroup\color{blue}
Both $\Sigma$ and $\mathfrak T$ are trace-free. Since $\nabla g=0$, for every scalar function $u$ and every $I\geq0$,
\[
\left\langle\nabla^{I+1}(ug),\nabla^{I+1}\Sigma\right\rangle_g=0,
\qquad
\left\langle\nabla^I(ug),\nabla^I\mathfrak T\right\rangle_g=0.
\]
Thus every displayed scalar multiple of $g$ in $P^I$ and $Q^I$ vanishes in the corresponding energy pairing. The time commutators need not be trace-free; they are retained below, together with the variation of the inverse metric in the energy. The constraint gives $\operatorname{tr}_g(\operatorname{Ric}+\frac29g)=|\Sigma|^2$; it does not change $\operatorname{tr}_g\mathfrak T=0$.

\endgroup
\noindent Now apply Proposition~\ref{wave} to the system (\ref{eq:wave1})-(\ref{eq:wave2}) to obtain 
\begingroup
\begin{align}
\mathscr E^{top}(T)={}&\mathscr E^{top}(T_0)
+\sum_{I\leq2}\int_{T_0}^T\int_{M(t)}
\left(\langle\nabla P^I,\nabla^{I+1}\Sigma\rangle
+\langle Q^I,\nabla^I\mathfrak T\rangle\right)\mu_g\,dt\notag\\
&+\int_{T_0}^T\mathscr{ER}(t)\,dt
+\frac12\left(\|\Sigma(T)\|_{L^2}^2-\|\Sigma(T_0)\|_{L^2}^2\right),\notag
\end{align}
\endgroup
where the error from the metric, the volume form and the outermost covariant derivative satisfies
\begin{align*}
 |\mathscr{ER}(t)|
 &\lesssim \left(\left\|\tfrac Nn-1\right\|_{L^\infty}
    +\tfrac{\varphi}{\tau}\|\Sigma\|_{L^\infty}\right)
 \sum_{I=0}^{2}\left(\|\nabla^{I+1}\Sigma\|_{L^2}^2+
                     \|\nabla^I\mathfrak T\|_{L^2}^2\right)\\
 &\quad+\sum_{I=0}^{2}\int_{M(t)}
 \left(\left|\nabla\left(\tfrac Nn-1\right)\right|
           +\tfrac{\varphi}{\tau}|\nabla(N\Sigma)|\right)
       |\nabla^I\Sigma|\,|\nabla^{I+1}\Sigma|\,\mu_g .
\end{align*}
The last line follows directly from the connection variation in the proof of Proposition~\ref{wave}.
Our goal is to control the error terms and the spacetime integral terms involving $P^{I}$ and $Q^{I}$. We do so schematically as follows:
\begingroup
\begin{align}
&\sum_{I\leq2}\int_{T_0}^T\int_{M(t)}
\langle\nabla P^I,\nabla^{I+1}\Sigma\rangle\mu_g\,dt\notag\\
={}&\underbrace{-(n-1)\sum_{I\leq2}\int_{T_0}^T\int_{M(t)}
|\nabla^{I+1}\Sigma|^2\mu_g\,dt}_{I\text{-decay term}}\notag\\
&+\sum_{I\leq2}\int_{T_0}^T\int_{M(t)}
\langle\nabla[\partial_t,\nabla^I]\Sigma,\nabla^{I+1}\Sigma\rangle\mu_g\,dt\notag\\
&-\sum_{I\leq2}\int_{T_0}^T\frac{\varphi}{\tau}\int_{M(t)}
\langle\nabla[\nabla^I,N]\mathfrak T,\nabla^{I+1}\Sigma\rangle\mu_g\,dt\notag\\
&+n\sum_{I\leq2}\int_{T_0}^T\frac{\varphi}{\tau}\int_{M(t)}
\left\langle\nabla^{I+1}\nabla^2\left(\frac Nn-1\right),\nabla^{I+1}\Sigma\right\rangle\mu_g\,dt\notag\\
&+2\sum_{I\leq2}\int_{T_0}^T\frac{\varphi}{\tau}\int_{M(t)}
\langle\nabla^{I+1}(N\Sigma_{ik}\Sigma^k{}_j),\nabla^{I+1}\Sigma\rangle\mu_g\,dt\notag\\
&-(n-2)\sum_{I\leq2}\int_{T_0}^T\int_{M(t)}
\left\langle\nabla^{I+1}\left(\left(\frac Nn-1\right)\Sigma\right),\nabla^{I+1}\Sigma\right\rangle\mu_g\,dt.
\end{align}
\endgroup
We estimate each term in this identity and then the terms involving $Q^I$.
All spatial norms in the following estimates are taken on $M(t)$.
By the bootstrap assumption,
\[
 \|\Sigma\|_{H^3}+\|\mathfrak T\|_{H^2}\lesssim\Gamma,
 \qquad \|\Sigma\|_{H^2}\le\mathds Y e^{-t},
 \qquad 0<\frac{\varphi}{\tau}\lesssim e^{-t}.
\]
In particular, $\mathds Y e^{-t}\le1$ for $a$ sufficiently large.
Proposition~\ref{elliptic} gives, at the orders that occur below,
\begin{align*}
 \left\|\tfrac Nn-1\right\|_{L^\infty}
 +\left\|\tfrac Nn-1\right\|_{H^1}
 &\lesssim \mathds Y^2 e^{-2t},\\
 \left\|\tfrac Nn-1\right\|_{H^3}
 +\left\|\tfrac Nn-1\right\|_{H^4}
 &\lesssim (1+\Gamma)\mathds Y^2 e^{-2t},\\
 \left\|\tfrac Nn-1\right\|_{H^5}
 &\lesssim (1+\Gamma)^2\mathds Y^2e^{-2t}
                  +\Gamma\mathds Y e^{-t}.
\end{align*}

\noindent First consider the lapse Hessian in $P^I$. At $I=0,1,2$
it carries, respectively, three, four and five spatial derivatives. In particular,
\begin{align*}
 &\sum_{I=0}^{2}\left|\int_a^T\frac{\varphi}{\tau}\int_{M(t)}
   \left\langle\nabla^{I+1}\nabla^2\!\left(\tfrac Nn-1\right),
                     \nabla^{I+1}\Sigma\right\rangle\mu_g\,dt\right|\\
 &\qquad\lesssim
   \Gamma\int_a^T\left((1+\Gamma)^2\mathds Y^2 e^{-3t}
                               +\Gamma\mathds Y e^{-2t}\right)dt .
\end{align*}
For the curvature coupling, the commutator $[\nabla^I,N]$ vanishes at
$I=0$. At $I=1,2$ each term after the outer derivative has the form
$\nabla^j N*\nabla^{I+1-j}\mathfrak T$, $1\le j\le I+1$.
At the highest order their $L^2$ norms are bounded by the sum
\[
 \|\nabla N\|_{L^\infty}\|\nabla^2\mathfrak T\|_{L^2}
 +\|\nabla^2 N\|_{L^4}\|\nabla\mathfrak T\|_{L^4}
 +\|\nabla^3N\|_{L^2}\|\mathfrak T\|_{L^\infty}.
\]
Proposition~\ref{sobolev} bounds this sum and the lower order terms by
$C\|N/n-1\|_{H^3}\|\mathfrak T\|_{H^2}$. Consequently,
\[
 \sum_{I=0}^{2}\left|\int_a^T\frac{\varphi}{\tau}\int_{M(t)}
  \langle\nabla[\nabla^I,N]\mathfrak T,\nabla^{I+1}\Sigma\rangle
   \mu_g\,dt\right|
 \lesssim \Gamma^2(1+\Gamma)\mathds Y^2\int_a^T e^{-3t}\,dt .
\]

\noindent The next two terms are products. Distributing the $I+1$
derivatives over the factors, using the fixed upper bound for $N$ and the tame estimate
of Section~\ref{sobolev2}, gives
\begin{align*}
 \|N\Sigma_{ik}\Sigma^k{}_j\|_{H^3}
 &\lesssim \|\Sigma\|_{L^\infty}\|\Sigma\|_{H^3}
   +\left\|\tfrac Nn-1\right\|_{H^3}\|\Sigma\|_{L^\infty}^2,\\
 \left\|\left(\tfrac Nn-1\right)\Sigma\right\|_{H^3}
 &\lesssim \left\|\tfrac Nn-1\right\|_{L^\infty}\|\Sigma\|_{H^3}
       +\left\|\tfrac Nn-1\right\|_{H^3}\|\Sigma\|_{L^\infty}.
\end{align*}
In the first line, for a constant $N$ the estimate is precisely the
first product on the right. For variable $N$ apply the product estimate
to $(N-n)\Sigma^2$ and use the uniform bound for $N-n$. Pairing with
$\nabla^{I+1}\Sigma$ and integrating gives, respectively,
\begin{align*}
 &C\Gamma\int_a^T e^{-t}
  \left(\Gamma\mathds Y e^{-t}+(1+\Gamma)\mathds Y^4e^{-4t}\right)dt,\\
 &C\Gamma\int_a^T
  \left(\Gamma\mathds Y^2e^{-2t}+(1+\Gamma)\mathds Y^3e^{-3t}\right)dt .
\end{align*}

\noindent We estimate the time commutator in $P^I$ together with the
outermost connection term of $\mathscr{ER}$. The identity
$[\partial_t,\nabla]\Sigma=\partial_t\Gamma[g]*\Sigma$
and the metric equation give, by Leibniz' rule, the products
\[
 \sum_{j=1}^{I+1}
 \left(\left|\nabla^j\!\left(\tfrac Nn-1\right)\right|
           +\frac{\varphi}{\tau}|\nabla^j(N\Sigma)|\right)
        |\nabla^{I+1-j}\Sigma|,\qquad I=0,1,2.
\]
Indeed $\nabla[\partial_t,\nabla^I]\Sigma+
 [\partial_t,\nabla]\nabla^I\Sigma=[\partial_t,\nabla^{I+1}]\Sigma$.
In the first product the largest order is bounded by
\[
 \|\nabla N\|_{L^\infty}\|\nabla^2\Sigma\|_{L^2}
 +\|\nabla^2N\|_{L^4}\|\nabla\Sigma\|_{L^4}
 +\|\nabla^3N\|_{L^2}\|\Sigma\|_{L^\infty}
 \lesssim (1+\Gamma)\mathds Y^3e^{-3t}.
\]
For the second product the corresponding three bounds are
\[
 \|\nabla(N\Sigma)\|_{L^4}\|\nabla^2\Sigma\|_{L^4},
 \quad \|\nabla^2(N\Sigma)\|_{L^2}\|\nabla\Sigma\|_{L^\infty},
 \quad \|\nabla^3(N\Sigma)\|_{L^2}\|\Sigma\|_{L^\infty}.
\]
The orders below three obey the same bounds with inhomogeneous norms.
Here the algebra estimate at order two and the tame estimate at order three give
\begin{align*}
 \|N\Sigma\|_{H^2}
 &\lesssim \bigl(1+(1+\Gamma)\mathds Y^2e^{-2t}\bigr)\mathds Y e^{-t},\\
 \|N\Sigma\|_{H^3}
 &\lesssim \Gamma+(1+\Gamma)\mathds Y^3e^{-3t}.
\end{align*}
The $L^2$ size of the products containing $N\Sigma$ is therefore at most
\[
 C e^{-t}\bigl(\Gamma\|N\Sigma\|_{H^2}
                     +\mathds Y e^{-t}\|N\Sigma\|_{H^3}\bigr).
\]
Each of these source bounds enters the spacetime integral multiplied
by at most $C\Gamma$. The remaining contributions to $\mathscr{ER}$
come from the differentiated inverse metric and volume form. Their
integral is bounded by
\[
 C\Gamma^2(\mathds Y+\mathds Y^2)\int_a^T e^{-2t}\,dt .
\]

\noindent We next estimate the spacetime integral involving $Q^I$.
Writing out the products from \eqref{eq:wave2} and omitting the
pure traces, its absolute value is bounded by a numerical constant
times the sum over $I=0,1,2$ of the following integrals:
\begin{align*}
 &\int_a^T\int_{M(t)} |\nabla^I\mathfrak T|\,
       \bigg( |[\partial_t,\nabla^I]\mathfrak T|
       +\frac{\varphi}{\tau}\big|[\nabla^I,N]\Delta_g\Sigma\big|
       +\frac{\varphi}{\tau}N\big|[\nabla^I,\Delta_g]\Sigma\big|\bigg)\mu_g\,dt,\\
 &\int_a^T\frac{\varphi}{\tau}\int_{M(t)}
   |\nabla^I\mathfrak T|\!\!
   \sum_{j+k+l=I}|\nabla^j N|\,
      \big(|\nabla^k\operatorname{Riem}|+|\nabla^k\mathfrak T|\big)
      |\nabla^l\Sigma|\,\mu_g\,dt,\\
 &\int_a^T\frac{\varphi}{\tau}\int_{M(t)}
    |\nabla^I\mathfrak T|\!\!
    \sum_{j+k=I}\left(|\nabla^{j+2}N|\,|\nabla^k\Sigma|
                 +|\nabla^{j+1}N|\,|\nabla^{k+1}\Sigma|\right)\mu_g\,dt,\\
 &\int_a^T\int_{M(t)}|\nabla^{I+2}(N/n-1)|\,|\nabla^I\mathfrak T|\,\mu_g\,dt .
\end{align*}
The last integral has no factor $\varphi/\tau$. We keep this distinction
in the following estimates.

\noindent First, at $I=1$ the time commutator contributes a contraction
of $\partial_t\Gamma[g]$ with $\mathfrak T$, and at $I=2$ it also
contains $\nabla\partial_t\Gamma[g]*\mathfrak T$ and
$\partial_t\Gamma[g]*\nabla\mathfrak T$.
For the part of $\partial_t g$ proportional to $(N/n-1)g$, at order one use
$\|\nabla N\|_{L^2}\|\mathfrak T\|_{L^\infty}$, while at order two
\eqref{eq:weighted-lapse-products} gives
\[
 \big\||\nabla N|\,|\nabla\mathfrak T|\big\|_{L^2}
 +\big\||\nabla^2 N|\,|\mathfrak T|\big\|_{L^2}
 \lesssim \mathds Y^2 e^{-2t}\|\mathfrak T\|_{H^2}.
\]
For the part proportional to $(\varphi/\tau)N\Sigma$, the two products
at order two are bounded by
\[
 C e^{-t}\bigl(\|\nabla(N\Sigma)\|_{L^4}\|\nabla\mathfrak T\|_{L^4}
          +\|\nabla^2(N\Sigma)\|_{L^2}\|\mathfrak T\|_{L^\infty}\bigr).
\]
Including the order-one term, the contribution of the time commutators is at most
\[
 C\Gamma^2\int_a^T\bigl(\mathds Y^2 e^{-2t}
                         +e^{-t}\|N\Sigma\|_{H^2}\bigr)\,dt .
\]

\noindent Second, consider the commutator with $N\Delta_g$.
For $[\nabla^I,N]\Delta_g\Sigma$ the differentiated lapse terms are
$\nabla N*\Delta_g\Sigma$ at order one and
$\nabla^2N*\Delta_g\Sigma$, $\nabla N*\nabla\Delta_g\Sigma$
at order two. The latter terms have the bounds
\[
 \|\nabla^2 N\|_{L^4}\|\Delta_g\Sigma\|_{L^4}
 +\|\nabla N\|_{L^\infty}\|\nabla\Delta_g\Sigma\|_{L^2}
 \lesssim (1+\Gamma)\mathds Y^2e^{-2t}\Gamma .
\]
The commutator $[\nabla^I,\Delta_g]\Sigma$ at $I=1$ consists of
$\operatorname{Riem}*\nabla\Sigma$ and
$\nabla\operatorname{Riem}*\Sigma$. At $I=2$ we have
\begin{align*}
 \|[\nabla^2,\Delta_g]\Sigma\|_{L^2}\lesssim{}&
 \|\operatorname{Riem}\|_{L^\infty}\|\nabla^2\Sigma\|_{L^2}
 +\|\nabla\operatorname{Riem}\|_{L^4}\|\nabla\Sigma\|_{L^4}\\
 &+\|\nabla^2\operatorname{Riem}\|_{L^2}\|\Sigma\|_{L^\infty}
 \lesssim (1+\Gamma)\mathds Y e^{-t}.
\end{align*}
We used \eqref{eq:Riem_bound_claim}, which follows from the constraint
in three dimensions. Combining these bounds, the two spatial commutators
contribute at most
\[
 C\Gamma(1+\Gamma)\int_a^T e^{-t}
             \big(\Gamma\mathds Y^2e^{-2t}+\mathds Y e^{-t}\big)\,dt .
\]

\noindent For the curvature products, terms with $j=0$ use the uniform
upper bound on $N$. The $j=1,2$ terms are covered by the $H^2$ algebra
inequality applied to $(N-n)\operatorname{Riem}*\Sigma$ and
$(N-n)\mathfrak T*\Sigma$. Their total, including $j=0$, is bounded by
\[
 C\Gamma(1+\Gamma)\int_a^T e^{-t}
     \big(1+(1+\Gamma)\mathds Y^2e^{-2t}\big)\mathds Y e^{-t}\,dt .
\]
For the terms with two lapse derivatives, distribute $I$ derivatives
between $\nabla^2 N$ and $\Sigma$. For the terms with one lapse
derivative distribute them between $\nabla N$ and $\nabla\Sigma$.
The $H^2$ product inequality, including the zeroth derivative, gives
\begin{align*}
 \|\nabla^2N*\Sigma\|_{H^2}
 &\lesssim \left\|\tfrac Nn-1\right\|_{H^4}\|\Sigma\|_{H^2}
 \lesssim (1+\Gamma)\mathds Y^3e^{-3t},\\
 \|\nabla N*\nabla\Sigma\|_{H^2}
 &\lesssim \left\|\tfrac Nn-1\right\|_{H^3}\|\Sigma\|_{H^3}
 \lesssim (1+\Gamma)\Gamma\mathds Y^2e^{-2t}.
\end{align*}
These two terms enter the integral multiplied by at most $C\Gamma e^{-t}$.
Finally the unweighted lapse Hessian contributes
\[
 (n-2)\sum_{I=0}^{2}\left|\int_a^T\int_{M(t)}
 \big\langle\nabla^I\nabla^2\!\left(\tfrac Nn-1\right),
       \nabla^I\mathfrak T\big\rangle\mu_g\,dt\right|
 \lesssim \Gamma(1+\Gamma)\mathds Y^2\int_a^T e^{-2t}\,dt .
\]

\begingroup\color{blue}
\noindent We now perform the time integrations. By \eqref{eq:bootstrap-size},
$1+\Gamma$ and $\mathds Y$ are bounded by fixed multiples of
$1+\mathcal I^0$. Since $(1+\mathcal I^0)e^{-t}\le e^{-\alpha a}\le1$,
each error integrand above is bounded by
$C_{A,B}(1+\mathcal I^0)^2 \big((1+\mathcal I^0)e^{-t}\big)^2$.
For example, the two fifth-derivative lapse contributions have integrands
bounded by fixed multiples of $(1+\mathcal I^0)^5e^{-3t}$ and
$(1+\mathcal I^0)^3 e^{-2t}$, respectively. The sum of the error integrals is therefore at most
\[
 C_{A,B}(1+\mathcal I^0)^4\int_a^T e^{-2t}\,dt
 =\tfrac12 C_{A,B}(1+\mathcal I^0)^4e^{-2a}
                   \big(1-e^{-2(T-a)}\big).
\]
The shear damping contributes a nonpositive term. The zeroth-order
endpoint in the integrated wave identity is bounded above by
$\tfrac12\|\Sigma(T)\|_{L^2}^2\le C\mathds Y^2e^{-2a}$.
Consequently
\[
 \mathscr E^{top}(T)
 \le \mathscr E^{top}(a)+C_{A,B}(1+\mathcal I^0)^4e^{-2a}
 \le \big(C_0+C_{A,B}e^{-2\alpha a}\big)(1+\mathcal I^0)^2.
\]
Here $C_0$ comes from the initial energy alone. Choose the fixed
bootstrap factor for $\Gamma$ using $C_0$ and the equivalence between
$(\mathscr E^{top})^{1/2}$ and $\mathcal O$. Then choose the factor for
$\mathds Y$ by \eqref{eq:suboptimal_recall}. With both factors fixed, the
last error gives a strict improvement for sufficiently large $a$.
This proves the proposition.
\par\endgroup

\end{proof}
\begin{corollary}
 $\mathcal{O}\lesssim 1+\mathcal{I}^{0}$.
\end{corollary}

\begin{remark}
 The constants are uniform up to the maximal existence time.   
\end{remark}

\section{Proof of the main results}
\label{mainproof}
\subsection{Proof of Theorem \ref{main}}
\label{proof1}
\noindent In this section, we prove the main theorem with the aid of the uniform estimates obtained in Section~\ref{energy}. Note that we have already presented the idea of the proof in Section~\ref{bootstrapargument}. Here we will use a contradiction argument. First, recall the statement of the main theorem:
\begin{theorem-non} [Global existence, $\Lambda>0, \sigma(M)\leq 0$]
Let $M$ be a closed connected oriented three-manifold with $\sigma(M)\leq0$, and fix $\Lambda>0$. Prescribe vacuum CMC data on the expanding branch $\tau_0<-\sqrt{3\Lambda}$, with the normalization of Section~\ref{mainsection}.

\noindent Fix a smooth background metric $\xi_0$ on $M$ and $C>1$.
\textcolor{blue}{Fix $0<\alpha<1$. For all sufficiently large $a$, assume that the initial size $\mathcal I^0>0$ satisfies}
\[
\textcolor{blue}{(1+\mathcal I^0)e^{-a}\le e^{-\alpha a}.}
\]

\noindent Let $(g_0, \Sigma_0)$ be smooth initial data satisfying the Einstein constraint equations at initial CMC time $T_0 = a$, written in CMC-transported spatial coordinates and satisfying:
\begin{align}
\label{eq:metric-restated}
C^{-1} \xi_0 \leq g_0 \leq C \xi_0,
\end{align}
\begin{align}
\label{eq:data-restated}
\sum_{I = 0}^{3} \| \nabla^I \Sigma_0 \|_{L^2(M)} + \sum_{I = 0}^{2} \left( \| \nabla^I \mathfrak{T}[g_0] \|_{L^2(M)} + \| e^a \nabla^I \Sigma_0 \|_{L^2(M)} \right) \leq \mathcal{I}^0,
\end{align}
where $\mathfrak{T}[g_0]$ denotes the renormalized trace-free spatial Ricci curvature tensor of $g_0$.

\noindent Then the Einstein-$\Lambda$ evolution equations admit a unique classical solution
\[
T \mapsto (g(T), \Sigma(T)) \in \mathcal{C}^\infty([T_0, \infty) \times M)
\]
in CMC-transported spatial coordinates, satisfying the constraint equations at each slice $T$ and obeying the following uniform a priori estimates for all $T \in [T_0, \infty)$:
\begin{align}
\sum_{I = 0}^{3} \| \nabla^I \Sigma(T) \|_{L^2(M)} + \sum_{I = 0}^{2} \left( \| \nabla^I \mathfrak{T}[g(T)] \|_{L^2(M)} + \| e^T \nabla^I \Sigma(T) \|_{L^2(M)} \right) \leq C_1(1 + \mathcal{I}^0),
\end{align}
\begin{align}
C_2^{-1} g_0 \leq g(T) \leq C_2 g_0.
\end{align}
The constants $C_1,C_2>0$ depend only on $\xi_0$, $C$, $\Lambda$ and the fixed time normalization, and are independent of $a$, $\mathcal I^0$ and $T$. The developed spacetime is future causally geodesically complete.

\noindent Moreover, the solution $(g(T), \Sigma(T))$ converges in the $C^\infty$ topology, as $T \to \infty$, to $(g_\infty,0)$, where $g_\infty$ is a Riemannian metric on $M$ of scalar curvature $-2/3$. More precisely,
\[
\Sigma(T) \to 0 \quad \text{and} \quad g(T) \to g_\infty, \quad \text{as } T \to \infty,
\]
with convergence holding in all Sobolev norms. In particular, the spacetime $(\widetilde{M}, \widehat{g})$ admits a future-complete CMC foliation asymptotic to a constant negative scalar curvature slice.  
\end{theorem-non}

\begingroup
\noindent Local existence follows from Theorem~\ref{thm:local-continuation}. The estimates of Section~\ref{energy} improve the bootstrap bounds on every interval on which the smooth solution exists.

\noindent Suppose that the maximal renormalized existence time is $T_*<\infty$.
On $[a,T_*)$ the factors $|\varphi|$, $|\tau|$, and their inverses are
bounded. Since $\varphi$ is constant on each slice,
\[
 \widetilde g=\varphi^{-2}g,\qquad
 \widetilde k=\frac{\tau}{3}\widetilde g+
       \varphi^{-1}\Sigma,\qquad
 \widetilde\nabla=\nabla.
\]
The estimates already proved therefore give uniform bounds for
the curvature and second-fundamental-form norms in
Theorem~\ref{thm:local-continuation}. Indeed,
\[
 \operatorname{Ric}(g)=\mathfrak T+
       \frac13\left(|\Sigma|_g^2-\frac23\right)g,
\]
and for $j\geq1$,
$\widetilde\nabla^j\widetilde k=\varphi^{-1}\nabla^j\Sigma$.
Uniform metric comparison gives an upper diameter bound and a positive
volume lower bound for the physical slices on this finite interval.
The three-dimensional curvature decomposition and the uniform Sobolev
inequality give a uniform sectional-curvature bound. The usual
injectivity-radius estimate for closed manifolds with bounded sectional
curvature, bounded diameter and volume bounded below then gives a
positive uniform injectivity radius. Finally,
\[
 \tau(T_*)=-\sqrt{C^2e^{-2T_*}+3\Lambda}< -\sqrt{3\Lambda}.
\]
Theorem~\ref{thm:local-continuation} extends the CMC development beyond
this value of $\tau$, and hence beyond $T_*$. Transporting the initial
spatial coordinates along the extended normal flow preserves the
zero-shift gauge. This contradicts maximality, so $T_*=\infty$.

\endgroup

\begingroup
\begin{lemma}[Higher derivatives and convergence in fixed coordinates]
\label{lem:higher-convergence}
Suppose that the data are smooth and the estimates through order three established in Section~\ref{energy} hold on the future development. The bounds below also hold on each smooth existence interval $[a,T_*)$ with the same constants and with the suprema taken over that interval. Then, for every
integer $m\geq2$, there is a finite constant $C_m$, depending on the initial
norms at the corresponding orders and the initial time $a$, such that
\begin{equation}\label{eq:higher-uniform}
 \sup_{T\geq a}\bigl(\|\Sigma(T)\|_{H^{m+1}(g(T))}
       +\|\mathfrak T(T)\|_{H^m(g(T))}\bigr)\leq C_m,
\end{equation}
\begin{equation}\label{eq:higher-decay}
 \|\Sigma(T)\|_{H^m(g(T))}\leq C_m e^{-T},\qquad
 \left\|\frac N3-1\right\|_{H^{m+2}(g(T))}\leq C_m e^{-2T}.
\end{equation}
Moreover, for each integer $r\geq0$,
\begin{equation}\label{eq:fixed-convergence-rate}
 \|\partial_Tg(T)\|_{H^r(g(a))}\leq C_r e^{-2T}.
\end{equation}
Consequently $g(T)$ converges in every fixed Sobolev norm, and hence smoothly,
to a positive-definite metric $g_\infty$, and $\Sigma(T)$ converges smoothly to
zero. No additional smallness assumption on the higher derivatives is required.
\end{lemma}

\begin{proof}
In this proof put $u=N/3-1$ and $c=\varphi/\tau$. Thus
$0<c\leq C e^{-T}$, $0<N\leq3$, and
\[
 -\Delta_gu+\frac13u=-|\Sigma|_g^2(1+u),\qquad
 \partial_Tg=2cN\Sigma+2ug.
\]
All norms in the first part of the proof are covariant norms for $g(T)$.
The uniform scalar Sobolev inequalities, Kato's inequality and interpolation
provide the tensor product estimates used below. In particular,
\[
 \|AB\|_{H^j}\leq C_j\bigl(
 \|A\|_{L^\infty}\|B\|_{H^j}
 +\|B\|_{L^\infty}\|A\|_{H^j}\bigr).
\]
The constants for these estimates are uniform in $T$. One can obtain the
interpolation inequalities directly by integration by parts with the metric
connection; no fixed-coordinate bounds on higher derivatives of $g$ are
being assumed here.

\noindent Define
\[
 \mathcal E_m=\frac12\sum_{j=0}^{m+1}\|\nabla^j\Sigma\|_{L^2}^2
              +\frac12\sum_{j=0}^{m}\|\nabla^j\mathfrak T\|_{L^2}^2.
\]
The already established estimates give \eqref{eq:higher-uniform} and
\eqref{eq:higher-decay} for $m=2$. Assume they hold with $m-1$ in place of
$m$, where $m\geq3$. In the induction step $C_m$ may depend on the bounds
already established at lower orders, but not on $T$ or on $\mathcal E_m(T)$.
In three dimensions the curvature decomposition and the constraint give
\[
 \|\operatorname{Riem}\|_{H^{m-1}}\leq C_m,
 \qquad
 \|\operatorname{Riem}\|_{H^m}
       \leq C_m(1+\mathcal E_m^{1/2}).
\]

\noindent We first record the lapse estimates needed at this step. Commuting the scalar
lapse equation $k$ times produces, besides derivatives of its right-hand side,
curvature terms of the form
\[
 \nabla^j\operatorname{Riem}*\nabla^{k-j}u,
 \qquad 0\leq j\leq k-1.
\]
Terms which vanish for a scalar are harmless in this list. The estimate
\[
 \|V\|_{H^2}\leq C\bigl(\|\Delta_gV\|_{L^2}
 +(1+\|\operatorname{Riem}\|_{L^\infty})\|V\|_{L^2}\bigr)
\]
applied to $V=\nabla^ku$, together with the lower-order elliptic estimates,
therefore gives
\begin{align}
 \|u\|_{H^{m+2}}&\leq C_m e^{-T},\label{eq:preliminary-lapse-high}\\
 \|u\|_{H^{m+3}}&\leq C_m e^{-T}(1+\mathcal E_m^{1/2}).
 \label{eq:preliminary-lapse-top}
\end{align}
Here are the derivative counts and the time factors in these estimates.
For $k=m$, the source satisfies
\[
 \||\Sigma|^2(1+u)\|_{H^m}
 \leq C_m\|\Sigma\|_{L^\infty}\|\Sigma\|_{H^m}
       +C_m\|\Sigma\|_{L^\infty}^2\|u\|_{H^m}
 \leq C_m e^{-T},
\]
where $\|\Sigma\|_{H^m}$ is controlled by $\mathcal E_{m-1}$.
The curvature commutators use only $\operatorname{Riem}$ through order
$m-1$ and $u$ through order $m$, and are $O(e^{-2T})$ by the induction
hypothesis. For $k=m+1$, the highest terms are
\[
 \Sigma*\nabla^{m+1}\Sigma,
 \qquad \nabla^m\operatorname{Riem}*\nabla u.
\]
Their $L^2$ norms are bounded respectively by
$C_m e^{-T}\mathcal E_m^{1/2}$ and
$C_m e^{-2T}(1+\mathcal E_m^{1/2})$.
All the remaining curvature terms contain at most $m-1$ derivatives of
curvature and at most $m+1$ derivatives of $u$. They are bounded by
$C_m e^{-2T}$. These observations prove
\eqref{eq:preliminary-lapse-high}--\eqref{eq:preliminary-lapse-top}.

\noindent Commute the evolution equations through order $m$, and pair them with
$\nabla^{j+1}\Sigma$ and $\nabla^j\mathfrak T$, for $0\leq j\leq m$,
as in the proof of the top-order estimate. The principal terms cancel by
integration by parts:
\[
 \int_M\langle\nabla(-cN\nabla^j\mathfrak T),
                       \nabla^{j+1}\Sigma\rangle\,\mu_g
 +\int_M\langle-cN\Delta_g\nabla^j\Sigma,
                       \nabla^j\mathfrak T\rangle\,\mu_g=0.
\]
The damping terms are nonpositive, and the pure-trace terms have zero pairing.
For completeness, the highest derivatives remaining after this cancellation
are accounted for as follows. The lapse term in the $\Sigma$ equation is
bounded by
\[
 c\|u\|_{H^{m+3}}\|\Sigma\|_{H^{m+1}}
       \leq C_m e^{-2T}(1+\mathcal E_m).
\]
The unweighted lapse Hessian in the $\mathfrak T$ equation is bounded by
\[
 \|u\|_{H^{m+2}}\|\mathfrak T\|_{H^m}
       \leq C_m e^{-T}(1+\mathcal E_m).
\]
The spatial curvature commutators contain at most
$\nabla^m\operatorname{Riem}*\Sigma$ and
$\nabla^{m-1}\operatorname{Riem}*\nabla\Sigma$, and carry the factor $c$.
The time commutators follow from
\[
 \partial_T\Gamma[g]=c\,g^{-1}*\nabla(N\Sigma)+g^{-1}*\nabla(ug).
\]
Thus their highest shear derivative is $\nabla^{m+1}\Sigma$, multiplied by
$c\Sigma$, and their highest lapse derivative is $\nabla^{m+1}u$.
Leibniz' rule and the preceding product estimates bound these terms,
the commutators with $N$, and the differentiated quadratic terms by
$C_m e^{-T}(1+\mathcal E_m)$. Finally, the derivatives of the contractions
and volume form have coefficients bounded by
$C(c\|\Sigma\|_{L^\infty}+\|u\|_{L^\infty})\leq C e^{-2T}$.
Including the zeroth-order shear equation, we obtain
\begin{equation}\label{eq:higher-energy-gronwall}
 \frac{d}{dT}\mathcal E_m\leq C_m e^{-T}(1+\mathcal E_m).
\end{equation}
Since $e^{-T}$ is integrable, Gronwall's inequality gives
\[
 1+\mathcal E_m(T)
 \leq (1+\mathcal E_m(a))\exp(C_m e^{-a}).
\]
This proves \eqref{eq:higher-uniform} at order $m$ without making the
higher initial norms small.

\noindent Now commute the damped shear equation only $m$ times and use the uniform
bound just obtained. The same product and time-commutator estimates give,
in the upper-derivative sense when the norm vanishes,
\[
 \frac{d}{dT}\|\Sigma\|_{H^m}
 +(2-C_m e^{-T})\|\Sigma\|_{H^m}\leq C_m e^{-T}.
\]
The integrating factor gives $\|\Sigma(T)\|_{H^m}\leq C_m e^{-T}$.
In the $k=m$ lapse estimate the source is therefore now $O(e^{-2T})$;
the curvature commutators were already of that order. Consequently
$\|u\|_{H^{m+2}}\leq C_m e^{-2T}$. This completes the induction.
The metric equation consequently yields, for every $r\geq0$,
\begin{equation}\label{eq:covariant-metric-velocity}
 \|\partial_Tg\|_{H^r(g(T))}\leq C_r e^{-2T}.
\end{equation}

\noindent We next justify passage to a fixed Sobolev topology. Let
$A=\nabla[g(T)]-\nabla[g(a)]$, regarded as a $(1,2)$-tensor. If
$h=\partial_Tg$, then
\[
 \partial_T A^k{}_{ij}
 =\frac12g^{k\ell}
   (\nabla_i h_{j\ell}+\nabla_j h_{i\ell}-\nabla_\ell h_{ij}).
\]
Equation \eqref{eq:covariant-metric-velocity} and the uniform Sobolev
inequalities give
$\|\nabla^j\partial_TA\|_{L^\infty(g(T))}\leq C_j e^{-2T}$ for every $j$.
Commuting $\partial_T$ with $\nabla^jA$ produces only contractions of
$\nabla^\ell\partial_TA$ with lower derivatives of $A$.
The metric variation has the same integrable bound. Since $A(a)=0$,
Gronwall's inequality, successively in $j$, gives
\[
 \sup_{T\geq a}\sum_{j=0}^r
       \|\nabla^jA(T)\|_{L^\infty(g(T))}<\infty
       \qquad\text{for every }r.
\]
Expanding $\nabla[g(a)]=\nabla[g(T)]-A$ and using the uniform comparison
of $g(T)$ and $g(a)$ now shows that the tensor Sobolev norms for these
two metrics are uniformly equivalent at each fixed order. Applying this to $h$ proves
\eqref{eq:fixed-convergence-rate}. Hence
\[
 \|g(T_2)-g(T_1)\|_{H^r(g(a))}
 \leq C_r\int_{T_1}^{T_2}e^{-2s}\,ds
 \leq C_r e^{-2T_1},\qquad T_2\geq T_1\geq a.
\]
The limits in the different Sobolev spaces agree. Sobolev embedding at
arbitrarily high order gives a smooth limit $g_\infty$, positive definite
by the uniform metric comparison, and smooth convergence of $\Sigma$ to
zero. Passing to the limit in
$R(g)+2/3=|\Sigma|_g^2$ gives $R(g_\infty)=-2/3$.
\end{proof}

\noindent This proves the asserted smooth convergence of the normalized metric to $g_\infty$.
\endgroup
\begingroup\color{blue}
It remains to prove future causal geodesic completeness. By
\eqref{eq:scaling} and the change of time variable, the physical
spacetime metric is
\[
\widehat g
=
-\frac{N^2}{\tau^2}\,dT^2
+\widetilde g_{ij}(T)\,dx^i dx^j,
\qquad
\widetilde g(T)=\varphi^{-2}g(T).
\]
Here $\widetilde g(T)$ denotes the physical spatial metric.
Set
\[
L:=\frac{N}{|\tau|},
\qquad
\mathbf n:=L^{-1}\partial_T.
\]
Thus $\mathbf n$ is the future-directed unit normal to the
constant-$T$ hypersurfaces.

\noindent Let $\mathcal C(\lambda)=(T(\lambda),x(\lambda))$ be a
future-directed causal geodesic, affinely parametrized, and
write $\alpha=d\mathcal C/d\lambda$. In the timelike case,
normalize the affine parameter so that
$\widehat g(\alpha,\alpha)=-1$; in the null case,
$\widehat g(\alpha,\alpha)=0$. Define
\[
E:=-\widehat g(\alpha,\mathbf n)
=L\frac{dT}{d\lambda}>0,
\qquad
\alpha=E\mathbf n+W,
\qquad
W\in TM_T.
\]
Then
\[
|W|_{\widetilde g}^2
=
E^2+\widehat g(\alpha,\alpha)
\le E^2.
\]

\noindent Let $\widetilde k$ denote the physical second
fundamental form, with the convention
\[
\widetilde k(X,Y)
=
-\widehat g(\widehat\nabla_X\mathbf n,Y)
\]
for vectors tangent to $M_T$. The CMC decomposition and the
rescaling give
\[
\widetilde k
=
\frac{\tau}{3}\widetilde g+\widetilde k^{\mathrm{TT}},
\qquad
|\widetilde k^{\mathrm{TT}}|_{\widetilde g}
=
|\varphi|\,|\Sigma|_g.
\]
Moreover,
\[
\widehat\nabla_{\mathbf n}\mathbf n
=
\nabla^{\widetilde g}\log L.
\]
Using $\widehat\nabla_\alpha\alpha=0$, we obtain
\[
\frac{dE}{d\lambda}
=
-E\,W(\log L)+\widetilde k(W,W).
\]
Since $dT/d\lambda=E/L$, it follows that
\[
\frac{dE}{dT}
=
-W(L)+\frac{L}{E}\widetilde k(W,W).
\]
The contribution of $(\tau/3)\widetilde g$ is nonpositive
because $\tau<0$. Therefore,
\[
\frac{d}{dT}\log E
\le
|\nabla^{\widetilde g}L|_{\widetilde g}
+
L|\widetilde k^{\mathrm{TT}}|_{\widetilde g}.
\]
Since $\tau$ and $\varphi$ are constant on each spatial slice,
\[
|\nabla^{\widetilde g}L|_{\widetilde g}
=
\frac{|\varphi|}{|\tau|}\,|\nabla N|_g.
\]
Consequently,
\[
\frac{d}{dT}\log E
\le
\frac{|\varphi|}{|\tau|}
\left(
\|\nabla N\|_{L^\infty(M,g)}
+
\|N\|_{L^\infty(M)}
\|\Sigma\|_{L^\infty(M,g)}
\right).
\]

\noindent The decay estimates already proved, together with
the uniform Sobolev inequalities and the elliptic estimate
for the lapse, imply
\[
\|\Sigma(T)\|_{L^\infty(M,g)}\le C e^{-T},
\qquad
\left\|\frac{N(T)}{3}-1\right\|_{L^\infty(M)}
+
\|\nabla N(T)\|_{L^\infty(M,g)}
\le C e^{-2T}.
\]
Here the constants may depend on the fixed initial data but
are independent of $T$. Since
$|\varphi|/|\tau|\le C e^{-T}$, we conclude that
\[
\frac{d}{dT}\log E\le C e^{-2T}.
\]
Integration gives a uniform bound
\[
E(T)\le C_E<\infty.
\]
Furthermore,
\[
L(T,x)
=
\frac{N(T,x)}{|\tau(T)|}
\longrightarrow
\frac{3}{\sqrt{3\Lambda}}>0
\]
uniformly on $M$ as $T\to\infty$.

\noindent Suppose that $\mathcal C$ had a finite maximal
future affine parameter $\lambda_*$. If $T$ remained bounded
as $\lambda\uparrow\lambda_*$, then the smoothness of the
physical metric on the corresponding compact time slab,
the positivity of $L$, and the bound $E\le C_E$ would keep
$(\mathcal C,\alpha)$ in a compact subset of the tangent
bundle. The geodesic equation would then extend $\mathcal C$
beyond $\lambda_*$, a contradiction. Hence $T\to\infty$ as
$\lambda\uparrow\lambda_*$. For all sufficiently large $T$,
however,
\[
\frac{d\lambda}{dT}
=
\frac{L}{E}
\ge\frac{c}{C_E}>0.
\]
Thus, writing $x(s)$ for the spatial component along the
geodesic at time $T=s$,
\[
\lambda(T)-\lambda(T_1)
=
\int_{T_1}^{T}\frac{L(s,x(s))}{E(s)}\,ds
\ge
\frac{c}{C_E}(T-T_1)
\longrightarrow\infty,
\]
again a contradiction. This proves future completeness for
every causal geodesic.
\par\endgroup

\hfill$\square$
\begingroup
\begin{remark}[Non-Einstein limits]
\label{noconvergence}
On a closed three-manifold admitting no hyperbolic metric, the smooth
limit of Theorem~\ref{main} necessarily has $\mathfrak T[g_\infty]\ne0$.
Indeed, a three-dimensional Einstein metric with $R=-2/3$ has constant
negative sectional curvature.

\noindent There are also examples on closed hyperbolizable manifolds. To construct
one within Section~4, start from a hyperbolic metric $\gamma$ with this
normalization. For a sufficiently large value of the construction
parameter $a$, take $g_{1,a}$ from
Corollary~\ref{cor:separation-scales} with $\alpha=1/2$ and set the shear to zero. The scalar
defect in \eqref{eq:Rsmallcor} is at most $Ce^{-a}$ in $H^2(\gamma)$.
The fixed-background inverse and contraction in the proof of
Corollary~\ref{cor:final-CMC}, with $\widehat\Sigma_a=0$, give
$g_*=\varphi_a^4g_{1,a}$ with $R(g_*)=-2/3$ and
$\|\varphi_a-1\|_{H^4(\gamma)}\le Ce^{-a}$.
By \eqref{eq:Tlargecor} and \eqref{eq:T-difference-cor},
$\|\mathfrak T[g_*]\|_{H^2(\gamma)}\ge ce^{a/2}-Ce^{-a}>0$
for all sufficiently large $a$. Choose one such metric and fix it as
$g_*$.

\noindent The construction parameter is now fixed. In the rest of this argument,
$a$ denotes the initial CMC time at which we prescribe the same
normalized data $(g_*,0)$. Their initial bound $\mathcal I^0$ is
independent of this initial time, so they satisfy Theorem~\ref{main}
for all sufficiently large $a$. Theorem~\ref{main} and
Proposition~\ref{elliptic} give
$\|\Sigma(T)\|_{H^2}\le C_*e^{-T}$,
$\|\mathfrak T(T)\|_{H^2}\le C_*$ and
$\|N/3-1\|_{H^4}\le C_*e^{-2T}$ with $C_*$ independent of $a$.
In \eqref{eq:T_evol} with $X=0$, the term containing $\Delta_g\Sigma$
has coefficient $O(e^{-T})$, and the unweighted lapse Hessian is
$O(e^{-2T})$ in $L^2$. The other products have the same or better
decay by the preceding bounds, Sobolev embedding and the
three-dimensional curvature identity. The uniform metric comparison
therefore gives
\[
 \|\partial_T\mathfrak T(T)\|_{L^2(g_*)}\le C_*e^{-2T},
 \qquad
 \|\mathfrak T[g_\infty]-\mathfrak T[g_*]\|_{L^2(g_*)}
 \le \tfrac12 C_*e^{-2a}.
\]
Taking $a$ so large that the last bound is less than
$\|\mathfrak T[g_*]\|_{L^2(g_*)}$ gives a non-Einstein limit.
\end{remark}
\endgroup

\subsection{Proof of Corollary \ref{nocollapse}}
\label{noncollapse}
\noindent \begin{proof}
\begingroup
Set $g_0=g(T_0)$. By Theorem~\ref{main} there is $A\ge1$ such that
$A^{-1}g_0\le g(T)\le A g_0$ for all $T\ge T_0$. Length comparison gives
\[
 B_{g_0}(x,r/\sqrt A)\subset B_{g(T)}(x,r)
             \subset B_{g_0}(x,\sqrt A\,r).
\]
In dimension three the volume densities satisfy
$A^{-3/2}\,d\mu_{g_0}\le d\mu_{g(T)}\le A^{3/2}\,d\mu_{g_0}$. Hence
\[
 A^{-3/2}\operatorname{Vol}_{g_0}\bigl(B_{g_0}(x,r/\sqrt A)\bigr)
 \le \operatorname{Vol}_{g(T)}\bigl(B_{g(T)}(x,r)\bigr)
 \le A^{3/2}\operatorname{Vol}_{g_0}(M).
\]
For each fixed $r>0$ the lower bound has a strictly positive infimum
over $x$ by compactness of $(M,g_0)$. This proves the assertion
uniformly in both the center and the time.
\par\endgroup
\end{proof}

\begin{theorem-non}[{\color{blue}Ringstr\"om's conclusion for the solutions of Theorem~\ref{main}}]
Let $(\widehat{M}^{3+1}, \widehat{g})$ be a globally hyperbolic Einstein vacuum spacetime with positive cosmological constant $\Lambda > 0$, admitting a global CMC foliation by compact spacelike hypersurfaces diffeomorphic to a closed $3$-manifold $M$ with $\sigma(M)\leq0$, and initial data satisfying Theorem~\ref{main}. Then there exists a CMC slice $M_{T}$ such that for every future-directed inextendible causal curve $\gamma$,
\[
M_{T} \not\subset J^{-}(\gamma).
\]
\end{theorem-non}

\begin{proof}
\begingroup\color{blue}
Let $\xi_0$ be the fixed smooth background metric appearing in
Theorem~\ref{main}. The uniform metric comparison gives
\[
g(T)\ge c\,\xi_0
\qquad\text{for all }T\ge T_0.
\]
In transported spatial coordinates, the physical spacetime
metric is
\[
\widehat g
=
-\frac{N^2}{\tau^2}\,dT^2
+\varphi^{-2}g_{ij}(T)\,dx^i dx^j.
\]
Therefore every future-directed causal curve, parametrized by
$T$, satisfies
\[
\left|\frac{dx}{dT}\right|_{g(T)}
\le\frac{N|\varphi|}{|\tau|}.
\]
The uniform bound for $N$, the lower bound
$|\tau|\ge\sqrt{3\Lambda}$, and
$|\varphi(T)|=C e^{-T}$ imply
\[
\left|\frac{dx}{dT}\right|_{\xi_0}\le C e^{-T},
\]
with a constant independent of the causal curve.

\noindent Let $\gamma$ be future-directed and inextendible.
Its $T$-coordinate tends to infinity: otherwise the preceding
estimate would give a spatial endpoint at a finite time, and
$\gamma$ could be continued causally from that endpoint.
Since
\[
\int_{T_0}^{\infty}C e^{-s}\,ds<\infty,
\]
the spatial projection of $\gamma$ converges to a point
$x_\infty(\gamma)\in M$.

\noindent Fix a sufficiently late slice $M_T$ and a point
$(T,x)\in J^-(\gamma)\cap M_T$. There is a causal curve from
$(T,x)$ to a point $\gamma(S)$ with $S\ge T$, where $S$
denotes the value of the time coordinate along $\gamma$.
Applying the spatial speed bound to this curve and then to
the future portion of $\gamma$, we obtain
\[
\begin{aligned}
d_{\xi_0}\bigl(x,x_\infty(\gamma)\bigr)
&\le
C\int_T^S e^{-s}\,ds
+
C\int_S^\infty e^{-s}\,ds\\
&=C e^{-T}.
\end{aligned}
\]
Consequently,
\[
J^-(\gamma)\cap M_T
\subseteq
\overline{B}_{\xi_0}\bigl(x_\infty(\gamma),C e^{-T}\bigr).
\]
Choose $T$ sufficiently large that
\[
2C e^{-T}<\operatorname{diam}_{\xi_0}(M).
\]
No closed ball of radius $C e^{-T}$ can then contain all of $M$.
Since this choice of $T$ is independent of $\gamma$, we have
\[
M_T\not\subset J^-(\gamma)
\]
for every future-directed inextendible causal curve $\gamma$.
\par\endgroup
\end{proof}

\section{Statements and Declarations}
\noindent\textit{Funding.} This work was supported by the Center of Mathematical Sciences and Applications, Department of Mathematics at Harvard University, and by the Beijing Institute of Mathematical Sciences and Applications of the Yau Mathematical Sciences Center at Tsinghua University.

\medskip
\noindent\textit{Competing interests.} The author has no conflicts to disclose.

\medskip
\noindent\textit{Data availability.} Data sharing is not applicable to this article as no new data were created or analyzed in this study.

\medskip
\noindent\textit{Use of artificial intelligence.} I used ChatGPT to help check calculations and revise parts of the proofs and exposition.


\begin{thebibliography}{99}

\bibitem{marsden} R. Abraham, J. E. Marsden, T. Ratiu, \textit{Manifolds, Tensor Analysis, and Applications}, second edition, Applied Mathematical Sciences, vol. 75, Springer, 1988.
\url{https://doi.org/10.1007/978-1-4612-1029-0}.

\bibitem{YCBY} A. Abrahams, A. Anderson, Y. Choquet-Bruhat, J. W. York, Jr., Geometrical hyperbolic systems for general relativity and gauge theories, \textit{Classical and Quantum Gravity}, vol. 14, A9--A22, 1997.
\url{https://doi.org/10.1088/0264-9381/14/1A/002}.

\bibitem{anderson2001long} M. T. Anderson, On Long-Time Evolution in General Relativity and Geometrization of 3-Manifolds, Communications in Mathematical Physics, vol. 222, 533--567, 2001.
\url{https://doi.org/10.1007/s002200100527}.

\bibitem{andersson2020nonlinear} L. Andersson, D. Fajman, Nonlinear stability of the Milne model with matter, \textit{Communications in Mathematical Physics}, vol. 378, 261--298, 2020.
\url{https://doi.org/10.1007/s00220-020-03745-w}.

\bibitem{elliptichyperbolic} L. Andersson, V. Moncrief, Elliptic-hyperbolic systems and the Einstein equations, \textit{Annales Henri Poincar\'e}, vol. 4, pages 1--34, 2003.
\url{https://doi.org/10.1007/s00023-003-0120-1}.

\bibitem{andersson2004future} L. Andersson, V. Moncrief, Future complete vacuum spacetimes, in P. T. Chru{\'s}ciel and H. Friedrich (eds.), \textit{The Einstein Equations and the Large Scale Behavior of Gravitational Fields}, Birkh{\"a}user, Basel, 299--330, 2004.
\url{https://doi.org/10.1007/978-3-0348-7953-8_8}.

\bibitem{andersson2011einstein} L. Andersson, V. Moncrief, Einstein spaces as attractors for the Einstein flow, \textit{Journal of Differential Geometry}, vol. 89, 1--47, 2011.
\url{https://doi.org/10.4310/jdg/1324476750}.

\bibitem{KK} V. Branding, D. Fajman, K. Kr{\"o}ncke, Stable cosmological Kaluza--Klein spacetimes, \textit{Communications in Mathematical Physics}, vol. 368, 1087--1120, 2019.
\url{https://doi.org/10.1007/s00220-019-03319-5}.

\bibitem{cheeger1} J. Cheeger, M. Gromov, Collapsing Riemannian manifolds while keeping their curvature bounded: I, Journal of Differential Geometry, vol. 23, 309--346, 1986.
\url{https://doi.org/10.4310/jdg/1214440117}.

\bibitem{cheeger2} J. Cheeger, M. Gromov, Collapsing Riemannian manifolds while keeping their curvature bounded: II, Journal of Differential Geometry, vol. 32, 269--298, 1990.
\url{https://doi.org/10.4310/jdg/1214445047}.

\bibitem{CBY1996Mixed} Y. Choquet-Bruhat and J. W. York, Jr., Mixed elliptic and hyperbolic systems for the Einstein equations, in G. Ferrarese (ed.), \textit{Gravitation, Electromagnetism and Geometrical Structures}, Pitagora Editrice, Bologna, 1996. Preprint:
\url{https://doi.org/10.48550/arXiv.gr-qc/9601030}.

\bibitem{C09} D. Christodoulou, \textit{The formation of black holes in general relativity}. EMS Monographs in Mathematics, European Mathematical Society (2009).
\url{https://doi.org/10.4171/068}.

\bibitem{christodoulou1993global} D. Christodoulou, S. Klainerman, \textit{The Global Nonlinear Stability of the Minkowski Space}, Princeton Mathematical Series, vol. 41, Princeton University Press, 1993.
\url{https://doi.org/10.1515/9781400863174}.

\bibitem{fajman2021stability} D. Fajman, J. Joudioux, J. Smulevici, The stability of the Minkowski space for the Einstein--Vlasov system, \textit{Analysis \& PDE}, vol. 14, 425--531, 2021.
\url{https://doi.org/10.2140/apde.2021.14.425}.

\bibitem{fajman2018cmc} D. Fajman, K. Kr{\"o}ncke, On the CMC-Einstein-$\Lambda$ flow, \textit{Classical and Quantum Gravity}, vol. 35, 195005, 2018.
\url{https://doi.org/10.1088/1361-6382/aad843}.

\bibitem{fajman2020stable} D. Fajman, K. Kr{\"o}ncke, Stable fixed points of the Einstein flow with positive cosmological constant, \textit{Communications in Analysis and Geometry}, vol. 28, 1533--1576, 2020.
\url{https://doi.org/10.4310/CAG.2020.v28.n7.a2}.

\bibitem{fajman2022cosmic} D. Fajman, L. Urban, Cosmic censorship near FLRW spacetimes with negative spatial curvature, \textit{Analysis \& PDE}, vol. 18, 1615--1713, 2025.
\url{https://doi.org/10.2140/apde.2025.18.1615}.

\bibitem{fajman2021attractors} D. Fajman, Z. Wyatt, Attractors of the Einstein-Klein-Gordon system, \textit{Communications in Partial Differential Equations}, vol. 46, pages 1--30, 2021.
\url{https://doi.org/10.1080/03605302.2020.1817072}.

\bibitem{fischer2000reduced} A. E. Fischer, V. Moncrief, The reduced Hamiltonian of general relativity and the $\sigma$-constant of conformal geometry, Mathematical and Quantum Aspects of Relativity and Cosmology: Proceedings of the Second Samos Meeting on Cosmology, Geometry and Relativity Held at Pythagoreon, Samos, Greece, 31 August--4 September 1998, 70--101, 2000.
\url{https://doi.org/10.1007/3-540-46671-1_4}.

\bibitem{fischer2001reduced} A. E. Fischer, V. Moncrief, The reduced Einstein equations and the conformal volume collapse of 3-manifolds, Classical and Quantum Gravity, vol. 18, 4493--4515, 2001.
\url{https://doi.org/10.1088/0264-9381/18/21/308}.

\bibitem{friedrich1986existence1} H. Friedrich, Existence and structure of past asymptotically simple solutions of Einstein's field equations with positive cosmological constant, \textit{Journal of Geometry and Physics}, vol. 3, 101--117, 1986.
\url{https://doi.org/10.1016/0393-0440(86)90004-5}.

\bibitem{friedrich1986existence} H. Friedrich, On the existence of $n$-geodesically complete or future complete solutions of Einstein's field equations with smooth asymptotic structure, \textit{Communications in Mathematical Physics}, vol. 107, 587--609, 1986.
\url{https://doi.org/10.1007/BF01205488}.

\bibitem{hadvzic2015global} M. Had{\v{z}}i{\'c}, J. Speck, The global future stability of the FLRW solutions to the Dust-Einstein system with a positive cosmological constant, \textit{Journal of Hyperbolic Differential Equations}, vol. 12, 87--188, 2015.
\url{https://doi.org/10.1142/S0219891615500046}.

\bibitem{KlainermanRodnianski2010} {S.~Klainerman and I.~Rodnianski,
On the breakdown criterion in general relativity,
\textit{Journal of the American Mathematical Society}, vol.~23,
345--382, 2010.}
\url{https://doi.org/10.1090/S0894-0347-09-00655-9}.

\bibitem{lefloch2016global} P. G. LeFloch, Y. Ma, \textit{The Global Nonlinear Stability of Minkowski Space for Self-Gravitating Massive Fields}, World Scientific, 2017.
\url{https://doi.org/10.1142/10730}.

\bibitem{lefloch2021nonlinear} P. G. LeFloch, C. Wei, Nonlinear stability of self-gravitating irrotational Chaplygin fluids in a FLRW geometry, \textit{Annales de l'Institut Henri Poincar\'e C, Analyse non lin\'eaire}, vol. 38, 787--814, 2021.
\url{https://doi.org/10.1016/j.anihpc.2020.09.005}.

\bibitem{LiimatainenSalo} T. Liimatainen and M. Salo, Nowhere conformally homogeneous manifolds and limiting Carleman weights, \textit{Inverse Problems and Imaging}, vol.~6, 523--530, 2012.
\url{https://doi.org/10.3934/ipi.2012.6.523}.

\bibitem{lindblad2010global} H. Lindblad, I. Rodnianski, The global stability of Minkowski space-time in harmonic gauge, \textit{Annals of Mathematics}, vol. 171, 1401--1477, 2010.
\url{https://doi.org/10.4007/annals.2010.171.1401}.

\bibitem{marsden1980maximal} J. E. Marsden, F. J. Tipler, Maximal hypersurfaces and foliations of constant mean curvature in general relativity, \textit{Physics Reports}, vol. 66, 109--139, 1980.
\url{https://doi.org/10.1016/0370-1573(80)90154-4}.

\bibitem{MM} V. Moncrief, P. Mondal, Could the universe have an exotic topology? \textit{Pure and Applied Mathematics Quarterly}, vol. 15, 921--966, 2019.
\url{https://doi.org/10.4310/PAMQ.2019.v15.n3.a7}.

\bibitem{mondal2019attractors} P. Mondal, Attractors of the `$n+1$' dimensional Einstein-$\Lambda$ flow, \textit{Classical and Quantum Gravity}, vol. 37, article 235002, 2020.
\url{https://doi.org/10.1088/1361-6382/abb9eb}.

\bibitem{mondal2024nonlinear} P. Mondal, The nonlinear stability of (n+1)-dimensional FLRW spacetimes, \textit{Journal of Hyperbolic Differential Equations}, vol. 21, 329--422, 2024.
\url{https://doi.org/10.1142/S0219891624500115}.

\bibitem{mondal2024global} P. Mondal, S.-T. Yau, Global exterior stability of the Minkowski space: Coupled Einstein--Yang--Mills perturbations, \textit{Journal of Mathematical Physics}, vol. 65, 042502, 2024.
\url{https://doi.org/10.1063/5.0185598}.

\bibitem{oliynyk2016future} T. A. Oliynyk, Future stability of the FLRW fluid solutions in the presence of a positive cosmological constant, \textit{Communications in Mathematical Physics}, vol. 346, 293--312, 2016.
\url{https://doi.org/10.1007/s00220-015-2551-1}.

\bibitem{oliynyk2021future} T. A. Oliynyk, Future Global Stability for Relativistic Perfect Fluids with Linear Equations of State $p=K\rho$ where $1/3<K<1/2$, \textit{SIAM Journal on Mathematical Analysis}, vol. 53, 4118--4141, 2021.
\url{https://doi.org/10.1137/20M1361195}.

\bibitem{penrose1999question} R. Penrose, The question of cosmic censorship, \textit{Journal of Astrophysics and Astronomy}, vol. 20, 233--248, 1999.
\url{https://doi.org/10.1007/BF02702355}.

\bibitem{porti2008geometrization} J. Porti, Geometrization of three manifolds and Perelman's proof, \textit{Rev. R. Acad. Cienc. Exactas F{\'\i}s. Nat. Ser. A Math. RACSAM}, vol. 102, 101--125, 2008.
\url{https://doi.org/10.1007/BF03191814}.

\bibitem{ringstrom2008future} H. Ringstr{\"o}m, Future stability of the Einstein-non-linear scalar field system, Inventiones mathematicae, vol. 173, 123--208, 2008.
\url{https://doi.org/10.1007/s00222-008-0117-y}.

\bibitem{rodnianski2018regime} I. Rodnianski, J. Speck, A regime of linear stability for the Einstein-scalar field system with applications to nonlinear Big Bang formation, \textit{Annals of Mathematics}, vol. 187, 65--156, 2018.
\url{https://doi.org/10.4007/annals.2018.187.1.2}.

\bibitem{schoenyau} R. Schoen, S.-T. Yau, On the Structure of Manifolds with Positive Scalar Curvature, \textit{Manuscripta Mathematica}, vol. 28, 159--183, 1979.
\url{https://doi.org/10.1007/BF01647970}.

\bibitem{taylor} M. Taylor, The global nonlinear stability of Minkowski space for the massless Einstein-Vlasov system, \textit{Annals of PDE}, vol. 3, article 9, 2017.
\url{https://doi.org/10.1007/s40818-017-0026-8}.

\bibitem{touati2023constraints} {\color{blue}A. Touati, High-Frequency Solutions to the Constraint Equations, \textit{Communications in Mathematical Physics}, vol. 402, 97--140, 2023.}
\url{https://doi.org/10.1007/s00220-023-04715-8}.

\bibitem{wald} R. Wald, Asymptotic behavior of homogeneous cosmological models in the presence of a positive cosmological constant, Physical Review D, vol. 28, 2118, 1983.
\url{https://doi.org/10.1103/PhysRevD.28.2118}.

\bibitem{wang} J. Wang, Future stability of the $1+3$ Milne model for the Einstein-Klein-Gordon system, \textit{Classical and Quantum Gravity}, vol. 36(22):225010, 2019.
\url{https://doi.org/10.1088/1361-6382/ab4dd3}.

\bibitem{york} J. W. York, Jr., Conformally invariant orthogonal decomposition of symmetric tensors on Riemannian manifolds and the initial-value problem of general relativity, Journal of Mathematical Physics, vol. 14, 456--464, 1973.
\url{https://doi.org/10.1063/1.1666338}.

\end{thebibliography}
\end{document}